\documentclass[12pt,a4paper,oneside]{report}
\usepackage[margin=3cm]{geometry}

\usepackage{amsmath}    			
\usepackage{amsfonts}
\usepackage{amssymb}
\usepackage{amsthm}					
\usepackage{mathrsfs}  				
\usepackage{enumitem}				
\usepackage{graphicx}
\usepackage{color}      			
\usepackage{subfigure}  			
\usepackage[font=small]{caption}	
\usepackage[compress]{cite}         
\usepackage[nodayofweek]{datetime}  
\usepackage{float}					
\graphicspath{ {figs/} }
\usepackage{tikz}					
	\usetikzlibrary{decorations.pathmorphing}
\usepackage{pgfplots}   			
	\pgfplotsset{width=10cm}

\usepackage{hyperref}
\hypersetup{
    bookmarks=true,         					
    pdftitle={Quantum Field Theory on 
              Rotating Black Hole Spacetimes},  
    pdfauthor={Hugo Ferreira},     				
    pdfsubject={Mathematical Physics},   		
    pdfcreator={Hugo Ferreira},   				
}

\definecolor{blue1}{RGB}{17,66,115}
\definecolor{blue2}{RGB}{21,83,145}
\definecolor{blue3}{RGB}{27,102,177}
\definecolor{blue4}{RGB}{30,115,200}
\definecolor{blue5}{RGB}{40,132,223}
\definecolor{blue6}{RGB}{69,148,226}

\newcommand{\dd}{\mathrm{d}}
\newcommand{\supp}{\operatorname{supp}}
 
\newcommand{\End}{\operatorname{End}} 
\newcommand{\Ker}{\operatorname{Ker}}
\newcommand{\ImC}{\operatorname{Im}}
\newcommand{\ReC}{\operatorname{Re}}
\newcommand{\dvol}{\operatorname{dvol}}
\newcommand{\Tr}{\operatorname{Tr}}
\renewcommand{\vec}[1]{\mathbf{#1}}
\newcommand{\Interior}{\operatorname{Int}}
\newcommand{\sech}{\operatorname{sech}}

\theoremstyle{plain}
\newtheorem{theorem}{Theorem}[section]
\newtheorem{lemma}[theorem]{Lemma}
\newtheorem{proposition}[theorem]{Proposition}

\theoremstyle{definition}
\newtheorem{definition}[theorem]{Definition}

\theoremstyle{remark}
\newtheorem{remark}[theorem]{Remark}

\begin{document}

\setcounter{secnumdepth}{2}
\setcounter{tocdepth}{2}

\bibliographystyle{JHEP-2}

\pagestyle{empty}
\thispagestyle{empty}
\include{arxiv-title}
\clearpage\mbox{}\clearpage  

\setcounter{page}{1}
\renewcommand{\thepage}{\roman{page}}

\pagestyle{empty}
\thispagestyle{empty}
\begin{titlepage}

\Huge 
\begin{center}

\vspace*{3cm}

\textbf{Quantum field theory \\
on rotating black hole spacetimes}

\vspace{1.5cm}

\Large

Hugo Ricardo Cola\c{c}o Ferreira, MASt

\vspace{4cm}

\large

Thesis submitted to The University of Nottingham \\
for the degree of Doctor of Philosophy

\vspace{1cm}

December 2015

\vspace{2cm}

\begin{figure}[h]
\centering
\includegraphics[width=5cm]{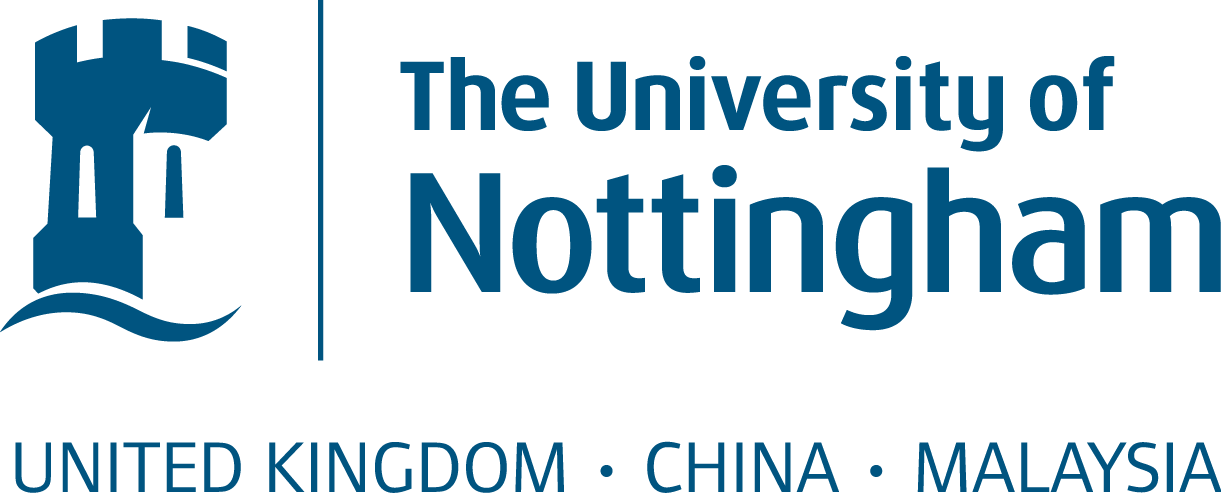}
\end{figure}

\end{center}

\end{titlepage}

\clearpage\mbox{}\clearpage  

\newpage
\pagestyle{plain}	
\thispagestyle{plain}	
\vspace*{5cm}
\begin{center}

\textbf{Abstract}
\end{center}

\normalsize

This thesis is concerned with the development of a general method to compute renormalised local observables for quantum matter fields, in a given quantum state, on a rotating black hole spacetime. The rotating black hole may be surrounded by a Dirichlet mirror, if necessary, such that a regular, isometry-invariant vacuum state can be defined. We focus on the case of a massive scalar field on a (2+1)-dimensional rotating black hole, but the method can be extended to other types of matter fields and higher-dimensional rotating black holes.

The Feynman propagator of the matter field in the regular, isometry-invariant state is written as a sum over mode solutions on the complex Riemannian section of the black hole. A Hadamard renormalisation procedure is implemented at the level of the Feynman propagator by expressing its singular part as a sum over mode solutions on the complex Riemannian section of rotating Minkowski spacetime. This allows us to explicitly renormalise local observables such as the vacuum polarisation of the quantum field.

The method is applied to the vacuum polarisation of a real massive scalar field on a (2+1)-dimensional warped AdS${}_3$ black hole surrounded by a mirror. Selected numerical results are presented, demonstrating the numerical efficacy of the method. The existence of classical superradiance and the classical linear mode stability of the warped AdS${}_3$ black hole to massive scalar field perturbations are also analysed.

\clearpage\mbox{}\clearpage  

\pagestyle{plain}
\thispagestyle{plain}

\begin{center}

\textbf{Acknowledgements}
\end{center}

\normalsize

First and foremost, I would like to thank my supervisor Dr.~Jorma Louko for all his help during the last four years and for reading this manuscript. I am very grateful for having had the opportunity to work with you!

I thank the Universitas 21 Network for awarding me a Universitas 21 Prize Scholarship for a month-long visit to the Gravitational Theory Group of the University of Maryland. I especially thank Prof.~Bei-Lok Hu for his hospitality.

Many thanks to Dr.~Vitor Cardoso, Dr.~Sam Dolan, Prof.~Christopher Fewster, Dr.~Carlos Herdeiro, Prof.~Bernard Kay, Prof.~Elizabeth Winstanley and Dr.~Helvi Witek for helpful discussions and comments during my PhD.

This thesis marks the end of a five-year adventure in the UK, which started in Cambridge and ended in Nottingham. I would like to thank Adam Fraser and Sagi Elster for a wonderful time in Cambridge, and also Emily Kirk, who has successfully endured my company for five years! I thank my fellow physicists and mathematicians Benito Aubry, Johnny Espin, James Gaunt, Sara Tavares and Carlos Scarinci for the great time discussing Physics (and sometimes other things) in Nottingham. And I also thank Andrew Yiakoumetti, Anja Andrejeva, Jyothika Kumar, Iker P\'{e}rez and Jennifer Kiefer for making my life in Nottingham much more enjoyable!

Finally, I wish to thank my parents Carlos and L\'{i}gia Ferreira for always being there and helping me achieve this important step in my life; my brother Andr\'{e} Ferreira for his friendship all these years; and my cousin Maria Em\'{i}lia Fernandes, with an academic greeting!

\vspace*{5ex}

The author acknowledges financial support from Funda\c{c}\~{a}o para a Ci\^{e}ncia e Tecnologia (FCT)-Portugal through Grant No.\ SFRH/BD/69178/2010.

\begin{center}
\includegraphics[width=\textwidth]{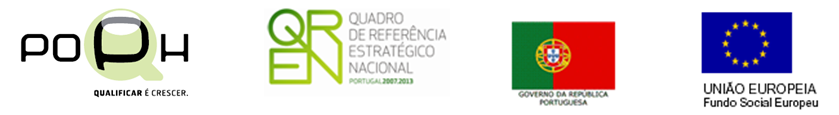}
\end{center}

\newpage
\phantomsection
\pagestyle{plain}
\thispagestyle{plain}
\tableofcontents
\cleardoublepage
\phantomsection

\newpage
\listoffigures

\newpage
\phantomsection
\addcontentsline{toc}{chapter}{Preface}
\addtocontents{toc}{\protect\vspace{-9pt}}
\pagestyle{plain}
\thispagestyle{plain}

\chapter*{Preface}

The research included in this thesis was carried out in the School of Mathematical Sciences of the University of Nottingham. Chapters~\ref{chap:localobservables} and \ref{chap:computation} are the outcomes of collaboration with Jorma Louko. Chapter~\ref{chap:classical-stability} is the outcome of the author's own research. 

The results of this research have lead to the following publications:
\begin{itemize}
\item H.~R.~C.~Ferreira, \emph{Stability of warped AdS${}_3$ black holes in Topologically
Massive Gravity under scalar perturbations}, \emph{Phys.Rev.} D87 (2013), no. 12
124013 (Ref.~\cite{Ferreira:2013zta}).
\item H.~R.~C.~Ferreira and J.~Louko, \emph{Renormalized vacuum polarization on rotating
warped AdS${}_3$ black holes}, \emph{Phys.Rev.} D91 (2015), no. 2 024038 (Ref.~\cite{Ferreira:2014ina}).
\item H.~R.~C.~Ferreira, \emph{Renormalized vacuum polarization of rotating black holes},
\emph{Int.J.Mod.Phys.} (2015) 1542007 (Ref.~\cite{Ferreira:2015ipa}).
\end{itemize}

\newpage
\phantomsection
\addcontentsline{toc}{chapter}{Notation}
\addtocontents{toc}{\protect\vspace{-9pt}}
\pagestyle{plain}
\thispagestyle{plain}

\chapter*{Notation}

In this thesis, we use metric signature $(-, +, \cdots, +)$. For the majority of the thesis, we use units such that $\hbar = c = G = k_B = 1$.

We will use abstract index notation, as presented in Section 2.4 of \cite{Wald:1984rg}. Greek indices $\mu$, $\nu$, etc.~refer to tensor components with respect to some coordinate basis, whereas abstract indices are Latin indices $a$, $b$, etc.~and are used to denote tensor equations which are valid in any basis.

The Riemann tensor, in a coordinate basis, is given by $${R^{\mu}}_{\nu\rho\sigma} = \partial_{\rho} \Gamma^{\mu}_{\nu\sigma} - \partial_{\sigma} \Gamma^{\mu}_{\nu\rho} + \Gamma^{\lambda}_{\nu\sigma} \Gamma^{\mu}_{\lambda\rho} - \Gamma^{\lambda}_{\nu\rho} \Gamma^{\mu}_{\lambda\sigma} \, , $$ and the Ricci tensor is defined by $R_{ab} = {R^{c}}_{acb}$.

The complex conjugate of a complex number $z$ is denoted by $\overline{z}$. The adjoint of an operator $T$ acting on a Hilbert space is denoted by $T^{\dagger}$. If $A$ and $B$ are sets, then $A \subset B$ means that $A$ is a subset of, or is included in, $B$.

Other notation and mathematical conventions are introduced in Chapter~\ref{chap:maths}.

\cleardoublepage

\pagestyle{myheadings}

\setcounter{page}{1}
\renewcommand{\thepage}{\arabic{page}}

\newpage
\phantomsection
\addcontentsline{toc}{chapter}{Introduction}
\addtocontents{toc}{\protect\vspace{-9pt}}

\chapter*{Introduction}
\label{chap:introduction}
\markboth{INTRODUCTION}{INTRODUCTION}

It would be an understatement to claim that the main principles of fundamental physics were completely overturned during the last century. At the time of writing up this thesis, the formulation of general relativity by Albert Einstein, which revolutionised the notions of time and space and replaced Newton's laws of gravitation, is celebrating its 100th anniversary. Moreover, starting during the 1930s, quantum field theory provided a new theoretical framework to understand the elementary constituents of matter and their interactions, which has culminated with the standard model of elementary particles in the 1970s. These theories have enjoyed a remarkable degree of experimental success and have allowed us to describe almost every single observation made to this day.

In spite of these major achievements, there remains a very important theoretical gap in our understanding of fundamental physics: in their current versions, general relativity and quantum field theory are not compatible and, as such, there is not currently a quantum theory of gravity. It has been proven prohibitively difficult to describe the gravitational field in the framework of quantum field theory, a strategy which was successful with the electromagnetic, weak and strong nuclear interactions.

During the last few decades, there have been several proposals for a theory of quantum gravity, most notably string gravity and loop quantum gravity. String theory claims to provide a unified description of the elementary particles and interactions, including the graviton and the gravitational interaction, having as the most basic physical constituent a one-dimensional object called a ``string'' \cite{becker2006string}. Loop quantum gravity attempts to describe the structure of spacetime as consisting of networks of finite loops, the so-called ``spin networks'' \cite{rovelli2004quantum}. Other approaches include asymptotic safety \cite{Niedermaier:2006wt} and causal dynamical triangulations \cite{Ambjorn:2013apa}.

One common feature of all these proposals is that they reduce to descriptions of quantised fields on classical curved backgrounds for energy levels way below the Planck scale,
\begin{equation*}
E_{\rm P} = \sqrt{\frac{\hbar c^5}{G}} \approx 1.22 \times 10^{16} \, \text{TeV} \, ,
\end{equation*}
where $\hbar$ is the reduced Planck's constant, $c$ is the speed of light in vacuum and $G$ is Newton's gravitational constant. The relevant regime for a full theory of quantum gravity is the one with energies of the order of the Planck energy or above and it concerns extreme situations such as neighbourhoods of black hole singularities and the Big Bang itself. On the other hand, the current limit of high energy experiments, such as the ones carried out in the Large Hadron Collider in CERN, is of the order of 10 TeV, about 15 orders of magnitude below the Planck scale. Therefore, for energy scales much smaller than the Planck scale it is natural to expect that the quantum effects of the gravitational field are negligible and that a description in which only the matter fields are quantised and the spacetime itself remains classical and fixed should provide a very good approximation to physical reality.

We can then think of \emph{quantum field theory on curved spacetimes} \cite{birrell1984quantum,fulling1989aspects,wald1994quantum} as a first step in the direction of formulating a theory of quantum gravity and an immediate generalisation of standard quantum field theory on flat spacetimes. The effects of the matter fields on the background geometry are ignored and, as such, the spacetime is fixed. We may improve the theory by including the backreaction effects of the matter fields on the background, which is the realm of \emph{semiclassical gravity}. Now, the spacetime is not fixed and its dynamics is given by the \emph{semiclassical Einstein equations},
\begin{equation*}
G_{ab} = \frac{8\pi G}{c^4} \langle T_{ab} \rangle \, ,
\end{equation*}
where $G_{ab}$ is the Einstein tensor and $\langle T_{ab} \rangle$ is the expectation value of the stress-energy tensor of a matter field in some quantum state, which acts as the source term. The computation of this local observable is then paramount in this framework. However, the stress-energy tensor is quadratic in the field operators, which are mathematically operator-valued distributions in the spacetime, hence, a renormalisation procedure is necessary to remove their short-distance singularity behaviour. We will return to this important point below.

Similarly to quantum field theory on curved spacetimes, semiclassical gravity breaks down at the Planck scale. But it also breaks down when the fluctuations of the stress-energy tensor become large, in which case the expectation value $\langle T_{ab} \rangle$ is no longer a good fit for the source term of the semiclassical Einstein equations. One expects that a new term encoding the stress-energy fluctuations should be added to the source term. A self-consistent approach to extend semiclassical theory to account for these quantum fluctuations is \emph{stochastic semiclassical gravity} \cite{Hu:2008rga}. This theory can be considered yet another step in the direction of quantum gravity.

In this thesis, we will focus on the framework of quantum field theory on curved spacetimes, with the intent of applying it to rotating black hole spacetimes. Historically, the study of quantum field theory on black hole backgrounds has mostly been restricted to asymptotically flat spacetimes, due to their relevance for astrophysics. Perhaps the most famous result is the celebrated Hawking effect \cite{Hawking:1974sw}, by which a black hole formed by stellar collapse emits thermal radiation. Recently, some attention has also been devoted to asymptotically anti-de Sitter (AdS) spacetimes, due to the AdS/CFT correspondence \cite{Aharony:1999ti}, but usually only in the classical regime, as this is sufficient in the context of the AdS/CFT correspondence, and hence few attempts have been made to study quantum field theory on these backgrounds. 

Besides the characteristics of the asymptotics of these black holes, a major part of the research has addressed static, spherical symmetrical geometries, where the isometries can be used to simplify computations. It was also in this setting that the first explicit calculations of renormalised local observables for a matter field on a black hole were performed, such as the vacuum polarisation and the expectation value of the stress-energy tensor \cite{Christensen:1976vb,Candelas:1980zt,Howard:1984qp,Anderson:1990jh,Anderson:1993if,Anderson:1994hg}. 

Static, spherical symmetric black holes have two key properties that can be utilised in the computation of local quantum observables. First, as we shall describe in detail in this thesis, for static spacetimes one can make use of the so-called ``Euclidean methods'' to simplify the computation of certain quantities such as the Feynman propagator for a given matter field. For instance, if one considers a scalar field $\Phi$, its Feynman propagator associated with a vacuum state $|0 \rangle$ is defined as
\begin{equation*}
G^{\rm F}(x,x') := i \, \langle 0 | \mathscr{T} \left( \Phi(x) \Phi(x') \right) | 0 \rangle \, ,
\end{equation*}
where $\mathscr{T}$ is the time-ordering operator. The Feynman propagator takes a crucial role in the renormalisation of the vacuum polarisation, $\langle \Phi^2(x) \rangle$, and of the expectation value of the stress-energy tensor, $\langle T_{ab} (x) \rangle$. The Euclidean method allows us to consider the Riemannian (or ``Euclidean'') section of the static spacetime (by means of a Wick rotation) on which the Green's distribution associated with the matter field equation is directly related to the Feynman propagator evaluated for a well defined state which is invariant under the isometries of the spacetime. In the specific case of a Schwarzschild black hole, this state is known as the Hartle-Hawking state \cite{Hartle:1976tp}. The Green's distribution is unique, due to the ellipticity of the matter field operator in the Riemannian manifold, and its computation may be done using standard techniques of the theory of Green's functions.

Second, and directly related to the previous point, a state like the Hartle-Hawking state in a Schwarzschild black hole is well known to exist for static black hole spacetimes \cite{Kay:1988mu}. Therefore, the method described above leads to the Feynman propagator evaluated for such a state, after which the renormalisation procedure can be applied to obtain the desired local observable. If we want the local observable to be evaluated with respect to another quantum state, it suffices to use the regular, isometry-invariant state as a reference and calculate the difference, which is finite without any further renormalisation.

Having said this, there have been attempts at considering stationary, but non-static, black hole spacetimes, with the main focus on Kerr \cite{Frolov:1982pi,Frolov:1984ra,Frolov:1986ut,Frolov:1989jh,Ottewill:2000qh,Duffy:2005mz,Casals:2012es}. In particular, the computation of the renormalised expectation value of the stress-energy tensor has proven to be very challenging and, so far, almost all calculations have only addressed the differences between expectation values for different quantum states \cite{Duffy:2005mz,Casals:2012es} and the large field mass limit \cite{Belokogne:2014ysa}. A notable exception is \cite{Steif:1993zv}, where the stress-energy tensor for the rotating BTZ black hole in 2+1 dimensions \cite{Banados:1992wn,Banados:1992gq} was renormalised with respect to AdS${}_3$, by using the fact that the black hole corresponds to AdS${}_3$ with discrete identifications, but this method cannot be used for more general classes of rotating black holes.

In comparison with the static, spherical symmetric case, we may summarise the main difficulties to compute renormalised local observables for matter fields on rotating black hole spacetimes in three points:
\begin{enumerate}[label={(\roman*)}]
\item the non-existence of generalisations of the Hartle-Hawking state, i.e.~a regular, isometry-invariant vacuum state;
\item the unavailability of Euclidean methods to simplify the computation of quantities such as the Feynman propagator;
\item the technical complexity of the computation to the lack of spherical symmetry.
\end{enumerate}

In this thesis we address each of the above points and provide a method to explicitly compute certain classes of local observables for quantised matter fields on a wide variety of rotating black hole spacetimes.

Concerning point (i), it has been shown that the Hartle-Hawking state for a scalar field in the Schwarszchild spacetime does not have a generalisation to the Kerr spacetime \cite{Kay:1988mu}. As reviewed in \cite{Ottewill:2000qh}, this is linked to the existence of a speed of light surface, outside of which no observer can co-rotate with the Kerr horizon, which does not exist in the Schwarszchild spacetime. An heuristic way to understand this point is to note that an observer on Schwarzschild co-rotating with the horizon would perform measurements with respect to the Hartle-Hawking state, which is, by definition, regular at the horizon. However, on Kerr, such observers cannot rotate with the same angular velocity as the Kerr horizon at and beyond the speed of light surface, as their worldlines would become null or spacelike. Given that the notion of a quantum state is a global notion, there cannot be a state which is regular at the horizon and defined everywhere in the exterior region of the black hole.

One way around this problem is to restrict the spacetime on which the matter field propagates so that it does not include the region from the speed of light surface to infinity. This can be done explicitly by inserting an appropriate timelike boundary which respects the isometries of the spacetime. The simplest example is a boundary on constant radial coordinate at which the matter field vanishes, i.e.~Dirichlet boundary conditions are imposed. We shall often call this boundary a ``mirror''. If the boundary is located between the horizon and the speed of light surface, then a vacuum state which is regular at the horizon and invariant under the isometries of the spacetime may be defined \cite{Duffy:2005mz}.

Regarding point (ii), the Euclidean methods used for static spacetimes to compute quantities such as the Feynman propagator, by performing the calculations on the Riemannian section of the spacetime, cannot be easily generalised to rotating spacetimes, since, in general, such a Riemannian section does not exist. This is the case for the Kerr black hole \cite{Woodhouse:1977-complex}. Note further that, even if such section with a positive definite metric existed, the Green's distribution associated with the matter field equation could not be related with the Feynman propagator evaluated for a regular, isometry-invariant vacuum state on the original spacetime, since such a state does not generally exist, cf.~point (i).

Nevertheless, even though Kerr and other rotating black holes do not admit a \emph{real} Riemannian section, the portions of their exterior regions between the horizon and the mirror we introduced previously do admit a \emph{complex} Riemannian section, which is obtained by means of a Wick rotation, but with no further analytical continuation of metric parameters \cite{Gibbons:1976ue,Brown:1990di,Moretti:1999fb} (a precise definition will be given in this thesis). The metric of the complex Riemannian section of the rotating black hole is complex-valued and the matter field operator is no longer elliptic as in the static case. However, analogously to the static case, the Green's distribution associated with the matter field equation which is regular at the horizon and satisfies the Dirichlet boundary condition at the mirror is unique and can be related to the Feynman propagator evaluated for the regular, isometry-invariant state.

Both in the static and stationary cases, the Green's distributions on the Riemannian sections are obtained as discrete sums over mode solutions of the defining differential equations. These sums are not convergent and a renormalisation procedure is required in order to subtract their short-distance (or high-frequency) divergences. An important property of these Green's distributions, and which is the basis of the so-called \emph{Hadamard renormalisation} \cite{Wald:1977up,Brown:1986tj,Decanini:2005eg}, is the fact they can be decomposed into a purely geometric part, which is singular in the coincidence limit, and a state-dependent part, which is regular in the coincidence limit. The idea is then to subtract the singular, purely geometric part, after which the renormalised local observables of interest, which involve the coincidence limit of the Green's distributions, can be obtained. 

Yet, this is easier said than done. The singular part of the Green's distribution is known in closed form for spacetimes of any dimension \cite{Decanini:2005eg}, whereas the full Green's distribution on a stationary spacetime is known only as a sum over mode solutions. It is, however, a highly non-trivial task to express the singular part of the Green's distribution as a mode sum, such that the short-distance divergences can be subtracted term by term. The strategy implemented in this thesis is to express the singular part of the Green's distribution as a sum over mode solutions on a spacetime for which the Green's distribution is known both in closed form (in terms of known functions) and as a mode sum, such as the Minkowski spacetime. This technical point will be fully explored in this thesis and we will argue that only the asymptotic approximations for the mode solutions for large values of the sum indices are needed in order to perform the subtraction. This is especially important for Kerr and higher-dimensional black holes for which the mode solutions have to be constructed fully numerically. In this way, one can remove the divergences of the Green's distribution for the rotating black hole spacetime and compute the renormalised local observables.

It remains to address point (iii) on the technical complexity of the computation. This is clearly manifest, for instance, on the fact that the partial differential equations describing matter fields propagating on Kerr can be separated into two ordinary differential equations, a radial and an angular part \cite{Teukolsky:1972my}, whereas only the radial part is necessary for fields propagating on Schwarzschild. In order to describe the method to compute local observables on rotating black holes discussed in this thesis without superfluous technical details, we will focus on rotating black hole spacetimes in 2+1 dimensions which are solutions of Einstein gravity or other modified theories of gravity.

(2+1)-dimensional gravity provides a convenient area to explore several aspects of black hole physics and quantum gravity \cite{carlip2003quantum}. Research on this field greatly increased after Einstein gravity in 2+1 dimensions was shown to be equivalent to a Chern-Simons gauge theory \cite{Achucarro:1987vz,Witten:1988hc}. The main advantage of focusing on this lower-dimensional setting is its technical simplicity and, in particular, the fact that many of quantities of interest can be obtained in closed form, such as the mode solutions of matter field equations. Even though Einstein gravity in 2+1 dimensions is a topological theory with no propagating degrees of freedom, it was possible to find a black hole solution, the Ba\~{n}ados-Teitelboim-Zanelli (BTZ) black hole, when the cosmological constant is negative \cite{Banados:1992wn,Banados:1992gq,Carlip:1995qv}. This spacetime is asymptotically AdS, and a vast amount of research has been done on it, partly inspired by the AdS/CFT correspondence after the late 1990s.

If one insists on having at least one propagating degree of freedom, one may consider a deformation of Einstein gravity called topologically massive gravity (TMG), which is obtained by adding a gravitational Chern-Simons term to the Einstein-Hilbert action with a negative cosmological constant \cite{Deser:1981wh,Deser:1982vy}. The resulting theory contains a massive propagating degree of freedom, although at the expense of being a third-order derivative theory. A very important property of this theory is that solutions of Einstein gravity, such as AdS${}_3$ and the BTZ black hole, are also solutions of TMG. Nevertheless, there are also new solutions and we focus on the warped AdS${}_3$ vacuum solutions and the warped AdS${}_3$ black hole solutions \cite{Nutku:1993eb,Gurses1994,Moussa:2003fc,Moussa:2008sj,Anninos:2008fx}. Mathematically, warped AdS${}_3$ spacetimes are Hopf fibrations of AdS${}_3$ over AdS${}_2$ where the fibre is the real line and the length of the fibre is ``warped'' \cite{Bengtsson:2005zj,Anninos:2008qb}. These solutions are thought to be perturbatively stable vacua of TMG in a wide region of the parameter space of the theory, in contrast to the AdS${}_3$ solution \cite{Anninos:2009zi}. Analogously to the BTZ black hole, the warped AdS${}_3$ black hole solutions are identifications of warped AdS${}_3$ vacuum solutions. In the limit in which the warping of spacetime vanishes, one recovers the BTZ black hole as a solution of TMG.

There are several reasons why the study of matter fields in warped AdS${}_3$ black hole spacetimes is interesting on its own right. These black holes are rotating (in fact, they do not have a static limit) and their causal structure resembles asymptotically flat spacetimes in the general case and AdS in the limit of no warping (which corresponds to the BTZ black hole) \cite{Jugeau:2010nq}. We then have at our disposal an example of a (2+1)-dimensional black hole whose asymptotic structure is very similar to Kerr and on which we can investigate the implementation of the method described in this thesis in a simpler setting. Note, however, that these black holes are not, strictly speaking, asymptotically flat, as they are asymptotic to the warped AdS${}_3$ vacuum solutions. Another particularly novel point is that these rotating black holes do not possess a stationary limit surface, but they nonetheless have a speed of light surface.

Henceforth, for the reasons given above, as an example on which to apply the general method to compute the renormalised local observables on rotating black holes, we will use the warped AdS${}_3$ black hole and consider a real massive scalar field propagating on this background. We will see that the use of this (2+1)-dimensional spacetime allows us to perform the calculations without having to deal with all the technical difficulties arising in its higher-dimensional analogues, namely most of the numerical computations --- the only numerics we will need is for the mode sums which give the Green's distributions associated with the scalar field equation. However, we emphasise that the implementation of our method does not require the knowledge of the mode solutions of the field equation in closed form, but only their asymptotic approximations for large values of the quantum numbers.

In closing, we should also note that this method is suitable to compute a wide class of local observables such as the vacuum polarisation of a scalar field, $\langle \Phi^2(x) \rangle$, but it turns out \emph{not} to be suitable for local observables such as the expectation value of the stress-energy tensor, $\langle T_{ab}(x) \rangle$. The main reason for this limitation is the impossibility of expressing the singular part of covariant derivatives of the Green's distributions for a rotating black hole as a sum over mode solutions, or derivatives of mode solutions, on Minkowski spacetime, for which the Green's distribution and its derivatives are known in closed form. As we will see in detail in the thesis, this comes essentially from the fact that the shift function of the metric of Minkowski written in some rotating coordinate system is a constant in spacetime, whereas the shift function of the metric of a rotating black hole is a function of the radial coordinate in some coordinate system. This makes the task of expressing the short-distance singular behaviour of covariant derivatives of the Green's distribution for the rotating black hole in terms of the short-distance singular behaviour of covariant derivatives of the Green's distribution for Minkowski impossible. We shall return to this point in the \hyperref[chap:conclusions]{Conclusions}.


\section*{Outline}

The outline of the thesis is as follows. It is divided in two main parts. Part I deals with the basics of quantum field theory on curved spacetimes, with particular focus on rotating black hole spacetimes, and the method to compute renormalised local observables for quantised matter fields propagating on rotating black holes. Part II introduces the (2+1)-dimensional warped AdS${}_3$ black hole solution and uses it as the background for an explicit computation of the vacuum polarisation for a massive scalar field on the Hartle-Hawking state.

Part I starts with Chapter~\ref{chap:maths}, which gives an overview of the mathematical tools used throughout the thesis, with the intent of establishing notation and stating the necessary essential results on the causal structure of spacetimes, stationary spacetimes, bi-tensors, symplectic and Hilbert spaces, distribution theory and hyperbolic and Green operators. 

In Chapter~\ref{chap:qftcst}, we present a detailed overview on quantum field theory on curved spacetimes. In particular, we explore the classical and quantum theories of a real scalar field on a globally hyperbolic spacetime in Section~\ref{sec:qftcst-realscalarfield}, before focusing on stationary spacetimes in Section~\ref{sec:qftcst-stationary} and on spacetimes with boundaries in Section~\ref{sec:qftcst-stwithboundaries}. We finish this chapter with a description of the Hadamard renormalisation procedure in Section~\ref{sec:qftcst-hadamardrenormalisation}.

Chapter~\ref{chap:localobservables} deals with the method to compute renormalised local observables in rotating black hole spacetimes and constitutes the most important new results in this thesis. In Section~\ref{sec:computation-scalarfield} we consider a massive scalar field on a (2+1)-dimensional rotating black hole surrounded with timelike boundaries (the ``mirrors'') and construct the regular, isometry-invariant vacuum state, which we call the Hartle-Hawking state. In order to obtain the Feynman propagator evaluated for this quantum state, in Section~\ref{sec:quasi-euclidean-method} we introduce the ``quasi-Euclidean method'' which allows us to obtain the complex Riemannian section of the exterior region of the rotating black hole, on which we get the Green's distribution associated with the scalar field equation, expressed as a sum over mode solutions. In Section~\ref{sec:renormalisation-procedure}, we implement the Hadamard renormalisation procedure to subtract the short-distance divergences of the Green's distribution. This is done by expressing the singular part of the Green's distribution on the rotating black hole as a sum over mode solutions on Minkowski. This culminates on Theorem~\ref{thm:matchingpolarisation}, where it is shown that the resulting mode sum is convergent in the coincidence limit. All this procedure allows us to obtain the vacuum polarisation for a scalar field on the (2+1)-dimensional rotating black hole, but we argue that this method can be straightforwardly extended to higher-dimensional rotating black hole spacetimes. We finish this chapter by explaining why this method is not suitable to renormalise the expectation value of the stress-energy tensor in Section~\ref{sec:stress-energy-tensor}.

Having developed the main method in the preceding chapters, in Part II the method is applied to explicitly compute the renormalised vacuum polarisation of a scalar field on a warped AdS${}_3$ black hole. 

We introduce the black hole solutions in Chapter~\ref{chap:wadsbh}, after a brief discussion of Einstein gravity and topologically massive gravity in 2+1 dimensions.

Before moving to the quantum theory, we first have a detailed look at some aspects of the classical theory of a scalar field on a warped AdS${}_3$ black hole in Chapter~\ref{chap:classical-stability}, in particular the existence of classical superradiance and the classical stability of the black hole to scalar field mode perturbations. We conclude that classical superradiance is indeed present, but that it does not lead to superradiant instabilities, even when the black hole is surrounded by a mirror, which is the case we are interested in the quantum theory. These stability results are new.

Finally, in Chapter~\ref{chap:computation} we use the method of Chapter~\ref{chap:localobservables} to compute the renormalised vacuum polarisation of a scalar field in the Hartle-Hawking state on a warped AdS${}_3$ black hole surrounded by a mirror. Selected numerical results are presented, demonstrating the numerical efficacy of the method.

The thesis is concluded with some final remarks about the research described above in the \hyperref[chap:conclusions]{Conclusions}. This is followed by four appendices, which deal with the complex Riemannian section of the Minkowski spacetime, WKB expansions, hypergeometric functions and classical black hole superradiance.

\pagestyle{headings}

\part{Fundamentals}


\chapter{Mathematical preliminaries}
\label{chap:maths}

The aim of this chapter is to present the mathematical tools which will be used throughout the thesis. It is assumed that the reader is familiar with the basic mathematics used in general relativity and quantum field theory. The intent here is to establish notation and present the necessary definitions and theorems, without many details and often without proofs. Relevant references to all the topics are provided.


\section{Spacetime and causal structure}

In this section, the basic ideas on the causal structure of spacetimes are presented, leading to the definition of a globally hyperbolic spacetime, the usual starting point for Quantum Field Theory. Standard references for this topic are chapter 8 of \cite{Wald:1984rg} and chapter 6 of \cite{hawking1973large}.

First, we start by recalling the basic definition of spacetime.

\begin{definition}
A \emph{spacetime} $(M,g,\mathfrak{o},\mathfrak{t})$ is a $d$-dimensional ($d \geq 2$) connected, orientable, time-orientable, smooth manifold $M$ equipped with a smooth Lorentzian metric $g$ of signature $(-,+,...,+)$, a choice of orientation $\mathfrak{o}$ and a choice of time orientation $\mathfrak{t}$.
\end{definition}

\begin{remark}
For convenience, a spacetime $(M,g,\mathfrak{o},\mathfrak{t})$ will often be denoted either by $(M,g)$ or more simply by $M$.
\end{remark}

The Lorentzian character of $g$ provides a causal structure to the spacetime $M$. For each point $p \in M$, denote the tangent space by $T_p M$. The following definitions concerning properties of tangent vectors and vector fields are standard.

\begin{definition} \label{def:causalityvectors}
A non-zero tangent vector $v_p \in T_p M$ is \emph{timelike} if $g(v_p,v_p) < 0$, \emph{null} if $g(v_p,v_p) = 0$ and \emph{spacelike} if $g(v_p,v_p) > 0$. A tangent vector is \emph{causal} if it is either timelike or null.
\end{definition}

\begin{remark}
Definition~\ref{def:causalityvectors} can be extended to vector fields if these satisfy the aforementioned properties for all $p \in M$.
\end{remark}

\begin{definition}
A Lorentzian manifold $M$ is said to be \emph{time-orientable} if there exists a smooth global timelike vector field on $M$.
\end{definition}

\begin{definition} \label{def:timeorientation}
Let $M$ be a time-orientable Lorentzian manifold. A \emph{time-orientation} $\mathfrak{t}$ is an equivalence class $[v]$ of timelike vector fields $v$ where $v \sim w$ if $g(v_p,w_p) < 0$ for all $p \in M$.
\end{definition}

\begin{definition}
Let $M$ be a time-orientable manifold and let $\mathfrak{t} = [v]$ be a time-orientation. A causal vector $u_p$ at $p \in M$ is \emph{future-directed} (resp.~\emph{past-directed}) if $g(u_p, v_p) < 0$ (resp.~$g(u_p, v_p) > 0$), for any $v \in \mathfrak{t}$.
\end{definition}

\begin{definition}
A smooth curve is called \emph{spacelike} (resp., \emph{timelike}, \emph{null}, \emph{causal}, \emph{future-directed}, \emph{past-directed}) if its tangent vector is everywhere spacelike  (resp., timelike, null, causal, future-directed, past-directed).
\end{definition}

\begin{remark} \label{rem:observer}
A timelike curve is sometimes referred to as a \emph{worldline} or a \emph{observer.}
\end{remark}

\begin{definition}
A point $p \in M$ is said to be the \emph{future} (resp., \emph{past}) \emph{endpoint} of a future- (resp., past-) directed curve $\gamma$ if for every neighbourhood $N$ of $p$ there exists a $t_0$ such that $\gamma(t) \in N$ for all $t > t_0$. The curve is said to be \emph{future} (resp., \emph{past}) \emph{inextendible} if it has no future (resp., past) endpoint.
\end{definition}

\begin{definition}
The \emph{chronological future} $I^+(U)$ (resp., \emph{chronological past} $I^-(U)$) of a subset $U \subset M$ is the set of all points which can be reached from $U$ by a future-directed (resp., past-directed) timelike curves.
\end{definition}

\begin{definition}
The \emph{causal future} $J^+(U)$ (resp., \emph{causal past} $J^-(U)$) of a subset $U \subset M$ is the set of all points which can be reached from $U$ by a future-directed (resp., past-directed) causal curves (see Fig.~\ref{fig:causal}). Their union $J(U) := J^+(U) \cup J^-(U)$ is called the \emph{causal shadow} of $U$.
\end{definition}

\begin{figure}[t!]
\begin{center}
\def\svgwidth{0.55\textwidth}
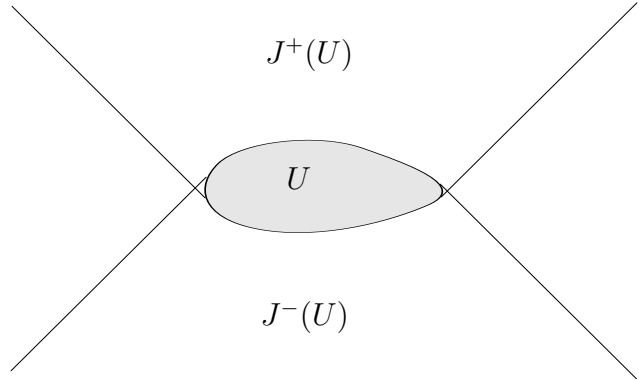
\caption[A set and its causal future and causal past.]{\label{fig:causal} A set $U \subset M$ and its causal future $J^+(U)$ and causal past $J^-(U)$.}
\end{center}
\end{figure}

\begin{remark}
Let $\dot{I}^{\pm}(U)$ and $\dot{J}^{\pm}(U)$ denote the boundaries of $I^{\pm}(U)$ and $J^{\pm}(U)$, respectively. It follows that $\overline{I^{\pm}(U)} = \overline{J^{\pm}(U)}$ and $\dot{I}^{\pm}(U) = \dot{J}^{\pm}(U)$.
\end{remark}

\begin{definition}
Two subsets $U$ and $V$ of $M$ are said to be \emph{causally separated} if $U \cap J(V) = \varnothing$.
\end{definition}

\begin{definition}
A subset $U \subset M$ is said to be \emph{achronal} if $U \cap I^+(U) = \varnothing$, i.e.~if each timelike curve in $M$ intersects $U$ at most once. 
\end{definition}

\begin{definition}
For any subset $U \subset M$, the \emph{future} (resp., \emph{past}) \emph{Cauchy development} or \emph{domain of dependence} $D^+(U)$ (resp., $D^-(U)$) of $U$ is the set of all points $p \in M$ such that every past (resp., future) inextendible causal curve through $p$ intersects $U$. Their union $D(U) := D^+(U) \cup D^-(U)$ is called \emph{Cauchy development} or \emph{domain of dependence} of $U$. 
\end{definition}

\begin{definition}
A closed achronal subset $\Sigma \subset M$ such that $D(\Sigma) = M$ is called a \emph{Cauchy surface}.  
\end{definition}

\begin{definition}
A spacetime is \emph{globally hyperbolic} if it has a Cauchy surface.
\end{definition}

As it will be seen in Section~\ref{sec:hyperbolicoperators}, well posed initial value problems for classical fields can be formulated when those fields propagate on globally hyperbolic spacetimes. Before stating a key theorem regarding the structure of globally hyperbolic spacetimes, the notion of a ``time function'' is introduced, which is also important for the definition of a stationary spacetime in Section~\ref{sec:stationaryspacetimes}.

\begin{definition} \label{def:timefunction}
A \emph{time function} is a continuous function $t : M \to \mathbb{R}$ such that $- \nabla^a t$ is a future-directed, timelike vector field.
\end{definition}

\begin{theorem} \label{thm:globhypspacetimefoliation}
Let $M$ be globally hyperbolic. Then, $M$ is isometric to $\mathbb{R} \times \Sigma$ endowed with the metric $\dd s^2 = - N^2 \, \dd t^2 + h_t$, where $t : \mathbb{R} \times \Sigma \to \mathbb{R}$ is a time function, $N$ is a smooth and strictly positive function on $\mathbb{R} \times \Sigma$, $t \mapsto h_t$ yields a one-parameter family of smooth Riemannian metrics and each $\{ t \} \times \Sigma$ is a spacelike smooth Cauchy surface of $M$.
\end{theorem}

\begin{proof}
See \cite{Bernal:2004gm,Bernal:2005qf}.
\end{proof}

\begin{figure}[t!]
\begin{center}
\def\svgwidth{0.65\textwidth}
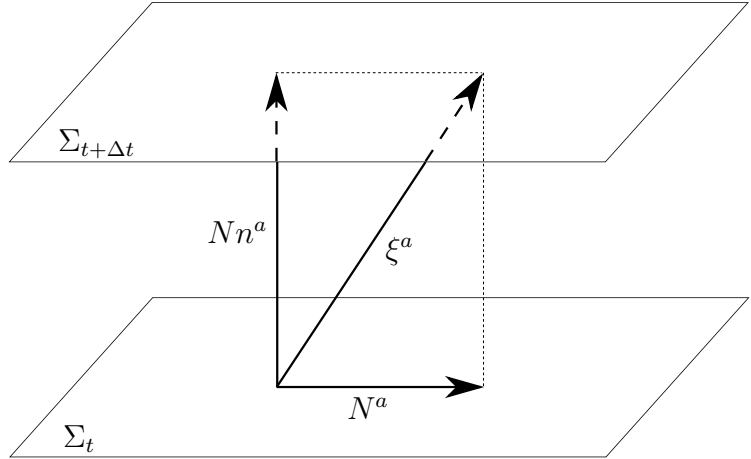
\caption[Normal and shift vectors to a spacelike surface.]{\label{fig:ADMvectors} A surface of constant $t$ and a surface of constant $t+\Delta t$, with the vectors $\xi^a$, $N n^a$ and $N^a$ in \eqref{eq:chinNN}.}
\end{center}
\end{figure}

\begin{remark}
This theorem allows us to perform the \emph{ADM decomposition} of globally hyperbolic spacetimes,
\begin{equation} \label{eq:ADMdecomposition}
\dd s^2 = - N^2 \, \dd t^2 + h_{ij} \left( \dd x^i + N^i \, \dd t \right) \left( \dd x^j + N^j \, \dd t \right) \, ,
\end{equation}
where the Latin indices $i, j = 1, \ldots , d-1$ are spatial indices. $N$ is called the \emph{lapse function} and $N^i$ is the \emph{shift vector}. If we denote $\xi^a = (\partial_t)^a$, we have that
\begin{equation} \label{eq:chinNN}
\xi^a = N n^a + N^a \, ,
\end{equation}
where $n^a$ is the future-directed unit normal vector to the Cauchy surfaces and $N^0 = 0$ (see Fig.~\ref{fig:ADMvectors}). For more details, see e.g.~chapter 10 of \cite{Wald:1984rg}.
\end{remark}

In this thesis, we will also be interested in the problem of constructing quantum field theories for certain non globally hyperbolic spacetimes, namely spacetimes with boundaries. Therefore, we relax the causality conditions on the spacetime, but still impose that the spacetime is stably causal.

\begin{definition} \label{def:stablycausal}
A spacetime $(M,g)$ is \emph{stably causal} if $g$ has a neighbourhood (see \cite{hawking1973large} for a precise definition) so that any spacetime $(M, \tilde{g})$, where $\tilde{g}$ belongs to such neighbourhood, does not contain any closed timelike curves.
\end{definition}

In other words, we require that arbitrarily small perturbations of the metric of a stably causal spacetime does not lead to spacetimes with closed timelike curves. Furthermore, the following can be shown.

\begin{proposition} \label{prop:stablycausaltimefunction}
A spacetime $(M,g)$ is stably causal if and only if there is a time function $t$ on $M$.
\end{proposition}

\begin{proof}
See Proposition 6.4.9 of \cite{hawking1973large}.
\end{proof}

We finish this section by introducing nomenclature for different types of ``compact'' regions of a globally hyperbolic spacetime and spaces of functions with support on these regions (see Fig.~\ref{fig:timespacecompact}).

\begin{definition}
Let $M$ be a globally hyperbolic spacetime. A subset $U \subset M$ is
\begin{enumerate}[label={(\roman*)}]
\item \emph{timelike compact} if $U \cap J(K)$ is compact for each compact $K \subset M$;
\item \emph{spacelike compact} if it is closed and $\exists$ compact $K \subset M$ such that $U \subset J(K)$.
\end{enumerate}
\end{definition}

\begin{definition} \label{def:spacetimelikecompactfunctions}
Let $M$ be a globally hyperbolic spacetime. We denote by
\begin{enumerate}[label={(\roman*)}]
\item $C^{\infty}_0(M)$ the space of smooth functions with compact support;
\item $C^{\infty}_{\rm tc}(M)$ the space of smooth functions with timelike compact support;
\item $C^{\infty}_{\rm sc}(M)$ the space of smooth functions with spacelike compact support.
\end{enumerate}
\end{definition}

\begin{figure}[t!]
\begin{center}
\subfigure[Timelike compact set $U$.]{
\def\svgwidth{0.4\textwidth}
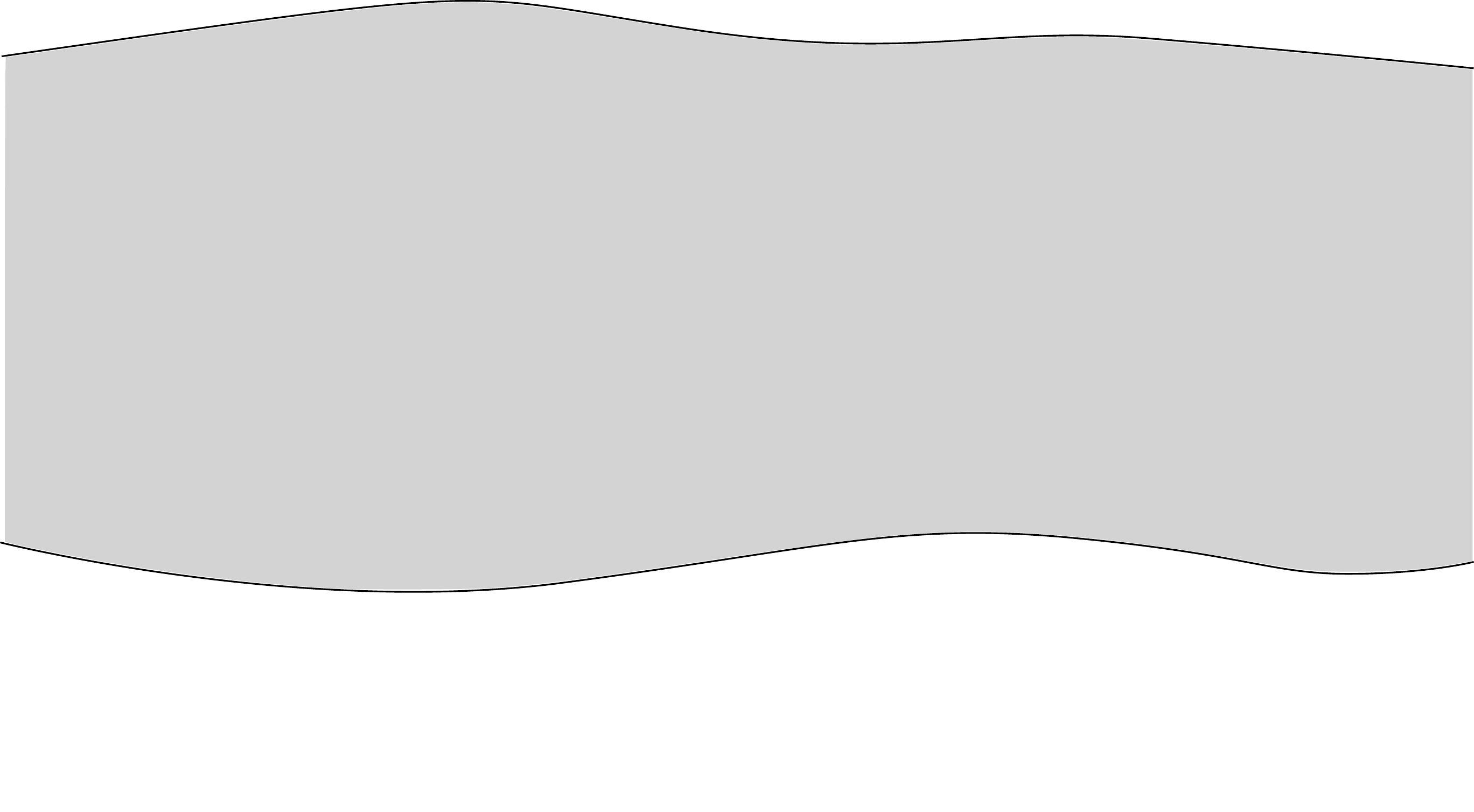 \vspace*{2ex}
}
\subfigure[Spacelike compact set $U$.]{ \hspace*{5ex}
\def\svgwidth{0.2\textwidth}
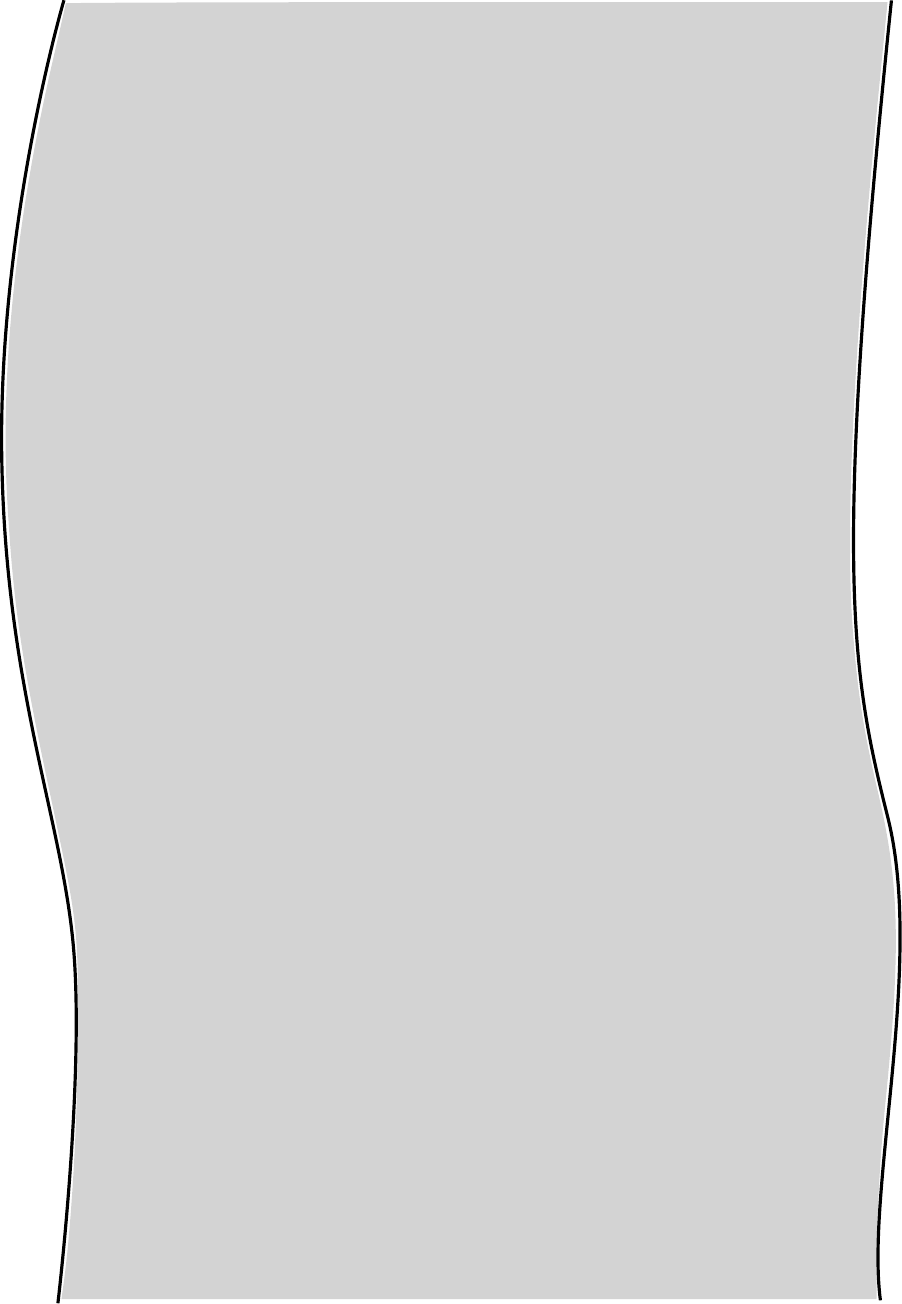 \hspace*{5ex}
}
\end{center}
\caption{\label{fig:timespacecompact} Timelike and spacelike compact sets.}
\end{figure}
%


\section{Stationary spacetimes}
\label{sec:stationaryspacetimes}

\subsection{Globally and locally stationary spacetimes}

It is important to clarify the definition of ``stationary'' spacetime used in this thesis. The strictest definition commonly found in the literature is the following.

\begin{definition} \label{def:globalstationary}
A spacetime $M$ is called \emph{globally stationary} if there exists a Killing vector field $\xi$ which is timelike everywhere in $M$.
\end{definition}

It then follows that, if $t$ is a time function (cf.~Definition~\ref{def:timefunction}) and $\xi = \partial_t$ is a Killing vector field, the metric of a stationary spacetime can be written as \eqref{eq:ADMdecomposition}, with $N$, $N^i$ and $h_{ij}$ being independent of $t$.

An important particular case of a stationary spacetime is a static spacetime.

\begin{definition}
A spacetime is \emph{static} if it is stationary and if there is a spacelike surface which is orthogonal to the orbits of the timelike Killing vector field.
\end{definition}

It follows then that the metric of a static spacetime can be written as \eqref{eq:ADMdecomposition} with vanishing $N^i$ and with $N$ and $h_{ij}$ being independent of $t$.

However, Definition~\ref{def:globalstationary} for a globally stationary spacetime needs to be relaxed if spacetimes such as the Kerr black hole is to be considered stationary, since it does not possess an everywhere timelike Killing vector field. It is common to relax the above definition for spacetimes which are asymptotically flat at null infinity by requiring that $\xi$ is timelike at least in a neighbourhood of null infinity. In this sense, the Kerr black hole is a stationary spacetime, although not globally.

A more general definition, that replaces the global assumption with a local assumption, is:

\begin{definition} \label{def:locallystationary}
A spacetime $M$ is called \emph{locally stationary} if, for any $p \in M$, there exists a neighbourhood $U \subset M$ centred in $p$ and a vector field $\xi$ which is Killing and timelike in $U$.
\end{definition}

In this thesis, we will be mostly interested in this more general class of stationary spacetimes.

\subsection{Riemannian sections of stationary spacetimes}
\label{sec:riemanniansections}

In many practical situations, given a Lorentzian manifold with a coordinate system which covers the whole manifold, it is convenient to perform analytical continuations in the coordinates such that one can consider a new manifold with a different signature, e.g.~a Riemannian manifold, where calculations are easier to carry out. If the analytical continuation exists and is well defined, one can then analytically continue the results back to the original manifold.

The basic idea is the following. One starts with a spacetime, i.e.~a real Lorentzian manifold, and performs an analytical continuation in one or more of the coordinates by allowing them to become complex-valued and by holomorphically extending the metric components into the complex domain. The resulting manifold is called the \emph{complexified spacetime}. One can then single out a subspace of interest which is a manifold on its own right. This subspace is called a \emph{section} of the complexified manifold.

Given a static spacetime, there is a natural section of its complexified manifold which is a real manifold and has Riemannian signature.

\begin{definition} \label{def:realriemanniansection}
Let $M$ be a static spacetime with analytic metric given by
\begin{equation}
\dd s^2 = g_{tt} \, \dd t^2 + h_{ij} \, \dd x^i \dd x^j \, ,
\end{equation}
where $t$ is a global time function (and, hence, $g_{tt}<0$). The \emph{real Riemannian section} is the manifold $M^{\mathbb{R}}$ with (real-valued) metric
\begin{equation}
\dd s^2_{\mathbb{R}} = g_{\tau\tau} \, \dd \tau^2 + h_{ij} \, \dd x^i \dd x^j \, ,
\end{equation}
where $\tau := it \in \mathbb{R}$ and, hence, $g_{\tau\tau} = - g_{tt} > 0$. The analytical continuation $t \to -i \tau$ is commonly known as \emph{Wick rotation}.
\end{definition}

\begin{remark}
In the literature, the real Riemannian section is also known as the \emph{Euclidean section}.
\end{remark}

The Wick rotation can be easily generalised to the case of stationary, but not static, spacetimes. However, the resulting section is not a real manifold anymore.

\begin{definition} \label{eq:complexRiemanniansection}
Let $M$ be a stationary, but not static, spacetime with analytic metric given by
\begin{equation}
\dd s^2 = g_{tt} \, \dd t^2 + h_{ij} \left( \dd x^i + N^i \, \dd t \right) \left( \dd x^j + N^j \, \dd t \right) \, ,
\end{equation}
where $t$ is a global time function (and, hence, $g_{tt}<0$). The \emph{complex Riemannian section} is the manifold $M^{\mathbb{C}}$ with (complex-valued) metric
\begin{equation}
\dd s^2_{\mathbb{C}} = g_{\tau\tau} \, \dd \tau^2 + h_{ij} \left( \dd x^i - i N^i \, \dd \tau \right) \left( \dd x^j - i N^j \, \dd \tau \right) \, ,
\end{equation}
where $\tau := it \in \mathbb{R}$ and, hence, $g_{\tau\tau} = - g_{tt} > 0$. 
\end{definition}

In this thesis, we will see that the complex Riemannian section allows us to vastly simplify calculations on stationary spacetimes, in a similar way to simplifications provided by the real Riemannian section of a static spacetime, as manifested by the numerous ``Euclidean methods'' found in the literature of quantum field theory of static spacetimes \cite{gibbons1993euclidean}.


\section{Bi-tensors}

In this section, we discuss \emph{bi-tensors}, which are objects that transform like tensors at $x$ and $x'$. A good reference for this topic is Chapter 2 of \cite{Poisson:2003nc}. 

We assume that $x$ and $x'$ belong to a geodesically convex neighbourhood.

\begin{definition} \label{def:geodconvexneighbourhood}
A \emph{geodesically convex neighbourhood} of $p \in M$ is a neighbourhood $N$ of $p$ such that, for all $q , \, q' \in N$, there exists a unique geodesic connecting $q$ and $q'$ which lies entirely within $N$.
\end{definition}

To the point $x$ we assign abstract indices $a$, $b$, etc, and to the point $x'$ we assign abstract indices $a'$, $b'$, etc. For instance, ${T_{ab'}}^{c'}(x,x')$ is a bi-tensor which transforms like a (0,1)-tensor at $x$ and like a (1,1)-tensor at $x'$. The same applies when taking covariant derivatives of bi-tensors. For example, if $T_{\cdots}(x,x')$ is a sufficiently regular bi-tensor, in $\nabla_a \nabla_{b'} T_{\cdots}(x,x')$, $\nabla_a$ corresponds to a covariant derivative with respect to $x$, while $\nabla_{b'}$ corresponds to a covariant derivative with respect to $x'$. Derivatives with respect to $x$ and $x'$ commute, i.e.~$\nabla_a \nabla_{b'} T_{\cdots}(x,x') = \nabla_{b'} \nabla_a T_{\cdots}(x,x')$.

We will be especially interested in the limit $x' \to x$ of a bi-tensor.

\begin{definition}
The \emph{coincidence limit} of a bi-tensor $T_{\cdots}(x,x')$, denoted by
\begin{equation} \label{eq:coincidencelimit}
[T_{\cdots}](x) := \lim_{x' \to x} T_{\cdots}(x,x') \, ,
\end{equation}
is a tensor at $x$, when such a limit exists, and is independent of the path $x' \to x$. (In \eqref{eq:coincidencelimit}, the \emph{Synge's bracket notation} is employed.)
\end{definition}

We consider two important bi-tensors.

\subsection{Synge's world function}

Let $\gamma_X : [0,1] \to M$ be the geodesic segment starting at a point $x \in M$, with $X \in T_x M$ being the tangent vector to the geodesic at $x$. Let $V \subset T_x M$ be the set of tangent vectors $X$ such that $\gamma_X(t)$ is well defined for $t \in [0,1]$.

\begin{definition} \label{def:exponentialmap}
The \emph{exponential map} is the map $\exp_x : V \to M$, $X \mapsto \gamma_X(1)$.
\end{definition}

Hence, with $\gamma_X(0) = x$ and $\gamma_X(1) = x'$, one has $\exp_x(X) = x'$.

\begin{definition}
The \emph{Synge's world function} $\sigma(x,x')$ is a bi-scalar given by
\begin{equation}
\sigma(x,x') := \frac{1}{2} \, g(x) \left( \exp_x^{-1}(x'), \exp_x^{-1}(x') \right) \, .
\end{equation}
\end{definition}

The Synge's world function gives the half squared geodesic distance between the points $x$ and $x'$. To see this, consider the geodesic segment $\gamma : [0,1] \to M$ connecting $x = \gamma(0)$ and $x' = \gamma(1)$, which is unique, since $x$ and $x'$ are assumed to be in a geodesically convex neighbourhood. The \emph{geodesic distance} between $x$ and $x'$ is given by
\begin{equation}
d(x,x') := \int_0^1 \dd t \, \sqrt{g(\gamma(t)) \left( \dot{\gamma}(t), \dot{\gamma}(t) \right)}
= \sqrt{g(x) \left( \dot{\gamma}(0), \dot{\gamma}(0) \right)} = \sqrt{g_{ab} \dot{\gamma}^a \dot{\gamma} ^b} \, ,
\end{equation}
since the integrand is constant along the geodesic. One has that $\exp_x \left(\dot{\gamma}(0)\right) = x'$, thus, it follows that
\begin{equation} \label{eq:Syngesworldfunction}
\sigma(x,x') = \frac{1}{2} \, d(x,x')^2 = \frac{1}{2} g_{ab} \dot{\gamma}^a \dot{\gamma}^b = \frac{1}{2} g_{a'b'} \dot{\gamma}^{a'} \dot{\gamma}^{b'} \, .
\end{equation}

Given this, the coincidence limit of the Synge's world function is
\begin{equation}
[\sigma] = 0 \, .
\end{equation}

Concerning the covariant derivatives of the Synge's world function, one has
\begin{equation} \label{eq:covderivativesigma}
\nabla_a \sigma(x,x') = g_{ab} \dot{\gamma}^b \, , \qquad
\nabla_{a'} \sigma(x,x') = g_{a'b'} \dot{\gamma}^{b'} \, ,
\end{equation}
and, consequently,
\begin{equation}
[\nabla_a \sigma] = [\nabla_{a'} \sigma] = 0 \, .
\end{equation}
By similar calculations, we also get
\begin{equation}
[\nabla_a \nabla_b \sigma] =  [\nabla_{a'} \nabla_{b'} \sigma] = g_{ab} \, , \qquad
[\nabla_a \nabla_{b'} \sigma] =  [\nabla_{a'} \nabla_b \sigma] = - g_{ab} \, .
\end{equation}

It also follows from \eqref{eq:covderivativesigma} that
\begin{equation} \label{eq:sigmamunorm}
\nabla_a \sigma \nabla^a \sigma = \nabla_{a'} \sigma \nabla^{a'} \sigma = 2 \sigma \, .
\end{equation}
Hence, $\nabla^a \sigma$ and $\nabla^{a'} \sigma$ are tangent vectors to the geodesic $\gamma$ at $x$ and $x'$, respectively, with length equal to the geodesic distance between $x$ and $x'$.

\begin{remark}
In the literature, it is common to find any of the following notations for $\nabla_a \sigma$: $\nabla_a \sigma = \sigma_{;a} = \sigma_a$, the last of which omits the semi-colon for the covariant derivative. In the rest of thesis, the notation $\sigma_{;a}$ is used, in order to avoid confusion.
\end{remark}

\subsection{Parallel propagator}

If $V \in T_{x'} M$ is a tangent vector at $x'$, it can be parallel transported to $x$ along the unique geodesic that links $x$ and $x'$. The parallel transported vector at $x$ is given by
\begin{equation} \label{eq:parallelpropagator}
V^a(x) =: {g^a}_{b'}(x,x') \, V^{b'} (x') \, .
\end{equation}
This relation defines the \emph{parallel propagator} ${g^a}_{b'}(x,x')$. Similarly,
\begin{equation}
V^{a'}(x') = {g^{a'}}_b(x,x') \, V^b (x) \, .
\end{equation}

The coincidence limit of the parallel propagator is given by
\begin{equation}
\left[ {g^a}_{b'} \right](x) = {g^a}_b(x) = \delta^a_b \, .
\end{equation}

\subsection{Covariant and non-covariant Taylor expansion}
\label{sec:maths-bitensorexpansions}

Not only we will be interested in the coincidence limit of bi-tensors, but we will also be interested in expanding a bi-tensor near the coincidence limit as a power series. There are two possible Taylor expansions: the covariant Taylor expansion, in which the expansion is performed in a covariant way, and the non-covariant Taylor expansion, which is expressed in terms of the coordinate separation of the points.

In curved spacetimes, instead of using the flat spacetime quantity $(x-x')^a$, the expansion near the coincidence limit can be done in powers of $\sigma^{;a}$, whose length coincides with the geodesic distance between $x$ and $x'$. One then defines the covariant Taylor expansion as follows.

\begin{definition}
The \emph{covariant Taylor expansion} of a bi-tensor $T_{\cdots}(x,x')$ is
\begin{equation} \label{eq:covarianttaylorexpansion}
T_{\cdots}(x,x') = \sum_{k=0}^{\infty} \frac{(-1)^k}{k!} \, t_{\cdots a_1 \cdots a_k}(x) \, \sigma^{;a_1}(x,x') \cdots \sigma^{;a_k}(x,x') \, ,
\end{equation}
where the coefficients $t_{\cdots a_1 \cdots a_k}(x)$ are tensors at $x$.
\end{definition}

\begin{proposition}
The first expansion coefficients in \eqref{eq:covarianttaylorexpansion} are given by
\begin{align}
t_{\cdots} &= \left[ T_{\cdots} \right] \, , \\
t_{\cdots a_1} &= - \left[ T_{\cdots ; a_1} \right] + t_{\cdots ; a_1} \, , \\
t_{\cdots a_1 a_2} &= \left[ T_{\cdots ; a_1 a_2} \right] - t_{\cdots ; a_1 a_2} + t_{\cdots a_1 ; a_2} + t_{\cdots a_2 ; a_1} \, .
\end{align}
\end{proposition}

\begin{proof}
The coefficients are obtained by repeated covariant differentiation of \eqref{eq:covarianttaylorexpansion} and taking the coincidence limit.
\end{proof}

\begin{remark}
As with conventional Taylor expansions, it is not true in general that the sum in the RHS of \eqref{eq:covarianttaylorexpansion} converges and that, when it does, it is equal to the LHS. In this thesis, we will only need the first few terms of the expansion and, hence, will treat it as an asymptotic expansion. In this way, we will not go into the details of the convergence of the covariant Taylor series which we will deal with.
\end{remark}

Sometimes, however, it is not practical, or even possible, to compute $\sigma^{;a}$ for a given curved spacetime, which makes it impossible to obtain the covariant Taylor expansion of a bi-tensor. Another possibility is to explicitly use a coordinate system in a chart that includes both $x$ and $x'$ and perform the expansion in the coordinate separation of the points.

\begin{definition}
The \emph{non-covariant Taylor expansion} of a bi-tensor $T_{\cdots}(x,x')$ is
\begin{equation} \label{eq:noncovarianttaylorexpansion}
T_{\cdots}(x,x') = \sum_{k=0}^{\infty} \frac{(-1)^k}{k!} \, \hat{t}_{\cdots \mu_1 \cdots \mu_k}(x) \, \Delta x^{\mu_1} \cdots \Delta x^{\mu_k} \, ,
\end{equation}
where $\hat{t}_{\cdots \mu_1 \cdots \mu_k}(x)$ are the components of tensors at $x$ and $\Delta x^{\mu} := x^{\mu} - x^{\mu'}$.
\end{definition}

As an example, one can express the bi-scalar $\sigma(x,x')$ in terms of a non-covariant Taylor expansion.

\begin{proposition} \label{prop:noncovariantexpansionsigma}
The non-covariant Taylor expansion of $\sigma(x,x')$ is
\begin{align} \label{eq:sigma-coord}
\sigma &= \tilde{\sigma}_{\alpha \beta} \, \Delta x^\alpha \Delta x^\beta
	 + \tilde{\sigma}_{\alpha \beta \gamma} \, \Delta x^\alpha \Delta x^\beta \Delta x^\gamma
	 + \tilde{\sigma}_{\alpha \beta \gamma \delta} \, \Delta x^\alpha \Delta x^\beta \Delta x^\gamma \Delta x^\delta \notag \\
&\quad+ \tilde{\sigma}_{\alpha \beta \gamma \delta \epsilon} \, \Delta x^\alpha \Delta x^\beta \Delta x^\gamma \Delta x^\delta  \Delta x^\epsilon
	 + \ldots
\end{align}
where $\tilde{\sigma}_{\mu_1 \cdots \mu_k} := \frac{(-1)^k}{k!} \, \hat{\sigma}_{\mu_1 \cdots \mu_k}$ are given by
\begin{subequations}
\begin{align}
\tilde{\sigma}_{\alpha \beta} &= \frac{1}{2} \, g_{\alpha \beta} \, , \\
\tilde{\sigma}_{\alpha \beta \gamma} &= - \frac{1}{4} \, g_{(\alpha \beta , \gamma)} \, , \\
\tilde{\sigma}_{\alpha \beta \gamma \delta} 
&= - \frac{1}{3} \left[ \tilde{\sigma}_{(\alpha \beta \gamma , \delta)} 
+ g^{\mu \nu} \left(\frac{1}{8} g_{(\alpha \beta ,|\mu|} g_{\gamma \delta) ,\nu}
			+ \frac{3}{2} g_{(\alpha \beta , |\mu|} \tilde{\sigma}_{|\nu| \gamma \delta)}
			+ \frac{9}{2} \tilde{\sigma}_{\mu (\alpha \beta} \tilde{\sigma}_{|\nu| \gamma \delta)}\right) \right] \, , \\
\tilde{\sigma}_{\alpha \beta \gamma \delta \epsilon}	
&= -\frac{1}{4} \Bigg[ \tilde{\sigma}_{(\alpha \beta \gamma \delta , \epsilon)} 
+ g^{\mu \nu} \Bigg(12 \tilde{\sigma}_{\mu (\alpha \beta} \tilde{\sigma}_{|\nu| \gamma \delta \epsilon)}
			+ 3 \tilde{\sigma}_{\alpha (\alpha \beta} \tilde{\sigma}_{\gamma \delta \epsilon) ,\beta}
			+ 2 g_{(\alpha \beta , |\mu|} \tilde{\sigma}_{|\nu| \gamma \delta \epsilon)} \notag \\
			&\hspace*{24ex} + \frac{1}{2} \tilde{\sigma}_{(\alpha \beta \gamma , |\mu|} g_{\delta \epsilon) ,\nu}\Bigg)\Bigg] \, .
\end{align}
\end{subequations}
\end{proposition}

\begin{proof}
One has that
\begin{align}
\label{eq:sigmaderiv-coord}
\sigma_{;\mu} &= 2 \tilde{\sigma}_{\mu \alpha} \Delta x^\alpha
	 + \left( \tilde{\sigma}_{\alpha \beta ,\mu} + 3 \tilde{\sigma}_{\mu \alpha \beta} \right)\Delta x^\alpha \Delta x^\beta
	 + (\tilde{\sigma}_{\alpha \beta \gamma ,\mu} + 4 \tilde{\sigma}_{\mu \alpha \beta \gamma}) \Delta x^\alpha \Delta x^\beta \Delta x^\gamma \notag \\
	&\quad + (\tilde{\sigma}_{\alpha \beta \gamma \delta ,\mu} + 5 \tilde{\sigma}_{\mu \alpha \beta \gamma \delta} )\Delta x^\alpha \Delta x^\beta \Delta x^\gamma \Delta x^\delta
	 + \ldots .
\end{align}
Substituting expansions (\ref{eq:sigma-coord}) and (\ref{eq:sigmaderiv-coord}) into \eqref{eq:sigmamunorm} and equating powers of $\Delta x^\alpha$, we get expressions for each coefficient in terms of the lower order coefficients. 
\end{proof}

Proposition~\ref{prop:noncovariantexpansionsigma} allows us to relate the coefficients of the covariant Taylor expansion \eqref{eq:covarianttaylorexpansion} and the coefficients of the non-covariant Taylor expansion \eqref{eq:noncovarianttaylorexpansion}.

\begin{proposition}
The first coefficients $\hat{t}_{\cdots \mu_1 \cdots \mu_k}$ of the non-covariant Taylor expansion \eqref{eq:noncovarianttaylorexpansion} can be expressed in terms of the coefficients $t_{\cdots \mu_1 \cdots \mu_k}$ of the covariant Taylor expansion \eqref{eq:covarianttaylorexpansion} as
\begin{subequations}
\begin{align}
\hat{t} &= t \, , \\
\hat{t}_{\alpha} &= t_{\alpha} \, , \\
\hat{t}_{\alpha\beta} &= t_{\alpha\beta} + t_{\mu} \Gamma^{\mu}_{\alpha \beta} \, , \\
\hat{t}_{\alpha\beta\gamma} &= t_{\alpha\beta\gamma} + 3 t_{(\alpha|\mu|} \Gamma^{\mu}_{\beta \gamma)} + 6 t_{\mu} \left( {\tilde{\sigma}_{(\alpha \beta \gamma)}}{}^{,\mu} + 4 {\tilde{\sigma}^{\mu}}{}_{\alpha \beta \gamma} \right) \, . 
\end{align}
\end{subequations}
\end{proposition}

\begin{proof}
The relations are obtained by substituting \eqref{eq:sigma-coord} into \eqref{eq:covarianttaylorexpansion} and equating the resulting expansion with \eqref{eq:noncovarianttaylorexpansion}.
\end{proof}

\subsection{The case of complex Riemannian manifolds}
\label{sec:nbhdcomplexRiemannian}

In this thesis, we will need to consider bi-tensors in complex Riemannian manifolds, which are obtained from real Lorentzian manifolds, as described in Section~\ref{sec:riemanniansections}. Here, we verify that the local geodesic structure of the Lorentzian manifold is preserved when going to the complex Riemannian section and, therefore, it is possible to generalise the concepts above. See \cite{Moretti:1999fb} for a detailed discussion.

Consider a complex Riemannian manifold $M^{\mathbb{C}}$ with metric $g^{\mathbb{C}}$, which was obtained from a real Lorentzian manifold $M$ with real analytic metric $g$, such that the metric component $g_{00} < 0$ and the inverse $g^{00} < 0$ in a coordinate system. The geodesic equations admit locally a unique solution with parameter $t \in \mathbb{C}$ satisfying given initial conditions. If we restrict $t$ to the real domain, $t \in \mathbb{R}$, we obtain a \emph{real-parameter geodesic segment} (the corresponding complex-parameter geodesic segment is obtained by analytical continuation).

We want to define an analogous notion of a geodesically convex neighbourhood introduced in Definition~\ref{def:geodconvexneighbourhood} which is valid for the the complex Riemannian manifold. For that, we need a series of intermediate definitions. Let $\gamma_X(t)$, $t \in [0,1]$, be the real-parameter geodesic segment starting at a point $p \in M^{\mathbb{C}}$, with $X \in T_p(M^{\mathbb{C}})$ being the tangent vector to the geodesic at $p$. Let $V \subset T_p(M^{\mathbb{C}})$ be the set of vectors $X$ such that $\gamma_X(t)$ is well defined for $t \in [0,1]$. The exponential map is defined as in Definition~\ref{def:exponentialmap} as the map $\exp_p : V \to M^{\mathbb{C}}$, $X \mapsto \gamma_X(1)$.

\begin{definition}
An open \emph{star-shaped neighbourhood} about $0$ of a vector space is such that, if $X$ belongs to the neighbourhood, then $\lambda X$, with $\lambda \in [0,1]$, also belongs to the neighbourhood.
\end{definition}

\begin{definition}
A \emph{normal neighbourhood} of $p \in M^{\mathbb{C}}$ is an open neighbourhood of $p$ with the form $N_p = \exp_p(S)$, with $S \subset V \subset T_p(M^{\mathbb{C}})$ an open star-shaped neighbourhood of $0 \in T_p(M^{\mathbb{C}})$.
\end{definition}

\begin{definition}
A \emph{totally normal neighbourhood} of $p \in M^{\mathbb{C}}$ is a neighbourhood of $p$, $O_p \subset M^{\mathbb{C}}$, such that, if $q \in O_p$, there is a normal neighbourhood of $q$, $N_q$, with $O_p \subset N_q$.
\end{definition}

We can now define the desired class of neighbourhoods.

\begin{definition} \label{def:convexnbhdcsection}
A \emph{geodesically linearly convex neighbourhood} of $p \in M^{\mathbb{C}}$ is a totally normal neighbourhood of $p$, $N_p \subset M^{\mathbb{C}}$, such that, for any $q, q' \in N_p$, there is only one real-parameter geodesic segment which links $q$ and $q'$ and which lies completely in $N_p$.
\end{definition}

\begin{proposition}
Given a complex Riemannian manifold with the properties described above, for any given point, there is always a geodesically linearly convex neighbourhood.
\end{proposition}

\begin{proof}
See Theorem 23 of \cite{Moretti:1999fb}.
\end{proof}

Therefore, we can extend the Synge's world function bi-scalar to a complex Riemannian manifold.

\begin{definition} \label{def:SyngesworldfunctionCsection}
Given a geodesically linearly convex neighbourhood $N \subset M^{\mathbb{C}}$, the \emph{complex Synge's world function} $\sigma(x,x')$ is given by
\begin{equation}
\sigma (x, x') := \frac{1}{2} g(x) \left( \exp_x^{-1} (x'), \exp_x^{-1} (x') \right) \, .
\end{equation}
\end{definition}

This reduces to the usual definition for real Riemannian and Lorentzian manifolds. In particular, suppose we choose $x$ and $x'$ in a way such that some of their coordinates in a given coordinate system are the same and the induced metric on the submanifold defined by this condition is either real Riemannian or Lorentzian. Then, we can use the usual definition as half of the square of the geodesic distance between $x$ and $x'$.


\section{Symplectic and Hilbert spaces}
\label{sec:qftcst-symplecticspaces}

A classical field theory on a globally hyperbolic spacetime is completely specified in terms of functions with values on an appropriate vector space, together with
\begin{enumerate}[label={(\roman*)}]
\item a non-degenerate bilinear or sesquilinear form, which specifies the kinematics; 
\item a partial differential operator, which specifies the dynamics.
\end{enumerate}
In this section, we will briefly describe the basics of point (i), leaving point (ii) to Section~\ref{sec:hyperbolicoperators}.

The main concern of this thesis will be with a scalar field on a curved spacetime. Classically, the phase spaces of neutral bosons (resp., charged bosons) are symplectic spaces (resp., charged symplectic spaces). We characterise these spaces below. We also give a brief characterisation of Hilbert spaces, which are crucial to the quantum field theory.

Much of this section follows closely \cite{dereziński2013mathematics}. More details on Hilbert spaces can also be found in \cite{reed1980methods}.

\subsection{Symplectic spaces}

Let $V$ be a vector space over the field $\mathbb{K}$ ($\mathbb{R}$ or $\mathbb{C}$).

\begin{definition} \label{def:bilinearform}
A \emph{bilinear form} on $V$ is a bilinear map $\nu : V \times V \to \mathbb{K}$,
\begin{equation}
(v_1, v_2) \mapsto \nu(v_1, v_2) \, .
\end{equation} 
\end{definition}

\begin{definition}
A bilinear form $\nu$ is \emph{non-degenerate} if $\Ker \nu = \{ 0 \}$.
\end{definition}

\begin{definition}
A bilinear form $\nu$ is \emph{symmetric} if
\begin{equation}
\nu(v_1, v_2) = \nu(v_2, v_1) \, , \qquad
v_1, v_2 \in V \, ,
\end{equation}
whereas it is \emph{anti-symmetric} if
\begin{equation}
\nu(v_1, v_2) = - \nu(v_2, v_1) \, , \qquad
v_1, v_2 \in V \, .
\end{equation}
\end{definition}

\begin{definition}
A symmetric form $\nu$ on a real vector space is \emph{positive definite} if $\nu(v, v) > 0$ for $v \neq 0$.
\end{definition}

\begin{remark}
A positive definite symmetric form is always non-degenerate.
\end{remark}

\begin{definition}
A non-degenerate anti-symmetric bilinear form is called a \emph{symplectic form}.
\end{definition}

\begin{definition} \label{def:symplecticspace}
The pair $(V, \nu)$, where $V$ is a vector space over $\mathbb{K}$ and $\nu$ is a symplectic form, is called a \emph{symplectic space}.
\end{definition}

\subsection{Charged symplectic spaces}

Let $V$ be a vector space over $\mathbb{C}$.

\begin{definition}
A \emph{sesquilinear form} on $V$ is a map $\beta : V \times V \to \mathbb{C}$,
\begin{equation}
(v_1, v_2) \mapsto \beta(v_1, v_2)
\end{equation} 
which is anti-linear in the first argument and linear in the second argument.
\end{definition}

\begin{remark}
In the definition above, the so-called ``physicist's convention'' was adopted. In the ``mathematician's convention'', the first argument would be linear, while the second argument would be anti-linear.
\end{remark}

\begin{definition}
A sesquilinear form $\beta$ is \emph{non-degenerate} if $\Ker \beta = \{ 0 \}$.
\end{definition}

\begin{definition}
A sesquilinear form $\beta$ is \emph{Hermitian} if
\begin{equation}
\overline{\beta(v_1, v_2)} = \beta(v_2, v_1) \, , \qquad
v_1, v_2 \in V \, ,
\end{equation}
whereas it is \emph{anti-Hermitian} if
\begin{equation}
\overline{\beta(v_1, v_2)} = - \beta(v_2, v_1) \, , \qquad
v_1, v_2 \in V \, .
\end{equation}
\end{definition}

\begin{remark} \label{rem:HermitianandantiHermitian}
If $\beta$ is Hermitian, then $i \beta$ is anti-Hermitian.
\end{remark}

\begin{definition} \label{def:scalarproduct}
A Hermitian form $\beta$ is \emph{positive definite} if $\beta(v, v) > 0$ for $v \neq 0$. In this case, it is also known as a \emph{scalar product} or an \emph{inner product} and denoted by $\langle \cdot | \cdot \rangle$.
\end{definition}

\begin{remark}
A positive definite Hermitian form is always non-degenerate.
\end{remark}

\begin{definition}
A non-degenerate anti-Hermitian form is called a \emph{charged symplectic form}.
\end{definition}

\begin{definition}
The pair $(V, \beta)$, where $V$ is a vector space over $\mathbb{C}$ and $\beta$ is a charged symplectic form, is called a \emph{charged symplectic space}.
\end{definition}

\subsection{Hilbert spaces}

The notion of a charged symplectic space $(V, \beta)$ is very closely related to the space $(V, i \beta)$, cf.~Remark~\ref{rem:HermitianandantiHermitian}, where $i \beta$ is an Hermitian form. If this form is positive definite, the space is an inner product space.

\begin{definition}
The pair $(V, \beta)$, where $V$ is a vector space over $\mathbb{C}$ and $\beta$ is a non-degenerate Hermitian form, is called a \emph{pseudo-unitary space}. If $\beta$ is positive definite, then $(V, \beta)$ is called a \emph{unitary space} or an \emph{inner product space}.
\end{definition}

\begin{definition}
An inner product space $(\mathscr{H}, \beta)$ is called a \emph{Hilbert space} if $\mathscr{H}$ is complete in the norm induced by the inner product $\beta$.
\end{definition}

\begin{remark}
All finite-dimensional inner product spaces are Hilbert spaces. If an infinite-dimensional inner product $V$ space fails to be complete, there exists a unique Hilbert space $\mathscr{H}$ such that $V$ is isomorphic to a dense subspace of $\mathscr{H}$. The Hilbert space $\mathscr{H}$ is called the \emph{Hilbert space completion} of $V$.
\end{remark}

\begin{remark}
From now on, we denote a Hilbert space $(\mathscr{H}, \beta)$ simply by $\mathscr{H}$ and write the endowed scalar product $\beta ( \cdot , \cdot) = \langle \cdot | \cdot \rangle$.
\end{remark}

Next, we describe the direct sum and the tensor product of Hilbert spaces.

\begin{definition}
The \emph{direct sum} of the Hilbert spaces $\mathscr{H}_1$ and $\mathscr{H}_2$ is the Hilbert space $\mathscr{H}_1 \oplus \mathscr{H}_2$ consisting of pairs $(u,v)$ with $u \in \mathscr{H}_1$ and $v \in \mathscr{H}_2$ and scalar product
\begin{equation} \label{eq:directsumscalarproduct}
\left\langle (u_1, v_1) | (u_2,v_2) \right\rangle := \langle u_1 | u_2\rangle_{\mathscr{H}_1} + \langle v_1 | v_2 \rangle_{\mathscr{H}_2} \, .
\end{equation}
To construct countable direct sums, let $\{\mathscr{H}_n\}_{n=1}^{\infty}$ be a sequence of Hilbert spaces and let $\mathscr{H}$ denote the set of sequences $\{u_n\}_{n=1}^{\infty}$, with $u_n \in \mathscr{H}_n$, which satisfy
\begin{equation}
\sum_{n=1}^{\infty} \langle u_n | u_n \rangle_{\mathscr{H}_n} < \infty \, .
\end{equation}
Then, $\mathscr{H}$ is a Hilbert space with the scalar product which is the natural generalisation of \eqref{eq:directsumscalarproduct}. $\mathscr{H}$ is denoted by
\begin{equation}
\mathscr{H} = \bigoplus_{n=1}^{\infty} \mathscr{H}_n \, .
\end{equation}
\end{definition}

To define the tensor product of two Hilbert spaces $\mathscr{H}_1$ and $\mathscr{H}_2$, first consider, for each $u_1 \in \mathscr{H}_1$ and $u_2 \in \mathscr{H}_2$, the bi-antilinear form $u_1 \otimes u_2 : \mathscr{H}_1 \times \mathscr{H}_2 \to \mathbb{C}$,
\begin{equation}
(u_1 \otimes u_2)(v_1,v_2) := \langle v_1 | u_1 \rangle \langle v_2 | u_2 \rangle \, .
\end{equation}
Let $\mathscr{S}$ denote the set of finite linear combinations of these bi-antilinear forms. Define a scalar product
\begin{equation}
\langle u_1 \otimes u_2 | v_1 \otimes v_2 \rangle := \langle u_1 | v_1 \rangle \langle u_2 | v_2 \rangle \, ,
\label{eq:tensorproductscalarproduct}
\end{equation}
and extend it by linearity to $\mathscr{S}$. $\mathscr{S}$ is now an inner product space.

\begin{definition}
The \emph{tensor product} $\mathscr{H}_1 \otimes \mathscr{H}_2$ of the Hilbert spaces $\mathscr{H}_1$ and $\mathscr{H}_2$ is defined as the completion of $\mathscr{S}$ under the scalar product defined in \eqref{eq:tensorproductscalarproduct}. By induction, the above construction can be extended to define the tensor product $$\bigotimes_{n=1}^{m} \mathscr{H}_n$$ of finitely many Hilbert spaces $\mathscr{H}_1, \, ... , \, \mathscr{H}_m$.
\end{definition}

An important application of the concepts of direct sum and tensor product is the definition of the Fock space.

\begin{definition} \label{def:Fockspace}
The \emph{Fock space} $\mathscr{F}(\mathscr{H})$ associated with a Hilbert space $\mathscr{H}$ is the Hilbert space
\begin{equation}
\mathscr{F}(\mathscr{H}) := \bigoplus_{n=0}^{\infty} \left( {\bigotimes^n} \mathscr{H} \right) \, ,
\end{equation}
where $\bigotimes^0 \mathscr{H} := \mathbb{C}$. The \emph{symmetric Fock space} $\mathscr{F}_{\rm s}(\mathscr{H})$ associated with a Hilbert space $\mathscr{H}$ is the subspace of $\mathscr{F}(\mathscr{H})$ defined by
\begin{equation}
\mathscr{F}_{\rm s}(\mathscr{H}) := \bigoplus_{n=0}^{\infty} \left( {\bigotimes^n}_{\rm s} \mathscr{H} \right) \, ,
\end{equation}
whereas the \emph{anti-symmetric Fock space} $\mathscr{F}_{\rm a}(\mathscr{H})$ associated with a Hilbert space $\mathscr{H}$ is the subspace of $\mathscr{F}(\mathscr{H})$ defined by
\begin{equation}
\mathscr{F}_{\rm a}(\mathscr{H}) := \bigoplus_{n=0}^{\infty} \left( {\bigotimes^n}_{\rm a} \mathscr{H} \right) \, .
\end{equation}
Here, $\bigotimes_{\rm s}$ and $\bigotimes_{\rm a}$ stand for the symmetrised and anti-symmetrised tensor product, respectively.
\end{definition}

Finally, we introduce the notion of orthonormal decomposition of a Hilbert space.

\begin{definition} \label{def:orthonormaldecomposition}
Let $\mathscr{H}$ be a Hilbert space. If $\mathscr{U} \subset \mathscr{H}$, then $\mathscr{U}^{\perp}$ denotes the \emph{orthogonal complement} of $\mathscr{U}$,
\begin{equation}
\mathscr{U}^{\perp} := \left\{ v \in \mathscr{H} : \langle u | v \rangle = 0, \, \text{for all} \, u \in \mathscr{U}  \right\} \, .
\end{equation}
\end{definition}

\begin{theorem} \label{thm:projectiontheorem}
Let $\mathscr{H}$ be a Hilbert space. If $\mathscr{U}$ is a closed vector subspace of $\mathscr{H}$, then $(\mathscr{U}^{\perp})^{\perp} = \mathscr{U}$ and, for any $v \in \mathscr{H}$, there exists a unique $u \in \mathscr{U}$ and $u' \in \mathscr{U}^{\perp}$ such that $v = u + u'$.
\end{theorem}

\begin{proof}
See e.g.~Theorem II.3 of \cite{reed1980methods}.
\end{proof}

\begin{remark}
Theorem~\ref{thm:projectiontheorem} is usually known as the \emph{projection theorem}. This provides a natural isomorphism between $\mathscr{U} \oplus \mathscr{U}^{\perp}$ and $\mathscr{H}$ given by $(u, u') \mapsto u + u'$. For simplicity, one writes $\mathscr{H} = \mathscr{U} \oplus \mathscr{U}^{\perp}$, the \emph{orthonormal decomposition} of $\mathscr{H}$.
\end{remark}

\vspace*{1ex}

\subsection{Operators on Hilbert spaces}

\begin{definition}
Let $\mathscr{H}_1$ and $\mathscr{H}_2$ be Hilbert spaces. An \emph{operator} $T : \mathscr{H}_1 \to \mathscr{H}_2$ is a linear map from a linear subspace $D(T) \subset \mathscr{H}_1$ to $\mathscr{H}_2$. The subspace $D(T)$ is called the \emph{domain} of the operator $T$ and we assume that it is dense in $\mathscr{H}_1$.
\end{definition}

\begin{definition}
Let $T : \mathscr{H}_1 \to \mathscr{H}_2$ be a densely defined operator. Let
\begin{equation}
D(T^{\dagger}) := \left\{ v \in \mathscr{H}_2 : \forall u \in D(T) \; \exists w \in \mathscr{H}_1 \; \langle v | T u \rangle_{\mathscr{H}_2} = \langle w | u \rangle_{\mathscr{H}_1} \right\} \, .
\end{equation}
For each such $v \in D(T^{\dagger})$, one can define the \emph{adjoint} operator $T^{\dagger}$ by $T^{\dagger} v := w$, i.e.
\begin{equation}
\langle v | T u \rangle_{\mathscr{H}_2} = \langle T^{\dagger} v | u \rangle_{\mathscr{H}_1}
\end{equation}
for all $u \in D(T)$.
\end{definition}

\begin{definition}
An operator $T : \mathscr{H} \to \mathscr{H}$ is called \emph{symmetric} or \emph{Hermitian} if $T \subset T^{\dagger}$, i.e.~if $D(T) \subset D(T^{\dagger})$ and $T u = T^{\dagger} u$, for all $u \in D(T)$, or equivalently if
\begin{equation}
\langle T u | v \rangle = \langle u | T v \rangle \, , \qquad 
\text{for all $u, \, v \in D(T)$.}
\end{equation}
\end{definition}

\begin{definition}
An operator $T : \mathscr{H} \to \mathscr{H}$ is called \emph{self-adjoint} if $T = T^{\dagger}$, i.e.~if $D(T) = D(T^{\dagger})$ and $T$ is Hermitian.
\end{definition}

\begin{definition}
The \emph{spectrum} of an operator $T : D(T) \subset \mathscr{H} \to \mathscr{H}$ is the set
\begin{equation}
\sigma(T) := \left\{ \lambda \in \mathbb{C} : T - \lambda \mathbb{I} \text{ is not a bijection with bounded inverse} \right\} \, .
\end{equation}
The spectrum can be decomposed into three disjoint sets.
\begin{enumerate}[label={(\roman*)}]
\item The \emph{point spectrum} is the set
\begin{equation}
\sigma_{\rm p}(T) := \left\{ \lambda \in \sigma(T) : T - \lambda \mathbb{I} \text{ is not injective} \right\} \, .
\end{equation}
Equivalently, $\lambda \in \sigma_{\rm p}(T)$ if there exists non-zero $u \in D(T)$ such that $T u = \lambda u$. $\lambda$ is called an \emph{eigenvalue} and $u$ is called an \emph{eigenvector}.
\item The \emph{continuous spectrum} is the set
\begin{align}
\sigma_{\rm c}(T) &:= \left\{ \lambda \in \sigma(T) : T - \lambda \mathbb{I} \text{ is not surjective and} \right. \notag \\ 
&\quad\;\; \left. \text{ $(T - \lambda \mathbb{I}) D(T)$ is dense on $D(T)$} \right\} \, .
\end{align}
\item The \emph{residual spectrum} is the set
\begin{align}
\sigma_{\rm r}(T) &:= \left\{ \lambda \in \sigma(T) : T - \lambda \mathbb{I} \text{ is not surjective and} \right. \notag \\ 
&\quad\;\; \left. \text{ $(T - \lambda \mathbb{I}) D(T)$ is not dense on $D(T)$} \right\} \, .
\end{align}
\end{enumerate}
\end{definition}

\begin{remark}
If $T$ is an operator acting on a finite-dimensional space, then the continuous and the residual spectrum of the operator are empty and its spectrum consists only of eigenvalues.
\end{remark}

If the spectrum for an operator is known, then the spectrum for its adjoint can be easily obtained.

\begin{proposition} \label{prop:spectrumadjoint}
If $T : D(T) \subset \mathscr{H} \to \mathscr{H}$ has spectrum $\sigma(T)$, then
\begin{equation}
\sigma(T^{\dagger}) = \{ \lambda : \overline{\lambda} \in \sigma(T) \} \, .
\end{equation}
\end{proposition}

In the cases in which the operators are Hermitian or self-adjoint, one can say more about their spectrum.

\begin{proposition}
If $T : D(T) \subset \mathscr{H} \to \mathscr{H}$ is Hermitian, then
\begin{enumerate}[label={(\roman*)}]
\item all eigenvalues of $T$ are real;
\item eigenvectors of $T$ corresponding to distinct eigenvalues are orthogonal;
\item the continuous spectrum of $T$ is real.
\end{enumerate}
If furthermore $T$ is self-adjoint, then 
\begin{enumerate}[label={(\roman*)}] \setcounter{enumi}{3}
\item the residual spectrum is empty.
\end{enumerate}
\end{proposition}

\begin{remark} \label{rem:antihermitianoperator}
Note that if $T$ is anti-Hermitian, i.e.~$T u = - T^{\dagger} u$ for all $u \in D(T)$, then it follows from Proposition~\ref{prop:spectrumadjoint} that all eigenvalues of $T$ are purely imaginary. It is also true that eigenvectors of $T$ corresponding to distinct eigenvalues are orthogonal.
\end{remark}

As seen above, not even all self-adjoint operators have a spectrum composed of only eigenvalues, as in finite-dimensional Hilbert spaces. However, there is a class of operators which enjoys this property.

\begin{definition}
An operator $T : \mathscr{H}_1 \to \mathscr{H}_2$ is called \emph{compact} if it takes bounded subsets of $\mathscr{H}_1$ into subsets of $\mathscr{H}_2$ whose closure is compact.
\end{definition}

\begin{proposition}
If $T : \mathscr{H} \to \mathscr{H}$ is a compact operator, then, except for the possible value 0, the spectrum of $T$ is entirely point spectrum.
\end{proposition}

Finally, we generalise the notion of trace of a matrix to the trace of an operator.

\begin{definition} \label{def:traceoperator}
Let $\mathscr{H}$ be a separable Hilbert space (i.e.~$\mathscr{H}$ contains a countable dense subset) and $\{ u_i \}_{i \in \mathscr{I}}$ be an orthonormal basis, where $\mathscr{I}$ is an index set. The \emph{trace} of a positive operator $T : \mathscr{H} \to \mathscr{H}$ is defined as
\begin{equation} \label{eq:traceoperator}
\Tr T := \sum_{i \in \mathscr{I}} \langle u_i | T u_i \rangle \, .
\end{equation}
The trace is independent of the orthonormal basis chosen. The operator $T$ is of \emph{trace class} if $\Tr \sqrt{T^{\dagger} T} < \infty$.
\end{definition}

The trace of a compact operator of trace class always exists.

\begin{proposition}
Let $T : \mathscr{H} \to \mathscr{H}$ be a compact operator of trace class. Then, the sum on the RHS of \eqref{eq:traceoperator} converges absolutely.
\end{proposition}

In the case of a compact positive operator, the trace is just the sum of all eigenvalues.

\begin{proposition} \label{prop:compactpositiveoperator}
Let $T : \mathscr{H} \to \mathscr{H}$ be a compact positive operator with non-zero eigenvalues $\lambda_i$, $i \in \mathbb{N}$. Then, its trace is given by
\begin{equation}
\Tr T = \sum_{i \in \mathbb{N}} \lambda_i \, .
\end{equation}
\end{proposition}


\section{Complexification of real vector spaces}

The procedure of quantising a classical field theory characterised by a real vector space with some additional structure (such as a symplectic structure) involves the complexification of the real vector space. Here, we present a very brief description of this procedure, which can be found e.g. in \cite{roman1992advanced}.

Let $V$ be a complex vector space, i.e.~a vector space over $\mathbb{C}$. If one restricts the scalars to be real, the resulting vector space $V_{\mathbb{R}}$ is a real vector space and is called the \emph{real form} of $V$.

Conversely, to each real vector space $V$, i.e.~a vector space over $\mathbb{R}$, one can associate a complex vector space $V^{\mathbb{C}}$.

\begin{definition} \label{def:complexification}
Let $V$ be a real vector space. The \emph{complexification} of $V$ is the complex vector space $V^{\mathbb{C}} := V \oplus V$ of ordered pairs, with
\begin{enumerate}[label={(\roman*)}]
\item addition
\begin{equation}
(v_1, v_2) + (w_1, w_2) = (v_1 + w_1, v_2 + w_2) \, ,
\end{equation}
\item scalar multiplication over $\mathbb{C}$ defined by
\begin{equation}
(x+iy) (v_1, v_2) = (x v_1 - y v_2, x v_2 + y v_1) \, ,
\end{equation}
\end{enumerate}
where $x, \, y \in \mathbb{R}$ and $v_1, \, v_2, \, w_1, \, w_2 \in V$.
\end{definition}

It is convenient to introduce the notation $v + i w$ for $(v,w) \in V^{\mathbb{C}}$, such that one can regard the complexification of $V$ as
\begin{equation}
V^{\mathbb{C}} = V \oplus i V = \left\{ v + i w : v, \, w \in V \right\} \, .
\end{equation}
Addition now resembles addition of complex numbers,
\begin{equation}
(v_1 + i v_2) + (w_1 + i w_2) = (v_1 + w_1) + i (v_2 + w_2) \, ,
\end{equation}
and the scalar multiplication resembles multiplication of complex numbers,
\begin{equation}
(x + iy) (v_1 + i v_2) = (x v_1 - y v_2) + i (x v_2 + y v_1) \, .
\end{equation}

Now consider a real vector space $V$ endowed with a bilinear form $\nu$, cf.~Definition~\ref{def:bilinearform}. The complexified space $V^{\mathbb{C}}$ can be endowed with a natural sesquilinear form $\nu_{\mathbb{C}}$.

\begin{definition} \label{def:canonicalsesquilinearextension}
Given a real vector space $V$ endowed with a bilinear form $\nu : V \times V \to \mathbb{R}$, the \emph{canonical sesquilinear extension} of $\nu$ to the complexified vector space $V^{\mathbb{C}}$ is the sesquilinear form $\nu_{\mathbb{C}} : V^{\mathbb{C}} \times V^{\mathbb{C}} \to \mathbb{C}$ defined by
\begin{equation}
\nu_{\mathbb{C}}(v_1 + i v_2, w_1 + i w_2) := \nu(v_1,w_1) + \nu(v_2,w_2) + i \left[ \nu(v_1,w_2) - \nu(w_1,v_2) \right] \, .
\end{equation}
\end{definition}

This extension maps (anti-)symmetric bilinear forms on $V$ to (anti-)Hermitian sesquilinear forms on $V^{\mathbb{C}}$.


\section{Distributions}
\label{sec:maths-distributions}

In this section, we present a brief overview of theory of distributions (or generalised functions). A more complete discussion can be found e.g.~in \cite{hormander1990analysis}. All the mathematical objects are assumed to be defined in an open set $U \subset \mathbb{R}^d$, which can be thought as a chart on the manifold $M$.

\begin{definition}
A \emph{distribution} is a continuous linear functional on $C_0^{\infty}(U)$, i.e.~a mapping of the form $f \mapsto \Phi(f)$, with $f \in C_0^{\infty}(U)$, such that, for every compact set $K \subset U$ there exists constants $C$ and $k$ such that 
\begin{equation}
\left| \Phi(f) \right| \leq C \sum_{|\alpha| \leq k} \sup \left| \partial^{\alpha} f \right| \, .
\end{equation}
Here, $\alpha = (\alpha_1, ..., \alpha_{d}) \in \mathbb{N}_0^d$ is a multi-index, $|\alpha| := \sum_{\mu=1}^{d} \alpha_{\mu}$, $\partial^{\alpha} := \prod_{\mu=1}^{d} \partial_{\mu}^{\alpha_{\mu}}$ and $\partial_{\mu}$ is the partial derivative with respect to $x^{\mu}$. The space of all distributions on $U$ is denoted by $C_0^{\infty}(U)^*$, the dual of $C_0^{\infty}(U)$.
\end{definition}

\begin{remark} \label{rem:distributionnotation}
Any locally integrable function $\Phi \in L^1_{\rm loc}(U)$ can be identified with a distribution by
\begin{equation} \label{eq:distributioneq}
f \mapsto \Phi(f) = \int_U \dd^d x \, f(x) \Phi(x) \, , \qquad f \in C_0^{\infty}(U) \, ,
\end{equation}
A distribution of this form is called a \emph{regular distribution}. However, not all distributions can be represented in this way. The most famous example is the Dirac delta distributions, $f \mapsto \delta(f)$, defined by
\begin{equation}
\delta(f) = f(0) \, , \qquad f \in C_0^{\infty}(U) \, .
\end{equation}
This is an example of a \emph{singular distribution}. However, it is useful to continue to represent distributions as in \eqref{eq:distributioneq}, so we pretend that there exists an ``object'' $\delta(x)$ such that
\begin{equation}
f \mapsto \delta(f) = \int_U \dd^d x \, f(x) \delta(x) \, , \qquad f \in C_0^{\infty}(U) \, .
\end{equation}
The ``object'' $\delta(x)$ is not a function and cannot be evaluated pointwise! It should only be thought as convenient notation, which allows us to use the language of ordinary functions when referring to distributions.
\end{remark}

\begin{remark} \label{rem:notationdistributions}
Another convenient notation for a distribution $\Phi \in C_0^{\infty}(U)^*$ is
\begin{equation}
f \mapsto \Phi(f) = (\Phi, f) \, , \qquad f \in C_0^{\infty}(U) \, .
\end{equation}
This pairing of the distribution $\Phi$ with the compactly supported function $f$ takes the same form as the scalar product in the Lebesgue space $L^2(U)$.
\end{remark}

\begin{definition}
If $\Phi \in C_0^{\infty}(U)^*$, then the \emph{support} of $\Phi$, denoted $\supp \Phi$, is the smallest closed subset $V \subset U$ such that $\Phi|_{U \setminus V} = 0$.
\end{definition}

\begin{remark}
If $\Phi \in L^1_{\rm loc}(U)$, the expression
\begin{equation}
\Phi(f) = \int_U \dd^d x \, f(x) \Phi(x) \, , \qquad f \in C_0^{\infty}(U) \, .
\end{equation}
is well defined for any $f \in C^{\infty}(U)$ if $\supp \Phi \cap \supp f$ has compact closure and is contained in $U$. It can furthermore be shown that the space of distributions in $U$ with compact support is the dual space of $C^{\infty}(U)$.
\end{remark}

\begin{definition} \label{def:distributionscompactsupp}
The space of distributions in $U$ with compact support is denoted by $C^{\infty}(U)^*$, the dual of $C^{\infty}(U)$.
\end{definition}

\begin{remark}
One has the following inclusions
\begin{gather}
C_0^{\infty}(U) \subset C^{\infty}(U)^* \subset C_0^{\infty}(U)^* \, , \\
C^{\infty}(U) \subset C_0^{\infty}(U)^* \, .
\end{gather}
It can further be shown that $C_0^{\infty}(U)$ is dense in $C^{\infty}(U)^*$ and $C_0^{\infty}(U)^*$ (see \cite{hormander1990analysis} for more details).
\end{remark}

Finally, we want to define differentiation of distributions. If $\Phi$ is such that $\partial_{x_i} \Phi$ is a regular distribution of the form \eqref{eq:distributioneq}, we have that
\begin{equation} 
\int_U \dd^d x \, f(x) \partial_{x_i} \Phi(x) = - \int_U \dd^d x \, \partial_{x_i}  f(x) \Phi(x)
\end{equation}
by integration by parts, since $f$ has compact support. For an arbitrary distribution, we have the following.

\begin{definition}
The \emph{partial derivative} $\partial_{x_i} \Phi$ of $\Phi \in C_0^{\infty}(U)^*$ is defined by
\begin{equation}
\left(\partial_{x_i} \Phi \right)(f) := - \Phi \left( \partial_{x_i} f \right) \, , \qquad f \in C_0^{\infty}(U) \, .
\end{equation}
\end{definition}


\section{Hyperbolic differential operators and Green operators}
\sectionmark{Hyperbolic operators and Green operators}
\label{sec:hyperbolicoperators}

A classical field on a fixed spacetime will obey a wave-like partial differential equation, subject to initial or boundary conditions, which specifies the dynamics of the field theory. In this section, a brief description of a  subclass of hyperbolic partial differential equations is given, of which the Klein-Gordon equation is an important example. There is a vast amount of literature on this topic, of which \cite{Bar:2007zz,Benini:2013fia} are just two examples which have been used here.

\begin{remark}
Formally, a classical field can be thought as a section of a vector bundle $E$ over the spacetime manifold $M$ with fibre $V$. In this section, we will only consider the case in which $M$ is globally hyperbolic and the vector bundle is trivial $E = M \times V$, such that the space of sections is isomorphic to $C^{\infty}(M;V)$. For example, for a real untwisted scalar field on a globally hyperbolic spacetime $M$, the relevant vector bundle is the line bundle $E = M \times \mathbb{R}$ and the scalar field is then a real-valued function on $M$. Therefore, we will not consider the most general case of non-trivial vector bundles, which is treated in detail in \cite{Bar:2007zz,Benini:2013fia}.
\end{remark}

\subsection{Normally hyperbolic operators}

Let $M$ be a $d$-dimensional globally hyperbolic spacetime and let $V$ be a vector space over $\mathbb{R}$ (the generalisation to $\mathbb{C}$ is straightforward).

\begin{definition}
A smooth $V$-valued function on $M$, $\Phi \in C^{\infty}(M;V)$, will be called a \emph{classical field}.
\end{definition}

The dynamics of a linear classical field will be given by a linear partial differential equation, whose building block is a linear partial differential operator.

\begin{definition}
A \emph{linear partial differential operator} of order at most $k \in \mathbb{N}_0$ is a linear map $L : C^{\infty}(M;V) \to C^{\infty}(M;V)$ such that, for all $p \in M$, there exists a coordinate chart $(U, \phi)$ centred at $p$ and a collection of smooth maps $A_{\alpha} : U \to \End(V)$ for which, given any $f \in C^{\infty}(M;V)$, one has
\begin{equation}
L f = \sum_{|\alpha| \leq k} A_{\alpha} \, \partial^{\alpha} f \quad \text{on $U$.}
\end{equation}
Here, $\alpha = (\alpha_0, ..., \alpha_{d-1}) \in \mathbb{N}_0^d$ is a multi-index, $|\alpha| := \sum_{\mu=0}^{d-1} \alpha_{\mu}$, $\partial^{\alpha} := \prod_{\mu=0}^{d-1} \partial_{\mu}^{\alpha_{\mu}}$ and $\partial_{\mu}$ is the partial derivative with respect to the coordinate $x^{\mu}$ from the chart $(U, \phi)$.
\end{definition}

\begin{remark}
A more general definition of a linear partial differential operator would be of a a linear map $L : C^{\infty}(M;V) \to C^{\infty}(M;V')$, where $V'$ is another vector space, but for our purposes, having $V' = V$ is enough.
\end{remark}

An important notion is the one of the formal adjoint of an operator. To define that, one adds an additional structure to the space $C^{\infty}(M;V)$.

\begin{definition}
A non-degenerate pairing $(\cdot, \cdot) : C_0^{\infty}(M;V) \times C^{\infty}(M;V) \to \mathbb{R}$ is
\begin{equation}
(f, g) := \int_M \dvol_M f \cdot g \, ,
\label{eq:pairingCinftyMV}
\end{equation}
where $\cdot : C_0^{\infty}(M;V) \times C^{\infty}(M;V) \to \mathbb{R}$ is a non-degenerate bilinear form and $\dvol_M$ is the metric-induced volume form on $M$.
\end{definition}

\begin{remark}
The pairing in \eqref{eq:pairingCinftyMV} can also be defined for $f, \, g \in C^{\infty}(M;V)$ for which $\supp f \cap \supp g$ is compact, so that the integral is well-defined.
\end{remark}

An example of such a pairing is $\langle \cdot | \cdot \rangle : C_0^{\infty}(M;\mathbb{R}) \times C^{\infty}(M;\mathbb{R})$ given by
\begin{equation}
\langle f | g \rangle = \int_M \dvol_M(x) \, f(x) \, g(x) \, .
\end{equation}
This takes the same form as the scalar product in the Lebesgue space $L^2(M, \dvol_M)$.

\begin{definition} \label{def:formaladjoint}
Given a linear partial differential operator $L : C^{\infty}(M;V) \to C^{\infty}(M;V)$, the \emph{formal adjoint} of $L$ is the linear partial differential operator $L^* : C^{\infty}(M;V) \to C^{\infty}(M;V)$ such that
\begin{equation}
\left( L^* f, g \right) = (f, L g)
\end{equation}
for all $f, \, g \in C^{\infty}(M;V)$ for which $\supp f \cap \supp g$ is non-empty and compact. If $L = L^*$, we call $L$ \emph{formally self-adjoint}.
\end{definition}

In the context of field theory, the focus is on linear partial differential operators which can be associated with an initial value problem. The class of operators of interest is the class of normally hyperbolic operators. To define these, one needs the concept of the principal symbol of a differential operator.

\begin{definition}
Let $L : C^{\infty}(M;V) \to C^{\infty}(M;V)$ be a linear partial differential operator of order $k$. Given $p \in M$ and a coordinate chart $(U, \phi)$ centred at $p$, the \emph{principal symbol} $S_L : T_p^*M \to \End(V)$ is defined locally as
\begin{equation}
S_L (\zeta) := \sum_{|\alpha| = k} A_{\alpha}(p) \, \zeta^{\alpha} \, .
\end{equation}
Here, $\zeta \in T_p^*M$, $\zeta^{\alpha} := \prod_{\mu=0}^{d-1} \zeta^{\alpha_{\mu}}_{\mu}$ and $\zeta_{\mu}$ are the components of $\zeta$ with respect to the chart $(U, \phi)$. 
\end{definition}

\begin{definition}
Given a Lorentzian manifold $(M,g)$, a second order linear differential operator $P : C^{\infty}(M;V) \to C^{\infty}(M;V)$ is called \emph{normally hyperbolic} if $S_L(\zeta) = g^{-1}(\zeta, \zeta) \, \mathbb{I}_V$ for all $\zeta \in T_p^*M$.
\end{definition}

\begin{remark}
In a given coordinate chart $(U, \phi)$, a normally hyperbolic operator $P$ is such that, for any $f \in C^{\infty}(M;V)$,
\begin{equation}
P f = g^{\mu\nu} \mathbb{I}_V \partial_{\mu} \partial_{\nu} f + A^{\mu} \partial_{\mu} f + A f \quad \text{on $U$,}
\end{equation}
where $A$, $A^{\mu} \in \End(V)$, $\mu = 0, ..., d-1$.
\end{remark}

\begin{remark}
The d'Alembert operator $\nabla^2 = g^{\mu\nu} \nabla_{\mu} \nabla_{\nu}$ and the Klein-Gordon operator $\nabla^2 - m^2 \mathbb{I}_V$, $m \in \mathbb{R}$, are examples of normally hyperbolic operators.
\end{remark}

\subsection{Cauchy problem}

The importance of normally hyperbolic operators in field theory is that, if a classical field is defined on a globally hyperbolic spacetime and if the partial differential operator associated to its field equation is normally hyperbolic, then one has a well-posed Cauchy problem, as described in the next theorem.

\begin{theorem} \label{thm:initialvalueproblem}
Let $M$ be a globally hyperbolic spacetime and let $\Sigma$ be a spacelike Cauchy surface whose future-directed unit normal vector field is denoted by $n$. Furthermore, let $V$ be a vector space and $P : C^{\infty}(M;V) \to C^{\infty}(M;V)$ a normally hyperbolic operator. Then, for any $j, \, \Phi_0, \, \Phi_1 \in C_0^{\infty}(M;V)$, the following Cauchy problem,
\begin{equation}
\left\{
\begin{array}{ll}
P \Phi = j \, , \\
\Phi \big|_{\Sigma} = \Phi_0 \, , \\
\nabla_n \Phi \big|_{\Sigma} = \Phi_1 \, ,
\end{array}
\right.
\label{eq:initialvalueproblem}
\end{equation}
admits a unique solution $\Phi \in C^{\infty}(M;V)$, such that
\begin{equation}
\supp \Phi \subset J \left( \supp \Phi_0 \cup \supp \Phi_1 \cup \supp j \right) \, .
\end{equation}
\end{theorem}

\begin{proof}
See e.g.~Theorems 3.2.11 and 3.2.12 of \cite{Bar:2007zz}.
\end{proof}

\begin{remark}
Even though the Cauchy problem given in \eqref{eq:initialvalueproblem} has a non-vanishing source term, in this thesis only the $j = 0$ case will be considered.
\end{remark}

\subsection{Green operators}

One important consequence of Theorem~\ref{thm:initialvalueproblem} is the existence and uniqueness of the so-called Green operators associated with a normally hyperbolic operator on a globally hyperbolic spacetime.

\begin{definition} \label{def:Greenoperators}
Let $L : C^{\infty}(M;V) \to C^{\infty}(M;V)$ be a linear partial differential operator. The linear maps $G_{\rm ret}, \, G_{\rm adv} : C_{\rm tc}^{\infty}(M;V) \to C^{\infty}(M;V)$ are the \emph{retarded} and \emph{advanced Green operators} for $L$, respectively, if, for any $f \in C_{\rm tc}^{\infty}(M;V)$,
\begin{enumerate}[label={(\roman*)}]
\item $L G_{\rm ret} f = L G_{\rm adv} f = f$;
\item $G_{\rm ret} L f = G_{\rm adv} L f = f$;
\item $\supp \left( G_{\rm ret} f \right) \subset J^+ \left( \supp f\right)$ and $\supp \left( G_{\rm adv} f \right) \subset J^- \left( \supp f\right)$.
\end{enumerate}
\end{definition}

\begin{figure}[t!]
\begin{center}
\def\svgwidth{0.55\textwidth}
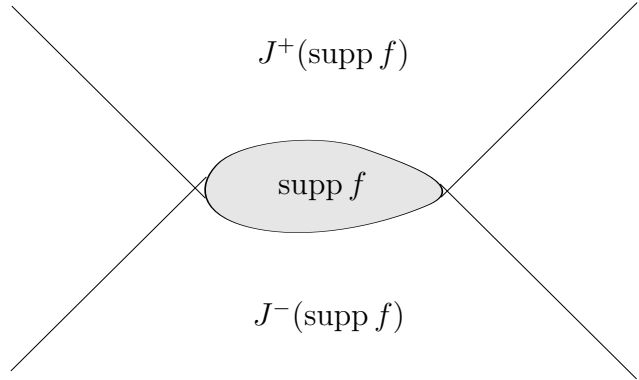
\caption{\label{fig:causalsuppf} The causal future and causal past of $\supp f$ for $f \in C^{\infty}_0(M)$.}
\end{center}
\end{figure}

Not all linear partial differential operators have Green operators.

\begin{definition}
A linear partial differential operator $L : C^{\infty}(M;V) \to C^{\infty}(M;V)$ is called \emph{Green hyperbolic} if it admits advanced and retarded Green operators.
\end{definition}

\begin{remark} \label{rem:uniqueGreenoperators}
A linear partial differential operator $L$ has unique Green operators if $L$ and its formal adjoint $L^*$ are Green hyperbolic. In particular, the Green operators are unique if $L$ is formally self-adjoint and Green hyperbolic. It can be shown that $G_{\rm ret}^* = G_{\rm adv}$ and $G_{\rm adv}^* = G_{\rm ret}$.
\end{remark}

Another convenient Green operator is the causal propagator.

\begin{definition} \label{def:causalpropagator}
The operator $G := G_{\rm adv} - G_{\rm ret}$ is called the \emph{causal propagator}.
\end{definition}

\begin{proposition} \label{prop:causalpropagator}
The causal propagator $G : C_{\rm tc}^{\infty}(M;V) \to C^{\infty}(M;V)$ satisfies, for any given $f \in C_{\rm tc}^{\infty}(M;V)$:
\begin{enumerate}[label={(\roman*)}]
\item $L G f = 0$;
\item $G L f = 0$;
\item $\supp \left( G f \right) \subset J \left( \supp f\right)$.
\end{enumerate}
\end{proposition}

\begin{proof}
It follows directly from the properties (i)-(iii) in Definition~\ref{def:Greenoperators} satisfied by the advanced and retarded Green operators.
\end{proof}

\begin{remark} \label{rem:Goncompactsuppfunctions}
Note that property (iii) in Proposition~\ref{prop:causalpropagator} implies that the causal propagator $G$ maps functions with compact support to functions with spacelike compact support (cf.~Definition~\ref{def:spacetimelikecompactfunctions}, see also Fig.~\ref{fig:causalsuppf}).
\end{remark}

\begin{remark} \label{rem:adjointGreenoperator}
It follows from Remark~\ref{rem:uniqueGreenoperators} that $G^* = - G$.
\end{remark}

\begin{remark} \label{rem:Gbidistribution}
The causal propagator can be regarded as a bi-distribution, $G \in C_{\rm tc}^{\infty}(M \times M;V)^*$, so that, for $f_1, \, f_2 \in C_{\rm tc}^{\infty}(M;V)$,
\begin{equation}
G(f_1,f_2) = \int_M \dvol_M (x) f_1(x) (G f_2) (x) \, ,
\end{equation}
where
\begin{equation} \label{eq:Gfx}
(G f)(x) := \int_M \dvol_M (x') \, G(x,x') f(x') \, ,
\end{equation}
and $G(x,x')$ is to be understood in the sense of distributions.
\end{remark}

It is an important fact that a normally hyperbolic operator is automatically a Green hyperbolic operator.

\begin{proposition}
Let $M$ be a globally hyperbolic spacetime. If $P : C^{\infty}(M;V) \to C^{\infty}(M;V)$ is normal hyperbolic, then it is also Green hyperbolic.
\end{proposition}

\begin{proof}
See Corollary 3.4.3 of \cite{Bar:2007zz}.
\end{proof}

\begin{remark}
The converse is not true in general. An important example is the Dirac operator, which is Green hyperbolic, but not normally hyperbolic (see e.g.~\cite{Benini:2013fia}).
\end{remark}

The importance of the Green operators comes from the fact that a solution $\Phi$ of a partial differential equation $P \Phi = 0$, where $P$ is Green hyperbolic, can be written as $\Phi = G f$, i.e., 
\begin{equation}
\Phi (x) = (G f)(x) = \int_M \dvol_M (x') \, G(x,x') f(x') \, ,
\end{equation}
where $G$ is the causal propagator associated with $P$ and $f$ is a function on the manifold. More precisely:

\begin{theorem}
Let $M$ be a globally hyperbolic spacetime, $V$ be a vector space and $P : C^{\infty}(M;V) \to C^{\infty}(M;V)$ be a Green hyperbolic operator with Green hyperbolic formal adjoint $P^*$, such that the associated causal propagator $G : C_{\rm tc}^{\infty}(M;V) \to C^{\infty}(M;V)$ is unique (cf.~Remark~\ref{rem:uniqueGreenoperators}). One has that
\begin{equation}
\Ker P = \ImC G = G \left[ C_{\rm tc}(M;V) \right] \, .
\end{equation}
\end{theorem}

\begin{proof}
See Theorem 3.4.7 of \cite{Bar:2007zz}.
\end{proof}

In other words, the space of smooth solutions of $P \Phi = 0$ is given by the image of the causal propagator $G$.

\begin{definition}
The space of smooth solutions of $P \Phi = 0$ will be denoted by $\mathscr{S}$, i.e.
\begin{equation}
\mathscr{S} := \left\{ \Phi \in C^{\infty}(M;V) : P \Phi = 0 \right\} \, .
\end{equation}
\end{definition}

When trying to endow the space of solutions with additional structure, such as a symplectric structure, it will be important to consider a vector subspace of $\mathscr{S}$ on which such structure can be well defined. According to Remark~\ref{rem:Goncompactsuppfunctions}, if one acts $G$ on $C^{\infty}_0(M;V)$, instead of $C^{\infty}_{\rm tc}(M;V)$, one obtains functions with spacelike compact support, i.e.~$G \left[C_0^{\infty}(M;V) \right] \subset C_{\rm sc}^{\infty}(M;V)$. Therefore, the space of smooth solutions with spacelike compact of $P \Phi = 0$ can be defined.

\begin{definition} \label{def:spacesolutionsspacelikecompactsupport}
The space of smooth solutions of $P \Phi = 0$ with spacelike compact support will be denoted by $\mathscr{S}_{\rm sc}$, i.e.
\begin{equation} 
\mathscr{S}_{\rm sc} := \left\{ \Phi \in C^{\infty}_{\rm sc}(M;V) : P \Phi = 0 \right\} \subset \mathscr{S} \, .
\end{equation}
\end{definition}

A solution $\Phi \in \mathscr{S}_{\rm sc}$ can then be written as $\Phi = G f$, with $f \in C_0^{\infty}(M;V)$, whereas a solution $\Phi \in \mathscr{S}$ can be written as $\Phi = G f'$, with $f' \in C_{\rm tc}^{\infty}(M;V)$. 

The subspace of solutions with spacelike compact support $\mathscr{S}_{\rm sc}$ can be endowed with additional structure, which in the case of a scalar field is a symplectic structure, as described in section~\ref{sec:qftcst-symplecticspaces}.


\chapter{Quantum field theory on curved spacetimes}
\label{chap:qftcst}

In this chapter, the classical and quantum theories of a real scalar field are described. Here, we will take a more mathematical and formal approach to the topic in comparison to the standard treatment given in physics textbooks such as \cite{birrell1984quantum}, and closer in spirit to \cite{fulling1989aspects,wald1994quantum,dereziński2013mathematics,Benini:2013fia}. In particular, we describe the classical theory in terms of the symplectic space of real solutions of the Klein-Gordon equation and its closely related phase space. We can then introduce the space of classical observables of the theory, which can be endowed with an algebraic structure, the Poisson bracket, and show how the Poisson bracket of two fields is given in terms of the symplectic structure of the space of solutions. The quantisation procedure then consists of finding an appropriate Hilbert space, the Fock space, and field operators which acts on elements of this space (the ``states'') and which obey a commutation relation which is analogue to the Poisson bracket of the classical theory.

As is well known, the choice of Hilbert space for the quantum theory is not unique and, worse than that, different choices are, in general, unitarily inequivalent. In Minkowski spacetime this is remedied by requiring that the ``vacuum states'' of the Hilbert space are invariant under the time translation invariance of the theory and something analogous can be done in the case of stationary spacetimes. In a general curved spacetime, however, no such natural choice is available. This is what is often meant by the lack of a natural definition of ``particles'' in quantum field theory on curved spacetimes.

One way to solve this theoretical problem is to modify our quantisation procedure by basically inverting the order of the steps described above. We could have started by constructing observables, such as the quantum fields, as elements of an abstract algebra, instead of operators acting on a Hilbert space. We then could have defined states as objects which associate with each observable a real number. This approach would have allowed us to treat all states on equal footing, even those arising from unitarily inequivalent choices of Hilbert spaces in the original approach. This approach to quantum field theory is known as \emph{algebraic quantum field theory} (for recent reviews see \cite{Benini:2013fia,Hollands:2014eia}).

In this thesis, the spacetimes will be interested in are stationary spacetimes, for which there are natural choices of Hilbert spaces, selected by the time translation symmetry of these spacetimes. Because of this, we will not take the more theoretically satisfying algebraic approach to the construction of the quantum field theory and instead use the more traditional Hilbert space approach.


\section{Real scalar field}
\label{sec:qftcst-realscalarfield}

In this section, we restrict our attention to the classical and quantum theories of a real scalar field on a generic globally hyperbolic spacetime, in which case the theories are very well understood and rigorous proofs are available. The case of a spacetime with boundaries, which is not as well understood and ultimately is the one of relevance for this thesis, will be treated in Section~\ref{sec:qftcst-stwithboundaries}.

\subsection{Classical field theory}
\label{sec:qftcst-classicaltheory}

Let $\Phi$ be a real scalar field on a globally hyperbolic spacetime $(M,g)$. The classical action is given by
\begin{equation}
S = \int_M \dvol_M \, \left( - \frac{1}{2} \nabla_a \Phi \nabla^a \Phi - \frac{1}{2} \left(m^2 + \xi R \right) \Phi^2 \right) \, ,
\end{equation}
where $m$ is the mass of the field, $R$ is the Ricci scalar, $\xi \in \mathbb{R}$ is the curvature coupling parameter and $\dvol_M(x) = \sqrt{-g} \, \dd^d x$. The field equation is the Klein-Gordon equation,
\begin{equation}
P \Phi := \left( \nabla^2 - m^2 - \xi R \right) \Phi = 0 \, ,
\label{eq:KGequation}
\end{equation}
where we have defined the differential operator $P := \nabla^2 - m^2 - \xi R$. A solution of \eqref{eq:KGequation} is fully determined by its Cauchy data at a Cauchy surface $\Sigma$. We then have the Cauchy problem
\begin{equation} \label{eq:scalarCauchyproblem}
\left\{
\begin{array}{ll}
P \Phi = 0 \, , \\
\Phi \big|_{\Sigma} = \Phi_0 \, , \\
\nabla_n \Phi \big|_{\Sigma} = \Phi_1 \, ,
\end{array}
\right.
\end{equation}
where $\Phi_0 , \, \Phi_1 \in C_0^{\infty}(\Sigma)$ and $n$ is the future-directed unit normal vector on $\Sigma$.

\subsubsection{Space of solutions}

As seen in Theorem~\ref{thm:initialvalueproblem}, the support of the solutions of the Klein-Gordon equation is contained in $J \left( \supp \Phi_0 \cup \supp \Phi_1 \right)$. It follows that a natural space of solutions to consider is the space $\mathscr{S}_{\rm sc}$ of smooth (real-valued) functions with spacelike compact support, as introduced in Definition~\ref{def:spacesolutionsspacelikecompactsupport}. We can endow $\mathscr{S}_{\rm sc}$ with a symplectic structure, $\sigma : \mathscr{S}_{\rm sc} \times \mathscr{S}_{\rm sc} \to \mathbb{R}$,
\begin{equation}
\sigma (\Phi_1, \Phi_2) := \int_{\Sigma} \dvol_{\Sigma} \, n_a J^a \left(\Phi_1, \Phi_2 \right) \, ,
\label{eq:symplecticform}
\end{equation}
with
\begin{equation}
J^a \left(\Phi_1, \Phi_2 \right) := \Phi_1 \nabla^a \Phi_2 - \Phi_2 \nabla^a \Phi_1 \, ,
\end{equation}
where $\Sigma$ is a spacelike Cauchy surface and $n$ is the future-directed unit normal vector on $\Sigma$. The pair $(\mathscr{S}_{\rm sc}, \sigma)$ is a symplectic space, cf.~Definition~\ref{def:symplecticspace}.

\begin{remark}
One has that
\begin{equation}
\nabla^a J_a = \Phi_1 \nabla^2 \Phi_2 - \Phi_2 \nabla^2 \Phi_1
= \left(m^2 + \xi R \right) \left( \Phi_1 \Phi_2 - \Phi_2 \Phi_1 \right) 
= 0 \, ,
\end{equation}
where the field equation \eqref{eq:KGequation} was used. Therefore, by the divergence theorem, the symplectic form $\sigma$ does not depend on the choice of Cauchy surface $\Sigma$. The vector-valued form $J^a$ is sometimes called a \emph{conserved current}.
\end{remark}

\begin{remark}
The restriction to the space $\mathscr{S}_{\rm sc}$ guarantees that the symplectic form as defined in \eqref{eq:symplecticform} is well defined.
\end{remark}

It is possible to relate the symplectic structure of the space of solutions to the causal propagator $G$, regarded as a bi-distribution (see Remark~\ref{rem:Gbidistribution}), as follows.

\begin{lemma} \label{lemma:relationGandsigma}
Let $\Phi = G f$ and $\Phi' = G f'$ be two solutions of the Klein-Gordon equation, $P \Phi = 0$, with $f, \, f' \in C_0^{\infty}(M)$. Then
\begin{equation}
G(f,f') = \sigma(\Phi, \Phi') \, .
\end{equation}
\end{lemma}

\begin{proof}
By exploiting the support properties of the advanced and retarded Green operators and by integrating by parts twice, one obtains
\begin{align}
G(f,f') &= \int_M \dvol_M f \, G f' \notag \\
&= \int_{J^+(\Sigma)} \dvol_M f \, \Phi' + \int_{J^-(\Sigma)} \dvol_M f \, \Phi' \notag \\
&= \int_{J^+(\Sigma)} \dvol_M \left(P G_{\rm adv} f \right) \Phi' + \int_{J^-(\Sigma)} \dvol_M \left(P G_{\rm ret} f \right) \Phi' \notag \\
&= - \int_{\Sigma} \dvol_{\Sigma} \nabla_n \left(G_{\rm adv} f \right) \Phi' + \int_{\Sigma} \dvol_{\Sigma} G_{\rm adv} f \, \nabla_n \Phi' \notag \\
&\quad + \int_{\Sigma} \dvol_{\Sigma} \nabla_n \left(G_{\rm ret} f \right) \Phi' - \int_{\Sigma} \dvol_{\Sigma} G_{\rm ret} f \, \nabla_n \Phi' \notag \\ 
&= \int_{\Sigma} \dvol_{\Sigma} \left( \Phi \nabla_n \Phi' - \Phi' \nabla_n \Phi \right) \notag \\
&= \sigma (\Phi, \Phi') \, .
\end{align}
\end{proof}

Lemma~\ref{lemma:relationGandsigma} can alternatively be written as
\begin{equation}
\int_M \dvol_M f \, G f' = \int_{\Sigma} \dvol_{\Sigma} \big( Gf \nabla_n (Gf') - Gf' \nabla_n (Gf) \big)
\end{equation}
for $f, \, f' \in C_0^{\infty}(M)$ and it will prove to very useful in the following to easily pass from ``integrals over $M$'' to ``integrals over $\Sigma$'' and vice-versa.

\subsubsection{Solution in terms of initial data}

It is possible to express the solution $\Phi$ of the Cauchy problem \eqref{eq:scalarCauchyproblem} in terms of the Cauchy data. Given a Cauchy surface $\Sigma$, one can define the operators
\begin{align}
\rho_0 &: C^{\infty}(M) \to C^{\infty}(\Sigma) \, , \qquad \rho_0(\Phi) := \Phi|_{\Sigma} \, , \\
\rho_1 &: C^{\infty}(M) \to C^{\infty}(\Sigma) \, , \qquad \rho_1(\Phi) := \nabla_n \Phi|_{\Sigma} \, .
\end{align}
Given a solution $\Phi$ of the Cauchy problem, these maps give the initial data $\Phi_0 = \rho_0(\Phi)$ and $\Phi_1 = \rho_1(\Phi)$. These operators have adjoints, $\rho_0^* , \, \rho_1^* : C^{\infty}(\Sigma)^* \to C^{\infty}(M)^*$, such that
\begin{equation} \label{eq:adjointoperators}
( \Psi , \rho_i \Phi)_{\Sigma} = ( \rho_i^* \Psi, \Phi)_M \, , \qquad
\Psi \in C^{\infty}(\Sigma), \, \Phi \in C^{\infty}(M) \, ,
\quad (i=0,1) \, ,
\end{equation}
where $(\cdot, \cdot)_{\Sigma}$ and $(\cdot, \cdot)_M$ are the pairings introduced in Remark~\ref{rem:notationdistributions} which define distributions on $\Sigma$ and $M$, respectively; and $C^{\infty}(\Sigma)^*$ and $C^{\infty}(M)^*$ are the spaces of compactly supported distributions on $\Sigma$ and $M$, respectively, cf.~Definition~\ref{def:distributionscompactsupp}.

The ``smeared field'' $\Phi(f) = (\Phi, f)_M$, thought as a distribution, is
\begin{align}
(\Phi, f)_M &= \int_M \dvol_M f \Phi \notag \\
&= \int_{\Sigma} \dvol_{\Sigma} \left[ Gf \nabla_n \Phi - \Phi \nabla_n (Gf) \right] \notag \\
&= \int_{\Sigma} \dvol_{\Sigma} \left[ \rho_0(Gf) \Phi_1 - \Phi_0 \rho_1(Gf) \right] \notag \\
&= \left( \Phi_1, \rho_0(Gf) \right)_{\Sigma} - \left( \Phi_0, \rho_1(Gf) \right)_{\Sigma} \, ,
\label{eq:smearedfieldinterm}
\end{align}
where Lemma~\ref{lemma:relationGandsigma} was used. Note that, in the last line, $\Phi_0 , \, \Phi_1 \in C_0^{\infty}(\Sigma)$ and also $\rho_0(Gf) , \, \rho_1(Gf) \in C_0^{\infty}(\Sigma)$, since $f \in C_0^{\infty}(M)$ implies that $Gf \in C_{\rm sc}^{\infty}(M)$. Now, using the adjoint operators $\rho_0^* , \, \rho_1^*$ and \eqref{eq:adjointoperators},
\begin{align}
(\Phi, f)_M = - \left( \rho_1^* \Phi_0, Gf \right)_{M} + \left( \rho_0^* \Phi_1, Gf \right)_{M} \, .
\end{align}

Remark~\ref{rem:adjointGreenoperator} tells us that the formal adjoint of $G$ with respect to some pairing \eqref{eq:pairingCinftyMV} is equal to $G^* = - G$. If we choose the pairing \eqref{eq:pairingCinftyMV} to be the one introduced in Remark~\ref{rem:notationdistributions} in the context of distributions, then the formal adjoint coincide with the one in \eqref{eq:adjointoperators} if we extend $G^*$ so that $G^* : C^{\infty}(M)^* \to C^{\infty}_0(M)^*$. Then, we can write
\begin{align}
(\Phi, f)_M &= \left( G \rho_1^* \Phi_0, f \right)_{M} - \left( G \rho_0^* \Phi_1, f \right)_{M} \, .
\end{align}
Hence, the solution $\Phi$ can be expressed in terms of its initial data as
\begin{equation} \label{eq:solutionintermsofic}
\Phi = G \rho_1^* \Phi_0 - G \rho_0^* \Phi_1 \, ,
\end{equation}
in the sense of distributions. Note, however, since $\Phi$ is smooth by Theorem~\ref{thm:initialvalueproblem}, \eqref{eq:solutionintermsofic} also holds in the sense of smooth functions.

\begin{remark}
This result can also be obtained in the more familiar ``unsmeared'' form. Starting with \eqref{eq:smearedfieldinterm} and using \eqref{eq:Gfx},
\begin{align}
\Phi(f) &= \int_{\Sigma} \dvol_{\Sigma}(x) \left[ {-\Phi(x)} \nabla_n (Gf)(x) + (Gf)(x) \nabla_n \Phi(x) \right] \notag \\
&= \int_{\Sigma} \dvol_{\Sigma}(x) \int_M \dvol_M(y) \left[ {-\Phi(x)} \nabla_n G(x,y) f(y) - G(x,y) f(y) \nabla_n \Phi(x) \right] \, ,
\end{align}
from which
\begin{equation}
\Phi(x) = \int_{\Sigma} \dvol_{\Sigma}(x') \, n^{a'} \left[ {-\nabla_{a'}}  G(x,x') \Phi(x') + G(x,x') \nabla_{a'} \Phi(x') \right] \, .
\end{equation}
\end{remark}

\subsubsection{Phase space and classical observables}

Having given a brief description of the space of solutions of the classical theory, we now discuss the phase space of the classical theory.

Since the spacetime $M$ under consideration is globally hyperbolic, Theorem~\ref{thm:globhypspacetimefoliation} guarantees that there exists a foliation of $M$ such that its metric can be given by
\begin{equation} \label{eq:ADMmetricchap2}
\dd s^2 = - N^2 \, \dd t^2 + h_{ij} \left(\dd x^i + N^i \, \dd t \right) \left(\dd x^j + N^j \, \dd t \right) \, .
\end{equation}
One can take the Cauchy surface $\Sigma$ to be a surface of constant $t$, with future-directed unit normal vector field $n$, and such that the metric on it is given by $h$. One has that $n = - N \dd t$, $\dvol_{\Sigma} = \sqrt{h} \, \dd^{d-1} x$ and $\sqrt{-g} = N \sqrt{h}$.

\begin{definition}
The \emph{canonical conjugate momentum} to $\Phi$ is the density
\begin{equation}
\Pi(x) := \frac{\delta S}{\delta (\partial_t \Phi(x))} \, ,
\end{equation}
evaluated at the Cauchy surface $\Sigma$.
\end{definition}

It follows that, at $\Sigma$,
\begin{equation}
\Pi = - \sqrt{-g} \, g^{t \mu} \partial_{\mu} \Phi = - N \sqrt{h} \, (dt)_{\nu} g^{\nu\mu} \partial_{\mu} \Phi = \sqrt{h} \, \nabla_n \Phi \, .
\end{equation}

The phase space is then the space described by the variables $(\Phi, \Pi)$.

\begin{definition}
The \emph{phase space} is the space $\mathscr{P} := C_0^{\infty}(\Sigma) \times \sqrt{h} \, C_0^{\infty}(\Sigma)$, where $\sqrt{h} \, C_0^{\infty}(\Sigma)$ denotes the space of smooth densities of compact support on $\Sigma$ of the form $\sqrt{h} \, f$, with $f \in C_0^{\infty}(\Sigma)$, such that a point in phase space corresponds to a specification of $\Phi(x)$ and $\Pi(x)$ on $\Sigma$.
\end{definition}

A classical observable can be thought as a functional on the phase space $\mathscr{P}$.

\begin{definition}
A \emph{classical observable} is a functional $\tilde{F}_f : \mathscr{P} \to \mathbb{R}$, labelled by a function $f \in C_0^{\infty}(M)$. For our purposes, we consider a class of classical observables of the form
\begin{equation} \label{eq:classobservablephasespace}
\tilde{F}_f(\Phi, \Pi) = \int_{\Sigma} \dd^{d-1} x \left( \Pi \, Gf - \Phi \sqrt{h} \, \nabla_n (Gf) \right) \, ,
\end{equation}
where $G$ is the causal propagator.
\end{definition}

By the well-posedness of the Cauchy problem \eqref{eq:scalarCauchyproblem}, every point $(\Phi, \Pi) \in \mathscr{P}$ of the phase space uniquely determines a solution $\Phi \in \mathscr{S}_{\rm sc}$. Therefore, a classical observable can be equivalently thought as a functional on the space of solutions $\mathscr{S}_{\rm sc}$. Using Lemma~\ref{lemma:relationGandsigma}, one can write classical observables of the form \eqref{eq:classobservablephasespace}, regarded as functionals on $\mathscr{S}_{\rm sc}$, as
\begin{equation}
F_f(\Phi) = \int_M \dvol_M(x) f(x) \Phi(x) \, .
\end{equation}

An important example of a classical observable of this class is the so-called ``smeared field'' $\mathcal{O}_f : \mathscr{S}_{\rm sc} \to \mathbb{R}$, $\Phi \mapsto \Phi(f) := F_f(\Phi)$. The ``smeared field'' $\Phi(f)$ has the interpretation of being the spacetime average of $\Phi(x)$, weighted by $f$. From another point of view, one may treat
\begin{equation}
f \mapsto \Phi(f) = \int_M \dvol_M(x) f(x) \Phi(x)
\end{equation}
as a distribution, $\Phi \in C^{\infty}(M)^*$, in which case $\Phi(x)$ is called the ``unsmeared field''.

The space of all classical observables can be endowed with an algebraic structure, the \emph{Poisson bracket}, which is induced by the symplectic structure of $\mathscr{S}_{\rm sc}$.

\begin{definition}
The \emph{Poisson bracket} of two classical observables $\tilde{F}_f , \, \tilde{F}_{f'} : \mathscr{P} \to \mathbb{R}$ is given by
\begin{equation}
\left\{ \tilde{F}_f, \tilde{F}_{f'} \right\} := \int_{\Sigma} \dd^{d-1} x \left( \frac{\delta \tilde{F}_f}{\delta \Phi} \frac{\delta \tilde{F}_{f'}}{\delta \Pi} - \frac{\delta \tilde{F}_f}{\delta \Pi} \frac{\delta \tilde{F}_{f'}}{\delta \Phi} \right) \, .
\end{equation}
\end{definition}

\begin{lemma}
One has
\begin{equation}
\left\{ \tilde{F}_f, \tilde{F}_{f'} \right\} = \sigma(Gf, Gf') = G(f,f') \, .
\end{equation}
\end{lemma}

\begin{proof}
Use \eqref{eq:classobservablephasespace} and Lemma~\ref{lemma:relationGandsigma}.
\end{proof}

It then follows that the Poisson bracket of two smeared fields $\Phi(f)$ and $\Phi(f')$ is
\begin{equation} \label{eq:Poissobracket}
\left\{ \Phi(f), \Phi(f') \right\} = G(f,f') \, .
\end{equation}
In terms of the ``unsmeared fields'',
\begin{equation}
\left\{ \Phi(x), \Phi(x') \right\} = G(x,x') \, .
\end{equation}
The aim of the quantisation procedure will be to find an analogous relation which is satisfied by the quantised scalar field.

\subsection{Quantum field theory}
\label{sec:qftcst-quantumtheory}

We now seek to find the quantum Klein-Gordon field theory. To do that, we apply Dirac's quantisation prescription, which consists of finding operators which act on a suitable Hilbert space.

In more detail, the aim of this prescription is to find operator-valued distributions $\Phi(f)$, with $f \in C_0^{\infty}(M)^{\mathbb{C}} \cong C_0^{\infty}(M;\mathbb{C})$, such that
\begin{enumerate}[label={(\roman*)}]
\item $f \mapsto \Phi(f)$ is linear;
\item $\Phi(Pf) = 0$ for all $f \in C_0^{\infty}(M;\mathbb{C})$;
\item $\Phi(f)^{\dagger} = \Phi(\overline{f})$ for all $f \in C_0^{\infty}(M;\mathbb{C})$;
\item $\left[ \Phi(f), \Phi(f') \right] = i G(f,f') \mathbb{I}$ for all $f, \, f' \in C_0^{\infty}(M;\mathbb{C})$.
\end{enumerate}

\vspace*{2ex}

\begin{remark}
Recall that 
\begin{equation}
C_0^{\infty}(M)^{\mathbb{C}} := C_0^{\infty}(M) \oplus i C_0^{\infty}(M) \cong C_0^{\infty}(M;\mathbb{C})
\end{equation}
is the complexification of the real vector space $C_0^{\infty}(M)$ (see Definition~\ref{def:complexification}).
\end{remark}

\vspace*{2ex}

\begin{remark}
The operator-valued distributions $\Phi(f)$ can be interpreted as the quantisation of the ``smeared fields'' $\Phi(f)$, which are real-valued distributions. Point (iv) above is then the result of the standard ``curly-bracket-to-square-bracket'' prescription from the Poisson bracket \eqref{eq:Poissobracket}. In terms of the quantised ``unsmeared fields'', one can rewrite the last three properties above as
\begin{enumerate}[label={(\roman*)}] \setcounter{enumi}{1}
\item $P \Phi(x) = 0$;
\item $\Phi(x)^{\dagger} = \Phi(x)$;
\item $\left[ \Phi(x), \Phi(x') \right] = i G(x,x') \mathbb{I}$.
\end{enumerate}
As noted in Remark~\ref{rem:distributionnotation}, the ``unsmeared fields'' $\Phi(x)$ should always be understood in the distribution sense.
\end{remark}

\subsubsection{Construction of the one-particle Hilbert space and Fock space}

The main problem to solve at this point is to identify the appropriate Hilbert space on which the operator-valued distributions $\Phi(f)$ act. In order to do that, we start with the symplectic space $(\mathscr{S}_{\rm sc},\sigma)$ of the classical theory, in which the classical solutions live.

The first step in our construction is the complexification of $\mathscr{S}_{\rm sc}$, 
\begin{equation}
\mathscr{S}_{\rm sc}^{\mathbb{C}} := \mathscr{S}_{\rm sc} \oplus i \mathscr{S}_{\rm sc} \, .
\end{equation}
This is isomorphic to the space of smooth complex-valued solutions with spacelike compact support of the Klein-Gordon equation. Then, one canonically extends the symplectic form $\sigma : \mathscr{S}_{\rm sc} \times \mathscr{S}_{\rm sc} \to \mathbb{R}$ defined in \eqref{eq:symplecticform} to $\sigma_{\mathbb{C}} : \mathscr{S}_{\rm sc}^{\mathbb{C}} \times \mathscr{S}_{\rm sc}^{\mathbb{C}} \to \mathbb{C}$, cf.~Definition~\ref{def:canonicalsesquilinearextension}. 
The canonical extension $\sigma_{\mathbb{C}}$ is anti-Hermitian.
 
It is convenient to define the Hermitian form $\tilde{\sigma}_{\mathbb{C}} : \mathscr{S}_{\rm sc}^{\mathbb{C}} \times \mathscr{S}_{\rm sc}^{\mathbb{C}} \to \mathbb{C}$,
\begin{equation} \label{eq:quantumHermitianform}
\tilde{\sigma}_{\mathbb{C}} (\Phi_1 , \Phi_2) := i \sigma_{\mathbb{C}} (\Phi_1, \Phi_2) \, .
\end{equation}
%
%
By using \eqref{eq:symplecticform}, one can show that
\begin{equation}
\tilde{\sigma}_{\mathbb{C}}(\Phi_1 , \Phi_2) = i \sigma \left(\overline{\Phi_1}, \Phi_2 \right) = i \int_{\Sigma} \dvol_{\Sigma} \, n_a J^a \left(\overline{\Phi_1}, \Phi_2 \right) \, ,
\end{equation}
where $\sigma$ has been extended to $\mathscr{S}_{\rm sc}^{\mathbb{C}}$ by linearity in each variable.

The Hermitian form $\tilde{\sigma}_{\mathbb{C}}$ is \emph{not} a scalar product (cf.~Definition~\ref{def:scalarproduct}) as it generally fails to be positive definite on $\mathscr{S}_{\rm sc}^{\mathbb{C}}$. Instead, consider a closed subspace $\mathscr{S}_{\rm sc}^{\mathbb{C}+} \subset \mathscr{S}_{\rm sc}^{\mathbb{C}}$ such that 
\begin{enumerate}[label={(\roman*)}]
\item $\tilde{\sigma}_{\mathbb{C}}$ is positive definite on $\mathscr{S}_{\rm sc}^{\mathbb{C}+}$;
\item $\mathscr{S}_{\rm sc}^{\mathbb{C}}$ is the span of $\mathscr{S}_{\rm sc}^{\mathbb{C}+}$ and $\overline{\mathscr{S}_{\rm sc}^{\mathbb{C}+}}$;
\item given any $\Phi^+ \in \mathscr{S}_{\rm sc}^{\mathbb{C}+}$ and $\Phi^- \in \overline{\mathscr{S}_{\rm sc}^{\mathbb{C}+}}$, then $\tilde{\sigma}_{\mathbb{C}}(\Phi^+,\Phi^-)=0$.
\end{enumerate}

From (i), $\tilde{\sigma}_{\mathbb{C}}$ is a scalar product on $\mathscr{S}_{\rm sc}^{\mathbb{C}+}$, and we denote $\tilde{\sigma}_{\mathbb{C}}(\cdot, \cdot) =: \langle \cdot | \cdot \rangle$. Given (iii), it is not difficult to check that the orthogonal complement $(\mathscr{S}_{\rm sc}^{\mathbb{C}+})^{\perp} = \overline{\mathscr{S}_{\rm sc}^{\mathbb{C}+}}$, the complex conjugate space. According to Theorem~\ref{thm:projectiontheorem}, $\mathscr{S}_{\rm sc}^{\mathbb{C}} = \mathscr{S}_{\rm sc}^{\mathbb{C}+} \oplus \overline{\mathscr{S}_{\rm sc}^{\mathbb{C}+}}$ and hence, if $\Phi \in \mathscr{S}_{\rm sc}^{\mathbb{C}}$, then it can be decomposed as $\Phi = \Phi^+ + \Phi^-$, with $\Phi^+ \in \mathscr{S}_{\rm sc}^{\mathbb{C}+}$ and $\Phi^- \in \overline{\mathscr{S}_{\rm sc}^{\mathbb{C}+}}$.

The subspace $\mathscr{S}_{\rm sc}^{\mathbb{C}+}$ with scalar product $\langle \cdot | \cdot \rangle$ is not necessarily complete in the norm induced by the scalar product.

\begin{definition}
Define $\mathscr{H}$ to be the completion of $\mathscr{S}_{\rm sc}^{\mathbb{C}+}$ in the norm induced by the scalar product $\langle \cdot | \cdot \rangle$. Then, $\mathscr{H}$ is a Hilbert space and is called the \emph{one-particle Hilbert space}.
\end{definition}

\begin{remark}
For a spacetime with time-translation symmetry, a natural choice of $\mathscr{H}$ is the space of complex positive frequency solutions, as detailed in the next section. For now, it is assumed that $\mathscr{H}$ is the completion of any space $\mathscr{S}_{\rm sc}^{\mathbb{C}+}$ satisfying the properties (i)-(iii) above.
\end{remark}

Given the one-particle Hilbert space $\mathscr{H}$, one constructs the (symmetric) Fock space, $\mathscr{F}_{\rm s}(\mathscr{H})$, as in Definition~\ref{def:Fockspace},
\begin{equation} \label{eq:Fockspace}
\mathscr{F}_{\rm s}(\mathscr{H}) = \bigoplus_{n=0}^{\infty} \left( {\bigotimes^n}_{\rm s} \mathscr{H} \right) \, ,
\end{equation}
where $\bigotimes_{\rm s}^0 \mathscr{H} := \mathbb{C}$. Elements of this Hilbert space are called states.

\begin{definition}
An element $\Psi \in \mathscr{F}_{\rm s}(\mathscr{H})$ of the Fock space,
\begin{equation}
\Psi = (\psi_0 , \psi_1, \psi_2, ...) \, ,
\end{equation}
with $\psi_n \in \bigotimes^n_{\rm s} \mathscr{H}$, is called a \emph{state}. A very common notation for an element of the Fock space $\Psi$ is $| \Psi \rangle$, such that an element $\Psi'$ of the dual space $\mathscr{F}_{\rm s}(\mathscr{H})^*$ is written as $\langle \Psi' |$. The state
\begin{equation}
| 0 \rangle = (1 , 0, 0, ...)
\end{equation}
is called a \emph{vacuum state}.
\end{definition}

\subsubsection{Quantum field operators}

The Fock space $\mathscr{F}_{\rm s}(\mathscr{H})$ is the desired Hilbert space on which the operator-valued distributions $\Phi(f)$ act. To see this, we first define the annihilation and creation operators.

\begin{definition}
Given any $\varphi \in \mathscr{H}$, the \emph{annihilation operator} $a(\overline{\varphi}) : \mathscr{F}_{\rm s}(\mathscr{H}) \to \mathscr{F}_{\rm s}(\mathscr{H})$ is defined by
\begin{equation}
a(\overline{\varphi}) |\Psi \rangle := \left( \langle \overline{\varphi} | \psi_1 \rangle , \, \sqrt{2} \, \langle \overline{\varphi} | \psi_2 \rangle , \, \ldots , \, \sqrt{n+1} \, \langle \overline{\varphi} | \psi_{n+1} \rangle , \, \ldots \right) \, ,
\end{equation}
whereas the \emph{creation operator} $a^{\dagger}(\varphi) : \mathscr{F}_{\rm s}(\mathscr{H}) \to \mathscr{F}_{\rm s}(\mathscr{H})$ is defined by
\begin{equation}
a^{\dagger}(\varphi) |\Psi \rangle := \left( 0, \, \varphi \, \psi_0, \, \ldots , \, \sqrt{n} \, \varphi \otimes_{\rm s} \psi_{n-1}, \, \ldots \right) \, .
\end{equation}
\end{definition}

\begin{remark}
The creation operator $a^{\dagger}(\varphi)$ is the adjoint of the annihilation operator $a(\overline{\varphi})$. The vacuum state $| 0 \rangle$ is such that $a(\overline{\varphi}) | 0 \rangle = (0, \, 0, \, \ldots ) \equiv 0$.
\end{remark}

\begin{remark}
These operators are unbounded operators, but their action on a state $|\Psi \rangle \in \mathscr{F}_{\rm s}(\mathscr{H})$ gives another state in the Fock space if $|\Psi \rangle$ is a terminating sequence, in which case the resulting state is also a terminating sequence.
\end{remark}

These operators satisfy commutation relations.

\begin{proposition}
The annihilation and commutation operators obey the following commutation relations
\begin{equation} \label{eq:acommutationrelations}
\left[ a(\varphi), a(\varphi') \right] = \left[ a^{\dagger}(\varphi), a^{\dagger}(\varphi') \right] = 0  \, , \qquad
\left[ a(\varphi), a^{\dagger}(\varphi') \right] = \langle \varphi , \varphi' \rangle \, \mathbb{I} \, ,
\end{equation}
with $\varphi, \, \varphi' \in \mathscr{H}$.
\end{proposition}

\begin{proof}
One has
\begin{align}
a^{\dagger}(\varphi') a(\overline{\varphi}) |\Psi \rangle &= a^{\dagger}(\varphi') \left( \ldots , \, \sqrt{n+1} \, \langle \overline{\varphi} | \psi_{n+1} \rangle , \, \ldots \right) \notag \\
&= \left( \ldots , \, \sqrt{n} \, \varphi' \otimes_{\rm s} \sqrt{n} \, \langle \overline{\varphi} | \psi_n \rangle, \, \ldots \right) \, ,
\end{align}
and
\begin{align}
a(\overline{\varphi}) a^{\dagger}(\varphi') |\Psi \rangle &= a(\overline{\varphi}) \left( \ldots , \, \sqrt{n} \, \varphi' \otimes_{\rm s} \psi_{n-1}, \, \ldots \right) \notag \\
&= \left( \ldots , \, \sqrt{n+1} \langle \overline{\varphi} | \sqrt{n+1} \, \varphi' \otimes_{\rm s} \psi_n \rangle , \, \ldots \right) \notag \\
&= \left( \ldots , \, (n+1) \langle \overline{\varphi} | \varphi' \otimes_{\rm s} \psi_n \rangle , \, \ldots \right) \notag \\
&= \left( \ldots , \, n \, \varphi' \otimes_{\rm s} \langle \overline{\varphi} | \psi_n \rangle + \langle \overline{\varphi} | \varphi' \rangle  \psi_n , \, \ldots \right) \, .
\end{align}
Hence,
\begin{equation}
\left[ a(\overline{\varphi}), a^{\dagger}(\varphi') \right] |\Psi \rangle = \langle \overline{\varphi} | \varphi' \rangle \, |\Psi \rangle \, .
\end{equation}
The other identities follow similarly.
\end{proof}

Let $\{ \Phi_i \}_{i \in \mathscr{I}}$, where $\mathscr{I}$ is some index set, be a orthonormal basis of $\mathscr{H}$. Together with their complex conjugates, they form a basis for the Hilbert space completion of $\mathscr{S}^{\mathbb{C}}_{\rm sc}$, such that
\begin{equation}
\langle \Phi_i | \Phi_j \rangle = \delta_{ij} \, , \qquad
\langle \overline{\Phi_i} | \overline{\Phi_j} \rangle = - \delta_{ij} \, , \qquad
\langle \Phi_i | \overline{\Phi_j} \rangle = 0 \, ,
\end{equation}
for any $i, \, j \in \mathscr{I}$. If one sets
\begin{equation}
a_i := a (\Phi_i) \, , \qquad a_i^{\dagger} := a^{\dagger} (\Phi_i) \, ,
\end{equation}
one can rewrite the commutation relations \eqref{eq:acommutationrelations} as
\begin{equation}
\left[ a_i , a_j \right] = \left[ a_i^{\dagger} , a_j^{\dagger} \right] = 0 \, , \qquad \left[ a_i, a_j^{\dagger} \right] = \delta_{ij} \, \mathbb{I} \, .
\end{equation}

We now have all we need to define the quantum field operator $\Phi(x)$.

\begin{definition} \label{def:quantumscalarfield}
The \emph{quantum scalar field} $\Phi(x)$ is an operator-valued distribution defined by
\begin{equation} \label{eq:quantumscalarfield}
\Phi(x) := \sum_{i \in \mathscr{I}} \left[ a_i \, \Phi_i(x) + a_i^{\dagger} \, \overline{\Phi_i(x)} \right] \, ,
\end{equation}
where $\{ \Phi_i \}_{i \in \mathscr{I}}$ is an orthonormal basis of $\mathscr{H}$.
\end{definition}

If follows from Definition~\ref{def:quantumscalarfield} that
\begin{equation}
a_i = \langle \Phi_i | \Phi \rangle \, , \qquad
a_i^{\dagger} = - \langle \overline{\Phi_i} | \Phi \rangle \, .
\end{equation}
More generally, given a complex classical solution $\varphi \in \mathscr{S}^{\mathbb{C}}_{\rm sc}$, we could have defined
\begin{equation} \label{eq:aalternativedefinition}
a(\varphi) = \langle \varphi | \Phi \rangle \, , \qquad
a^{\dagger}(\varphi) = - \langle \overline{\varphi} | \Phi \rangle \, .
\end{equation}

Using the canonical commutation relations \eqref{eq:acommutationrelations}, one can derive the commutation relations for the field and its canonical conjugate momentum on a Cauchy surface $\Sigma$.

\begin{proposition}
On a Cauchy surface $\Sigma$ of constant time coordinate $t$, the quantum field $\Phi$ and its canonical conjugate momentum $\Pi$ satisfy the canonical commutation relations
\begin{equation}
\left[ \Phi(t,x), \Phi(t,x') \right] = \left[ \Pi(t,x), \Pi(t,x') \right] = 0 \, , \qquad
 \left[ \Phi(t,x), \Pi(t,x') \right] = i \delta(x,x') \mathbb{I} \, ,
\end{equation}
with $x , \, x'$ representing coordinates on $\Sigma$ and the Dirac delta distribution $\delta(x,x')$ is a density in the second argument.
\end{proposition}

\begin{proof}
Dropping the dependence on $t$ for notational simplicity, for $\varphi , \, \psi \in \mathscr{H}$, 
\begin{align}
\left[ a(\varphi), a(\psi) \right] &= \langle \varphi | \Phi \rangle \langle \psi | \Phi \rangle - \langle \psi | \Phi \rangle \langle \varphi | \Phi \rangle \notag \\
&= - \int_{\Sigma} \dvol_{\Sigma}(x) \left[ \overline{\varphi(x)} \nabla_n \Phi(x) - \Phi(x) \nabla_n \overline{\varphi(x)}\right] \notag \\
&\quad \times \int_{\Sigma} \dvol_{\Sigma}(x') \left[ \overline{\psi(x')} \nabla_n \Phi(x') - \Phi(x') \nabla_n \overline{\psi(x')}\right]
- (\varphi \leftrightarrow \psi) \notag \\
&= - \int_{\Sigma} \dd^{d-1}x \left[ \overline{\varphi(x)} \Pi(x) - \Phi(x) \nabla_n \overline{\varphi(x)}\right] \notag \\
&\quad \times \int_{\Sigma} \dd^{d-1}x' \left[ \overline{\psi(x')} \Pi(x') - \Phi(x') \nabla_n \overline{\psi(x')}\right]
- (\varphi \leftrightarrow \psi) \notag \\
&= - \int_{\Sigma} \dd^{d-1}x \int_{\Sigma} \dd^{d-1}x' \left\{ \overline{\varphi(x)} \overline{\psi(x')} \left[ \Pi(x),\Pi(x') \right] \right. \notag \\
&\quad \left. + \nabla_n \overline{\varphi(x)} \nabla_n \overline{\psi(x')} \left[ \Phi(x), \Phi(x') \right]
- \overline{\varphi(x)} \nabla_n \overline{\psi(x')} \left[ \Pi(x), \Phi(x') \right] \right. \notag \\
&\quad \left. - \nabla_n \overline{\varphi(x)} \overline{\psi(x')} \left[ \Phi(x),\Pi(x') \right]  \right\} \notag \\
&= 0 \, .
\end{align}
Similarly,
\begin{align}
\left[ a^{\dagger}(\varphi), a^{\dagger}(\psi) \right] &= \langle \overline{\varphi} | \Phi \rangle \langle \overline{\psi} | \Phi \rangle - \langle \overline{\psi} | \Phi \rangle \langle \overline{\varphi} | \Phi \rangle \notag \\
&= - \int_{\Sigma} \dd^{d-1}x \int_{\Sigma} \dd^{d-1}x' \left\{ \varphi(x) \psi(x') \left[ \Pi(x),\Pi(x') \right] \right. \notag \\
&\quad \left. + \nabla_n \varphi(x) \nabla_n \psi(x') \left[ \Phi(x), \Phi(x') \right]
- \varphi(x) \nabla_n \psi(x') \left[ \Pi(x), \Phi(x') \right] \right. \notag \\
&\quad \left. - \nabla_n \varphi(x) \psi(x') \left[ \Phi(x),\Pi(x') \right]  \right\} \notag \\
&= 0 \, ,
\end{align}
and
\begin{align}
\left[ a(\varphi), a^{\dagger}(\psi) \right] &= - \langle \varphi | \Phi \rangle \langle \overline{\psi} | \Phi \rangle + \langle \overline{\psi} | \Phi \rangle \langle \varphi | \Phi \rangle \notag \\
&= \int_{\Sigma} \dd^{d-1}x \int_{\Sigma} \dd^{d-1}x' \left\{ \overline{\varphi(x)} \psi(x') \left[ \Pi(x),\Pi(x') \right] \right. \notag \\
&\quad \left. + \nabla_n \overline{\varphi(x)} \nabla_n \psi(x') \left[ \Phi(x), \Phi(x') \right]
- \overline{\varphi(x)} \nabla_n \psi(x') \left[ \Pi(x), \Phi(x') \right] \right. \notag \\
&\quad \left. - \nabla_n \overline{\varphi(x)} \psi(x') \left[ \Phi(x),\Pi(x') \right]  \right\} \notag \\
&= \langle \varphi | \psi \rangle \, .
\end{align}
The commutation relations follow.
\end{proof}

\begin{remark}
Note that the ``smeared'' form of these commutation relations is
\begin{equation} \label{eq:smearedcommutationrelationonSigma}
\left[ \Phi(f), \Phi(g) \right] = \left[ \Pi(f), \Pi(g) \right] = 0 \, , \qquad 
\left[ \Phi(f), \Pi(g) \right] = i (f,g)_{\Sigma} \mathbb{I} \, ,
\end{equation}
where $f, \, g \in C^{\infty}_0 (\Sigma)$ and
\begin{equation}
(f,g)_{\Sigma} = \int_{\Sigma} \dvol_{\Sigma} f g \, .
\end{equation}
\end{remark}

Finally, we can show that the quantum field obeys the desired commutation relation.

\begin{proposition}
The quantum scalar field obeys the canonical commutation relation, for $f, \, g \in C^{\infty}_0 (M)$,
\begin{equation}
\left[ \Phi(f) , \Phi(g) \right] = i G(f,g) \mathbb{I} \, .
\end{equation}
\end{proposition}

\begin{proof}
By using \eqref{eq:solutionintermsofic}, the commutation relations \eqref{eq:smearedcommutationrelationonSigma} and Lemma~\ref{lemma:relationGandsigma}, 
\begin{align}
\left[ \Phi(f), \Phi(g) \right] 
&= \left[ (G \rho_1^* \Phi_0,f)_M - (G \rho_0^* \Phi_1, f)_M, \, (G \rho_1^* \Phi_0,f)_M - (G \rho_0^* \Phi_1,f)_M \right] \notag \\
&= \left[ - (\Phi_0, \rho_1 G f)_{\Sigma} + (\Phi_1, \rho_0 G f)_{\Sigma}, \, - (\Phi_0, \rho_1 G g)_{\Sigma} + (\Phi_1, \rho_0 G g)_{\Sigma} \right] \notag \\
&= \left[ - \Phi (\rho_1 G f) + \Pi(\rho_0 G f), \, - \Phi (\rho_1 G g) + \Pi(\rho_0 G g)\right] \notag \\
&= - i (\rho_1 G f, \rho_0 G g)_{\Sigma} \, \mathbb{I} + i (\rho_0 G f, \rho_1 G g)_{\Sigma} \, \mathbb{I} \notag \\ 
&= - i \int_{\Sigma} \dvol_{\Sigma} \left( Gg \nabla_n (Gf) - Gf \nabla_n (Gg) \right) \mathbb{I} \notag \\
&= i G(f,g) \mathbb{I} \, .
\end{align}
\end{proof}

\begin{remark}
The ``unsmeared field'' $\Phi(x)$ satisfies the commutation relation
\begin{equation}
\left[ \Phi(x), \Phi(x') \right] = i G(x,x') \mathbb{I} \, .
\end{equation}
When given an orthonormal basis $\{ \Phi_i \}_{i \in \mathscr{I}}$ of $\mathscr{H}$, this is equivalent to
\begin{align}
\left[ \Phi(x), \Phi(x') \right]
&= \sum_{i,j \in \mathscr{I}} \left\{ \left[ a_i, a_j \right] \Phi_i(x) \, \Phi_j (x')
+ \left[ a_i^{\dagger}, a_j^{\dagger} \right] \overline{\Phi_i(x)} \, \overline{\Phi_j (x')} \right. \notag \\
&\qquad\qquad \left. + \left[ a_i, a_j^{\dagger} \right] \Phi_i(x) \, \overline{\Phi_j (x')} 
+ \left[ a_i^{\dagger}, a_j \right] \overline{\Phi_i(x)} \, \Phi_j (x') \right\} \notag \\
&= \sum_{i \in \mathscr{I}} \left\{ \Phi_i(x) \, \overline{\Phi_i (x')} - \overline{\Phi_i(x)} \, \Phi_i (x') \right\} \mathbb{I} \notag \\
&= i G(x,x') \mathbb{I} \, . \label{eq:commutatorphi}
\end{align}
\end{remark}

We have then finished the Dirac's prescription to construct the quantum scalar field theory. We have constructed the Fock space $\mathscr{F}_{\rm s}(\mathscr{H})$ in \eqref{eq:Fockspace} and defined the quantum field as an operator-valued distribution $\Phi(f)$, whose ``unsmeared'' form is given by \eqref{eq:quantumscalarfield}. The quantum field satisfies the properties (i)-(iv) given in the beginning of this section.


\section{The case of stationary spacetimes}
\label{sec:qftcst-stationary}

In this section, we restrict our attention to stationary spacetimes, as defined in section~\ref{sec:stationaryspacetimes}. The basic idea to construct the quantum field theory is to choose the one-particle Hilbert space to be the subspace of complex solutions which are positive frequency with respect to the timelike Killing vector field, in a straightforward generalisation of quantum field theory on Minkowski spacetime. There are, however, some subtle technical points which need to be carefully considered and which are explored in detail in \cite{Ashtekar:1975zn,Kay:1978yp}. Here, we give an heuristic discussion of the construction and then a brief overview of the Green's distributions associated with the field equation which will be needed.

\subsection{Positive frequency solutions}
\label{sec:qftcst-positivefrequency}

In the last section it was described how to construct the one-particle Hilbert space $\mathscr{H}$ as the completion of a subspace $\mathscr{S}^{\mathbb{C}+}_{\rm sc}$ of the space $\mathscr{S}^{\mathbb{C}}_{\rm sc}$ of complex classical solutions such that the Hermitian form $\tilde{\sigma}_{\mathbb{C}}$ defined in \eqref{eq:quantumHermitianform} is positive definite (and, hence, a scalar product $\langle \cdot | \cdot \rangle$), $\mathscr{S}^{\mathbb{C}+}_{\rm sc}$ and its complex conjugate span the space of complex solutions and $\langle \Phi^+ | \Phi^- \rangle = 0$ for $\Phi^+ \in \mathscr{S}_{\rm sc}^{\mathbb{C}+}$ and $\Phi^- \in \overline{\mathscr{S}_{\rm sc}^{\mathbb{C}+}}$. However, there are many choices of such Hilbert spaces $\mathscr{H}$ and, therefore, the quantum field theory ultimately depends on our choice of $\mathscr{H}$.

In quantum field theory, which deals with infinite-dimensional vector spaces of solutions, different choices of $\mathscr{H}$ yield, in general, unitarily inequivalent theories. For a detailed discussion of this fact we refer e.g. to section 4.4 of \cite{wald1994quantum}. A simple way to visualise this point is to consider two such choices of one-particle Hilbert spaces, $\mathscr{H}_1$ and $\mathscr{H}_2$. Then, any solution $\varphi \in \mathscr{H}_2$ can be decomposed as $\varphi = \psi + \overline{\xi}$, with $\psi , \, \xi \in \mathscr{H}_1$. The annihilation operator $a(\varphi)$, which acts on the Fock space $\mathscr{F}_{\rm s}(\mathscr{H}_2)$, can then be written as
\begin{equation}
a(\varphi) = \langle \varphi | \Phi \rangle = \langle \psi + \overline{\xi} | \Phi \rangle
= a(\psi) - a^{\dagger}(\xi) \, ,
\end{equation}
where \eqref{eq:aalternativedefinition} was used. Let $|0 \rangle \in \mathscr{F}_{\rm s}(\mathscr{H}_1)$ be the vacuum state of the Fock space defined by $\mathscr{H}_1$, such that $a(\psi) |0 \rangle = 0$. It is clear that $a(\varphi) |0 \rangle \neq 0$, i.e. the vacuum state defined using $\mathscr{H}_1$ is not equivalent to the vacuum state defined using $\mathscr{H}_2$. We conclude that the definition of a vacuum state depends on the choice of the one-particle Hilbert space $\mathscr{H}$.

This choice-dependence is also true for Minkowski spacetime. However, in this case, there is a natural choice of $\mathscr{H}$, consisting of the subspace of positive frequency solutions (to be defined below), which arises from the time translation invariance (and, ultimately, from the Poincar\'{e} invariance) of the classical theory. In a general curved spacetime, there is no natural criterion, such as symmetries of the theory, for a unique choice of $\mathscr{H}$.

However, for a globally stationary spacetime, as defined in Definition~\ref{def:globalstationary}, there exists a time translation symmetry that can be used in an analogous way to Minkowski spacetime to select a natural choice of $\mathscr{H}$.

Let $M$ be a globally stationary spacetime and let $\xi$ be the future-directed timelike Killing vector field. Given the time translation symmetry, the Lie derivative with respect to $\xi$, $\mathcal{L}_{\xi}$, commutes with the Klein-Gordon operator $P$ and, therefore, it maps the space of complex smooth solutions $\mathscr{S}^{\mathbb{C}}$ to itself. Furthermore, it can be shown that:

\begin{proposition}
$\mathcal{L}_{\xi}$ is anti-Hermitian with respect to the Hermitian form $\tilde{\sigma}_{\mathbb{C}}$.
\end{proposition}

\begin{proof}
One wants to show that $\mathcal{L}_{\xi}^{\dagger} = - \mathcal{L}_{\xi}$, i.e.
\begin{equation}
\tilde{\sigma}_{\mathbb{C}} (\Phi , \mathcal{L}_{\xi} \Psi ) = - \tilde{\sigma}_{\mathbb{C}} (\mathcal{L}_{\xi} \Phi , \Psi ) \, ,
\end{equation}
for $\Phi , \, \Psi \in \mathscr{S}^{\mathbb{C}}_{\rm sc}$. On the Cauchy surface $\Sigma$ on which the Hermitian form is evaluated, one has that
$\xi^a = N n^a + N^a$, according to \eqref{eq:chinNN}, hence
\begin{equation}
\mathcal{L}_{\xi} \Phi = N \, \nabla_n \Phi + \nabla_N \Phi \, .
\end{equation}
Applying \eqref{eq:quantumHermitianform},
\begin{align}
\tilde{\sigma}_{\mathbb{C}} (\Phi , \mathcal{L}_{\xi} \Psi )
&= i \int_{\Sigma} \dvol_{\Sigma} \left( \overline{\Phi} \, \nabla_n \mathcal{L}_{\xi} \Psi
- \mathcal{L}_{\xi} \Psi \, \nabla_n \overline{\Phi} \right) \notag \\
&= i \int_{\Sigma} \dvol_{\Sigma} \left[ \overline{\Phi} \, ( N \nabla_n \nabla_n \Psi + \nabla_n \nabla_N \Psi )
- (N \, \nabla_n \Psi + \nabla_N \Psi ) \, \nabla_n \overline{\Phi} \right] \, .
\end{align}
By using integration by parts and the torsion-free property of the connection,
\begin{align}
\tilde{\sigma}_{\mathbb{C}} (\Phi , \mathcal{L}_{\xi} \Psi )
&= i \int_{\Sigma} \dvol_{\Sigma} \left[ ( N \nabla_n \nabla_n \overline{\Phi} + \nabla_n \nabla_N \overline{\Phi} ) \, \Psi
- (N \, \nabla_n \overline{\Phi} + \nabla_N \overline{\Phi} ) \, \nabla_n \Psi \right] \notag \\
&= i \int_{\Sigma} \dvol_{\Sigma} \left( \nabla_n \mathcal{L}_{\xi} \overline{\Phi} \,  \Psi
- \mathcal{L}_{\xi} \overline{\Phi} \, \nabla_n \Psi \right) \notag \\
&=  - \tilde{\sigma}_{\mathbb{C}} (\mathcal{L}_{\xi} \Phi , \Psi ) \, .
\end{align}
Hence, $\mathcal{L}_{\xi}$ is anti-Hermitian.
\end{proof}

An immediate and important consequence is:

\begin{proposition} \label{prop:Lchiimaginaryeigenvalues}
$\mathcal{L}_{\xi}$ has purely imaginary eigenvalues and eigenvectors for distinct eigenvalues are orthogonal.
\end{proposition}

\begin{proof}
It follows directly from Remark~\ref{rem:antihermitianoperator}.
\end{proof}

A positive frequency solution is then defined to be an eigenfunction of $\mathcal{L}_{\xi}$ whose eigenvalue is purely negative imaginary.

\begin{definition}
A mode solution $\Phi \in \mathscr{S}^{\mathbb{C}}$ is of \emph{positive frequency} if it is an eigenfunction of $\mathcal{L}_{\xi}$ such that
\begin{equation}
\mathcal{L}_{\xi} \Phi = - i \omega \Phi \, , \qquad \omega > 0 \, .
\end{equation}
A general solution is of positive frequency if it can be expressed as a linear combination of mode solutions of positive frequency.
\end{definition}

Let $\mathscr{S}^{\mathbb{C}+}$ denote the subspace of positive frequency solutions. Solutions in $\overline{\mathscr{S}^{\mathbb{C}+}}$ are called \emph{negative frequency} solutions.

\begin{remark} \label{rem:positivefrequencydense}
Note that, even though positive frequency solutions cannot have spacelike compact support \cite{wald1994quantum}, the space $\mathscr{S}^{\mathbb{C}}_{\rm sc}$ is dense in $\mathscr{S}^{\mathbb{C}} = \mathscr{S}^{\mathbb{C}+} \oplus \overline{\mathscr{S}^{\mathbb{C}+}}$.
\end{remark}

\begin{proposition} \label{prop:choiceofpositivefreq}
One has that
\begin{enumerate}[label={(\roman*)}]
\item the Hermitian form $\tilde{\sigma}_{\mathbb{C}}$ is positive definite on $\mathscr{S}^{\mathbb{C}+}$ (and hence defines a scalar product $\langle \cdot | \cdot \rangle$);
\item $\mathscr{S}^{\mathbb{C}}$ is the span of $\mathscr{S}^{\mathbb{C}+}$ and $\overline{\mathscr{S}^{\mathbb{C}+}}$;
\item given any $\Phi^+ \in \mathscr{S}^{\mathbb{C}+}$ and $\Phi^- \in \overline{\mathscr{S}^{\mathbb{C}+}}$, then $\tilde{\sigma}_{\mathbb{C}}(\Phi^+,\Phi^-)=0$.
\end{enumerate}
\end{proposition}

\begin{proof}
Item (iii) follows from Proposition~\ref{prop:Lchiimaginaryeigenvalues}. This shows that $\mathscr{S}^{\mathbb{C}}$ can be orthogonally decomposed as $\mathscr{S}^{\mathbb{C}} = \mathscr{S}^{\mathbb{C}+} \oplus \overline{\mathscr{S}^{\mathbb{C}+}}$ and, hence, (ii).

It remains to prove (i). We want to show that, for positive frequency $\Phi$,
\begin{equation} \label{eq:KGscalarproductproof1}
\tilde{\sigma}_{\mathbb{C}}(\Phi, \Phi)
= i \int_{\Sigma} \dvol_{\Sigma} \left( \overline{\Phi} \nabla_n \Phi - \Phi \nabla_n \overline{\Phi} \right) \geq 0 \, .
\end{equation}
The metric of the stationary spacetime can be written as \eqref{eq:ADMdecomposition},
\begin{equation} \label{eq:ADMdecompositionproof}
\dd s^2 = - N^2 \, \dd t^2 + h_{ij} \left( \dd x^i + N^i \, \dd t \right) \left( \dd x^j + N^j \, \dd t \right) \, .
\end{equation}
with $\xi = \partial_t$ and $N$, $N^i$ and $h_{ij}$ being independent of $t$. Note that the unit normal vector field is given by
\begin{equation}
n = \frac{1}{N} \left( \xi - N^i \partial_i \right) \, ,
\end{equation}
and thus \eqref{eq:KGscalarproductproof1} is equivalent to
\begin{equation}
\tilde{\sigma}_{\mathbb{C}}(\Phi, \Phi)
= i \int_{\Sigma} \dd^{d-1} x \frac{\sqrt{h}}{N} \left( \overline{\Phi} \partial_t \Phi - \Phi \partial_t \overline{\Phi} -\overline{\Phi} N^i \partial_i \Phi + \Phi N^i \partial_i \overline{\Phi} \right) \geq 0 \, .
\end{equation}

In the static case ($N^i = 0$), point (i) follows easily from $\partial_t \Phi = - i \omega \Phi$ with $\omega > 0$, 
\begin{equation}
\tilde{\sigma}_{\mathbb{C}}(\Phi, \Phi)
= 2 \omega \int_{\Sigma} d^{d-1} x \frac{\sqrt{h}}{N} |\Phi|^2 \geq 0 \, .
\end{equation}

In the non-static case, a bit more work is needed. Here, we just sketch the proof, by considering the special case for which the spacelike surfaces $\Sigma$ are compact. Then, the spectrum is discrete and, by orthogonality, it suffices to consider positive frequency modes with fixed $\omega$. (If $\Sigma$ is not compact, we need to consider wave-packets of positive frequency, which are localised in spacetime, as the ones described in Section~\ref{sec:HHstateconstruction}, instead of mode solutions of sharp frequency.)

Start with
\begin{equation} \label{eq:KGscalarproductproof2}
0 \equiv \int_{\Sigma} \dd^{d-1} x \, N \sqrt{h} \, \left( \overline{\Phi} P \Phi \right) \, ,
\end{equation}
where $P = \nabla^2 - m^2 - \xi R$ is the Klein-Gordon operator and $\Phi$ is a positive frequency solution, $\partial_t \Phi = - i \omega \Phi$ with $\omega > 0$. For a stationary spacetime with metric \eqref{eq:ADMdecompositionproof} the Klein-Gordon operator is
\begin{align}
P &= \frac{1}{\sqrt{-g}} \, \partial_{\mu} \left( g^{\mu\nu} \sqrt{-g} \, \partial_{\nu} \right) - m^2 - \xi R \notag \\
&= - \frac{1}{N^2} \, \partial_t^2 + \frac{N^i}{N^2} \, \partial_t \partial_i + \frac{1}{N \sqrt{h}} \, \partial_i \left( \frac{N^i \sqrt{h}}{N} \partial_t \right) \notag \\
&\quad + \frac{1}{N \sqrt{h}} \, \partial_i \left[ N \sqrt{h} \left( h^{ij} - \frac{N^i N^j}{N^2} \right) \partial_j \right] - m^2 - \xi R \, .
\end{align}
Substituting in \eqref{eq:KGscalarproductproof2} gives
\begin{align}
0 &=  \int_{\Sigma} \dd^{d-1} x \, N \sqrt{h} \, \left[ \frac{\omega^2}{N^2} \, |\Phi|^2 - i \omega \frac{N^i}{N^2} \overline{\Phi} \partial_i \Phi + i \omega \frac{N^i}{N^2} \Phi \partial_i \overline{\Phi} \right. \notag \\
&\hspace{19ex} \left. - \left( h^{ij} - \frac{N^i N^j}{N^2} \right) \partial_i \overline{\Phi} \partial_j \Phi - \left( m^2 + \xi R \right) |\Phi|^2 \right] \, ,
\label{eq:KGscalarproductproof3}
\end{align}
where integration by parts was used in the third and fourth terms and the boundary terms vanish, given the compactness of $\Sigma$.

At this point, we use the fact that $\partial_t$ is timelike, which translates into
\begin{equation} \label{eq:KGscalarproductproof4}
- N^2 + h_{ij} N^i N^j < 0 \, .
\end{equation}
At each point of $\Sigma$, choose Riemann normal coordinates such that $h_{ij} = \delta_{ij}$ and $N^i=(N^1, 0, ..., 0)$. Then, \eqref{eq:KGscalarproductproof4} shows that $|N^1| < N$ at that point, from which we can conclude that the quadratic form given by
\begin{equation}
h^{ij} - \frac{N^i N^j}{N^2} \, , \qquad i, j = 1, \ldots , d-1,
\end{equation}
is positive. On physical grounds, we also assume that $m^2 + \xi R \geq 0$, such that there are no tachyonic instabilities. Then it follows that
\begin{align}
\int_{\Sigma} \dd^{d-1} x \, N \sqrt{h} \, \left[ \left( h^{ij} - \frac{N^i N^j}{N^2} \right) \partial_i \overline{\Phi} \partial_j \Phi + \left( m^2 + \xi R \right) |\Phi|^2 \right] \geq 0 \, .
\end{align}
Given \eqref{eq:KGscalarproductproof3} and recalling that $\omega > 0$, we can conclude that
\begin{align} \label{eq:KGscalarproductproof5}
\int_{\Sigma} \dd^{d-1} x \, \frac{\sqrt{h}}{N} \, \left[ \omega \, |\Phi|^2 - i \, \overline{\Phi} N^i \partial_i \Phi + i \, \Phi N^i \partial_i \overline{\Phi} \right] \geq 0 \, .
\end{align}

Therefore,
\begin{align} 
\tilde{\sigma}_{\mathbb{C}}(\Phi, \Phi) &= \int_{\Sigma} \dd^{d-1} x \, \frac{\sqrt{h}}{N} \, \left[ 2 \omega \, |\Phi|^2 - i \, \overline{\Phi} N^i \partial_i \Phi + i \, \Phi N^i \partial_i \overline{\Phi} \right] \notag \\
&\geq \int_{\Sigma} \dd^{d-1} x \, \frac{\sqrt{h}}{N} \, \left[ \omega \, |\Phi|^2 - i \, \overline{\Phi} N^i \partial_i \Phi + i \, \Phi N^i \partial_i \overline{\Phi} \right]
\geq 0 \, ,
\end{align}
which proves (i).
\end{proof}

It follows that a natural choice for the one-particle Hilbert space $\mathscr{H}$ for a globally stationary spacetime is the subspace of positive frequency solutions.

Let $\{ \Phi_i \}_{i \in \mathscr{I}}$ denote an orthonormal basis of $\mathscr{H}$ and let $t$ be a time function such that $\xi = \partial_t$ and $x = (t, \vec{x})$. One can write a positive frequency mode solution as
\begin{equation} \label{eq:positivefreqmodesolution}
\Phi_i (x) = e^{- i \omega_i t} \, \phi_i(\vec{x}) \, , \qquad
\omega_i > 0 \, .
\end{equation}
These expressions are written with a notation appropriate for the case in which the index set $\mathscr{I}$ is discrete, for notational simplicity, but they should be thought to also include the continuous case.

\subsection{Green's distributions}
\label{sec:Greensfunctions}

In this section, we present a brief description of the Green's distributions (and other closely related distributions) associated to the Klein-Gordon equation and their relation to the expectation values of products of the fields. We start by considering the case of systems with zero temperature before introducing the thermal Green's distributions.
This brief overview follows parts of \cite{birrell1984quantum,fulling1987temperature,fulling1989aspects}.

\subsubsection{Zero temperature Green's distributions}

To start with two examples, the bi-distributions $G_{\rm ret}$ and $G_{\rm adv}$ introduced in Definition~\ref{def:Greenoperators} satisfy
\begin{equation}
P_x \, G_{\rm ret}(x,x') = P_x \, G_{\rm adv}(x,x') = \frac{\delta(x,x')}{\sqrt{-g(x)}} \, ,
\end{equation}
in the sense of distributions, where $P_x$ is the Klein-Gordon operator at the point $x$ and the Dirac delta distribution $\delta(x,x')$ is a density in the first argument. Hence, $G_{\rm ret}$ and $G_{\rm adv}$ are \emph{Green's distributions} associated with the Klein-Gordon equation, $P \Phi = 0$. As is normal with hyperbolic differential operators, there are several different Green's distributions associated with a hyperbolic equation.

The bi-distribution $G$ introduced in Definition~\ref{def:causalpropagator}, on the other hand, satisfies the homogeneous equation
\begin{equation}
P_x \, G(x,x') = 0 \, ,
\end{equation}
and is not, strictly speaking, a Green's distribution. It is common practice, however, to designate this and other related distributions as ``Green's distributions'' or (misleadingly) ``Green's functions''.

Above, it was shown that $G(x,x')$ was related to products of two field operator-valued distributions $\Phi(x)$ and $\Phi(x')$ by the relation \eqref{eq:commutatorphi},
\begin{equation}
i G(x,x') \mathbb{I} = [ \Phi(x), \Phi(x') ] \, .
\end{equation}
Given the support properties of $G_{\rm ret}$ and $G_{\rm adv}$, one has
\begin{align}
i G_{\rm ret}(x,x') \mathbb{I} &= - \Theta(t-t') \, [ \Phi(x), \Phi(x') ] \, , \\
i G_{\rm adv}(x,x') \mathbb{I} &= \Theta(t'-t) \, [ \Phi(x), \Phi(x') ] \, ,
\end{align}
where $t$ is a time function on the spacetime and $\Theta$ is the Heaviside function. These Green's distributions are characterised by their support properties in spacetime and can be defined for any globally hyperbolic spacetime. Furthermore, they give the expectation value of different products of the quantum field, e.g.
\begin{equation}
\langle \Psi | [ \Phi(x), \Phi(x') ] | \Psi \rangle = i G(x,x') \, ,
\end{equation}
where $|\Psi \rangle$ is a normalised quantum state.

Another Green's distribution which will be important is the following.

\begin{definition}
The \emph{Feynman propagator} $G^{\rm F}$ associated with a vacuum state $|0 \rangle$ is defined as the expectation value of the time-ordered product of fields,
\begin{equation} \label{eq:Feynmanpropdef}
G^{\rm F}(x,x') := i \langle 0 | \mathscr{T} \left( \Phi(x) \Phi(x') \right) | 0 \rangle \, ,
\end{equation}
where $\mathscr{T}$ is the time-ordering operator,
\begin{equation}
\mathscr{T} \left( \Phi(x) \Phi(x') \right) := \Theta(t-t') \Phi(x) \Phi(x') + \Theta(t'-t) \Phi(x') \Phi(x) \, .
\end{equation}
\end{definition}

The Feynman propagator satisfies
\begin{equation} \label{eq:GFdiffeq}
P_x \, G^{\rm F}(x,x') = - \frac{\delta(x,x')}{\sqrt{-g(x)}} \, ,
\end{equation}
and, by definition, depends on the quantum state being considered.

\begin{definition}
The \emph{Wightman two-point functions} $G^{\pm}$ associated with a vacuum state $|0 \rangle$ are defined as,
\begin{equation}
G^+(x,x') := \langle 0 | \Phi(x) \Phi(x')| 0 \rangle \, , \qquad
G^-(x,x') := \langle 0 | \Phi(x') \Phi(x)| 0 \rangle \, .
\end{equation}
\end{definition}

They satisfy the homogeneous equation
\begin{equation}
P_x \, G^{\pm}(x,x') = 0 \, .
\end{equation}
On a stationary spacetime with time function $t$, the Wightman two-point functions can be expressed in terms of mode solutions of the form \eqref{eq:positivefreqmodesolution} as
\begin{equation}
G^{\pm}(x,x') = \sum_{i \in \mathscr{I}} e^{\mp i \omega_i (t-t')} \, \phi_i(x) \, \overline{\phi_i (x')} \, .
\end{equation}
If $\Delta t := t - t'$, one verifies that, as a function of $\Delta t$, $G^{\pm}$ is analytic when $\ImC [\Delta t] \in \mathbb{R}^{\mp}$. Moreover, if $|\Delta t| < |\Delta \vec{x}| \neq 0$, i.e.~if the points are spacelike separated, $G^+(x,x') = G^-(x,x')$ on the real axis. Therefore, there exists an holomorphic function $\mathscr{G}$ of $\Delta t$ on the cut complex plane $\mathbb{C} \setminus \left( (-\infty, - |\Delta \vec{x}|) \cup ( |\Delta \vec{x}|, \infty) \right)$ such that
\begin{equation}
\mathscr{G}(x,x') = \begin{cases}
G^+(x,x') \, , & \ImC [\Delta t] < 0 \, , \\
G^-(x,x') \, , & \ImC [\Delta t] > 0 \, ,
\end{cases}
\end{equation}
and both equalities hold when $\ImC [\Delta t] = 0$ and $|{\ReC [\Delta t]|} < |\Delta \vec{x}|$ (see Fig.~\ref{fig:complexplanenonthermal}). In other words, the Wightman two-point distribution $G^{\pm}$ is the boundary value of $\mathscr{G}$ as $\Delta t$ approaches the real axis from below/above:
\begin{equation}
G^{\pm}(\Delta t; \vec{x},\vec{x'}) = \lim_{\epsilon \to 0+} \mathscr{G}(\Delta t \mp i \epsilon ; \vec{x},\vec{x'}) \, .
\end{equation}

\begin{figure}[th!]
\begin{center}
{\small
\begin{tikzpicture}[scale=1]

\draw (-4,0) -- (-3.9,0);
\draw (-1,0) -- (1,0)
	node[pos=1, below, yshift=-1ex]{$|\Delta \vec{x}|$}
	node[pos=0, below, yshift=-1ex]{$-|\Delta \vec{x}|$}
	;
\draw (-1,-0.12) -- (-1,0.12);
\draw (1,-0.12) -- (1,0.12);
\draw[->] (3.9,0) -- (4,0)
	node[pos=1, below, yshift=-1ex]{$\ReC [\Delta t]$}	
	;
\draw[->] (0,-2) -- (0,2)
	node[pos=1, left, yshift=-1ex]{$\ImC [\Delta t]$}	
	;

\draw[decorate,decoration=zigzag] (1,0) -- (3.9,0);
\draw[decorate,decoration=zigzag] (-3.9,0) -- (-1,0);

\draw[->] (2.5,0.75) -- (2.5,0.25)
	node[pos=0, above, xshift=.5ex]{$G^-$}
	;
\draw[->] (-2.5,0.75) -- (-2.5,0.25)
	node[pos=0, above, xshift=.5ex]{$G^-$}
	;
\draw[->] (-2.5,-0.75) -- (-2.5,-0.25)
	node[pos=0, below, xshift=.5ex]{$G^+$}
	;
\draw[->] (2.5,-0.75) -- (2.5,-0.25)
	node[pos=0, below, xshift=.5ex]{$G^+$}
	;
\draw[->] (0.6,1) -- (0.15,1)
	node[pos=0, right]{$G^{\rm E}$}
	;

\end{tikzpicture}
}
\end{center}
\caption[Cut complex plane for the holomorphic function $\mathscr{G}$.]{\label{fig:complexplanenonthermal}Cut complex plane for the holomorphic function $\mathscr{G}$.
}
\end{figure}

On the imaginary axis, $\Delta t = i \Delta \tau \in i \mathbb{R}$, one has
\begin{equation}
\mathscr{G}(x,x') = \sum_{i \in \mathscr{I}} e^{- \omega_i |\Delta \tau|} \, \phi_i(x) \, \overline{\phi_i (x')} \, .
\end{equation}
Even though each mode term is not holomorphic in $\Delta \tau$, the series has an holomorphic limit. Let's assume for a moment that the spacetime is static. Then,
\begin{equation}
G^{\rm E}(\tau, \vec{x}; \tau', \vec{x'}) := \mathscr{G}(i \tau, \vec{x}; i \tau', \vec{x'}) \, ,
\end{equation}
where $\tau - \tau' = \Delta \tau$, can be shown to be the Green's distribution satisfying
\begin{equation} \label{eq:EuclideanGreenfunction}
\left( \square_x - m^2 \right) G^{\rm E}(x,x') = - \frac{\delta(x,x')}{\sqrt{g^{\rm E}(x)}} \, ,
\end{equation}
where $\square_x$ is the d'Alembertian on a Riemannian manifold with metric $g^{\rm E}$. In the static case, the operator $\square_x - m^2$ is elliptic and, hence, has a unique Green's distribution, $G^{\rm E}$, which is called the \emph{Euclidean Green's distribution}. If the spacetime is stationary, but not static, the operator is no longer elliptic in general and, therefore, uniqueness of the Green's distribution does not necessarily follow.

Note that the Feynman propagator $G^{\rm F}$ can be obtained from $G^{\rm E}$ by a rigid rotation of the domain from the imaginary axis to the real axis in a counter-clockwise direction,
\begin{align}
G^{\rm F}(\Delta t; \vec{x},\vec{x'}) &= i \lim_{\theta \to \pi/2 -} \mathscr{G}(- i \Delta t e^{i \theta} ; \vec{x},\vec{x'}) \label{eq:GFGholo} 
= \begin{cases}
i \, G^+(\Delta t; \vec{x},\vec{x'}) \, , & \Delta t > 0 \, , \\
i \, G^-(\Delta t; \vec{x},\vec{x'}) \, , & \Delta t < 0 \, ,
\end{cases}
\end{align}
which agrees with \eqref{eq:Feynmanpropdef}. Eq.\eqref{eq:GFGholo} is usually written as
\begin{equation}
G^{\rm F}(t, \vec{x}; t', \vec{x}') = i \, G^{\rm E} (t, \vec{x}; t', \vec{x}') \, .
\end{equation}
This relation is a crucial part for the ``Euclidean methods'' used to obtain the Feynman propagator on static spacetimes.

\subsubsection{Non-zero temperature Green's distributions}

The Green's distributions discussed so far have been computed for pure quantum states such as the vacuum state, and hence are appropriate for systems at zero temperature. We now turn to thermal equilibrium states. The expectation value of an operator $A$ for a thermal equilibrium state at temperature $T = 1/\beta$ corresponding to a time-independent Hamiltonian $H$ is given by the \emph{Gibbs formula}
\begin{equation}
\langle A \rangle_{\beta} := \frac{\Tr \left( e^{-\beta H} A \right)}{\Tr \left( e^{-\beta H} \right)} \, .
\end{equation}
Here, we assume that the \emph{density operator} $\rho := e^{-\beta H}$ is of trace class (see Definition~\ref{def:traceoperator}). This implies that $H$ must be an operator with purely point spectrum $\{ E_i \}_{i \in \mathscr{I}}$ and that
\begin{equation}
Z := \Tr \left( e^{-\beta H} \right) = \sum_{i \in \mathscr{I}} e^{- \beta E_i} < \infty \, .
\end{equation}

\begin{remark}
For a massive scalar field on a stationary spacetime with metric given by \eqref{eq:ADMmetricchap2}, the Hamiltonian $H$ is given by:
\begin{align}
H &= \frac{1}{2} \int_{\Sigma} \dd^{d-1} x \, N \sqrt{h} \left[ \left(1 - \frac{N_i N^i}{N^2} \tilde{\Pi}^2 \right)  + h^{ij} \left( \partial_i \Phi + \frac{N^i}{N} \tilde{\Pi} \right) \left( \partial_j \Phi + \frac{N^j}{N} \tilde{\Pi} \right) \right. \notag \\
&\hspace*{20ex} \left. + (m^2 + \xi R) \Phi^2 \right] \, ,
\end{align}
where $\Pi =: \sqrt{h} \, \tilde{\Pi}$. If $\Sigma$ is compact, it can be shown that $H$ is a positive, compact operator and its trace is just the sum of its eigenvalues, cf.~Proposition~\ref{prop:compactpositiveoperator}. More details can be found in \cite{fulling1987temperature}.
\end{remark}

Given this, one can define the thermal Wightman two-point functions as 
\begin{equation}
G^+_{\beta}(x,x') := \langle \Phi(x) \Phi(x') \rangle_{\beta} \, , \qquad
G^-_{\beta}(x,x') := \langle \Phi(x') \Phi(x) \rangle_{\beta} \, .
\end{equation}
In terms of mode solutions, it can be shown that they are given by
\begin{equation} \label{eq:thermalWightmanfunction}
G^{\pm}_{\beta}(x,x') = \sum_{i \in \mathscr{I}} \frac{\phi_i(x) \, \overline{\phi_i (x')}}{1 - e^{-\beta \omega_i}} \left( e^{\mp i \omega_i \Delta t} + e^{- \beta \omega_i} e^{\pm i \omega_i \Delta t}\right) \, .
\end{equation}
Similarly to the zero temperature case, as functions of $\Delta t$, $G^+_{\beta}$ and $G^-_{\beta}$ are analytic when $- \beta < \ImC [\Delta t] < 0$ and $0 < \ImC [\Delta t] < \beta$, respectively. Moreover, if $|\Delta t| < |\Delta \vec{x}| \neq 0$, i.e.~if the points are spacelike separated, $G^+_{\beta}(x,x') = G^-_{\beta}(x,x')$ on the real axis. Therefore, there exists an holomorphic function $\mathscr{G}_{\beta}$ of $\Delta t$ on the region $\{ z \in \mathbb{C} : |{\ImC [z]}| < \beta \} \setminus \left( (-\infty, - |\Delta \vec{x}|) \cup ( |\Delta \vec{x}|, \infty) \right)$ such that
\begin{equation}
\mathscr{G}_{\beta}(x,x') = \begin{cases}
G^+_{\beta}(x,x') \, , & - \beta < \ImC [\Delta t] < 0 \, , \\
G^-_{\beta}(x,x') \, , & 0 < \ImC [\Delta t] < \beta \, ,
\end{cases}
\end{equation}
and both equalities hold when $\ImC [\Delta t] = 0$ and $|{\ReC [\Delta t]|} < |\Delta \vec{x}|$ (see Fig.~\ref{fig:complexplanethermal}).

\begin{figure}
\begin{center}
{\small
\begin{tikzpicture}[scale=1]

\draw (-4,0) -- (-3.9,0);
\draw (-1,0) -- (1,0)
	node[pos=1, below, yshift=-1ex]{$|\Delta \vec{x}|$}
	node[pos=0, below, yshift=-1ex]{$-|\Delta \vec{x}|$}
	;
\draw (-1,-0.12) -- (-1,0.12);
\draw (1,-0.12) -- (1,0.12);
\draw[->] (3.9,0) -- (4,0)
	node[pos=1, below, yshift=-1ex]{$\ReC [\Delta t]$}	
	;
\draw[->] (0,-3.5) -- (0,3.5)
	node[pos=1, left, yshift=-1ex]{$\ImC [\Delta t]$}	
	;

\draw[decorate,decoration=zigzag] (1,0) -- (3.9,0)
	node[pos=0.5, above, yshift=3ex]{$G^-_{\beta}$}
	node[pos=0.5, below, yshift=-3ex]{$G^+_{\beta}$}	
	;
\draw[decorate,decoration=zigzag] (1,1.5) -- (3.9,1.5)
	node[pos=0.5, above, yshift=3ex]{$G^+_{\beta}$}
	;
\draw[decorate,decoration=zigzag] (1,-1.5) -- (3.9,-1.5)
	node[pos=0.5, below, yshift=-3ex]{$G^-_{\beta}$}
	;
\draw[decorate,decoration=zigzag] (-3.9,0) -- (-1,0)
	node[pos=0.5, above, yshift=3ex]{$G^-_{\beta}$}
	node[pos=0.5, below, yshift=-3ex]{$G^+_{\beta}$}	
	;
\draw[decorate,decoration=zigzag] (-3.9,1.5) -- (-1,1.5)
	node[pos=0.5, above, yshift=3ex]{$G^+_{\beta}$}
	;
\draw[decorate,decoration=zigzag] (-3.9,-1.5) -- (-1,-1.5)
	node[pos=0.5, below, yshift=-3ex]{$G^-_{\beta}$}
	;

\draw[->] (0.6,2.5) -- (0.15,2.5)
	node[pos=0, right]{$G^{\rm E}_{\beta}$}
	;

\draw (-0.12,1.5) -- (0.12,1.5)
	node[pos=0,left]{$\beta$}
	;
\draw (-0.12,-1.5) -- (0.12,-1.5)
	node[pos=0,left]{$-\beta$}
	;

\draw[dashed] (1,0) -- (1,1.5);
\draw[dashed] (1,-0.85) -- (1,-1.5);
\draw[dashed] (-1,0) -- (-1,1.5);
\draw[dashed] (-1,-0.85) -- (-1,-1.5);

\node at (2.4,3.15) {$\vdots$};
\node at (-2.5,3.15) {$\vdots$};
\node at (2.4,-3) {$\vdots$};
\node at (-2.5,-3) {$\vdots$};

\end{tikzpicture}
}
\end{center}
\caption{\label{fig:complexplanethermal}Cut complex plane for the holomorphic function $\mathscr{G}_{\beta}$.}
\end{figure}

From \eqref{eq:thermalWightmanfunction}, one can derive the important property of thermal Green's distributions,
\begin{equation}
G^+_{\beta}(\Delta t - i \beta ; \vec{x}, \vec{x}') = G^-_{\beta}(\Delta t ; \vec{x}, \vec{x}') \, ,
\end{equation}
which is usually known as the \emph{KMS condition}. This allows us to analytically continue $\mathscr{G}_{\beta}$ to $\mathbb{C} \setminus \{ z \in \mathbb{C} : \ImC [z] = N \beta , \, N \in \mathbb{Z}, \text{ and } |{\ReC [z]|} > | \Delta \vec{x} | \}$, by
\begin{equation} \label{eq:thermalGreenperiodicity}
\mathscr{G}_{\beta}(\Delta t ; \vec{x}, \vec{x}') = \mathscr{G}_{\beta}(\Delta t + i N \beta ; \vec{x}, \vec{x}') \, , \qquad N \in \mathbb{Z} \, .
\end{equation}

On the imaginary axis, one has 
\begin{equation}
\mathscr{G}_{\beta}(x,x') = \sum_{i \in \mathscr{I}} \frac{\phi_i(x) \, \overline{\phi_i (x')}}{1 - e^{-\beta \omega_i}} \left( e^{- \omega_i \Delta \tau} + e^{- \beta \omega_i} e^{\omega_i \Delta \tau}\right) \, ,
\end{equation}
for $0 < \Delta \tau < \beta$ (in the rest of the axis the expression can be obtained by using \eqref{eq:thermalGreenperiodicity}). Then, assuming that the spacetime is static,
\begin{equation} \label{eq:thermalEuclideanGreenfunction}
G^{\rm E}_{\beta}(\tau, \vec{x}; \tau', \vec{x'}) := \mathscr{G}_{\beta}(i \tau, \vec{x}; i \tau', \vec{x'}) \, ,
\end{equation}
where $\tau - \tau' = \Delta \tau$, can be shown to be the Green's distribution satisfying
\begin{equation}
\left( \square_x - m^2 \right) G^{\rm E}_{\beta}(x,x') = - \frac{\delta(x,x')}{\sqrt{g^{\rm E}(x)}} \, ,
\end{equation}
such that $G^{\rm E}_{\beta}(\Delta \tau + \beta ; \vec{x}, \vec{x'}) = G^{\rm E}_{\beta}(\Delta \tau ; \vec{x}, \vec{x'})$. That is, it is the Euclidean Green's distribution for the elliptic operator $\square_x - m^2$ acting on the cylinder $S^1 \times \Sigma$ of radius $\beta$. As in the non-thermal case, if the spacetime is stationary, but not static, this operator is no longer elliptic and there might not be a unique Green's distribution.

As before, the thermal Feynman propagator can be obtained by
\begin{align}
G^{\rm F}_{\beta}(\Delta t; \vec{x},\vec{x'}) = i \lim_{\theta \to \pi/2 -} \mathscr{G}_{\beta}(- i \Delta t e^{i \theta} ; \vec{x},\vec{x'})
\end{align}
which we will write simply as
\begin{equation} \label{eq:GFGEthermal}
G^{\rm F}_{\beta}(t, \vec{x}; t', \vec{x}') = i \, G^{\rm E}_{\beta} (t, \vec{x}; t', \vec{x}') \, .
\end{equation}

Finally, we note that the thermal and non-thermal Green's distributions can be related in the following way. It can be shown (see e.g.~\cite{birrell1984quantum}) that
\begin{equation} \label{eq:imagemsumG}
\mathscr{G}_{\beta}(\Delta t; \vec{x},\vec{x'}) = \sum_{N = - \infty}^{\infty} \mathscr{G}(\Delta t + i N \beta; \vec{x},\vec{x'}) \, ,
\end{equation}
i.e.~the thermal Green's distributions can be obtained as an imaginary-time image sum of the zero temperature Green's distributions. We will make use of this relation in Appendix~\ref{app:Minkowski} to write the thermal Green's distribution on the Minkowski spacetime in terms of its zero-temperature Green's distribution.

\subsection{Rotating black hole spacetimes}
\label{section:qftcst-rotatingbhs}

In this section, we focus on stationary black hole spacetimes, by which we mean black hole spacetimes which are locally stationary. The Kerr black hole in four dimensions is the most notable example. In contrast with the Schwarzschild black hole, the exterior region of the Kerr black hole does \emph{not} have a global timelike Killing vector field and, therefore, is not a globally stationary manifold in its own right. In the usual Boyer-Lindquist coordinates $(t, r, \theta, \phi)$, the Killing vector $\partial_t$ is timelike for $r > r_{\mathcal{S}}$, where $r = r_{\mathcal{S}}$ is the radial location of the \emph{stationary limit surface}, and spacelike in the region given by $r_+ < r < r_{\mathcal{S}}$ (the \emph{ergoregion}), where $r = r_+$ is the radial location of the event horizon. If we instead consider the generator $\chi = \partial_t + \Omega_{\mathcal{H}} \, \partial_{\phi}$ of the horizon (where $\Omega_{\mathcal{H}}$ is the angular velocity of the horizon), then $\chi$ is a Killing vector field and is timelike in the region $r_+ < r < r_{\mathcal{C}}$, where $r = r_{\mathcal{C}}$ is the radial location of the \emph{speed of light surface}, and is spacelike in the region given by $r > r_{\mathcal{C}}$. (More details about these statements for the Kerr spacetime can be found e.g.~in~\cite{Duffy:2005mz}.)

The quantisation procedure described above for globally stationary spacetimes, which chooses for the one-particle Hilbert space the subspace of positive frequency solutions of the field equation, is therefore not applicable to the exterior region of Kerr. The non-existence of an everywhere timelike Killing vector field in the exterior region of the spacetime is directly related to the non-existence of a well defined quantum vacuum state which is regular at the horizon and is invariant under the isometries of the spacetime. For the Kerr spacetime, this was noted by Frolov and Thorne \cite{Frolov:1989jh} and was proven in a seminal paper by Kay and Wald \cite{Kay:1988mu}. 

It is then expected that a state with these properties can be defined if we restrict the spacetime such that the scalar field does not have access to the region from the speed of light surface to infinity. This can be done by inserting a mirror-like, timelike boundary which respects the isometries of the spacetime. The simplest example is a boundary $\mathcal{M}$ at constant radius $r = r_{\mathcal{M}}$, on which the scalar field satisfies Dirichlet boundary conditions. If we choose the radius such that $r_{\mathcal{M}} \in (r_+, r_{\mathcal{C}})$, then the horizon generator $\chi$ is a timelike Killing vector field up to the boundary, and a vacuum state with the above properties is expected to be well defined. The introduction of timelike boundaries was suggested in \cite{Frolov:1989jh} and explored in 
\cite{Duffy:2005mz}. As far as we know, no rigorous proof of the existence of such a state is available. However, the heuristic arguments given above strongly suggest that such a conjecture is expected to be true and we shall take it as an assumption from this point onwards.

\begin{figure}[t!]
\begin{center}
{\small
\def\svgwidth{0.45\textwidth}
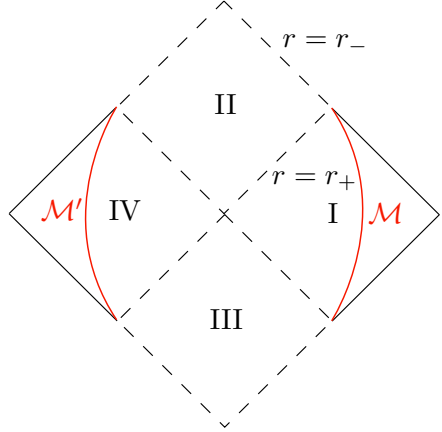
}
\caption[Carter-Penrose diagram of manifold with boundaries.]{\label{fig:CPdiagrammainfoldboundary} Carter-Penrose diagram of manifold $M$ described in the text.}
\end{center}
\end{figure}

For definiteness, we define the spacetime with timelike boundaries, $M$, to be the one constructed in the following way. Consider the non-extremal Kerr spacetime, $0 < |a| < M$, and let:
\begin{itemize}
\item region I be the exterior region;
\item region II be the black hole region;
\item region III be the white hole region;
\item region IV be the other asymptotically flat region.
\end{itemize}
The maximal analytical extension of Kerr comprises more regions, which we will not consider (see \cite{Wald:1984rg} for more details). In region I we insert a boundary $\mathcal{M}$ at constant radius $r = r_{\mathcal{M}}$, with $r_{\mathcal{M}} \in (r_+, r_{\mathcal{C}})$, on which Dirichlet boundary conditions are imposed. We denote by $\widetilde{\text{I}}$ the portion of the region I from the horizon up to the mirror. In region IV, a similar boundary $\mathcal{M}'$ is inserted, which can be obtained by the action of a discrete isometry $J$ which takes points in region I to points in region IV by a reflection about the bifurcation surface. In a similar way, a region $\widetilde{\text{IV}}$ is defined. We take as the new manifold $M$ of interest the union of regions $\widetilde{\text{I}}$, II, III and $\widetilde{\text{IV}}$ (see Fig.~\ref{fig:CPdiagrammainfoldboundary}).

The resulting manifold $M$ is \emph{not} globally hyperbolic and, hence, the quantisation procedure described above for globally hyperbolic spacetimes is not applicable. In section~\ref{sec:qftcst-stwithboundaries} it is described how to construct a quantum field theory on a spacetime with timelike boundaries. The upshot is that one expects that a isometry-invariant vacuum state which is regular at the horizons can be defined on the manifold $M$. The explicit construction of such a state is deferred to section~\ref{sec:computation-scalarfield}.

\begin{remark}
Note that if only the timelike boundary $\mathcal{M}$ had been introduced in region I, but no corresponding timelike boundary $\mathcal{M}'$ in region IV, it is conjectured that there is no isometry-invariant vacuum state which is regular at the horizons \cite{Kay:2015iwa}.
\end{remark}

\section{The case of spacetimes with boundaries}
\label{sec:qftcst-stwithboundaries}

So far, we have dealt with globally hyperbolic spacetimes, on which the Cauchy problem describing the dynamics of a scalar field is well posed. However, there are examples of non globally hyperbolic spacetimes on which we might be interested in constructing a quantum field theory, such as spacetimes with boundaries. In the last section, we argued that an isometry-invariant quantum state which is regular at the horizon of a rotating black hole spacetime can be defined if we restrict the spacetime by inserting appropriate timelike boundaries which respect the isometries of the spacetime. Another important example of spacetime with timelike boundaries is the anti-de Sitter (AdS) spacetime, whose spatial infinity provides a natural timelike boundary. A quantisation scheme for a scalar field on AdS was introduced by Avis, Isham and Storey \cite{Avis:1977yn} and the main idea is that appropriate boundary conditions need to be introduced so that there is a unique classical solution to the field equation. Since then, there have been several similar attempts (e.g.~\cite{Ishibashi:2004wx,Seggev:2003rp}) to construct consistent quantum field theories for non globally hyperbolic stationary spacetimes.

Therefore, instead of imposing global hyperbolicity, consider a stably causal, stationary spacetime. Recall from Proposition~\ref{prop:stablycausaltimefunction} that a stably causal spacetime has a time function. If, furthermore, it is stationary the following can be shown.

\begin{proposition}
If $M$ is a stably causal, stationary spacetime, then there
exists a spacelike surface $\Sigma$ which intersects each orbit of the timelike Killing vector field exactly once.
\end{proposition}

\begin{proof}
See Proposition 3.1 of \cite{Seggev:2003rp}.
\end{proof}

Even though a stably causal, stationary spacetime which is not globally hyperbolic does not possess a Cauchy surface on which initial data can be prescribed, it has spacelike surfaces which intersect the orbits of the timelike Killing vector field only once. However, initial data on such a surface is not enough to have a well posed initial value problem.

For concreteness, consider the spacetime $M$ constructed in the previous section, consisting of the portions of regions I, II, III and IV of the extended Kerr black hole between two timelike boundaries $\mathcal{M}$ and $\mathcal{M}'$. This spacetime is not globally hyperbolic, but it is stably causal and each region $\widetilde{\text{I}}$ and $\widetilde{\text{IV}}$ is a stably causal, stationary spacetime in its own right. On each region there is a spacelike surface which intersects the orbits of the timelike Killing vector field of each region exactly once. Assume that two such surfaces on regions $\widetilde{\text{I}}$ and $\widetilde{\text{IV}}$ meet each other at the bifurcation surface and let $\Sigma$ denote the whole spacelike surface on $M$, including the points at the timelike boundaries. We shall call this surface an ``initial-value surface''. It can be thought as the restriction of a Cauchy surface in the original globally hyperbolic spacetime without boundaries which passes through the bifurcation surface to the regions $\widetilde{\text{I}}$ and $\widetilde{\text{IV}}$, together with the bifurcation surface and the timelike boundaries (see Fig.~\ref{fig:CPdiagrammainfoldboundarysurface}). By construction, this ``initial-value surface'' is compact.

\begin{figure}[t!]
\begin{center}
{\small
\def\svgwidth{0.45\textwidth}
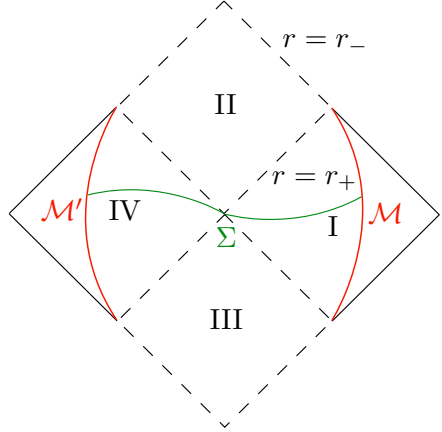
}
\caption[Carter-Penrose diagram of manifold with boundaries and ``initial-value surface''.]{\label{fig:CPdiagrammainfoldboundarysurface} Carter-Penrose diagram of manifold $M$ described in the text with an ``initial-value surface'' $\Sigma$.}
\end{center}
\end{figure}

One now considers the following mixed Dirichlet-Cauchy problem:
\begin{equation} \label{eq:mixedDirichletCauchyproblem}
\left\{
\begin{array}{ll}
P \Phi = 0 \, , \\
\Phi \big|_{\Sigma} = \Phi_0 \, , \\
\nabla_n \Phi \big|_{\Sigma} = \Phi_1 \, , \\
\Phi \big|_{\mathcal{M}} = \Phi \big|_{\mathcal{M}'} = 0 \, ,
\end{array}
\right.
\end{equation}
where $\Phi_0 , \, \Phi_1 \in C^{\infty}_0(\Sigma)$. Given that $\Sigma$ is compact and the boundaries are timelike, standard results on mixed Dirichlet-Cauchy problems guarantee the well posedness of \eqref{eq:mixedDirichletCauchyproblem} (see e.g.~chapter 24 of \cite{hormander2007analysis}).

Hence, even though the spacetime under consideration is not globally hyperbolic, upon imposing Dirichlet boundary conditions on the timelike boundaries the mixed Dirichlet-Cauchy problem is well posed. We then expect (and will assume) that our previous results derived for a globally hyperbolic spacetime to be carried over for the manifold $M$. Namely, we expect that the space of spacelike compact solutions to \eqref{eq:mixedDirichletCauchyproblem} to be a symplectic space with symplectic form given by \eqref{eq:symplecticform}, where the integral is evaluated on an ``initial-value surface'' $\Sigma$. Moreover, we expect that the construction of the quantum field theory from this symplectic space to remain valid and, therefore, we can obtain a Fock space whose vacuum state is regular and isometry-invariant, as discussed in the previous discussion. We also expect this vacuum state to be ``physically acceptable'' in the sense of having the defining features of a ``Hadamard state'', as defined for a globally hyperbolic spacetime. To conclude the discussion of the quantum scalar field theory, it remains to define this notion of ``physically acceptable'' state and discuss the issue of renormalisation of local observables, which is done in the last section of this chapter.

\section{Hadamard renormalisation}
\label{sec:qftcst-hadamardrenormalisation}

The objective of this thesis is to compute the expectation value of local observables for a given quantum field theory on a rotating black hole spacetime. Since observables in quantum field theory are self-adjoint operator-valued distributions, problems are bound to arise if we want to consider observables which are non-linear in the fields, such as $\Phi^2(x)$ and the stress-energy tensor $T_{ab}(x)$, which is given by
\begin{align}
T_{ab} &= (1-2\xi) \nabla_a \Phi \nabla_b \Phi + \left( 2\xi - \frac{1}{2} \right) g_{ab} \nabla^c \Phi \nabla_c \Phi - 2 \xi \Phi \nabla_b \nabla_a \Phi \notag \\ 
&\quad + 2 \xi g_{ab} \Phi \nabla^c \nabla_c \Phi + \xi \left( R_{ab} - \frac{1}{2} g_{ab} R \right) \Phi^2 - \frac{1}{2} g_{ab} m^2 \Phi^2 \, .
\end{align}
Since $\Phi$ is a distribution on spacetime, these observables involve taking the product of two distributions at the same spacetime point, which is not a well defined operation. Therefore, some kind of renormalisation procedure is necessary. In this section, we will describe the Hadamard renormalisation, which is an extension of the standard ``point-splitting method'' which uses the so-called Hadamard representation of the Green's distributions. References for this part are \cite{Decanini:2005eg,wald1994quantum}.

\subsection{The case of globally hyperbolic spacetimes}

First, note that both expectation values $\langle \Phi^2(x) \rangle$ and $\langle T_{ab} (x) \rangle$, with respect to a given quantum state $| \Psi \rangle$, can be given as spacetime limits of the Feynman propagator $G^{\rm F}$ associated with $| \Psi \rangle$. One has
\begin{equation}
\langle \Phi^2(x) \rangle = -i \lim_{x' \to x} G^{\rm F}(x, x') \, ,
\end{equation}
and
\begin{equation}
\langle T_{ab}(x) \rangle = \lim_{x' \to x} \mathcal{T}_{ab'}(x,x') \left[ -i \, G^{\rm F} (x,x') \right] \, ,
\end{equation}
where $\mathcal{T}_{ab'}(x,x')$ is an operator-valued bi-tensor constructed by point splitting,
\begin{align}
\mathcal{T}_{ab'} &= (1-2\xi) \, {g_{b}}^{b'} \nabla_{a} \nabla_{b'} + \left(2\xi - \frac{1}{2} \right) g_{ab} g^{cd'} \nabla_{c} \nabla_{d'} - 2 \xi \, {g_{a}}^{a'} {g_{b}}^{b'} \nabla_{a'} \nabla_{b'} \notag \\
&\quad +2\xi \, g_{ab} \nabla_{\rho} \nabla^{\rho} + \xi \left( R_{ab} - \frac{1}{2} g_{ab} R \right) - \frac{1}{2} g_{ab} m^2 \, ,
\end{align}
where ${g^a}_{b'}$ is the parallel propagator, as introduced in \eqref{eq:parallelpropagator}. These expectation values have been regularised by point splitting and now we are dealing with well-defined bi-distributions.

For simplicity, the prescription is described in detail for $\langle \Phi^2(x) \rangle$ and, at the end, the results are also presented for $\langle T_{ab} (x) \rangle$. The basic idea of the prescription is to ``subtract'' the short-distance singular behaviour of the bi-distribution, in this case the Feynman propagator $G^{\rm F}$. For that, one expands the Feynman propagator for small geodesic distance between $x$ and $x'$.

Assume that $x$ and $x'$ are in a geodesically convex neighbourhood, cf.~Definition~\ref{def:geodconvexneighbourhood}. Then, with respect to a class of quantum states to be defined below, the Feynman propagator has a \emph{Hadamard expansion} which depends on the spacetime dimension $d$.

\begin{enumerate}
\item For even $d \geq 4$, the Hadamard expansion of $G^{\rm F}$ is given by
\begin{align} \label{eq:Hadamardexpansioneven}
G^{\rm F}(x,x') &= i \alpha_d \left[ \frac{U(x,x')}{(\sigma(x,x')+i\epsilon)^{d/2-1}} + V(x,x') \log \left[ \sigma(x,x') + i\epsilon \right] + W(x,x') \right] \, .
\end{align}
\item For odd $d \geq 3$, the Hadamard expansion of $G^{\rm F}$ is given by
\begin{align} \label{eq:Hadamardexpansionodd}
G^{\rm F}(x,x') &= i \alpha_d \left[ \frac{U(x,x')}{(\sigma(x,x')+i\epsilon)^{d/2-1}} + W(x,x') \right] \, .
\end{align}
\end{enumerate}

In both cases, one has that
\begin{equation}
\alpha_d = \frac{\Gamma(d/2-1)}{2(2\pi)^{d/2}} \, ,
\end{equation}
$U(x,x')$, $V(x,x')$ and $W(x,x')$ are smooth symmetric bi-scalars which are regular when $x' \to x$, $\sigma(x,x')$ is the Synge's world function \eqref{eq:Syngesworldfunction}, and $i \epsilon$ with $\epsilon \to 0+$ is introduced to give $G^{\rm F}$ a singularity structure consistent with its definition as a time-ordered product.

These expansions of the Feynman propagator are only valid when evaluated for a special class of quantum states.

\vspace*{0.5ex}

\begin{definition} \label{def:Hadamardstate}
A quantum state for which the short-distance singularity of the Feynman propagator is given by either \eqref{eq:Hadamardexpansioneven} or \eqref{eq:Hadamardexpansionodd} is called a \emph{Hadamard state}.
\end{definition}

\vspace*{1ex}

\begin{remark}
The heuristic idea behind this definition is that the short-distance singularity structure of the Feynman propagator (equivalently, of the two-point function $\langle \Phi(x) \Phi(x') \rangle$) on a curved spacetime should be as close as possible to that on Minkowski spacetime. As seen below, the singular terms in the expansion of the Feynman propagator only depend on the local geometry (and not on the quantum state being considered) and, hence, it seems reasonable to require that a physically acceptable state has a two-point function $\langle \Phi(x) \Phi(x') \rangle$ with the same short-distance singularity structure as on Minkowski spacetime.
\end{remark}

\vspace*{0.5ex}

\begin{remark}
There is a more modern definition of Hadamard states introduced by Radzikowski \cite{Radzikowski:1996pa} which uses microlocal analysis techniques and looks into the singularity structure of the two-point function. It is equivalent to Definition~\ref{def:Hadamardstate} and the latter is sufficient for the purposes of this thesis.
\end{remark}

\vspace*{0.5ex}

\begin{remark}
Results from \cite{Fulling:1978ht} and \cite{Fulling:1981cf} show that every globally hyperbolic spacetime admit a wide class of Hadamard states. Analogous results are not available for spacetimes with boundaries, in which the definition of a Hadamard state needs to be modified (see Section~\ref{sec:renormalisationspacetimeswithboundaries}).
\end{remark}

The key point of these expansions of the Feynman propagator evaluated for Hadamard states is that it can be shown that the bi-scalars $U(x,x')$ and $V(x,x')$ only depend on the geometry along the geodesic joining $x$ and $x'$, whereas the bi-scalar $W(x,x')$ contains the quantum state dependence of the Feynman propagator (see e.g.~\cite{Decanini:2005eg}). Hence, the Hadamard expansion of the Feynman propagator contains a purely geometrical part which is \emph{singular} when $x' \to x$,
\begin{align} \label{eq:Hadsingularparteven}
G_{\rm Had}(x,x') := i \alpha_d \left[ \frac{U(x,x')}{(\sigma(x,x')+i\epsilon)^{d/2-1}} + V(x,x') \log \left[ \sigma(x,x') + i\epsilon \right]\right] \, ,
\end{align}
for even $d \geq 4$, and
\begin{align} \label{eq:Hadsingularpartodd}
G_{\rm Had}(x,x') := i \alpha_d \left[ \frac{U(x,x')}{(\sigma(x,x')+i\epsilon)^{d/2-1}} \right] \, ,
\end{align}
for odd $d \geq 3$, and a state dependent part, which is \emph{regular} when $x' \to x$,
\begin{align} \label{eq:regularpart}
G^{\rm F}_{\rm reg}(x,x') &= i \alpha_d \, W(x,x') \, .
\end{align}
The singular part \eqref{eq:Hadsingularparteven} or \eqref{eq:Hadsingularpartodd} is called the \emph{Hadamard singular part} of the Feynman propagator.

Given the Hadamard expansion of the Feynman propagator and its singular, non-state dependent Hadamard part, the final step of the Hadamard renormalisation procedure is to subtract this part from the Feynman propagator and use the regular part to define the renormalised local observables.

\begin{definition} \label{def:renormalisedquantities}
The \emph{renormalised vacuum polarisation} $\langle \Phi^2(x) \rangle_{\rm ren}$ with respect to a Hadamard state is defined as
\begin{equation}
\langle \Phi^2(x) \rangle_{\rm ren} := -i \lim_{x' \to x} G^{\rm F}_{\rm reg}(x, x') \, .
\end{equation}
The \emph{renormalised expectation value of the stress-energy tensor} $\langle T_{ab}(x) \rangle_{\rm ren}$ with respect to a Hadamard state is defined as
\begin{equation} \label{eq:renstressenergytensor}
\langle T_{ab}(x) \rangle_{\rm ren} := \lim_{x' \to x} \mathcal{T}_{ab'}(x,x') \left[ -i \, G^{\rm F}_{\rm reg}(x,x') \right] + \tilde{\Theta}_{ab}(x) \, ,
\end{equation}
where $\tilde{\Theta}_{ab}(x)$ is a state independent tensor which only depends on the local geometry and the parameters $m^2$ and $\xi$ of the scalar field and which guarantees that $\langle T_{ab}(x) \rangle_{\rm ren}$ is conserved, i.e.~$\nabla^a \langle T_{ab}(x) \rangle_{\rm ren} = 0$.
\end{definition}

\begin{remark}
The term $\tilde{\Theta}_{ab}(x)$ in \eqref{eq:renstressenergytensor} is necessary as the renormalisation procedure fails to provide a conserved tensor.  Additionally, an extra conserved tensor (denoted by $\Theta_{ab}(x)$ in \cite{Decanini:2005eg}), which also only depends on the local geometry and the parameters $m^2$ and $\xi$ of the scalar field, can be added to \eqref{eq:renstressenergytensor}, since the renormalised stress-energy tensor is defined only up to a local, conserved tensor \cite{Wald:1977up,Wald:1978pj,wald1994quantum}. This leaves an intrinsic ambiguity in the definition of the renormalised expectation value of the stress-energy tensor, which cannot be corrected without a full theory of quantum gravity or an experiment.
\end{remark}

As noted above, the Hadamard singular part of the Feynman propagator is purely geometrical and, hence, does not depend on which quantum state the Feynman propagator is being evaluated. It is possible to explicitly compute the Hadamard singular part (up to the required order in $\sigma$) in terms of the local geometry and the parameters $m^2$ and $\xi$ of the scalar field. The results for $2 \leq d \leq 6$ can be found in \cite{Decanini:2005eg}, here we present the results for $d = 3$ which are needed for this thesis.

\begin{proposition} \label{prop:Uexpansion}
For $d=3$, the covariant expansion of the Hadamard singular part \eqref{eq:Hadsingularpartodd} of the Feynman propagator is obtained using
\begin{equation} \label{eq:biscalarU}
U = U_0 + U_1 \sigma + \mathcal{O}(\sigma^2) \, ,
\end{equation}
with
\begin{align}
U_0 &= u_0 - u_{0a} \, \sigma^{;a} + \frac{1}{2!} u_{0ab} \, \sigma^{;a} \sigma^{;b} - \frac{1}{3!} u_{0abc} \, \sigma^{;a} \sigma^{;b} \sigma^{;c} + \mathcal{O}(\sigma^2) \, , \\
U_1 &= u_1 - u_{1a} \, \sigma^{;a} + \mathcal{O}(\sigma) \, ,
\end{align}
and where the coefficients are given by
\begin{equation}
u_0 = 1 \, , \quad u_{0a} = 0 \, , \quad u_{0ab} = \frac{1}{6} R_{ab} \, , \quad
u_{0abc} = \frac{1}{4} R_{(ab;c)} \, ,
\end{equation}
and
\begin{equation}
u_1 = m^2 + \left(\xi - \frac{1}{6}\right) R \, , \quad u_{1a} = \frac{1}{2} \left(\xi - \frac{1}{6}\right) R_{;a} \, .
\end{equation}
Here, $R$ is the Ricci scalar and $R_{ab}$ is the Ricci tensor of the spacetime.
\end{proposition}

\begin{remark}
Note that the expansion of the Hadamard singular part is only required up to $\sigma^2$ for the computation of the expectation value of the stress-energy tensor \eqref{eq:renstressenergytensor} since only two covariant derivatives are taken before taking the coincidence limit. This also allows us to bypass the discussion of the convergence of this expansion.
\end{remark}

\subsection{The case of spacetimes with timelike boundaries}
\label{sec:renormalisationspacetimeswithboundaries}

The above discussion of the Hadamard renormalisation and Hadamard states was formulated for globally hyperbolic spacetimes. Even though the focus was to analyse and ultimately subtract the short-distance singularity structure of the Feynman propagator, the concept of a quantum state is global, as evidenced by the non-existence of a natural vacuum state in a general curved spacetime.

Here, we present a modification of the definition of a Hadamard state for the case of a spacetime with boundaries which was first proposed by \cite{Kay:2015iwa}. For that, first we recall the notion of a causally convex set.

\begin{definition}
A subset $U$ of a spacetime $M$ is a \emph{causally convex set} if, whenever two points $x$ and $x'$ of $U$ can be connected by a causal curve in $U$, the portion of the causal curve between $x$ and $x'$ lies entirely in $U$.
\end{definition}

We note that this definition differs from the Definition~\ref{def:geodconvexneighbourhood} of a geodesically convex set. We then define an Hadamard state in a spacetime with boundaries in the following way. Let $\Interior M := M \setminus \partial M$ denote the interior of the spacetime.

\begin{definition}
If, for any causally convex subset $U$ of $\Interior M$ which is a globally hyperbolic spacetime on its own right, a quantum state is Hadamard in the usual sense in $U$, then we say that the quantum state is \emph{Hadamard} in $M$.
\end{definition}

Given this definition, the Hadamard renormalisation procedure described above is performed in exactly the same way. In Chapter~\ref{chap:localobservables}, an explicit implementation of the Hadamard renormalisation of local observables such as $\langle \Phi^2(x) \rangle$ for a scalar field on a rotating black hole spacetime is presented.

\chapter{Renormalised local observables in rotating black hole spacetimes}
\chaptermark{Renormalised local observables}
\label{chap:localobservables}

In this chapter, we present a method to renormalise a class of local observables for a scalar field on a rotating black hole. We will focus on (2+1)-dimensional spacetimes for simplicity, but we argue that this method should be easily generalised to a wide range of rotating black hole spacetimes in four and more spacetime dimensions. Additionally, the details of the computation will focus on the renormalised vacuum polarisation $\langle \Phi^2(x) \rangle$. In Chapter~\ref{chap:computation}, this computation will be made explicit for the case of a warped AdS${}_3$ black hole. Finally, at the end of this chapter, it is explained why this method fails to renormalise local observables such as the expectation values of the stress-energy tensor for a rotating black hole spacetime.

The contents of this chapter were published on \cite{Ferreira:2014ina,Ferreira:2015ipa}.

\section{Scalar field on a rotating black hole}
\label{sec:computation-scalarfield}

In this first section, the outline of the calculation of the renormalised vacuum polarisation for a scalar field on a (2+1)-dimensional rotating black hole is displayed. Namely, we identify the quantum state of interest and apply the renormalisation procedure explained in Section~\ref{sec:qftcst-hadamardrenormalisation}.

\subsection{(2+1)-dimensional rotating black hole}
\label{sec:rotatingBHgeneral}

For concreteness, in the following we consider a generic (2+1)-dimensional stationary black hole spacetime, whose metric is of the form given below. This is the case of the warped AdS${}_3$ black hole described in Chapter~\ref{chap:wadsbh}. However, it will be argued that the method should be applicable to a wide range of rotating black hole spacetimes in three and more spacetime dimensions.

We choose spherical coordinates $(t,r,\theta)$, where $t \in (-\infty, \infty)$, $r \in (0, \infty)$ and $(t,r,\theta) \sim (t,r,\theta + 2\pi)$. In these coordinates, the metric of a (2+1)-dimensional stationary black hole, according to \eqref{eq:ADMdecomposition}, can be written as
\begin{equation}
\dd s^2 = - N(r)^2 \, \dd t^2 + g_{rr}(r) \, \dd r^2 + g_{\theta\theta}(r) \left( \dd \theta + N^{\theta}(r) \, \dd t \right)^2 \, ,
\end{equation}
where $N(r)$ is the lapse function and $N^{\theta}(r)$ is the shift function. We make the following remarks about the black hole:
\begin{enumerate}
\item It is assumed that there is $r_+>0$ where $g^{rr}(r_+) = 0$, $N(r_+) = 0$ and $\sqrt{g_{rr}(r_+)} \, N(r_+)$ is finite, such that $r = r_+$ is the location of the \emph{event horizon}. The region of the spacetime in which $r>r_+$ is called the \emph{exterior region}.
\item In this coordinate system, $\partial_t$ and $\partial_{\theta}$ are Killing vector fields. Even though it might seem natural to assume that $\partial_t$ is timelike in some region $r > r'$, with $r' > r_+$, we do not make such a requirement. Instead, it is only required that the spacetime be locally stationary, cf.~Definition~\ref{def:locallystationary}, i.e.~any point of the spacetime must have a neighbourhood in which there is a timelike Killing vector field. In particular, there exists a Killing vector field of the form
\begin{equation} \label{eq:horizongeneratorchi}
\chi = \partial_t + \Omega_{\mathcal{H}} \, \partial_{\theta}
\end{equation}
which is timelike in the region $(r_+, r_{\mathcal{C}})$, with $r_{\mathcal{C}} > r_+$ and which generates the event horizon. This vector field is then null at the horizon and at the surface located at $r = r_{\mathcal{C}}$, which is called the \emph{speed of light surface}. The constant $\Omega_{\mathcal{H}}$ is then interpreted as the \emph{angular velocity of the horizon} with respect to the coordinate system.
\end{enumerate}

As seen in Section~\ref{section:qftcst-rotatingbhs}, for most cases of physical interest, such as the Kerr spacetime in four dimensions, there is \emph{not} any Killing vector field which is timelike everywhere in the exterior region $r > r_+$.  As a consequence, in the context of quantum field theory, there is not a well defined quantum vacuum state which is regular at the horizon and is invariant under the isometries of the spacetime. However, we argued that a vacuum state with these properties can be defined if we restrict the spacetime by inserting a mirror-like boundary which respects the isometries of the spacetime. The simplest example is a boundary $\mathcal{M}$ at constant radius $r = r_{\mathcal{M}}$, with $r_+ < r_{\mathcal{M}} < r_{\mathcal{C}}$, on which the field satisfies Dirichlet boundary conditions, as in this region $\chi$ is a timelike Killing vector field.

Given this remark, we add the following comments: 
\begin{enumerate} \setcounter{enumi}{2}
\item We assume that there exists a Dirichlet boundary $\mathcal{M}$ at $r = r_{\mathcal{M}}$ and focus only on the portion of the exterior region from the horizon up to the boundary, the region $\widetilde{\text{I}}$ defined in Section~\ref{section:qftcst-rotatingbhs}, which from now is what is meant by \emph{exterior region} of the black hole.
\item In region $\widetilde{\text{I}}$ of the black hole spacetime, it is convenient to consider ``co-rotating coordinates'' $(\tilde{t}=t, \, r, \,\tilde{\theta} = \theta - \Omega_{\mathcal{H}} t)$, such that the Killing vector field $\chi$ is given by $\chi = \partial_{\tilde{t}}$ and the metric is given by
\begin{equation}
\dd s^2 = - N(r)^2 \, \dd \tilde{t}^2 + g_{rr}(r) \, \dd r^2 + g_{\theta\theta}(r) \left( \dd \tilde{\theta} + \big( N^{\theta}(r) + \Omega_{\mathcal{H}} \big) \dd \tilde{t} \right)^2 \, .
\label{eq:metriccorotatingcoords}
\end{equation}
The coordinate $\tilde{t}$ is, by definition, a time function in this region.
\end{enumerate}

\subsection{Scalar field and Hartle-Hawking state}
\label{sec:HHstateconstruction}

We now turn to the theory of a real massive scalar field $\Phi$ on the exterior region of the rotating black hole. It obeys the Klein-Gordon equation \eqref{eq:KGequation},
\begin{equation}
P \Phi = \left(\nabla^2 - m^2 - \xi R \right) \Phi = 0 \, .
\label{eq:KGequation2}
\end{equation}

A thorough study of the classical theory for a general curved spacetime was done in Section~\ref{sec:qftcst-classicaltheory}, followed by a discussion of the quantisation procedure in Section~\ref{sec:qftcst-quantumtheory}. The aim now is to identify the one-particle Hilbert space $\mathscr{H}$ on which an isometry-invariant, regular state can be defined (as the vacuum state of the corresponding Fock space $\mathscr{F}_{\rm s}(\mathscr{H})$) and to construct the quantised scalar field as an operator-valued distribution which acts on the Fock space as in Definition~\ref{def:quantumscalarfield}. Given that the spacetime under consideration is (locally) stationary, the natural choice for the one-particle Hilbert space $\mathscr{H}$ is the space of positive frequency solutions of the Klein-Gordon equation with respect to a timelike Killing vector field, cf.~Section~\ref{sec:qftcst-positivefrequency}.

First, note that, using the co-rotating coordinates $(\tilde{t}, r, \tilde{\theta})$, $\partial_{\tilde{t}}$ and $\partial_{\tilde{\theta}}$ are Killing vector fields. Hence, there are \emph{mode solutions} of \eqref{eq:KGequation2} of the form
\begin{equation}
\Phi_{\tilde{\omega} k}(\tilde{t},r,\tilde{\theta}) = e^{-i \tilde{\omega} \tilde{t} + i k \tilde{\theta}} \, \phi_{\tilde{\omega} k}(r) \, ,
\label{eq:modesolutions}
\end{equation}
where $\tilde{\omega} \in \mathbb{R}$ and $k \in \mathbb{Z}$, such that a general (complex) solution can be written as
\begin{equation}
\Phi(\tilde{t}, r, \tilde{\theta}) = \int_{\mathbb{R}} \dd \tilde{\omega} \, \sum_{k \in \mathbb{Z}} N_{\tilde{\omega} k} \, \Phi_{\tilde{\omega} k}(\tilde{t},r,\tilde{\theta}) \, ,
\end{equation}
where $N_{\tilde{\omega} k}$ is a normalisation constant.

\begin{remark}
The mode solutions \eqref{eq:modesolutions} with $\tilde{\omega}>0$ form a basis for the space of positive frequency solutions of \eqref{eq:KGequation2} with respect to Killing vector $\chi = \partial_{\tilde{t}}$. This is \emph{not} the space used to define the desired vacuum state.
\end{remark}

We want to find a vacuum state which is regular at the horizons of the extended spacetime and invariant under the spacetime isometries. The one-particle Hilbert space consisting of positive frequency solutions with respect to the affine parameters of the horizon satisfies these conditions.

The affine parameters of the horizon are given by the usual Kruskal coordinates, similarly to the Kerr spacetime in four dimensions. Given the coordinate system $(\tilde{t}, r, \tilde{\theta})$, define the tortoise coordinate $r_*$ by
\begin{equation} \label{eq:tortoisecoordgeneral}
\frac{\dd r_*}{\dd r} = \frac{\sqrt{g_{rr}(r)}}{N(r)^2} \, ,
\end{equation}
such that $r \in (r_+ , \infty)$ is mapped to $r_* \in (-\infty, \infty)$, and define the two null coordinates
\begin{equation}
v := \tilde{t} + r_* \, , \qquad u := \tilde{t} - r_* \, .
\end{equation}
Define a new angular coordinate $\psi$ by
\begin{equation}
\dd \psi = \dd \tilde{\theta} + \left( N^{\theta}(r) + \Omega_{\mathcal{H}} \right) \dd \tilde{t} \, .
\end{equation}
The metric written in coordinates $(v, r, \psi)$ is given by
\begin{equation}
\dd s^2 = - N(r)^2 \dd v^2 + 2 \sqrt{g_{rr}(r)} N(r) \dd v \dd r + g_{\theta\theta}(r) \dd \psi \, .
\end{equation}

The Killing vector field $\chi$ in these coordinates is given by $\chi = \partial_v$. Since it is null at the horizon, the latter is a Killing horizon. To find its surface gravity $\kappa_+$, we compute, at $r=r_+$,
\begin{equation}
\chi^{\mu} \nabla_{\mu} \chi_{\nu} = - \chi^{\mu} \nabla_{\nu} \chi_{\mu}
= - \frac{1}{2} \, \nabla_{\nu} \! \left(\chi^2\right) = - \frac{1}{2} \, \partial_{\nu} N^2
= N \partial_r N \, \delta_{\nu r} \, ,
\end{equation}
and use the relation valid at $r=r_+$,
\begin{equation}
\chi^{\mu} \nabla_{\mu} \chi_{\nu} = \kappa_+ \chi_{\nu} \, ,
\end{equation}
to conclude that
\begin{equation} \label{eq:surfacegravitygeneral}
\kappa_+ = \left. \frac{\partial_r N}{\sqrt{g_{rr}}} \right|_{r=r_+} \, .
\end{equation}

The coordinate $v$ is a non-affine parameter along (part of) the horizon. To find an affine parameter, let $\tilde{\chi}$ denote an affinely parametrised generator of the horizon, i.e. $\tilde{\chi}^a \nabla_a \tilde{\chi}^b = 0$ on the horizon. It is easy to check that $\tilde{\chi} = e^{\kappa_+ v} \chi$ is such a generator. Hence, if we denote by $V$ the affine parameter, such that $\tilde{\chi} = \partial_V$, then $V \propto e^{\kappa_+ v}$.  Analogously, we have that $U \propto e^{- \kappa_+ u}$ is an affine parameter along (part of) the horizon.

This suggests that we define the Kruskal-like coordinates as
\begin{equation}
V: = e^{\kappa_+ v} \, , \qquad U := - e^{- \kappa_+ u} \, .
\end{equation}
These new coordinates allow us to analytically extend the spacetime in the standard way. In the extended spacetime, the horizon is now given by either $V = 0$ or $U = 0$ and the original exterior region is given by $V>0$ and $U<0$. Denote the surface $V = 0$ by $\mathcal{H}^-$ (the \emph{past horizon}) and the surface $U=0$ by $\mathcal{H}^+$ (the \emph{future horizon}). Their intersection is the \emph{bifurcation surface}. Denote the original exterior region by region I. As in Section~\ref{section:qftcst-rotatingbhs}, we assume that there is a Dirichlet boundary $\mathcal{M}$ at $r = r_{\mathcal{M}}$ in the original exterior region and denote by $\tilde{\text{I}}$ the portion of the region I from the horizon up to the boundary. We also assume that there is another Dirichlet boundary $\mathcal{M}'$ in the region given by $V<0$ and $U>0$ (region IV), which can be obtained by the action of a discrete isometry $J : (U,V, \tilde{\theta}) \mapsto (-U,-V,\tilde{\theta})$, taking points from region I to region IV by a reflection about the bifurcation surface. Denote by $\widetilde{\text{IV}}$ the portion of region IV up to the boundary $\mathcal{M}'$. 

From the results above, one can conclude that $V$ is an affine parameter along $\mathcal{H}^+$, whereas $U$ is an affine parameter along $\mathcal{H}^-$. In other words, the vector field $\partial_V$ is tangent to affinely parametrised null geodesics along $\mathcal{H}^+$, whereas $\partial_U$ is tangent to affinely parametrised null geodesics along $\mathcal{H}^-$. Therefore, denote by $\mathscr{H}$ the one-particle Hilbert space of solutions such that
\begin{enumerate}[label={(\roman*)}]
\item when restricted to $\mathcal{H}^+$ are positive frequency with respect to $\partial_V$;
\item when restricted to $\mathcal{H}^-$ are positive frequency with respect to $\partial_U$;
\item the Dirichlet boundary condition at $\mathcal{M}$ and $\mathcal{M}'$ is satisfied.
\end{enumerate}
Let $\mathscr{F}_{\rm s}(\mathscr{H})$ denote the Fock space associated with $\mathscr{H}$, cf.~Definition~\ref{def:Fockspace}. We then define:

\vspace*{0.4ex}

\begin{definition}
The \emph{Hartle-Hawking state} $|H\rangle$ is the vacuum state of $\mathscr{F}_{\rm s}(\mathscr{H})$.
\end{definition}

\vspace*{0.8ex}

\begin{remark}
We call this state the Hartle-Hawking state as this state satisfies the same defining properties as the state defined on a Schwarzschild black hole \cite{Hartle:1976tp}, namely the regularity at the horizons and the isometry invariance. Hence, this state can be thought as the natural generalisation to the rotating case with mirrors.
\end{remark}

Having defined the Fock space of interest, one can define the quantum field in the way detailed in Section~\ref{sec:qftcst-quantumtheory}. For that, one picks an orthonormal basis of $\mathscr{H}$. In the following, such a basis is constructed.
\begin{enumerate}
\item Let $\big\{ \Phi^{\tilde{\text{I}}}_{\tilde{\omega} k} \big\}_{\tilde{\omega}>0, k \in \mathbb{Z}}$ be the orthonormal basis of mode solutions of the form \eqref{eq:modesolutions} on region $\tilde{\text{I}}$ which satisfy the Dirichlet boundary condition at $\mathcal{M}$.
\item Define mode solutions on region $\widetilde{\text{IV}}$, $\Phi^{\widetilde{\text{IV}}}_{\tilde{\omega} k}$, by the action of the isometry $J$,
\begin{equation}
\Phi^{\widetilde{\text{IV}}}_{\tilde{\omega} k}(x) := \overline{\Phi^{\tilde{\text{I}}}_{\tilde{\omega} k} \left( J^{-1}(x) \right)} \, , \qquad
x \in \widetilde{\text{IV}} \, .
\end{equation}
It is understood that the modes $\Phi^{\tilde{\text{I}}}_{\tilde{\omega} k}$ only have support on region $\tilde{\text{I}}$ and that the modes $\Phi^{\widetilde{\text{IV}}}_{\tilde{\omega} k}$ only have support on region $\widetilde{\text{IV}}$.
\item In the union $\tilde{\text{I}} \cup \widetilde{\text{IV}}$, define the new mode solutions $\Phi^{\text{L}}_{\tilde{\omega}k}$ and $\Phi^{\text{R}}_{\tilde{\omega}k}$ by
\begin{subequations}
\begin{align}
\Phi^{\text{L}}_{\tilde{\omega}k}(x) &:= \frac{1}{\sqrt{1-e^{-2\pi \tilde{\omega} /\kappa_+}}} \left( \Phi^{\widetilde{\text{IV}}}_{\tilde{\omega}k}(x) + e^{-\pi \tilde{\omega} /\kappa_+} \overline{\Phi^{\widetilde{\text{I}}}_{\tilde{\omega}k}(x)} \right) , \quad x \in \widetilde{\text{I}} \cup \widetilde{\text{IV}}  , \\
\Phi^{\text{R}}_{\tilde{\omega}k}(x) &:= \frac{1}{\sqrt{1-e^{-2\pi \tilde{\omega} /\kappa_+}}} \left( \Phi^{\widetilde{\text{I}}}_{\tilde{\omega}k}(x) + e^{-\pi \tilde{\omega} /\kappa_+} \overline{\Phi^{\widetilde{\text{IV}}}_{\tilde{\omega}k}(x)} \right) , \quad x \in \widetilde{\text{I}} \cup \widetilde{\text{IV}}  .
\end{align}
\end{subequations}
These L and R modes can be analytically extended across the horizons.
\end{enumerate}

\vspace*{0.5ex}

\begin{proposition} \label{prop:LRmodespositivefreq}
The L and R mode solutions are of positive frequency with respect to the affine parameters of $\mathcal{H}^+$ and $\mathcal{H}^-$.
\end{proposition}

\vspace*{0.3ex}

\begin{proof}
We only show that $\Phi^{\rm R}_{\tilde{\omega}k}$ is of positive frequency with respect to the affine parameter $U$ of $\mathcal{H}^-$. For that, we want to decompose $\Phi^{\rm R}_{\tilde{\omega}k}$ into its positive and negative frequency parts with respect to $U$ and show that the latter vanishes. The Fourier transform of $\Phi^{\rm R}_{\tilde{\omega}k}$ with respect to $U$ is
\begin{equation}
\tilde{\Phi}^{\rm R}_{\tilde{\omega}k}(\sigma, r, \tilde{\theta}) = \int_{-\infty}^{\infty} \dd U \, e^{i \sigma U} \, \Phi^{\rm R}_{\tilde{\omega}k} (U,r, \tilde{\theta}) \, ,
\label{eq:FouriertransfU}
\end{equation}
where
\begin{align*}
\Phi^{\rm R}_{\tilde{\omega}k} (U,r, \tilde{\theta}) 
&= \frac{1}{2\pi} \int_{-\infty}^{\infty} \dd \sigma \, e^{-i \sigma U} \, \tilde{\Phi}^{\rm R}_{\tilde{\omega}k}(\sigma, r, \tilde{\theta}) \notag \\
&= \frac{1}{2\pi} \int_{0}^{\infty} \dd \sigma \, e^{-i \sigma U} \, \tilde{\Phi}^{\rm R}_{\tilde{\omega}k}(\sigma, r, \tilde{\theta}) + \frac{1}{2\pi} \int_{0}^{\infty} \dd \sigma \, e^{i \sigma U} \, \tilde{\Phi}^{\rm R}_{\tilde{\omega}k}(-\sigma, r, \tilde{\theta}) \, .
\end{align*}
We decomposed $\Phi^{\rm R}_{\tilde{\omega}k}$ in its positive and negative frequency parts with respect to $U$. If $\tilde{\Phi}^{\rm R}_{\tilde{\omega}k}(-\sigma, r, \tilde{\theta}) = 0$ for $\sigma > 0$, then $\Phi^{\rm R}_{\tilde{\omega}k}$ is of positive frequency with respect to $U$.

Suppose that $\Phi^{\rm R}_{\tilde{\omega}k}(U,r, \tilde{\theta})$ is analytic in the lower half of the complex $U$-plane and, furthermore, that
\begin{equation}
\lim_{R \to \infty} \, \max_{\theta \in (-\pi, 0)} \left| \Phi^{\rm R}_{\tilde{\omega}k} (R e^{i \theta} ,r, \tilde{\theta}) \right| = 0 \, .
\label{eq:decaylhpU}
\end{equation}
Then, we can apply Jordan's lemma to \eqref{eq:FouriertransfU} when $\sigma < 0$ and close the integration contour in the lower half plane to conclude that $\tilde{\Phi}^{\rm R}_{\tilde{\omega}k}(\sigma, r, \tilde{\theta}) = 0$ when $\sigma < 0$.

To show this, first note that, on the past horizon $\mathcal{H}^-$, $\Phi^{\tilde{\text{I}}}_{\tilde{\omega}k}$ is of the form
\begin{equation}
\Phi^{\tilde{\text{I}}}_{\tilde{\omega}k}(U, r_+, \tilde{\theta}) 
\sim e^{i \frac{\tilde{\omega}}{\kappa_+} \log(-U) + i k \tilde{\theta}} \, \Theta(-U) \, ,
\end{equation}
where $\Theta(U)$ is the Heaviside function, whereas
\begin{equation}
\Phi^{\widetilde{\text{IV}}}_{\tilde{\omega}k}(U, r_+, \tilde{\theta}) 
\sim e^{-i \frac{\tilde{\omega}}{\kappa_+} \log(U) - i k \tilde{\theta}} \, \Theta(U) \, .
\end{equation}

Thus, on $\mathcal{H}^-$,
\begin{align}
\Phi^{\rm R}_{\tilde{\omega}k}(U,r_+, \tilde{\theta}) \sim \left[ e^{i \frac{\tilde{\omega}}{\kappa_+} \log(-U)} \Theta(-U) +  e^{-\frac{\pi\tilde{\omega}}{\kappa_+}}  e^{i \frac{\tilde{\omega}}{\kappa_+} \log(U)} \Theta(U) \right] e^{i k \tilde{\theta}} \, .
\end{align}
To check that $\Phi^{\rm R}_{\tilde{\omega}k}(U,r_+, \tilde{\theta})$ is analytic in the lower half of the complex $U$-plane, we extend the logarithm in the complex plane by taking the branch cut to lie in the upper half plane. Then,
\begin{equation}
e^{i \frac{\tilde{\omega}}{\kappa_+} \log(U)} 
= e^{i \frac{\tilde{\omega}}{\kappa_+} \left[ \log(-U) - i \pi \right]} 
= e^{\frac{\pi\tilde{\omega}}{\kappa_+}} e^{i \frac{\tilde{\omega}}{\kappa_+} \log(-U)} \, ,
\end{equation}
and we can write
\begin{align}
\Phi^{\rm R}_{\tilde{\omega}k}(U,r_+, \tilde{\theta}) \sim e^{i \frac{\tilde{\omega}}{\kappa_+} \log(-U)} \, e^{i k \tilde{\theta}} \, ,
\end{align}
for all $U$. This is analytic in the lower half of the complex $U$-plane.

However, it does not satisfy \eqref{eq:decaylhpU}. This is due to the fact that we are dealing with non-normalisable mode solutions with sharp frequencies, which leads that in
\begin{equation}
\tilde{\Phi}^{\rm R}_{\tilde{\omega}k}(\sigma, r_+, \tilde{\theta}) 
\sim \int_0^{\infty} \dd U \, e^{i \sigma U + i\frac{\tilde{\omega}}{\kappa_+} \log(-U)} \, e^{i k \tilde{\theta}}
\end{equation}
the integral does not converge absolutely. 

Instead, we should consider \emph{wave-packets} constructed as superpositions of positive frequency modes, such as the ones in \cite{Hawking:1974sw,Wald:1975kc}
\begin{equation}
\Phi_{jn} = \frac{1}{\sqrt{\epsilon}} \int_{j \epsilon}^{(j+1)\epsilon} d\omega \, e^{-i \frac{2\pi n}{\epsilon}} \Phi_{\omega} \, ,
\end{equation}
where $j \in \mathbb{N}$, $n \in \mathbb{Z}$, $\epsilon > 0$ and $\Phi_{\omega}$ is a mode solution generated by data of the form $e^{i \omega u}/\sqrt{\omega}$ at future null infinity. These wave-packets are made of frequencies within $\epsilon$ of $j \epsilon$ and are peaked around retarded time $u = \frac{2\pi n}{\epsilon}$ with spread $\sim 1/\epsilon$. Then, the additional integration over the frequencies required to construct the wave-packets make the above integrals over $U$ convergent. 
\end{proof}

Hence, we take the one-particle Hilbert space $\mathscr{H}$ to consist of the L and R mode solutions. The quantum scalar field $\Phi(x)$ is then defined to be
\begin{equation}
\Phi(x) = \sum_{k=-\infty}^{\infty} \int_0^{\infty} \dd \tilde{\omega} \left[ a^{\text{L}}_{\tilde{\omega}k} \Phi^{\text{L}}_{\tilde{\omega}k}(x) + a^{\text{R}}_{\tilde{\omega}k} \Phi^{\text{R}}_{\tilde{\omega}k}(x) + \text{h.c.} \right] \, ,
\end{equation}
where h.c.~stands for ``hermitian conjugate''. The Hartle-Hawking state $|H\rangle$ is such that $a^{\text{L}}_{\tilde{\omega}k} |H \rangle = a^{\text{R}}_{\tilde{\omega}k} |H \rangle = 0$. Moreover, one can define the Feynman propagator $G^{\rm F}$ evaluated for this quantum state by \eqref{eq:Feynmanpropdef}.

\begin{remark} \label{rem:HHstatethermal}
Note that the Hartle-Hawking state is the \emph{vacuum} state of the Fock space $\mathscr{F}_{\rm s}(\mathscr{H})$ associated with the one-particle Hilbert space $\mathscr{H}$ consisting of the L and R modes. As shown by
\cite{Hartle:1976tp,Israel:1976ur}, it is however a \emph{thermal} state with respect to the horizon generator Killing vector field $\chi$, introduced in \eqref{eq:horizongeneratorchi}. Therefore, the Feynman propagator $G^{\rm F}$ evaluated for the Hartle-Hawking state, when written using the coordinates $(\tilde{t}, r, \tilde{\theta})$ (which only cover region $\tilde{\text{I}}$), is a thermal Green's distribution, as defined in Section~\ref{sec:Greensfunctions}.
\end{remark}

\subsection{Hadamard renormalisation}

In Section~\ref{sec:qftcst-hadamardrenormalisation} the Hadamard renormalisation of the vacuum polarisation $\langle \Phi^2(x) \rangle$ was described. We concluded that the renormalized vacuum polarization $\langle \Phi^2(x) \rangle$ in any Hadamard state is given by
\begin{equation} \label{eq:Grengeneral}
\langle \Phi^2(x) \rangle_{\rm ren} := - i \lim_{x' \to x} \left[ G^{\text{F}}(x,x') - G_{\text{Had}}(x,x') \right] \, ,
\end{equation}
where $G_{\text{Had}}$ is the Hadamard singular part of the Feynman propagator $G^{\text{F}}$, evaluated for the Hadamard state. In the current three-dimensional setting, it is given by
\begin{equation} \label{eq:Hadamardsingularpartfull}
G_{\text{Had}}(x,x') = \frac{i}{4\sqrt{2}\pi} \frac{U(x,x')}{\sqrt{\sigma(x,x')+i\epsilon}} \, ,
\end{equation}
with the bi-scalar $U$ given by \eqref{eq:biscalarU}. This bi-scalar can be expressed as a covariant Taylor expansion (see section~\ref{sec:maths-bitensorexpansions}) as
\begin{equation}
U(x,x') = \sum_{k=0}^{\infty} U_k(x,x') \, \sigma^{;k}(x,x') \, .
\end{equation}
For the computation of the vacuum polarization, it is sufficient to know the zeroth term, $U_0(x,x') = 1 + \mathcal{O}(\sigma)$, thus,
\begin{equation} \label{eq:Hadamardsingularpartfirsterm}
G_{\text{Had}}(x,x') = \frac{i}{4\sqrt{2}\pi} \frac{1}{\sqrt{\sigma(x,x')+i\epsilon}} + \mathcal{O}(\sigma^{1/2}) \, .
\end{equation}

At this stage, we are faced with two important technical difficulties. To perform the subtraction in \eqref{eq:Grengeneral}, we need to compute the Feynman propagator $G^{\rm F}$ evaluated for the Hartle-Hawking state and the state-independent Hadamard singular part $G_{\rm Had}$. The former is usually obtained as a sum over mode solutions of the differential equation \eqref{eq:GFdiffeq} satisfied by $G^{\rm F}$, whereas the latter is given in closed form by \eqref{eq:Hadamardsingularpartfull}. This implies that, unless we are able to express the mode sum in closed form, which is generally not possible, we need to express $G_{\rm Had}$ as a sum over mode solutions, such that the short-distance divergence can be subtracted term by term. This will be done in Section~\ref{sec:renormalisation-procedure}.

First, however, we need to compute the Feynman propagator $G^{\rm F}$ and write it as sum over mode solutions of \eqref{eq:GFdiffeq}. If the background spacetime were static, the standard technique to obtain the Feynman propagator would be to consider the real Riemannian section of the static spacetime, cf.~Definition~\ref{def:realriemanniansection}, and obtain the Euclidean Green's distribution $G^{\rm E}$ which satisfies the differential equation \eqref{eq:thermalEuclideanGreenfunction}. The Euclidean Green's distribution is unique, due to the ellipticity of the Klein-Gordon operator in the real Riemannian section, and it can be easily computed using standard Green's functions techniques. The Feynman propagator for the original static spacetime can then be obtained by using \eqref{eq:GFGEthermal}.

Since our spacetime is stationary, but non static, this technique is no longer valid. If we instead consider the complex Riemannian section of the spacetime, cf.~Definition~\ref{eq:complexRiemanniansection}, there is no guarantee that, in general, there is a unique Green's distribution that solves the differential equation, as the Klein-Gordon operator is not elliptic in the complex Riemannian section. In the next section, we describe how the complex Riemannian section can still be used to compute the Feynman propagator in this case.

\section{Quasi-Euclidean method}
\label{sec:quasi-euclidean-method}

In this section, we present the ``quasi-Euclidean method'' to compute the Feynman propagator for a scalar field in the Hartle-Hawking state on a rotating black hole spacetime. This is a generalisation of the ``Euclidean method'' used for static spacetimes and involves the complex Riemannian section of the exterior region of the rotating black hole spacetime with a timelike boundary. Ideas similar to the ones presented in Section~\ref{sec:qftcst-stwithboundaries} allows us to conclude that there exists a unique Green's distribution associated with the Klein-Gordon equation in the complex Riemannian section which can obtained as a mode sum using standard Green's functions techniques.

The complex Riemannian section of certain rotating spacetimes has been briefly discussed in \cite{Gibbons:1976ue,Frolov:1982pi,Brown:1990di} in the context of the Kerr-Newman black hole. In \cite{Moretti:1999fb}, a more general concept of ``local Wick rotation'' is discussed for any Lorentzian manifold, even without a timelike Killing vector field, as long as its metric is a locally analytic function of the coordinates.

\subsection{Complex Riemannian section}

In Section~\ref{sec:riemanniansections} the real Riemannian section of a static spacetime was defined. In short, a static spacetime can be thought of as a real Lorentzian section of a complex manifold, for which it is always possible to find a real Riemannian section by performing an appropriate analytical continuation. For a (2+1)-dimensional static spacetime whose metric in coordinates $(t,r,\theta)$ is
\begin{equation}
\dd s^2 = - N(r)^2 \, \dd t^2 + g_{rr}(r) \, \dd r^2 + g_{\theta\theta}(r) \, \dd \theta^2 \, ,
\end{equation}
where $t$ is a global time function, one can obtain the real Riemannian section by performing a Wick rotation $t \to -i \tau \in i \mathbb{R}$,
\begin{equation}
\dd s^2_{\mathbb{R}} = N(r)^2 \, \dd \tau^2 + g_{rr}(r) \, \dd r^2 + g_{\theta\theta}(r) \, \dd \theta^2 \, .
\end{equation}

The analytic continuation procedure does unfortunately not have an immediate generalization to spacetimes that are stationary but not static which generates a real Riemannian section. For the exterior of a rotating black hole, one issue is that the exterior need not have a globally timelike Killing vector even when each point in the exterior has a neighbourhood with such a Killing vector, i.e.~it is locally stationary, as we saw in Section~\ref{section:qftcst-rotatingbhs}. A second issue is that there may exist no analytic continuation in the coordinates that results in a real Riemannian section. Both of these issues are present in Kerr (for which the absence of a real section with a positive definite metric was shown in~\cite{Woodhouse:1977-complex}) and in the $(2+1)$-dimensional warped AdS${}_3$ black holes considered in Chapter~\ref{chap:wadsbh}. It is possible to obtain a positive definite metric by analytically continuing not just the coordinates but also the parameters (e.g.~the angular momentum parameter in Kerr \cite{hawking1979general}), but the physical relevance of continuing parameters seems debatable \cite{Brown:1990di}. 

If we only consider region~$\widetilde{\text{I}}$ of the (2+1)-dimensional rotating black hole spacetime, there exists an everywhere timelike Killing vector field, $\chi = \partial_{\tilde{t}}$. If we perform a Wick rotation $\tilde{t} = -i \tau \in i \mathbb{R}$, the metric \eqref{eq:metriccorotatingcoords} becomes
\begin{equation} \label{eq:metriccomplexRiemanniansectiongen}
\dd s^2_{\mathbb{C}} = N(r)^2 \, \dd \tau^2 + g_{rr}(r) \, \dd r^2 + g_{\theta\theta} \left( \dd \tilde{\theta} - i \, \big( N^{\theta}(r) + \Omega_{\mathcal{H}} \big) \dd \tau \right)^2 \, .
\end{equation}
This is the complex-valued metric $g^{\mathbb{C}}$ of the complex Riemannian section $I^{\mathbb{C}}$ of a complex manifold, in which region $\widetilde{\text{I}}$ is a real Lorentzian section, cf.~Definition~\ref{eq:complexRiemanniansection}. 

This metric is regular at $r=r_+$ if $\tau$ is periodic with period $2\pi/\kappa_+$, where $\kappa_+$ is the surface gravity of the black hole obtained in \eqref{eq:surfacegravitygeneral}. Otherwise, there would be a conical singularity at $r=r_+$. The resulting manifold has two periodic directions and a third direction that is also compact due to the boundary at $r = r_{\mathcal{M}}$.

\subsection{Green's distribution in the Riemannian section}

In the real Lorentzian section of the rotating black hole, we defined the Feynman propagator $G^{\rm F}$ evaluated for the Hartle-Hawking state in the usual way as a bi-distribution on $M$. In the complex Riemannian section $I^{\mathbb{C}}$, we find the Green's distribution $G$ associated with the Klein-Gordon equation (which should \emph{not} be confused with the causal propagator $G$ introduced in Definition~\ref{def:causalpropagator}). Given the construction of $I^{\mathbb{C}}$, the results obtained in this section will only be relevant for region $\widetilde{\text{I}}$ of the original spacetime.

In the complex Riemannian section $I^{\mathbb{C}}$, the Green's distribution $G$ associated with the Klein-Gordon equation satisfies the distributional equation
\begin{equation} \label{eq:GCfunctioneq}
\left( \nabla^2 - m^2 - \xi R \right) G(x,x') = - \frac{\delta^3(x,x')}{\sqrt{g(x)}} = - \frac{\delta(\tau-\tau') \delta(r-r') \delta(\tilde{\theta}-\tilde{\theta}')}{\sqrt{g(x)}} \, ,
\end{equation}
where $g(x) := \det(g^{\mathbb{C}}_{\mu\nu})$ and $\nabla^2 := (g^{\mathbb{C}})^{\mu\nu} \nabla_{\mu} \nabla_{\nu}$.

In contrast to the real Lorentzian section, there is a \emph{unique} solution to this equation in the complex Riemannian section which satisfies the following boundary conditions: 
\begin{enumerate}[label={(\roman*)}]
\item $G(x,x')$ is regular at $r = r_+$; 
\item $G(x,x')$ satisfies the Dirichlet boundary conditions at $r = r_{\mathcal{M}}$.
\end{enumerate}
This follows from the uniqueness results for boundary value problems in compact manifolds. Note that two of the directions of the complex spacetime are periodic, while the third direction is compact due to the existence of the timelike boundary. In contrast, on static spacetimes without boundaries (and suitable asymptotic properties at infinity), the real Riemannian section has a unique Euclidean Green's distribution, due to the ellipticity of the Klein-Gordon operator, as previously remarked.

Given the periodicity conditions of $\tau$ and $\tilde{\theta}$, one has
\begin{align}
\delta (\tau - \tau') &= \frac{\kappa_+}{2\pi} \sum_{n = -\infty}^{\infty} e^{i \kappa_+ n (\tau - \tau')} \, ,
\label{eq:deltataugeneral} \\
\delta (\tilde{\theta} - \tilde{\theta}') &= \frac{1}{2\pi} \sum_{k = -\infty}^{\infty} e^{i k(\tilde{\theta} - \tilde{\theta}')} \, ,
\label{eq:deltathetageneral}
\end{align}
understood as distributional identities. We can write the Green's distribution $G(x,x')$ as a sum over modes $G_{nk}(r,r')$
\begin{equation}
G(x,x') = \frac{\kappa_+}{4\pi^2} \sum_{n = -\infty}^{\infty} e^{i \kappa_+ n (\tau - \tau')} \sum_{k = -\infty}^{\infty} e^{i k(\tilde{\theta} - \tilde{\theta}')} \, G_{nk}(r,r') \, .
\label{eq:Greenfunctiongen0}
\end{equation}
By using \eqref{eq:deltataugeneral}, \eqref{eq:deltathetageneral} and \eqref{eq:Greenfunctiongen0} in the field equation \eqref{eq:GCfunctioneq}, we obtain an ordinary differential equation for $G_{nk}$,
\begin{equation} \label{eq:radialGreenseq}
\left[ \frac{1}{\sqrt{g}} \frac{\dd }{\dd r} \left( \sqrt{g} \, g^{rr} \frac{\dd }{\dd r} \right) - \frac{\left( \kappa_+ n + i k \big( N^{\theta} + \Omega_{\mathcal{H}} \big) \right)^2}{N^2} - \frac{k^2}{g_{\tilde{\theta}\tilde{\theta}}} - m^2 - \xi R \right] G_{nk} = - \frac{\delta(r-r')}{\sqrt{g}} \, .
\end{equation}
The solutions of this equation can be given in terms of solutions of the corresponding homogeneous equation. Let $p_{nk}$ be the solution of the homogeneous equation which is regular at the horizon and let $q_{nk}$ be the solution of the homogeneous equation which satisfies the Dirichlet boundary condition at the timelike boundary. Then, by the standard theory of Green's functions (e.g.~Chapter 10 of \cite{arfken2012mathematical}), the radial part of the Green's distributions is given by
\begin{equation} \label{eq:Gnkgeneral}
G_{nk}(r,r') = C_{nk} \, p_{nk}(r_<) \, q_{nk}(r_>) \, ,
\end{equation}
where $r_< := \min \{ r, r' \}$, $r_> := \max \{ r, r' \}$ and $C_{nk}$ is the normalization constant determined by the Wronskian relation
\begin{equation} \label{eq:Cnkdefinition}
C_{nk} \left( p_{nk} \frac{\dd  q_{nk}}{\dd r} - q_{nk} \frac{\dd  p_{nk}}{\dd r} \right) = \frac{1}{\sqrt{g} \, g^{rr}} \, .
\end{equation}

Hence, we have found the unique solution of Eq.~\eqref{eq:GCfunctioneq} which satisfies the boundary conditions of regularity at the horizon and Dirichlet condition at the timelike boundary. The Green's distribution is expressed as a sum over mode solutions of the differential equation \eqref{eq:radialGreenseq}. For convenience, we rewrite \eqref{eq:Greenfunctiongen0} as
\begin{equation}
G(x,x') =: \sum_{n = -\infty}^{\infty} e^{i \kappa_+ n (\tau - \tau')} \sum_{k = -\infty}^{\infty} e^{i k(\tilde{\theta} - \tilde{\theta}')} \, G_{nk}^{\rm BH}(r,r') \, .
\label{eq:Greenfunctiongen1}
\end{equation}

In the current (2+1)-dimensional case, it is generally possible to find the mode solutions \eqref{eq:Gnkgeneral} in closed form in terms of known functions, whereas we need to resort to numerical methods for four or more dimensions. In Chapter~\ref{chap:computation}, we will explicitly find the mode solutions for a scalar field on a warped AdS${}_3$ black hole.

\section{Renormalisation procedure}
\label{sec:renormalisation-procedure}

In the last section, we described how to compute the Green's distribution associated with the Klein-Gordon equation in the complex Riemannian section of the exterior region of the rotating black hole. As noted at the end of Section~\ref{sec:computation-scalarfield}, we also need to express the Hadamard singular part $G_{\rm Had}$ of the Feynman propagator as a sum over mode solutions, such that the short-distance divergences can be subtracted term by term. Before that, we need to make sense of $G_{\rm Had}$ in the complex Riemannian section.

In Section~\ref{sec:nbhdcomplexRiemannian} we verified that the local geodesic structure of a real Lorentzian manifold is preserved when going to the complex Riemannian section and, in particular, we can define a notion of a geodesically linearly convex neighbourhood as in Definition~\ref{def:convexnbhdcsection} and generalise the definition of the Synge's world function, as in Definition~\ref{def:SyngesworldfunctionCsection}.

Henceforth, we can write the Hadamard singular part of the Green's distribution $G$ in the complex Riemannian section as
\begin{equation}
G_{\text{Had}}(x,x') = \frac{1}{4\sqrt{2}\pi} \frac{1}{\sqrt{\sigma(x,x')}} + \mathcal{O}(\sigma^{1/2}) \, .
\label{eq:Hadamardsingpart}
\end{equation}

In an analogous way to the real Lorentzian case, we now subtract the Hadamard singular part from the Green's function $G$ and then take the coincidence limit to obtain the vacuum polarization at $x \in \widetilde{\text{I}}$,
\begin{equation}
\langle \Phi^2(x) \rangle = \lim_{x' \to x} \left[ G(x,x') - G_{\text{Had}}(x,x') \right] \, .
\label{eq:vacuumpolCsection}
\end{equation}
(In a slight abuse of notation, on the RHS of the equation $x, x' \in I^{\mathbb{C}}$, such that $x \in I^{\mathbb{C}}$ is the result of a Wick rotation of $x \in \widetilde{\text{I}}$.)

By construction, the Green's distribution $G$ is regular at $r = r_+$, satisfies the Dirichlet boundary conditions at $r = r_{\mathcal{M}}$ and is invariant under the spacetime isometries. Therefore, together with \eqref{eq:GFGEthermal}, after analytically continuing back to the Lorentz section, $\langle \Phi^2(x) \rangle$ as given by \eqref{eq:vacuumpolCsection} is the vacuum polarisation for a scalar field in the Hartle-Hawking state.

\subsection{Subtraction of the Hadamard singular part}

It remains to perform the subtraction in \eqref{eq:vacuumpolCsection} before the coincidence limit can be taken. As $G$ is known only as the mode sum \eqref{eq:Greenfunctiongen0}, the evaluation of $\langle \Phi^2(x) \rangle$ requires $G_{\text{Had}}$ to be rewritten as a mode sum that can be combined with \eqref{eq:Greenfunctiongen0} so that the divergences in the coincidence limit get subtracted under the sum term by term. For a general spacetime, it is not known how to express $G_{\text{Had}}$ as a mode sum. 

We accomplish this in the following way. The Hadamard singular part incorporates the short-distance singular behaviour of the Green's distribution for (the complex Riemannian section of) a rotating black hole, which should be of the same form as the singular behaviour of the Green's distribution for the (complex Riemannian section of) Minkowski spacetime, given that we are dealing with Hadamard states. A good thing about Minkowski spacetime is that the zero temperature Green's distribution for a scalar field is known in closed form
\begin{equation} \label{eq:GMclosedform}
G_0^{\mathbb{M}}(x,x') = \frac{1}{4\pi} \frac{e^{-m \sqrt{2\sigma(x,x')}}}{\sqrt{2\sigma(x,x')}} \, .
\end{equation}
This means that the thermal Green's distribution $G^{\mathbb{M}}$ can be expressed both as an imaginary-time image sum of the zero temperature Green's distribution using \eqref{eq:imagemsumG} and also as a sum over mode solutions, say,
\begin{equation} 
G^{\mathbb{M}}(x,x') = \sum_{n = -\infty}^{\infty} e^{i \kappa_+ n (\tau - \tau')} \sum_{k = -\infty}^{\infty} e^{i k(\tilde{\theta} - \tilde{\theta}')} \, G^{\mathbb{M}}_{nk}(x,x') \, ,
\end{equation}
as described in Appendix~\ref{app:Minkowski}. And, of course, we know how to write its Hadamard singular part $G^{\mathbb{M}}_{\rm Had}(x,x')$ in closed form, as in \eqref{eq:Hadamardsingpart}. This allow us to write the Hadamard singular part as in \eqref{eq:GHadMinkappendix},
\begin{equation}
G_{\text{Had}}^{\mathbb{M}}(x,x') =  \sum_{n = -\infty}^{\infty} e^{i \kappa_+ n (\tau - \tau')} \sum_{k = -\infty}^{\infty} e^{i k(\tilde{\theta} - \tilde{\theta}')} \, G^{\mathbb{M}}_{nk}(x,x') - G_{\text{reg}}^{\mathbb{M}}(x,x') \, ,
\label{eq:GnkMink}
\end{equation}
where $G_{\text{reg}}^{\mathbb{M}}(x,x')$ is a finite term when the coincidence limit is taken, which is obtained in Appendix~\ref{app:Minkowski}. Eq.~\eqref{eq:GnkMink} expresses the Hadamard singular part of the Green's distribution for the Minkowski spacetime as a sum over mode solutions, plus a regular term which can be easily computed since we know the Minkowski Green's distribution in closed form \eqref{eq:GMclosedform}.

In general, it is not possible to obtain the Green's distribution for the rotating black hole spacetime in closed form and, therefore, it is not possible to write its Hadamard singular part as in \eqref{eq:GnkMink}. However, as we argued above, the Hadamard singular parts of the Green's distributions for the rotating black hole and Minkowski spacetime are essentially of the same form, just given in different coordinate systems. As we show below, it should be possible to express $G_{\text{Had}}(x,x')$ of the black hole in terms of $G^{\mathbb{M}}_{\rm Had}(x,x')$ which, in turn, we know how to write in terms of a mode sum, as in \eqref{eq:GnkMink}! Therefore, we are able to subtract the short-distance divergences of the black hole Green's distribution by using a sum over mode solutions of the Minkowski Green's distributions differential equation.

To explain this procedure in detail, it is convenient at this stage to consider a particular choice of point separation. Assume that the black hole metric is given in coordinates $(\tau,r,\tilde{\theta})$, whereas the Minkowski metric is given in coordinates $(\tau,\rho,\tilde{\theta})$. Now, consider the case of angular separation in each spacetime, such that for the black hole case $x = (\tau, r, 0)$ and $x' = (\tau, r, \tilde{\theta})$, with $\tilde{\theta} > 0$, and similarly for the Minkowski case.

The expansion of the Hadamard singular parts for small $\tilde{\theta}$ are
\begin{align}
G_{\text{Had}}(x,x') &= \frac{1}{4\pi} \frac{1}{\sqrt{g_{\tilde{\theta}\tilde{\theta}}(r)}} \, \frac{1}{\tilde{\theta}} + \mathcal{O}(\tilde{\theta}) \, , 
\label{eq:GHadleadingterm} \\
G_{\text{Had}}^{\mathbb{M}}(x,x') &= \frac{1}{4\pi} \frac{1}{\sqrt{g^{\mathbb{M}}_{\tilde{\theta}\tilde{\theta}}(\rho)}} \, \frac{1}{\tilde{\theta}} + \mathcal{O}(\tilde{\theta}) \, ,
\end{align}
where $g^{\mathbb{M}}_{\tilde{\theta}\tilde{\theta}}(\rho) = \rho^2$ is the $\tilde{\theta}\tilde{\theta}$-component of the metric for the rotating Minkowski spacetime (see \eqref{eq:Minkmetric}). 
We are free to make the identification
\begin{equation} \label{eq:gthetathetaidentification}
g^{\mathbb{M}}_{\tilde{\theta}\tilde{\theta}}(\rho) \equiv \gamma(r)^{-2} \, g_{\tilde{\theta}\tilde{\theta}}(r) \, ,
\end{equation}
where $\gamma(r) > 0$ is a function to be specified. This identification provides a matching between the two radial coordinates, $\rho = \rho(r) = \gamma(r)^{-1} \sqrt{g_{\tilde{\theta}\tilde{\theta}}(r)}$, and allows us to express $G_{\text{Had}}(x,x')$ of the black hole in terms of $G^{\mathbb{M}}_{\rm Had}(x,x')$, as argued in the beginning of this section.

Given this identification, we can now write
\begin{align} \label{eq:Gdiff}
G(x,x') - G_{\text{Had}}(x,x') &= \sum_{k = -\infty}^{\infty} e^{i k \tilde{\theta}} \sum_{n = -\infty}^{\infty}  \left[ G^{\text{BH}}_{nk}(x,x') - \gamma(r)^{-1} G^{\mathbb{M}}_{nk}(x,x') \right] \notag \\
&\quad + \gamma(r)^{-1} G_{\text{reg}}^{\mathbb{M}}(x,x') + \mathcal{O}(\tilde{\theta}) \, .
\end{align}
We have succeeded in writing $G(x,x') - G_{\text{Had}}(x,x')$ as a sum over the difference of mode solutions, plus a regular term which is finite in the coincidence limit.

At this point, note that the Minkowski Green's distribution has several free parameters: $T_{\mathbb{M}}$ (temperature of the scalar field), $\Omega_{\mathbb{M}}$ (angular velocity of the coordinate system) and $m_{\mathbb{M}}^2$ (squared mass of the scalar field), besides the unspecified factor $\gamma$ we introduced above (for more details on these parameters, see Appendix~\ref{app:Minkowski}). However, the combination
\begin{align}
\sum_{k = -\infty}^{\infty} e^{i k \tilde{\theta}} \sum_{n = -\infty}^{\infty}  \gamma(r)^{-1} G^{\mathbb{M}}_{nk}(x,x') - \gamma(r)^{-1} G_{\text{reg}}^{\mathbb{M}}(x,x') = \gamma(r)^{-1} G_{\text{Had}}^{\mathbb{M}}(x,x')
\end{align}
(see \eqref{eq:GnkMink}) is unchanged if any of these parameters are modified, since $G_{\text{Had}}^{\mathbb{M}}(x,x')$ is independent of them. Therefore, these parameters can be chosen such that the double sum in \eqref{eq:Gdiff} is convergent when $\tilde{\theta} \to 0$. To check the convergence of the double sum, we need to analyse the asymptotic behaviour of the summand for large values of $n$ and $k$.

\subsection{Large quantum number behaviour}

In order to check the convergence of the double sum \eqref{eq:Gdiff} in the coincidence limit, we need to obtain the asymptotic behaviour of the summand for large values of the quantum numbers $n$ and $k$.
  
For a black hole spacetime with metric \eqref{eq:metriccomplexRiemanniansectiongen} in the complex Riemannian section, the Klein-Gordon equation
\begin{equation}
\left(\nabla^2 - m^2 - \xi R \right) \Phi(\tau, r, \tilde{\theta}) = 0 \, ,
\label{eq:KGeqCsection}
\end{equation}
together with the ansatz $\Phi_{n k}(\tau,r,\tilde{\theta}) = e^{i \kappa_+ n \tau + i k \tilde{\theta}} \, \phi_{n k}(r)$, leads to
\begin{equation}
\left[ \frac{1}{\sqrt{g}} \frac{\dd }{\dd r} \left( \sqrt{g} \, g^{rr} \frac{\dd }{\dd r} \right) - \frac{\left( \kappa_+ n + i k \big( N^{\theta} + \Omega_{\mathcal{H}} \big) \right)^2}{N^2} - \frac{k^2}{g_{\tilde{\theta}\tilde{\theta}}} - m^2 - \xi R \right] \phi_{n k} = 0 \, .
\label{eq:radialfieldeq}
\end{equation}

Let $\phi_{n k}^1$ and $\phi_{n k}^2$ be two independent solutions of the radial equation \eqref{eq:radialfieldeq}.
Define a new radial coordinate $\xi$ such that the equation can be written in the form
\begin{equation} \label{eq:fieldeqxi}
\frac{\dd^2 \phi_{n k}(\xi)}{\dd \xi^2} - \left( \chi_{nk}^2 (\xi) + \eta^2 (\xi) \right) \phi_{n k}(\xi) = 0 \, ,
\end{equation}
and the Wronskian relation is given by
\begin{equation} \label{eq:wronskianxi}
\phi_{nk}^1(\xi) \frac{\dd \phi_{nk}^2(\xi)}{\dd \xi} - \phi_{nk}^2(\xi) \frac{\dd \phi_{nk}^1(\xi)}{\dd \xi} = \frac{1}{C_{n k}} \, ,
\end{equation}
where $C_{n k}$ is a constant and $\chi_{nk}^2 (\xi)$ contains all the $n$ and $k$ dependence and is large whenever $n^2+k^2$ is large. From \eqref{eq:radialfieldeq} we obtain
\begin{equation} \label{eq:xircoords}
\frac{\dd }{\dd \xi} = \sqrt{g} \, g^{rr} \frac{\dd }{\dd r}
\end{equation}
and
\begin{equation} \label{eq:chietadefinitions}
\chi_{nk}^2 = g_{\tilde{\theta}\tilde{\theta}} \left( \kappa_+ n + i k \big( N^{\theta} + \Omega_{\mathcal{H}} \big) \right)^2 + N^2 k^2 \, , \qquad
\eta^2 = g_{\tilde{\theta}\tilde{\theta}} N^2 \left( m^2+\xi R \right) \, .
\end{equation}

We are interested in obtaining the large $\chi_{n k}$ expansion of the quantity
\begin{equation}
\mathcal{G}_{n k} (\xi) := C_{n k} \, \phi_{n k}^1(\xi) \, \phi_{n k}^2(\xi) \, .
\end{equation}
This asymptotic expansion can be obtained using a WKB method and this is described in Appendix~\ref{app:WKBexpansions}. From Proposition~\ref{prop:Gxiexpansion}, the asymptotic expansion of $\mathcal{G}_{n k}(\xi)$ for large values of $\chi_{n k}$ is
\begin{equation} \label{eq:Gxiexpansionmaintext}
\mathcal{G}_{n k} (\xi) = \frac{1}{2 \chi_{n k}} - \frac{\eta^2}{4 \chi_{n k}^3} - \frac{(\chi_{n k}^2)''}{16 \chi_{n k}^5} + \frac{5 [(\chi_{n k}^2)']^2}{64 \chi_{n k}^7} + \mathcal{O}(\chi_{n k}^{-5}) \, ,
\end{equation}
where the prime represents derivative with respect to $\xi$.

For our case of interest, $\mathcal{G}_{n k}$ plays the role of the radial part of the Green's distributions, given by \eqref{eq:Gnkgeneral}, whereas $\phi_{n k}^1$ and $\phi_{n k}^2$ correspond to $p_{nk}$ and $q_{nk}$ for both the black hole and Minkowski spacetimes.

\subsection{Fixing of the Minkowski free parameters}

We now have all we need to determine the choice of parameters of the Minkowski Green's distribution which makes the mode sum in \eqref{eq:Gdiff} convergent in the coincidence limit. This unique choice is the following.

\begin{theorem} \label{thm:matchingpolarisation}
If the parameters $\gamma$, $T_{\mathbb{M}}$ and $\Omega_{\mathbb{M}}$ are chosen as
\begin{equation}
\gamma(r) = N(r) \, , \qquad T_{\mathbb{M}} = \frac{\kappa_+}{2\pi} \, , \qquad \Omega_{\mathbb{M}} = N^{\theta}(r) + \Omega_{\mathcal{H}} \, ,
\label{eq:matching}
\end{equation}
then the double sum in \eqref{eq:Gdiff} is finite in the coincidence limit.
\end{theorem}

\begin{proof}
First, we obtain the leading terms in the asymptotic expansions of the summands in \eqref{eq:Gdiff}, using \eqref{eq:Gxiexpansionmaintext}. The summand $G^{\text{BH}}_{nk}(x,x)$ of the Green's distribution $G(x,x')$ in \eqref{eq:Greenfunctiongen1} has the following asymptotic expansion for large $\chi_{nk}$
\begin{equation}
G^{\text{BH}}_{nk}(x,x') = \frac{\kappa_+}{4\pi^2} \frac{1}{2\chi_{nk}} + \mathcal{O}\left(\chi_{nk}^{-3}\right) \, .
\end{equation}
Analogously, for the Minkowski Green's distribution, the summand $G^{\mathbb{M}}_{nk}(x,x')$ in \eqref{eq:GnkMink} has the asymptotic expansion
\begin{equation}
G^{\mathbb{M}}_{nk}(x,x') = \frac{T_{\mathbb{M}}}{2\pi} \frac{1}{2\chi^{\mathbb{M}}_{nk}} + \mathcal{O}\left((\chi^{\mathbb{M}}_{nk})^{-3}\right) \, ,
\end{equation}
where
\begin{equation} \label{eq:chiMink}
\chi^{\mathbb{M}}_{nk} (\rho)^2 = \rho^2 \left( 2\pi T_{\mathbb{M}} n + i k \Omega_{\mathbb{M}} \right)^2 + k^2 \, .
\end{equation}

The double sum in \eqref{eq:Gdiff} will be finite in the coincidence limit if the leading term in the asymptotic expansion of the summand vanishes, that is, if the term of order $\chi_{nk}^{-1}$ of the expansion of $G^{\text{BH}}_{nk}(x,x)$ cancels with the term of order $\left(\chi^{\mathbb{M}}_{nk}\right)^{-1}$ of the expansion of $\gamma(r)^{-1} \, G^{\mathbb{M}}_{nk}(x,x)$. This only occurs if the free parameters $\gamma$, $T_{\mathbb{M}}$ and $\Omega_{\mathbb{M}}$ are chosen as
\begin{equation}
\gamma(r) = N(r) \, , \qquad T_{\mathbb{M}} = \frac{\kappa_+}{2\pi} \, , \qquad \Omega_{\mathbb{M}} = N^{\theta}(r) + \Omega_{\mathcal{H}} \, .
\end{equation}

To show that the double sum is indeed finite in the coincidence limit, we need to check that the double sum of the remaining terms in the asymptotic expansion of the summand, which are $\mathcal{O}\left(\chi_{nk}^{-3}\right)$, is finite.

It is enough to consider
\begin{equation}
\Delta G(r) := \sideset{}{'}\sum_{k,n} \left[ G^{\text{BH}}_{nk}(r,r) - \gamma(r)^{-1} \, G^{\mathbb{M}}_{nk}(\rho(r),\rho(r)) \right] \, ,
\end{equation}
where $\sum_{n,k}'$ stands for the double sum over $k$ and $n$ excluding the $k=n=0$ term.

The first terms in the WKB-like expansion cancel each other, thus
\begin{equation}
\Delta G(r) = \sideset{}{'}\sum_{k,n} \left\{ G^{\text{BH}(2)}_{nk}(r) + \mathcal{O}(\chi_{nk}^{-5}) - \gamma(r)^{-1} \, \left[ G^{\mathbb{M}(2)}_{nk}(\rho) + \mathcal{O}((\chi_{nk}^{\mathbb{M}})^{-5}) \right] \right\} \, ,
\end{equation}
where $G^{\text{BH}(2)}_{nk}$ and $G^{\mathbb{M}(2)}_{nk}$ are the terms of the expansion of order $\chi_{nk}^{-3}$ and $(\chi_{nk}^{\mathbb{M}})^{-3}$, respectively.
With the choice \eqref{eq:matching}, one has
\begin{equation}
\chi_{nk}^{\mathbb{M}}(r)^2 = \frac{\chi_{nk}(r)^2}{N(r)^2} \, .
\end{equation}
Therefore
\begin{equation}
\Delta G(r) = \sideset{}{'}\sum_{k,n} \left[ \frac{\mathcal{W}(r)}{\chi_{nk}(r)^3}  + \mathcal{O}(\chi_{nk}^{-5}) \right]  \, ,
\end{equation}
where $\mathcal{W}(r)$ does not depend on $n$ and $k$.

\newpage 

Note that:
\begin{align}
\sideset{}{'}\sum_{k,n} \left| \frac{\mathcal{W}(r)}{\chi_{nk}(r)^3} \right|
&\propto \sideset{}{'}\sum_{k,n} \left| g_{\tilde{\theta}\tilde{\theta}}(r) \left(\kappa_+ n + i k (N^{\theta}(r) + \Omega_{\mathcal{H}})\right)^2 + N(r)^2 k^2 \right|^{-3/2} \notag \\
&= \sideset{}{'}\sum_{k,n} \Big\{ \left[ g_{\tilde{\theta}\tilde{\theta}}(r) \kappa_+^2 n^2 + \left( N(r)^2 - g_{\tilde{\theta}\tilde{\theta}} (N^{\theta}(r) + \Omega_{\mathcal{H}})^2 \right) k^2 \right]^2 \notag \\
&\qquad\qquad\quad\; + 4 g_{\tilde{\theta}\tilde{\theta}}(r)^2 (N^{\theta}(r) + \Omega_{\mathcal{H}})^2 \kappa_+^2 n^2 k^2 \Big\}^{-3/4} \notag \\
&\leq \sideset{}{'}\sum_{k,n} \left[ g_{\tilde{\theta}\tilde{\theta}}(r) \kappa_+^2 n^2 + \left( N(r)^2 - g_{\tilde{\theta}\tilde{\theta}}(r) (N^{\theta}(r) + \Omega_{\mathcal{H}})^2 \right) k^2 \right]^{-3/2} \, .
\end{align}
Lemma~\ref{lemma:convergenceseries} below shows that the latter series is convergent. This proves the absolute convergence of
\begin{equation}
\sideset{}{'}\sum_{k,n} \frac{\mathcal{W}(r)}{\chi_{nk}(r)^3} \, .
\end{equation}

Finally, since
\begin{equation}
\lim_{|\chi_{\text{BH}}| \to \infty} \frac{\left| \frac{\mathcal{W}(r)}{\chi_{nk})(r)^3}  + \mathcal{O}(\chi_{nk}^{-5}) \right|}{\left| \frac{\mathcal{W}(r)}{\chi_{nk}(r)^3} \right|}
= 1 \, ,
\end{equation}
the limit comparison test implies the absolute convergence of
\begin{equation}
\sideset{}{'}\sum_{k,n} \left[ \frac{\mathcal{W}(r)}{\chi_{nk}(r)^3}  + \mathcal{O}(\chi_{nk}^{-5}) \right] \, .
\end{equation}

Therefore, we conclude that the $\Delta G(r)$ is finite.
\end{proof}

\vspace*{1ex}

\begin{lemma} \label{lemma:convergenceseries}
Let $A$, $B>0$. Then,
\begin{align}
S := \sideset{}{'}\sum_{k,n} \frac{1}{(A n^2 + B k^2)^{3/2}} < \infty
\end{align}
where $\sum_{n,k}'$ stands for the double sum over $k, n \in \mathbb{Z}$ excluding the $k=n=0$ term.
\end{lemma}

\begin{proof}
We write
\begin{equation}
S = \sum_{k \in \mathbb{Z}} S_k \, ,
\end{equation}
where
\begin{equation}
S_k := \begin{cases} \displaystyle\sum_{n \in \mathbb{Z}} \frac{1}{(A n^2 + B k^2)^{3/2}} \, , & k \neq 0 \, , \\
\displaystyle\sum_{n \in \mathbb{Z} \setminus \{0 \}} \frac{1}{A^{3/2} \, n^3} \, , & k = 0 \, .
\end{cases}
\end{equation}
Each $S_k$ is clearly finite, and $S_{-k} = S_k$. For $k>0$ we have
\begin{equation}
k^2 \, S_k = \sum_{n=-\infty}^{\infty} \frac{1}{\left[B + A (\frac{n}{k})^2 \right]^{3/2}} \frac{1}{k} \; \xrightarrow{\;\; k \to \infty \;\;} \; \frac{2}{B \sqrt{A}} \, ,
\label{eq:seriesRiemannsum}
\end{equation}
since the series in \eqref{eq:seriesRiemannsum} becomes the Riemann sum for the integral
\begin{equation}
\int_{-\infty}^{\infty} \dd t \, \frac{1}{(B + A t^2)^{3/2}} = \frac{2}{B \sqrt{A}} \, .
\end{equation}
Thus,
\begin{equation}
S_k \sim \frac{2}{B \sqrt{A}} \frac{1}{k^2} \, , \qquad |k| \to \infty \, .
\end{equation}
so that $S$ is finite.
\end{proof}

\begin{remark}
The choice \eqref{eq:matching} for the parameters of the Minkowski Green's distribution corresponds to have the temperature $T_{\mathbb{M}}$ of the scalar field in Minkowski to match the Hawking temperature of the black hole and to have the angular velocity $\Omega_{\mathbb{M}}$ to be equal to the one measured by a locally non-rotating observer at radius $r$ in the black hole spacetime.
\end{remark}

The key aspect of the proof is that, in order to remove the divergences, we only need to know the asymptotic behaviour of the Green's distribution summands $G^{\text{BH}}_{nk}(r,r)$ and $G^{\mathbb{M}}_{nk}(\rho,\rho)$ for large values of $n$ and $k$, and not the full solutions. This implies that, apart from technical difficulties, this method can be applied to black holes in four or more dimensions, for which although we can only obtain the Green's distributions numerically, the asymptotic expansions of the summands for large quantum numbers can be explicitly computed using the above procedure.

Setting the parameters as in \eqref{eq:matching}, it is now possible to take the coincidence limit $\tilde{\theta} \to 0$ of \eqref{eq:Gdiff} and compute the renormalized vacuum polarization \eqref{eq:vacuumpolCsection}. In Part II of the thesis, as an example, we present the results for the particular case of the warped AdS${}_3$ black hole.



\section{Expectation value of the stress-energy tensor}
\label{sec:stress-energy-tensor}

In the previous section, we successfully implemented a method to compute the renormalised vacuum polarisation of a scalar field on a rotating black hole by subtracting the short-distance divergence of the Feynman propagator using a sum over Minkowski modes with the same singularity structure. The next physically interesting local observable is the renormalised expectation value of the stress-energy tensor, $\langle T_{ab} (x) \rangle$. In this section, we demonstrate why the method described in this chapter cannot be used to renormalise the stress-energy tensor.

First, recall that we defined the renormalised expectation value of the stress-energy tensor in Definition~\ref{def:renormalisedquantities} as
\begin{equation}
\langle T_{ab}(x) \rangle = \lim_{x' \to x} \mathcal{T}_{ab'}(x,x') \left[ -i \left( G^{\rm F}(x,x') - G_{\rm Had}(x,x') \right) \right] + \Theta_{ab}(x) \, ,
\end{equation}
where $\Theta_{ab}(x)$ is a state independent tensor which ensures that $\langle T_{ab}(x) \rangle$ is covariantly conserved and
\begin{align}
\mathcal{T}_{ab'} &= (1-2\xi) \, {g_{b}}^{b'} \nabla_{a} \nabla_{b'} + \left(2\xi - \frac{1}{2} \right) g_{ab} g^{cd'} \nabla_{c} \nabla_{d'} - 2 \xi \, {g_{a}}^{a'} {g_{b}}^{b'} \nabla_{a'} \nabla_{b'} \notag \\
&\quad +2\xi \, g_{ab} \nabla_{\rho} \nabla^{\rho} + \xi \left( R_{ab} - \frac{1}{2} g_{ab} R \right) - \frac{1}{2} g_{ab} m^2 \, .
\end{align}

Besides having to carefully perform the coincidence limit of $G^{\rm F}(x,x') - G_{\rm Had}(x,x')$ as in the vacuum polarisation computation, for the stress-energy tensor we also need to consider terms of the form $\nabla_a \nabla_b \left[ G^{\rm F}(x,x') - G_{\rm Had}(x,x') \right]$. Here, we show that the implementation of our renormalisation method, in particular the formulation of $G_{\rm Had}(x,x')$ as a sum over Minkowski modes, fails to subtract the short-distance divergences of $\nabla_{\tilde{\theta}} \nabla_{\tilde{\theta}} G^{\rm F}(x,x')$ and, hence, the whole unrenormalised stress-energy tensor.

\begin{remark} \label{rem:shiftBHstressenergy}
A heuristic argument that suggests the method indeed fails to renormalise the stress-energy tensor is that to match the short-distance divergences of $\nabla_a \nabla_b G_{\rm Had}(x,x')$ on the rotating black hole spacetime to the the short-distance divergences of $\nabla_a \nabla_b G^{\mathbb{M}}_{\rm Had}(x,x')$ on Minkowski spacetime (in rotating coordinates) it will be necessary not only to identify components of the metric tensors (as in \eqref{eq:gthetathetaidentification}) but also \emph{derivatives} of those components. However, the shift function $N^{\theta}$ (present in the $t \tilde{\theta}$-component of the metric) is a function of the radial coordinate $r$ for the rotating black hole, whereas it is a constant ($- \Omega_{\mathbb{M}}$) for the Minkowski spacetime (in rotating coordinates). We then expect that the matching between derivatives of the metric components with respect to the radial coordinates might not be possible. This is indeed the case, as we show explicitly in the following.
\end{remark}

\subsection{Derivatives of the Hadamard singular part}

Here, we work again in the complex Riemannian sections of both the black hole and Minkowski spacetimes. First, we compute the double covariant derivatives of the Hadamard singular part, $\nabla_{\mu} \nabla_{\nu} G_{\rm Had}(x,x')$. Recall that the Hadamard singular part in three dimensions is given by
\begin{equation}
G_{\rm Had}(x,x') = \frac{1}{4\sqrt{2} \pi} \frac{U(x,x')}{\sqrt{\sigma(x,x')}} \, ,
\end{equation}
where the bi-scalar $U$ is given in \eqref{eq:biscalarU}, thus,
\begin{equation}
\frac{U}{\sigma^{1/2}} = \frac{U_0}{\sigma^{1/2}} + U_1 \sigma^{1/2} + \mathcal{O}(\sigma^{3/2}) \, ,
\end{equation}
and $U_0$ and $U_1$ given in Proposition~\ref{prop:Uexpansion}. To simplify the expressions in the following, we only consider spacetimes with a constant Ricci scalar. In three dimensions, this is not a major restriction, as all solutions of Einstein gravity satisfy this property (this is not necessarily true of modified theories of gravity). The examples in Part II of this thesis have constant Ricci scalars.

The first covariant derivative of $U \, \sigma^{-1/2}$ up to $\mathcal{O}(\sigma)$ is
\begin{equation}
\left(\frac{U}{\sigma^{1/2}}\right)_{;\mu} = \frac{U_{0;\mu}}{\sigma^{1/2}} - \frac{1}{2} \frac{U_0}{\sigma^{3/2}} \, \sigma_{;\mu} + \frac{1}{2} \frac{U_1}{\sigma^{1/2}} \, \sigma_{;\mu} + \mathcal{O}(\sigma) \, .
\end{equation}
Note that the term $U_{1;\mu} \, \sigma^{1/2} = \mathcal{O}(\sigma)$.

The double covariant derivative of $U \, \sigma^{-1/2}$ up to $\mathcal{O}(\sigma^{1/2})$ is
\begin{align}
\left(\frac{U}{\sigma^{1/2}}\right)_{;\mu\nu} &= \frac{U_{0;\mu\nu}}{\sigma^{1/2}} - \frac{1}{2} \frac{U_{0;\nu}}{\sigma^{3/2}} \, \sigma_{;\mu} - \frac{1}{2} \frac{U_{0;\mu}}{\sigma^{3/2}} \, \sigma_{;\nu} + \frac{3}{4} \frac{U_0}{\sigma^{5/2}} \sigma_{;\mu} \sigma_{;\nu} - \frac{1}{2} \frac{U_0}{\sigma^{3/2}} \, \sigma_{;\mu\nu} \notag \\
&\quad -\frac{1}{4} \frac{U_1}{\sigma^{3/2}} \sigma_{;\mu} \sigma_{;\nu}  + \frac{1}{2} \frac{U_1}{\sigma^{1/2}} \, \sigma_{;\mu\nu} + \mathcal{O}(\sigma^{1/2}) \, .
\end{align}

Concerning the derivatives of $U_0$ and $U_1$ up to the required order:
\begin{align}
U_{0;\mu} &= u_{0 ab} \sigma^{;a} {\sigma^{;b}}_{\mu} + \frac{1}{2} \left( u_{0 ab;\mu} \sigma^{;a} \sigma^{;b} - u_{0 abc} \sigma^{;a} \sigma^{;b} {\sigma^{;c}}_{\mu} \right) + \mathcal{O}(\sigma^{3/2}) \, ; \\
U_{0;\mu\nu} &= u_{0 ab} \left( {\sigma^{;a}}_{\nu} {\sigma^{;b}}_{\mu} + \sigma^{;a} {\sigma^{;b}}_{\mu\nu} \right) + u_{0 ab;\mu} \sigma^{;a} {\sigma^{;b}}_{\nu} + u_{0 ab;\nu} \sigma^{;a} {\sigma^{;b}}_{\mu} \notag \\
&\quad - u_{0 abc} \sigma^{;a} {\sigma^{;b}}_{\nu} {\sigma^{;c}}_{\mu} + \mathcal{O}(\sigma) \, ; \\
U_{1;\mu} &= \mathcal{O}(\sigma^{1/2}) \, ; \\
U_{1;\mu\nu} &= \mathcal{O}(\sigma^{0}) \, .
\end{align}

\subsection{Attempt to renormalise $\nabla_{\tilde{\theta}} \nabla_{\tilde{\theta}} G$}

At this stage, it is convenient to consider again angular separation of the points  in spacetime, such that $x = (\tau, r, 0)$ and $x' = (\tau, r, \tilde{\theta})$, with $\tilde{\theta} > 0$, for the black hole case and similarly for the Minkowski case.

Hence, we can expand the Synge's world function and its derivatives in $\tilde{\theta}$, using Proposition~\ref{prop:noncovariantexpansionsigma},
\begin{align}
\sigma &= \frac{1}{2} g_{\tilde{\theta} \tilde{\theta}} \, \tilde{\theta}^2 + \tilde{\sigma}_{\tilde{\theta}\tilde{\theta}\tilde{\theta}\tilde{\theta}} \, \tilde{\theta}^4 + \mathcal{O}(\tilde{\theta}^6) \, , \\
\sigma_{;\mu} &= g_{\mu \tilde{\theta}} \, \tilde{\theta}+ \left( \frac{1}{2} g_{\tilde{\theta}\tilde{\theta},\mu} + 3 \tilde{\sigma}_{\tilde{\theta}\tilde{\theta}\mu} \right) \tilde{\theta}^2 + \left( \tilde{\sigma}_{\tilde{\theta}\tilde{\theta}\tilde{\theta},\mu} + 4 \tilde{\sigma}_{\tilde{\theta}\tilde{\theta}\tilde{\theta}\mu} \right) \tilde{\theta}^3 \notag \\
&\quad + \left( \tilde{\sigma}_{\tilde{\theta}\tilde{\theta}\tilde{\theta}\tilde{\theta},\mu} + 5 \tilde{\sigma}_{\tilde{\theta}\tilde{\theta}\tilde{\theta}\tilde{\theta}\mu} \right) \tilde{\theta}^4 + \mathcal{O}(\tilde{\theta}^5) \, , \\
\sigma_{;\mu\nu} &= g_{\mu\nu} + \left( 2g_{\tilde{\theta}(\mu,\nu)} - \Gamma^{\lambda}_{\mu\nu} g_{\tilde{\theta}\lambda} + 6 \tilde{\sigma}_{\tilde{\theta}\mu\nu} \right) \tilde{\theta} \notag \\
&\quad + \left( \frac{1}{2} g_{\tilde{\theta}\tilde{\theta},\mu\nu} - \frac{1}{2} \Gamma^{\lambda}_{\mu\nu} g_{\tilde{\theta}\tilde{\theta},\lambda} + 6 \tilde{\sigma}_{\tilde{\theta}\tilde{\theta}(\mu,\nu)} - 3 \Gamma^{\lambda}_{\mu\nu} \tilde{\sigma}_{\tilde{\theta}\tilde{\theta}\lambda} + 12 \tilde{\sigma}_{\tilde{\theta}\tilde{\theta}\mu\nu} \right) \notag \\
&\quad + \left( \tilde{\sigma}_{\tilde{\theta}\tilde{\theta}\tilde{\theta},\mu\nu} - \Gamma^{\lambda}_{\mu\nu} \tilde{\sigma}_{\tilde{\theta}\tilde{\theta}\tilde{\theta},\lambda} + 8 \tilde{\sigma}_{\tilde{\theta}\tilde{\theta}\tilde{\theta}(\mu,\nu)} - 4 \Gamma^{\lambda}_{\mu\nu} \tilde{\sigma}_{\tilde{\theta}\tilde{\theta}\tilde{\theta}\lambda} + 20 \tilde{\sigma}_{\tilde{\theta}\tilde{\theta}\tilde{\theta}\tilde{\theta}\mu\nu} \right) \tilde{\theta}^3 \notag \\
&\quad + \mathcal{O}(\tilde{\theta}^4) \, .
\end{align}

Using the above expressions, one obtains that, for angular separation,
\begin{equation} \label{eq:thetathetaGHad}
\nabla_{\tilde{\theta}} \nabla_{\tilde{\theta}} G_{\mathrm{Had}} (x,x') = \frac{1}{2 \pi} \,  \frac{1}{\sqrt{g_{\tilde{\theta}\tilde{\theta}}}} \frac{1}{\tilde{\theta}^{3}} - \frac{1}{32\pi} \, \frac{g^{rr} \, \left( \partial_r g_{\tilde{\theta}\tilde{\theta}} \right)^2}{g_{\tilde{\theta}\tilde{\theta}}^{3/2}} \, \frac{1}{\tilde{\theta}} + \mathcal{O}(\tilde{\theta}^0) \, . 
\end{equation}

Some remarks:
\begin{enumerate}
\item As expected, $\nabla_{\tilde{\theta}} \nabla_{\tilde{\theta}} G_{\mathrm{Had}}$ is more divergent at the coincidence limit, having a leading divergent term which goes as $\tilde{\theta}^{-3}$, than $G_{\mathrm{Had}}$, which only has a term that goes as $\tilde{\theta}^{-1}$, cf.~\eqref{eq:GHadleadingterm}.
\item Note, however, that the term that goes as $\tilde{\theta}^{-3}$ is essentially of the same form as the term of $G_{\mathrm{Had}}$ that goes as $\tilde{\theta}^{-1}$ and, hence, the identification \eqref{eq:gthetathetaidentification} made to renormalise the vacuum polarisation is sufficient to subtract the leading term divergence of $\nabla_{\tilde{\theta}} \nabla_{\tilde{\theta}} G$.
\item It is easy to verify, however, that with the identification \eqref{eq:gthetathetaidentification} in place, the term in $\tilde{\theta}^{-1}$ cannot be subtracted by the corresponding Minkowski term, i.e.
\begin{equation*}
\nabla_{\tilde{\theta}} \nabla_{\tilde{\theta}} G(x,x') - \gamma^{-1} \nabla_{\tilde{\theta}} \nabla_{\tilde{\theta}} G^{\mathbb{M}}_{\rm Had}(x,x')
\end{equation*}
still retains a term of order $\tilde{\theta}^{-1}$ when expanded in $\tilde{\theta}$. One possible workaround is to perform the following subtraction instead,
\begin{equation} \label{eq:deltaGthetatheta}
\Delta_{\tilde{\theta}\tilde{\theta}}G(x,x') := \nabla_{\tilde{\theta}} \nabla_{\tilde{\theta}} G(x,x') - \gamma^{-1} \nabla_{\tilde{\theta}} \nabla_{\tilde{\theta}} G^{\mathbb{M}}_{\rm Had}(x,x') - a \, G^{\mathbb{M}}_{\mathrm{Had}}(x,x') \, ,
\end{equation}
where $a$ is a function of $r$ to be specified such that the term of order $\tilde{\theta}^{-1}$ is cancelled.
\end{enumerate}

At this stage it remains to find a choice of the parameters of the Minkowski's Green's distribution such that the mode sums in \eqref{eq:deltaGthetatheta} are convergent, if such a choice exists. If 
$G(x,x')$ is written as in \eqref{eq:Greenfunctiongen1}, then
\begin{align}
\nabla_{\tilde{\theta}} \nabla_{\tilde{\theta}} G(x,x')
&= \partial_{\tilde{\theta}}^2 G(x,x') - \Gamma^{\lambda}_{\tilde{\theta}\tilde{\theta}} \partial_{\lambda} G(x,x') \notag \\
&= \sum_{k = - \infty}^{\infty} e^{i k \tilde{\theta}} \sum_{n = - \infty}^{\infty} \Big[ {-k^2} \, G_{nk}^{\rm BH}(r,r) - \Gamma^r_{\tilde{\theta}\tilde{\theta}} \, \partial_r G_{nk}^{\rm BH}(r,r')\big|_{r' = r} \Big] \, ,
\end{align}
and similarly for the Minkowski Green's distribution. Hence,
\begin{align}
\Delta_{\tilde{\theta}\tilde{\theta}}G (x,x')
&= \sum_{k = - \infty}^{\infty} e^{i k \tilde{\theta}} \sum_{n = - \infty}^{\infty} 
\Big\{ {-k^2} \left[ G_{nk}^{\rm BH}(r,r) - \gamma^{-1} \, G_{nk}^{\mathbb{M}}(\rho, \rho) \right] \notag \\
&\quad - \Gamma^r_{\tilde{\theta}\tilde{\theta}} \, \partial_r G_{nk}^{\rm BH}(r,r')\big|_{r' = r} + \gamma^{-1} \, \Gamma^{\rho}_{\tilde{\theta}\tilde{\theta}} \, \partial_{\rho} G_{nk}^{\mathbb{M}}(\rho,\rho')\big|_{\rho' = \rho} \notag \\
&\quad - a \, G_{nk}^{\mathbb{M}}(\rho,\rho) \Big\}\Big|_{\rho = \rho(r)} \, .
\label{eq:deltaGthetathetamodesum}
\end{align}

Some extra remarks:
\begin{enumerate} \setcounter{enumi}{3}
\item Note that the first line of \eqref{eq:deltaGthetathetamodesum} can be made convergent in the coincidence limit if we fix the parameters of the Minkowski's Green's distribution as in \eqref{eq:matching}, by Theorem~\ref{thm:matchingpolarisation}. Hence, the fixing made to renormalise the vacuum polarisation also applies to the first line of \eqref{eq:deltaGthetathetamodesum}.
\item To check the convergence of the remaining terms, we need the asymptotic expansion of $\partial_r G_{nk}^{\rm BH}(r,r')|_{r' = r}$ and $\partial_{\rho} G_{nk}^{\mathbb{M}}(\rho,\rho')|_{\rho' = \rho}$ for large values of $n$ and $k$. In Appendix~\ref{app:WKBexpansions}, Proposition~\ref{prop:WKBexpansionGprime} shows that
\begin{equation}
\mathcal{G}_{n k}' (\xi) := C_{n k} \, \frac{\dd \phi_{n k}^1(\xi)}{\dd \xi} \, \phi_{n k}^2(\xi)
\end{equation}
has the asymptotic expansion for large values of $\chi_{n k}$
\begin{equation} \label{eq:Gnkprimemaintext}
\mathcal{G}_{n k}' (\xi) = \frac{1}{2} - \frac{(\chi_{n k}^2)'}{8 \chi_{n k}^3} + \mathcal{O}(\chi_{n k}^{-3}) \, .
\end{equation}
For our case of interest, $\mathcal{G}_{n k}'$ plays the role of the radial partial derivative of the radial part of the Green's distributions, given by 
\begin{equation}
\partial_r G_{nk}(r,r')|_{r' = r} = C_{nk} \frac{\dd  p_{nk}(r)}{\dd r} q_{nk}(r) \, ,
\end{equation}
whereas $\phi_{n k}^1$ and $\phi_{n k}^2$ correspond to $p_{nk}$ and $q_{nk}$ for both the black hole and Minkowski spacetimes.
\end{enumerate}

Let's focus for a moment on the $\left(\chi_{n k}^2(\xi)\right)'/(8 \chi_{n k}^3(\xi))$ term of the asymptotic expansion for the rotating black hole case. Using \eqref{eq:xircoords},
\begin{equation}
\frac{\partial_{\xi}\left(\chi_{n k}^2(\xi)\right)}{8 \chi_{n k}^3(\xi)}
= \sqrt{g} \, g^{rr} \, \frac{\partial_{r}\left(\chi_{n k}^2(r)\right)}{8 \chi_{n k}^3(r)} \, .
\end{equation}
Recall that $\chi_{n k}$ is defined by \eqref{eq:chietadefinitions},
\begin{equation}
\chi_{nk}^2 = g_{\tilde{\theta}\tilde{\theta}} \left( \kappa_+ n + i k \big( N^{\theta} + \Omega_{\mathcal{H}} \big) \right)^2 + N^2 k^2 \, ,
\end{equation}
and hence we can write
\begin{align}
\partial_r \left(\chi_{nk}^2\right) &= \partial_r g_{\tilde{\theta}\tilde{\theta}} \left( \frac{\chi_{nk}^2-N^2 k^2}{g_{\tilde{\theta}\tilde{\theta}}} \right)
+2i g_{\tilde{\theta}\tilde{\theta}} \partial_r N^{\theta} k \sqrt{\frac{\chi_{nk}^2-N^2 k^2}{g_{\tilde{\theta}\tilde{\theta}}}} + \partial_r (N^2) k^2 \notag \\
&= \frac{\partial_r g_{\tilde{\theta}\tilde{\theta}}}{g_{\tilde{\theta}\tilde{\theta}}} \chi_{nk}^2 
+ \left( \partial_r (N^2) - N^2 \frac{\partial_r g_{\tilde{\theta}\tilde{\theta}}}{g_{\tilde{\theta}\tilde{\theta}}} \right) k^2 
+ 2i \sqrt{g_{\tilde{\theta}\tilde{\theta}}} \partial_r N^{\theta} k \sqrt{\chi_{nk}^2-N^2 k^2} 
\end{align}
and
\begin{align} \label{eq:derivativechiBH}
\frac{\partial_r \left(\chi_{nk}^2\right)}{\chi_{nk}^3} 
&= \frac{\partial_r g_{\tilde{\theta}\tilde{\theta}}}{g_{\tilde{\theta}\tilde{\theta}}} \frac{1}{\chi_{nk}} 
+ \left( \partial_r (N^2) - N^2 \frac{\partial_r g_{\tilde{\theta}\tilde{\theta}}}{g_{\tilde{\theta}\tilde{\theta}}} \right) \frac{k^2}{\chi_{nk}^3}
+ 2i \sqrt{g_{\tilde{\theta}\tilde{\theta}}} \partial_r N^{\theta} \frac{\sqrt{\chi_{nk}^2-N^2 k^2}}{\chi_{nk}^3} \, .
\end{align}

Similarly, for the Minkowski case, using \eqref{eq:chiMink},
\begin{align} \label{eq:derivativechiMink}
\frac{\partial_{\rho} \left((\chi_{nk}^{\mathbb{M}})^2\right)}{(\chi_{nk}^{\mathbb{M}})^3} 
&= \frac{2}{\rho} \frac{1}{\chi_{nk}^{\mathbb{M}}} 
- \frac{2}{\rho} \frac{k^2}{(\chi_{nk}^{\mathbb{M}})^3} \, .
\end{align}

The double sum in \eqref{eq:deltaGthetathetamodesum} will be convergent in the coincidence limit if we are able to match all the terms in \eqref{eq:derivativechiBH} and the constant term in \eqref{eq:Gnkprimemaintext} to similar terms coming from the Minkowski summand. However, by comparing \eqref{eq:derivativechiBH} and \eqref{eq:derivativechiMink}, it is clear that there is \emph{no} term on the Minkowski side that can cancel the third term in the RHS of \eqref{eq:derivativechiBH}, i.e.~there is no term in the asymptotic expansions of the Minkowski terms for large $\chi_{nk}^{\mathbb{M}}$ which goes as
\begin{equation*}
\frac{\sqrt{(\chi_{nk}^{\mathbb{M}})^2 - k^2}}{(\chi_{nk}^{\mathbb{M}})^3} \, .
\end{equation*}
This comes down to the existence of $\partial_r N^{\theta}$ in the third term in the RHS of \eqref{eq:derivativechiBH}, as we alluded in Remark~\ref{rem:shiftBHstressenergy}.

In conclusion, our method to renormalise local observables by subtracting the short-distance divergences of the rotating black hole Green's distribution and its derivatives using the Minkowski spacetime Green's distribution does \emph{not} work if the computation of these local observables involves radial derivatives of the metric components, in particular if it involves terms containing $\partial_r N^{\theta}$, given that the shift function for the Minkowski metric components in rotating coordinates does not depend on the radial coordinate. Therefore, this method is suitable to renormalise local observables such as $\langle \Phi^2(x) \rangle$, which do not involve derivatives of the metric components. In the particular case of static spacetimes, in which there is a coordinate system such that the shift function vanishes, this method is still applicable for the renormalisation of the stress-energy tensor.

\part{Application}


\chapter{Warped \texorpdfstring{A\MakeLowercase{d}S$\mathbf{{}_3}$}{AdS3} black holes}
\chaptermark{Warped AdS${}_3$ black holes}
\label{chap:wadsbh}

In this chapter it is our aim to introduce the warped AdS${}_3$ black hole solutions which are used as the background spacetimes on which to apply the method described in Part I of this thesis. This is one of several possible choices of rotating black hole solutions and no particular physical significance is attached to this choice, apart from providing a simpler technical arena on which to renormalise the vacuum polarisation of a scalar field. Given this mindset, this chapter only exposes the basic ideas of (2+1)-dimensional gravity and the warped AdS${}_3$ solutions which are necessary to complete the computation, and does not attempt to give an exhaustive review of the research carried out in these topics in the last decades. For the latter, appropriate references are given in each section.

\section{2+1 gravity and topologically massive gravity}
\label{sec:wads-gravity-TMG}

In this section, we present a brief overview of (2+1)-dimensional classical gravity. In particular, we describe Einstein gravity in 2+1 dimensions and emphasise the main differences to the theory in 3+1 dimensions, namely the fact that there are no propagating degrees of freedom in 2+1 dimensions. We then introduce an extension of Einstein gravity, called topologically massive gravity, which has a propagating degree of freedom and new interesting solutions, such as the warped AdS${}_3$ solutions.

A standard reference for both the classical and quantum aspects of (2+1)-dimensional gravity is \cite{carlip2003quantum}, whereas for topologically massive gravity a few significant references are \cite{Deser:1982vy,Deser:1981wh,Deser:1982sw,Carlip:2008jk,Carlip:2008eq,Deser:2009ki}.

\subsection{2+1 gravity}

Here, we give a brief description of the main features of general relativity in 2+1 dimensions. This theory is described by the Einstein-Hilbert action
\begin{align}
S_{\text{E-H}} &= \frac{1}{16\pi G} \int \dd^3 x \, \sqrt{-g} \, \left( R - 2 \Lambda \right) \, ,
\end{align}
where $G$ is Newton's gravitational constant, $g$ is the determinant of the metric, $R$ is the Ricci scalar and $\Lambda$ is the cosmological constant.

A significant difference between 2+1 and 3+1 dimensions is the fact that in 2+1 dimensions the Riemann tensor $R_{abcd}$ is fully determined by the Ricci tensor $R_{ab}$,
\begin{equation}
R_{abcd} = g_{ac} R_{bd} + g_{bd} R_{ac} - g_{bc} R_{ad} - g_{ad} R_{bc} - \frac{1}{2} \left( g_{ac} g_{bd} - g_{ad} g_{cd} \right) R \, .
\end{equation}
This implies that any vacuum solution has constant curvature,
\begin{equation}
R_{ab} = 2 \Lambda \, g_{ab} \, .
\end{equation}
Therefore, in 2+1 gravity, there are no local degrees of freedom, only possibly global degrees of freedom, if the topology of the spacetime is non-trivial (e.g.~by performing identifications) \cite{carlip2003quantum}.

Additionally, if the cosmological constant is zero, there is no length scale in 2+1 dimensions. To see this, note that $G M$ is dimensionless in 2+1 dimensions. An important consequence of this fact is that there cannot be asymptotically flat black hole solutions of Einstein gravity, as the Schwarszchild radius would be a multiple of $G M$. If the cosmological constant is not zero, then there is a natural length scale, the cosmological length $\ell$, given by
\begin{equation}
\Lambda = \pm \frac{1}{\ell^2} \, .
\end{equation}
Then, in principle, it should be possible to find asymptotically AdS${}_3$ and dS${}_3$ black hole solutions by making identifications in AdS${}_3$ and dS${}_3$. This is indeed the case with AdS${}_3$, on which specific identifications yield the BTZ black hole solution \cite{Banados:1992wn,Banados:1992gq}. There are no known asymptotically dS${}_3$ black hole solutions.

If we want to use 2+1 gravity as a simpler arena to explore black hole physics, we might have gone too far in the simplification. However, there exist extensions of (2+1)-dimensional Einstein gravity which restore local degrees of freedom and whose dynamics are closer to the physically interesting case of 3+1 dimensions. In the following section, we consider one of such extensions, topologically massive gravity.

\subsection{Topologically massive gravity}

We now consider a deformation of (2+1)-dimensional Einstein gravity called topologically massive gravity (hereby denoted TMG), which is obtained by adding a gravitational Chern-Simons term to the Einstein-Hilbert action with a negative cosmological constant \cite{Deser:1982vy,Deser:1981wh,Carlip:2008jk,Carlip:2008eq,Deser:2009ki}. The Chern-Simons term creates a propagating, massive, spin 2 degree of freedom. In this sense, it is closer in spirit to general relativity in (3+1)-dimensions and can provide useful insight to some of the challenging problems of the higher dimensional theory.

The action of TMG in 2+1 spacetime dimensions is then
\begin{equation}
S = S_{\text{E-H}} + S_{\text{C-S}} \, ,
\end{equation}
with
\begin{align}
S_{\text{E-H}} &= \frac{1}{16\pi G} \int \dd^3 x \, \sqrt{-g} \, \left( R - 2 \Lambda \right) \, , \\
S_{\text{C-S}} &= \frac{1}{32\pi G \mu} \int \dd^3 x \, \sqrt{-g} \, \epsilon^{\mu\nu\rho} \, \Gamma^{\delta}_{\mu\lambda} \left( \partial_{\nu} \Gamma^{\lambda}_{\delta\rho} + \frac{2}{3} \Gamma^{\lambda}_{\nu\gamma} \Gamma^{\gamma}_{\rho\delta} \right) \, .
\end{align}
Here, $\mu$ is the Chern-Simons coupling, $g$ is the determinant of the metric, $\Gamma^{\delta}_{\mu\lambda}$ are the Christoffel symbols, and $\epsilon^{\mu\nu\rho}$ is the Levi-Civita tensor in three dimensions.

By linearising the action, it can be shown that TMG has a single massive propagating degree of freedom of squared mass $\mu^2$ \cite{Deser:1981wh}. This theory has third time derivative dependence, however it is ghost-free and unitary \cite{Deser:1981wh}.

A key feature of TMG is that it retains all of Einstein gravity solutions, including AdS${}_3$ and the BTZ black hole in the case of negative cosmological constant. Nevertheless, there also exist new solutions, such as the warped AdS${}_3$ vacuum solutions and warped AdS${}_3$ black hole solutions \cite{Nutku:1993eb,Gurses1994,Moussa:2003fc,Moussa:2008sj,Anninos:2008fx,Bengtsson:2005zj,Anninos:2008qb}, which are introduced in the next section. Similarly to the BTZ solution, the latter are obtained from the former by global identifications.

\section{Warped AdS$\mathbf{{}_3}$ solutions}
\sectionmark{Warped AdS${}_3$ solutions}
\label{sec:wads-wads-solution}

As noted above, there are non-Einstein solutions to TMG, and the simplest ones are the warped AdS${}_3$ solutions. These solutions are thought to be perturbatively stable vacua of TMG in a wide region of the parameter space of the theory, in contrast to the AdS${}_3$ solution \cite{Anninos:2009zi}. Mathematically, warped AdS${}_3$ spacetimes are Hopf fibrations of AdS${}_3$ over AdS${}_2$ where the fibre is the real line and the length of the fibre is ``warped'' \cite{Bengtsson:2005zj,Anninos:2008qb,Jugeau:2010nq}. This is the Lorentzian version of the warping of $S^3$ in the Riemannian setting, in which $S^3$ is warped along the Hopf fibres, which form a congruence of linked geodesic circles in $S^3$. In the Lorentzian case, there are actually two analogues, since AdS${}_3$ can be warped along Hopf fibres which may be spacelike or timelike (we will only focus on the spacelike case). And, in each case, the Hopf fibres can be either ``squashed'' or ``stretched''.

In the following, we introduce the warped AdS${}_3$ solutions in Section~\ref{sec:wads-wads-solution-noBH} and the black hole solutions in Section~\ref{sec:wadsbh-wads3bhsolutions}.

\subsection{Warped AdS$\mathbf{{}_3}$ solutions}
\label{sec:wads-wads-solution-noBH}

Before introducing the warped AdS${}_3$ solutions, we describe AdS${}_3$ in an unusual coordinate system.

\subsubsection{AdS$\mathbf{{}_3}$ in fibred coordinates}

\begin{definition}
Three-dimensional anti-de Sitter AdS${}_3$ is defined as the surface
\begin{equation}\label{eq:AdS3surface}
-U^2 - V^2 + X^2 + Y^2 = - \ell^2 \, ,
\end{equation}
embedded in the four-dimensional flat space $\mathbb{M}^{2,2}$ with metric

\begin{equation} \label{eq:AdS3flatspace}
\dd s^2 = - \dd U^2 - \dd V^2 + \dd X^2 + \dd Y^2 \, .
\end{equation}
\end{definition}

\begin{remark}
The topology of AdS${}_3$ is $\mathbb{R}^2 \times S^1$, with $S^1$ corresponding to timelike circles $U^2 + V^2 = \text{constant}$. The universal covering space is obtained by unwrapping $S^1$, which removes the closed timelike circles.
\end{remark}

To analyse the isometry group of AdS${}_3$, first note that the independent Killing vector fields of $\mathbb{M}^{2,2}$ are given by
\begin{equation}
J_{\mu\nu} = x_{\nu} \partial_{\mu} - x_{\mu} \partial_{\nu} \, , \qquad P_{\mu} = \partial_{\mu} \, ,
\end{equation} 
with $x^{\mu}=(U,V,X,Y)$. A general Killing vector field $\xi$ can then be written as
\begin{equation}
\xi = \frac{1}{2} \omega^{\mu\nu} J_{\mu\nu} + \omega^{\mu} P_{\mu} = \omega^{\mu\nu} x_{\nu} \partial_{\mu} + \omega^{\mu} \partial_{\mu} \, ,
\end{equation}
with $\omega^{\mu\nu} = - \omega^{\nu\mu}$. In detail, we can identify the spacelike and timelike rotations,
\begin{equation} \label{eq:M22rotations}
J_{UV} = V \partial_U - U \partial_V \, , \quad J_{XY} = Y \partial_X - X \partial_Y \, , 
\end{equation}
the four linearly independent boosts,
\begin{equation} \label{eq:M22boosts}
B_{UX} = U \partial_X + X \partial_U \, , \quad B_{UY} = U \partial_Y + Y \partial_U \, , \quad
\text{etc.} \, ,
\end{equation}
and the four translations,
\begin{equation} \label{eq:M22translations}
P_U = \partial_U \, , \quad P_V = \partial_V \, , \quad P_X = \partial_X \, , \quad P_Y = \partial_Y \, .
\end{equation}
These Killing vectors are the generators of $ISO(2,2)$, the isometry group of $\mathbb{M}^{2,2}$.

The isometry group of AdS${}_3$ is the subgroup of the isometry of $\mathbb{M}^{2,2}$ which leaves the AdS${}_3$ surface \eqref{eq:AdS3surface} invariant. Of the isometries above, only the translations \eqref{eq:M22translations} do not leave \eqref{eq:AdS3surface} invariant, therefore, the isometry group of AdS${}_3$ is $SO(2,2)$, which is generated by the two rotations \eqref{eq:M22rotations} and the four boosts \eqref{eq:M22translations}.

The connected component of $SO(2,2)$, $SO_0(2,2)$, is the direct product
\begin{equation}
SO_0(2,2) = SL(2,\mathbb{R})_{\rm L} \otimes SL(2,\mathbb{R})_{\rm R} / \mathbb{Z}^2 \, .
\end{equation}
To see this, it is useful to describe AdS${}_3$ as the group manifold of $SL(2,\mathbb{R})$,
\begin{equation}
SL(2,\mathbb{R}) = \left\{ A = \frac{1}{\ell}
\begin{pmatrix}
U+X & Y-V \\
Y+V & U-X
\end{pmatrix} : \;
\det(A) = 1 .
\right\}
\end{equation}
The condition $\det(A) = 1$ is invariant under the transformation
\begin{equation}
A \mapsto A' = B A C^{-1} \, , \quad B \in SL(2,\mathbb{R})_{\rm L} \, , \quad C \in SL(2,\mathbb{R})_{\rm R} \, .
\end{equation}
Hence, any element $G \in SO_0(2,2)$ may be identified with an equivalence class of two elements in the direct product $SL(2,\mathbb{R})_{\rm L} \otimes SL(2,\mathbb{R})_{\rm R}$,
\begin{equation}
G \sim (B, C) \sim (-B, -C) \, .
\end{equation}

It is then convenient to group the set of Killing vector fields of AdS${}_3$ into two mutually commuting sets. Define the right- and left-invariant Killing vector fields, $\xi_i^{\rm L}$ and $\xi_i^{\rm R}$ respectively,
\begin{align}
\xi_0^{\rm L} &= - \frac{1}{2} \left( J_{UV} + J_{XY} \right) \, ,
& \xi_0^{\rm R} &= - \frac{1}{2} \left( B_{UV} - B_{XY} \right) \, , \\
\xi_1^{\rm L} &= - \frac{1}{2} \left( B_{UY} - B_{VX} \right) \, ,
& \xi_1^{\rm R} &= - \frac{1}{2} \left( B_{UX} + B_{VY} \right) \, , \\
\xi_2^{\rm L} &= - \frac{1}{2} \left( B_{UX} + B_{VY} \right) \, ,
& \xi_2^{\rm R} &= - \frac{1}{2} \left( B_{UY} + B_{VX} \right) \, .
\end{align}
They satisfy
\begin{equation}
\left[ \xi_i^{\rm L}, \xi_j^{\rm L} \right] = {\epsilon_{ij}}^k \, \xi_k^{\rm L} \, , \quad
\left[ \xi_i^{\rm R}, \xi_j^{\rm R} \right] = {\epsilon_{ij}}^k \, \xi_k^{\rm R} \, , \quad
\left[ \xi_i^{\rm L}, \xi_j^{\rm R} \right] = 0 \, ,
\end{equation}
where $i,j,k=0,1,2$ and $\epsilon_{012}=1$. These vectors fields form bases $\{ \xi_i^{\rm L} \}_{i=0}^2$ and $\{ \xi_i^{\rm R} \}_{i=0}^2$ for (the sections of) $T \, SL(2,\mathbb{R})_{\rm L}$ and $T \, SL(2,\mathbb{R})_{\rm R}$, respectively, constituting the Maurer-Cartan frames (see Chapter~5 of \cite{nakahara2003geometry} for more details).

We can also define the dual one-forms $\theta^i_{\rm L}$ and $\theta^i_{\rm R}$, such that $\theta^i_{\rm L}(\xi_j^{\rm L}) = \delta^i_j$ and $\theta^i_{\rm R}(\xi_j^{\rm R}) = \delta^i_j$. These form the Maurer-Cartan co-frames. These one-forms satisfy the Maurer-Cartan structure equations,
\begin{equation}
\dd \theta^i_{\rm L} = - \frac{1}{2} {\epsilon_{jk}}^i \, \theta^j_{\rm L} \wedge \theta^k_{\rm L} \, ,
\end{equation}
and similarly for $\theta^i_{\rm R}$. From these equations, it follows that the Lie derivatives of these one-forms with respect to the Killing vector fields are given by
\begin{equation}
\mathcal{L}_{\xi_i^{\rm L}} \theta^j_{\rm L} = {{\epsilon_i}^j}_k \, \theta^k_{\rm L} \, , \quad
\mathcal{L}_{\xi_i^{\rm R}} \theta^j_{\rm R} = {{\epsilon_i}^j}_k \, \theta^k_{\rm R} \, , \quad
\mathcal{L}_{\xi_i^{\rm L}} \theta^j_{\rm R} = \mathcal{L}_{\xi_i^{\rm R}} \theta^j_{\rm L} = 0 \, .
\end{equation}
For instance,
\begin{align}
\mathcal{L}_{\xi_i^{\rm L}} \theta^j_{\rm L} 
= \iota_{\xi_i^{\rm L}} \dd \theta^j_{\rm L} + \dd \left( \iota_{\xi_i^{\rm L}} \theta^j_{\rm L} \right) 
= - \frac{1}{2} {\epsilon_{k l}}^j \, \iota_{\xi_i^{\rm L}} \left( \theta^k_{\rm L} \wedge \theta^l_{\rm L} \right) 
= {{\epsilon_i}^j}_k \, \theta^k_{\rm L} \, ,
\end{align}
where $\iota_{\xi_i^{\rm L}} \theta^j_{\rm L} = \theta^j_{\rm L} (\xi_i^{\rm L}) = \delta_i^j$ is the interior product of $\theta^j_{\rm L}$ with respect to $\xi_i^{\rm L}$. Therefore, the dual one-forms $\theta^i_{\rm L}$ and $\theta^i_{\rm R}$ are left- and right- invariant, respectively.

The dual one-forms allows us to write an invariant metric for the group manifold, the \emph{Killing metric}, given by
\begin{equation}
\dd s^2 = \frac{\ell^2}{4} \, \eta_{ij} \, \theta^i_{\rm L} \otimes \theta^j_{\rm L} \, .
\end{equation}

At this stage, we introduce the parametrisation
\begin{align}
U &= \cosh \left( \frac{\sigma}{2} \right) \cosh \left( \frac{u}{2} \right) \cos \left( \frac{\tau}{2} \right) + \sinh \left( \frac{\sigma}{2} \right) \sinh \left( \frac{u}{2} \right) \sin \left( \frac{\tau}{2} \right) \, , \\
V &= \cosh \left( \frac{\sigma}{2} \right) \cosh \left( \frac{u}{2} \right) \sin \left( \frac{\tau}{2} \right) - \sinh \left( \frac{\sigma}{2} \right) \sinh \left( \frac{u}{2} \right) \cos \left( \frac{\tau}{2} \right) \, , \\
X &= \cosh \left( \frac{\sigma}{2} \right) \sinh \left( \frac{u}{2} \right) \cos \left( \frac{\tau}{2} \right) + \sinh \left( \frac{\sigma}{2} \right) \cosh \left( \frac{u}{2} \right) \sin \left( \frac{\tau}{2} \right) \, , \\
Y &= \cosh \left( \frac{\sigma}{2} \right) \sinh \left( \frac{u}{2} \right) \sin \left( \frac{\tau}{2} \right) - \sinh \left( \frac{\sigma}{2} \right) \cosh \left( \frac{u}{2} \right) \cos \left( \frac{\tau}{2} \right) \, ,
\end{align}
with $u, \, \sigma \in \mathbb{R}$ and $\tau \sim \tau + 4\pi$. The right-invariant vector fields are
\begin{align}
\xi_0^{\rm L} &= - \sinh(u) \partial_{\sigma} - \cosh (u) \sech (\sigma) \partial_{\tau} + \cosh (u) \tanh (\sigma) \partial_u \, , \\
\xi_1^{\rm L} &= - \cosh(u) \partial_{\sigma} - \sinh (u) \sech (\sigma) \partial_{\tau} + \sinh (u) \tanh (\sigma) \partial_u \, , \\
\xi_2^{\rm L} &= \partial_u \, ,
\end{align}
the left-invariant fields are
\begin{align}
\xi_0^{\rm R} &= \partial_{\tau} \, , \\
\xi_1^{\rm R} &= \sin(\tau) \partial_{\sigma} - \cos (\tau) \tanh (\sigma) \partial_{\tau} + \cos (\tau) \sech (\sigma) \partial_u \, , \\
\xi_2^{\rm R} &= - \cos(\tau) \partial_{\sigma} + \sin (\tau) \tanh (\sigma) \partial_{\tau} + \sin (\tau) \sech (\sigma) \partial_u \, ,
\end{align}
and the left-invariant one-forms are
\begin{align}
\theta^0_{\rm L} &= - \cosh (u) \cosh (\sigma) \dd \tau + \sinh (u) \dd \sigma \, , \\
\theta^1_{\rm L} &= \sinh (u) \cosh (\sigma) \dd \tau - \cosh (u) \dd \sigma \, , \\
\theta^2_{\rm L} &= \dd u + \sinh (\sigma) \dd \tau \, .
\end{align}
The Killing metric is then given by
\begin{equation}
\dd s^2 = \frac{\ell^2}{4} \left[ - \cosh (\sigma)^2 \, \dd\tau^2 + \dd\sigma^2 + (\dd u+\sinh (\sigma) \, \dd\tau)^2 \right] \, .
\end{equation}
Unwrapping $\tau \in \mathbb{R}$ gives the covering space of AdS${}_3$. This is the AdS${}_3$ metric given in \emph{fibred coordinates}, as it is expressed as a Hopf fibration of the real line over AdS${}_2$.

\begin{remark}
The metric of AdS${}_3$ in the standard global coordinates $(t,\rho,\phi)$ is
\begin{equation}
\dd s^2 = \frac{\ell^2}{4} \left( - \cosh (\rho)^2 \, dt^2 + d\rho^2 + \sinh^2 (\rho) \, d\phi^2 \right) \, ,
\end{equation}
with $t, \, \rho \in \mathbb{R}$ and $\phi \sim \phi + 2\pi$.
\end{remark}

\begin{remark}
The coordinate system $(\tau,\sigma,u)$ is only one example of a fibred coordinate system, which we use in the following. For other possibilities, see e.g.~\cite{Anninos:2008fx}.
\end{remark}

\subsubsection{Spacelike warped AdS$\mathbf{{}_3}$}

In order to obtain a warped AdS${}_3$ spacetime, we multiply the fibre in the direction of $\xi_2^{\rm L} = \partial_u$ by a warp factor. The warping can take either the shape of ``stretching'' if the warp factor is positive or ``squashing'' if the warp factor is negative. Since the warping is made in the direction of the spacelike $\partial_u$, we call the resulting spacetime spacelike warped AdS${}_3$.

\begin{definition}
The \emph{spacelike warped} AdS${}_3$ spacetime has metric
\begin{equation} \label{eq:spacelikeWAdS}
\dd s^2 = \frac{\ell^2}{\nu^2+3} \left( - \theta^0_{\rm R} \otimes \theta^0_{\rm R} + \theta^1_{\rm R} \otimes \theta^1_{\rm R} + \frac{4\nu^2}{\nu^2+3} \theta^2_{\rm R} \otimes \theta^2_{\rm R} \right) \, ,
\end{equation}
where $\nu = \mu \ell/3$. For $\nu^2>1$ we have \emph{spacelike stretched} AdS${}_3$, for $\nu^2<1$ we have \emph{spacelike squashed} AdS${}_3$.
\end{definition}

In fibred coordinates $(\tau, \sigma, u)$, the metric is
\begin{equation} \label{eq:spacelikeWAdS}
\dd s^2 = \frac{\ell^2}{\nu^2+3} \left[- \cosh (\sigma)^2 \, \dd\tau^2 + \dd\sigma^2 + \frac{4\nu^2}{\nu^2+3} (\dd u+\sinh (\sigma) \, \dd\tau)^2 \right] \, .
\end{equation}

The isometry group of AdS${}_3$, which locally is $SL(2,\mathbb{R})_{\rm L} \otimes SL(2,\mathbb{R})_{\rm R}$, is broken by the warping and is only generated by $\xi_2^{\rm L}$ and $\xi_i^{\rm R}$, $i=0,1,2$. Hence, the isometry group of spacelike warped AdS${}_3$ is $U(1)_{\rm L} \otimes SL(2,\mathbb{R})_{\rm R}$.

\begin{remark}
Both AdS${}_3$ and spacelike warped AdS${}_3$ are solutions of TMG, but the latter is \emph{not} a vacuum solution of Einstein gravity in 2+1 dimensions.
\end{remark}

\begin{remark}
Besides the spacelike warped AdS${}_3$ spacetime, there exist also timelike and null warped AdS${}_3$ spacetimes. For more details, see e.g.~\cite{Anninos:2008fx}.
\end{remark}

\subsection{Warped AdS$\mathbf{{}_3}$ black hole solutions}
\label{sec:wadsbh-wads3bhsolutions}

Black hole solutions which are asymptotically warped AdS${}_3$ and do not have CTCs have only been found in the spacelike stretched case. The spacelike stretched black hole metric in coordinates $(t,r,\theta)$ is \cite{Anninos:2008fx}
\begin{equation}
\dd s^2 = \dd t^2 + \frac{\ell^2 \dd r^2}{4 R(r)^2 N(r)^2} + 2 R(r)^2 N^{\theta}(r) \dd t d\theta + R(r)^2 \dd \theta^2 \, ,
\label{eq:metricbh}
\end{equation}
with $r \in (0,\infty)$, $t \in (-\infty,\infty)$, $(t,r,\theta) \sim (t,r,\theta + 2\pi)$ and
\begin{align}
R(r)^2 &= \frac{r}{4} \left[ 3(\nu^2-1)r + (\nu^2+3)(r_+ + r_-) - 4\nu \sqrt{r_+ r_-(\nu^2+3)} \right] \, , \\
N(r)^2 &= \frac{(\nu^2+3)(r-r_+)(r-r_-)}{4R(r)^2} \, , \\
N^{\theta}(r) &= \frac{2\nu r - \sqrt{r_+ r_- (\nu^2+3)}}{2 R(r)^2} \, .
\end{align}
We can also write the metric in ADM form as
\begin{equation}
\dd s^2 = -N(r)^2 \, \dd t^2 + \frac{\ell^2 \dd r^2}{4 R(r)^2 N(r)^2} + R(r)^2 \left( \dd \theta + N^{\theta}(r) \, \dd t \right)^2 \, .
\end{equation}

In the rest of this section, some of the more important features of these black holes that will be needed in later chapters are briefly described. More details can be found in \cite{Anninos:2008fx,Ferreira:2013zta}.

\begin{enumerate}
\item There are outer and inner horizons at $r = r_+$ and $r = r_-$, respectively, and a singularity in the causal structure located at $r = \bar{r}_0 := \max\{ 0, r_0 \}$, with
\begin{equation}
r_0 = \frac{4\nu \sqrt{r_+ r_- (\nu^2+3)} - (\nu^2+3)(r_++r_-)}{3(\nu^2-1)} \, ,
\end{equation}
such that $0 \leq \bar{r}_0 \leq r_- \leq r_+$. The dimensionless constant $\nu = \mu \ell/3$ is greater than unity for the spacelike stretched black hole and in this context is usually known as the \emph{warp factor}. In the limit $\nu \to 1$ the metric reduces to the metric of the BTZ black hole in a rotating frame.
\item In this coordinate system, the vector fields $\partial_t$ and $\partial_{\theta}$ are Killing vector fields, however, $\partial_t$ is spacelike \emph{everywhere} in the spacetime. Consequently, this black hole does not have a stationary limit surface and its ergoregion extends to infinity. Therefore, no observers follow orbits of $\partial_t$ in the exterior region.
\item Notwithstanding the previous point, one can still consider observers following orbits of the (non Killing) vector field $\xi(r) = \partial_t + \Omega(r) \, \partial_{\theta}$ at a given radius $r$, which is timelike as long as
\begin{equation}
\Omega_-(r) < \Omega(r) < \Omega_+(r) \, ,
\label{eq:conditionalmoststationary}
\end{equation}
with
\begin{equation}
\Omega_{\pm}(r) = - \frac{2}{2\nu r - \sqrt{r_+ r_- (\nu^2+3)} \pm \sqrt{(r-r_+)(r-r_-)(\nu^2+3)}} \, .
\end{equation}
$\Omega(r)$ is negative for all $r > r_+$, approaches zero as $r \to +\infty$, and tends to
\begin{equation}
\Omega_{\mathcal{H}} = - \frac{2}{2\nu r_+ - \sqrt{r_+ r_- (\nu^2+3)}}
\label{eq:omegaH}
\end{equation}
as $r \to r_+$. In view of this, we can take $\Omega_{\mathcal{H}}$ as the \emph{angular velocity of the horizon} with respect to stationary observers in the limit they approach infinity.

One example of such a timelike vector field in the exterior region is 
\begin{equation} \label{eq:xiLNRO}
\xi(r) = \partial_t - N^{\theta}(r) \, \partial_{\theta} \, .
\end{equation}
Observers following orbits of $\xi(r)$ are known as locally non-rotating observers (LNRO) or zero angular momentum observers (ZAMO), since the $\theta$-component of the one-form $\xi_a$ is
\begin{equation}
\xi_{\theta} = g_{\theta \mu} \xi^{\mu} = R(r)^2 N^{\theta}(r) + R(r)^2 \big( {- N^{\theta}(r)} \big) = 0 \, .
\end{equation}
We will further consider these observers in Section~\ref{sec:existence-superradiance}. The vector field $\xi(r)$ is a representative of the time-orientation of the exterior region of the spacelike stretched black hole, cf.~Definition~\ref{def:timeorientation}.
\item Note that, even though the Killing vector field $\partial_t$ is spacelike in the exterior region, $t$ is a \emph{time function}, in the sense of Definition~\ref{def:timefunction}. To see this, let $\eta^a := - \nabla^a t$.  One has that $\eta^a$ is timelike,
\begin{equation}
\eta^2 = g^{\mu\nu} \eta_{\mu} \eta_{\nu} = g^{tt} = - \frac{1}{N(r)^2} < 0 \, ,
\end{equation}
and is future-directed,
\begin{align}
g_{\mu\nu} \eta^{\mu} \xi^{\nu} = g_{\mu\nu} \left(- g^{\mu t} \right) \xi^{\nu} 
= - 1 < 0 \, ,
\end{align}
where $\xi(r)$ in \eqref{eq:xiLNRO} was used as a representative of the time-orientation of the exterior region. Therefore, in this region, $t$ is increasing along worldlines of timelike curves and, furthermore, constant-$t$ surfaces are spacelike.
\item Similarly to the Kerr spacetime, there is a \emph{speed of light surface}, beyond which an observer cannot co-rotate with the event horizon. It is located at the surface where the Killing vector field which generates the horizon,
\begin{equation}
\chi = \partial_t + \Omega_{\mathcal{H}} \, \partial_{\theta} \, ,
\end{equation}
is null,
\begin{equation} \label{eq:solsurface}
r = r_{\mathcal{C}} = \frac{4\nu^2 r_+ - (\nu^2+3) r_-}{3(\nu^2-1)} \, .
\end{equation}
\item The spacelike stretched black hole can be obtained as the quotient of spacelike stretched AdS${}_3$ under a discrete subgroup of the isometry group, the same way the BTZ black hole is a quotient of AdS${}_3$ \cite{Banados:1992wn,Banados:1992gq}. The discrete subgroup is the one generated by the Killing vector $\partial_{\theta}$, which in terms of the original fibred coordinates $(\tau, \sigma, u)$ is given by
\begin{equation}
\partial_{\theta} = \frac{(\nu^2+3)(r_+-r_-)}{4} \xi^{\rm R}_2 
+ \frac{\nu^2+3}{4} \left( r_+ + r_- - \frac{\sqrt{r_+ r_- (\nu^2+3)}}{\nu} \right) \xi^{\rm L}_2 \, ,
\end{equation}
such that the identification of points $x$ of the spacelike stretched AdS${}_3$ is
\begin{equation} \label{eq:identpoints}
x \sim \exp(2\pi\lambda \, \partial_{\theta}) \, x \, , \qquad \lambda \in \mathbb{Z} \, .
\end{equation}
Across the spacelike stretched AdS${}_3$ spacetime, $\partial_{\theta}$ can be spacelike, null or timelike. The spacelike stretched black hole is then the region where $\partial_{\theta}$ is spacelike, which is geodesically incomplete. The boundaries are the surfaces where $\partial_{\theta}$ is null and they correspond to the singularity $r = \bar{r}_0$ in the causal structure. The region where $\partial_{\theta}$ is timelike would have closed timelike curves upon the identification \eqref{eq:identpoints}.

Therefore, the spacelike stretched black hole is locally equivalent to spacelike stretched AdS${}_3$. Another explicit way to see this is by a local coordinate transformation from the spacelike stretched AdS${}_3$ metric \eqref{eq:spacelikeWAdS} to the spacelike stretched black hole metric \eqref{eq:metricbh},
\begin{subequations}
\begin{align}
\tau &= \arctan \left[ \frac{2\sqrt{(r-r_+)(r-r_-)}}{2r-r_+-r_-} \sinh \left( \frac{\nu^2+3}{4}(r_+-r_-) \theta \right) \right] \, , \\
u &= \frac{\nu^2+3}{4\nu} \left[ 2t + \left( \nu (r_+-r_-) - \sqrt{r_+ r_-(\nu^2+3)} \right) \theta \right] \notag \\
&\quad - \arctan \left[ \frac{r_+ + r_- - 2r}{r_+ - r_-} \coth \left( \frac{\nu^2+3}{4}(r_+-r_-) \theta \right) \right] \, , \\
\sigma &= \text{asinh} \left[ \frac{2\sqrt{(r-r_+)(r-r_-)}}{2r-r_+-r_-} \cosh \left( \frac{\nu^2+3}{4}(r_+-r_-) \theta \right) \right] \, ,
\end{align}
\end{subequations}
valid for $\nu > 1$ and for the non-extremal case $r_+ > r_-$ (more details on this and the extremal case can be found in Ref.~\cite{Anninos:2008fx}).
\item Using the standard procedure, the Carter-Penrose diagrams for these black hole spacetimes were obtained in Ref.~\cite{Jugeau:2010nq} and are shown in Fig.~\ref{fig:CPdiagrams}. We see that the causal structure is very similar to that of asymptotically flat spacetimes in 3+1 dimensions. Indeed, the diagrams for the cases $r_0 < r_- < r_+$ and $r_0 < r_- = r_+$ are exactly the same as those for the standard and extreme Reissner-Nordstr\"{o}m black holes, while the one for the case $r_0 = r_- < r_+$ is identical to that for the Kruskal spacetime. For this reason, one may expect the behaviour of matter fields on these spacetimes to be qualitatively similar to that on the asymptotically flat ones.

In the next chapter, we will focus on the case of a spacelike stretched black hole for which $r_0 < r_- < r_+$.
\end{enumerate}

\begin{figure}[ht!]
\begin{center}
{\small
\subfigure[$r_0<r_-<r_+$]{
\begin{tikzpicture}[scale=0.56]

\node (I)    at ( 2,0)   {I};
\node (IV)   at (-2,0)   {IV};
\node (II)   at (0, 2) {II};
\node (III)  at (0,-2) {III};

\path  
  (IV) +(90:2)  coordinate (IVtop)
       +(-90:2) coordinate (IVbot)
       +(0:2)   coordinate (IVright)
       +(180:2) coordinate (IVleft)
       ;
       
\draw (IVleft) -- (IVtop); 
\draw[dashed] (IVtop) -- (IVright);
\draw[dashed] (IVright) -- (IVbot); 
\draw (IVbot) -- (IVleft);

\path 
   (I) +(90:2)  coordinate (Itop)
       +(-90:2) coordinate (Ibot)
       +(180:2) coordinate (Ileft)
       +(0:2)   coordinate (Iright)
       ;

\draw[dashed] (Ileft) -- (Itop)
	node[midway, below right, xshift=-1mm] {\scriptsize {$r=r_+$}}; 
\draw (Itop) -- (Iright);
\draw (Iright) -- (Ibot); 
\draw[dashed] (Ibot) -- (Ileft);

\path (Itop) +(90:4) coordinate (IItopright);
\path (IVtop)+(90:4) coordinate (IItopleft);
\node (IItop) at (0,4) {};

\draw[dashed] (Itop)  -- (IItopleft)
	node[pos=0.25, below, xshift=-3mm] {\scriptsize {$r=r_-$}};
\draw[dashed] (IVtop) -- (IItopright);
\draw[decorate,decoration=zigzag] (IVtop) -- (IItopleft);
\draw[decorate,decoration=zigzag] (Itop) -- (IItopright)
      node[midway, right, inner sep=2mm] {\scriptsize {$r=r_0$}};

\path (IItopright) +(2,2)  coordinate (IItoptopright);
\path (IItopright) +(-2,2) coordinate (IItoptopmiddle);
\path (IItopleft)  +(-2,2) coordinate (IItoptopleft);

\draw[dashed] (IItopright)  -- (IItoptopmiddle);
\draw[dashed] (IItopleft)   -- (IItoptopmiddle);
\draw         (IItopright)  -- (IItoptopright);
\draw         (IItopleft)  -- (IItoptopleft);

\path (Ibot) +(-90:4) coordinate (IIIbotright);
\path (IVbot)+(-90:4) coordinate (IIIbotleft);
\node (IIIbot) at (0,-4) {};

\draw[dashed] (Ibot)  -- (IIIbotleft);
\draw[dashed] (IVbot) -- (IIIbotright);
\draw[decorate,decoration=zigzag] (IVbot) -- (IIIbotleft);
\draw[decorate,decoration=zigzag] (Ibot) -- (IIIbotright);

\path (IIIbotright) +(2,-2)  coordinate (IIIbotbotright);
\path (IIIbotright) +(-2,-2) coordinate (IIIbotbotmiddle);
\path (IIIbotleft)  +(-2,-2) coordinate (IIIbotbotleft);

\draw[dashed] (IIIbotright)  -- (IIIbotbotmiddle);
\draw[dashed] (IIIbotleft)   -- (IIIbotbotmiddle);
\draw         (IIIbotright)  -- (IIIbotbotright);
\draw         (IIIbotleft)  -- (IIIbotbotleft);

\end{tikzpicture}
}
%
%
%
%
\subfigure[$r_0=r_-<r_+$]{
\begin{tikzpicture}[scale=0.71]

\node (I)    at ( 2,0)   {I};
\node (IV)   at (-2,0)   {IV};
\node (II)   at (0, 1.2) {II};
\node (III)  at (0,-1.2) {III};

\path  
  (IV) +(90:2)  coordinate (IVtop)
       +(-90:2) coordinate (IVbot)
       +(0:2)   coordinate (IVright)
       +(180:2) coordinate (IVleft)
       ;
       
\draw (IVleft) -- (IVtop); 
\draw[dashed] (IVtop) -- (IVright);
\draw[dashed] (IVright) -- (IVbot); 
\draw (IVbot) -- (IVleft);

\path 
   (I) +(90:2)  coordinate (Itop)
       +(-90:2) coordinate (Ibot)
       +(180:2) coordinate (Ileft)
       +(0:2)   coordinate (Iright)
       ;

\draw[dashed] (Ileft) -- (Itop)
	node[midway, below right] {\scriptsize {$r=r_+$}}; 
\draw (Itop) -- (Iright);
\draw (Iright) -- (Ibot); 
\draw[dashed] (Ibot) -- (Ileft);

\draw[decorate,decoration=zigzag] (IVtop) -- (Itop)
      node[midway, above, inner sep=1mm] {\scriptsize {$r=r_0$}};

\draw[decorate,decoration=zigzag] (IVbot) -- (Ibot)
      node[midway, below, inner sep=2mm] {\scriptsize {$r=r_0$}};

\end{tikzpicture}
}
%
%
%
%
\subfigure[$r_0<r_-=r_+$]{
\begin{tikzpicture}[scale=0.7]
      
\draw[decorate,decoration=zigzag] (-2,-6) -- (-2,6)
	node[midway, right, inner sep=1mm] {\scriptsize {$r=r_0$}};
\draw (2,6) -- (0,4) -- (2,2) -- (0,0) -- (2,-2) -- (0,-4) -- (2,-6);
\draw[dashed] (-2,6) -- (0,4) -- (-2,2) -- (0,0) --
 			      node[midway, below right, xshift=-1mm, yshift=-1mm] {\scriptsize {$r=r_+$}}
 			  (-2,-2) -- (0,-4) -- (-2,-6);
\end{tikzpicture}
}
}
\end{center}
\caption[Carter-Penrose diagrams of the spacelike stretched black hole.]{\label{fig:CPdiagrams} Carter-Penrose diagrams of the spacelike stretched black hole spacetime for different values of $r_0$, $r_-$, and $r_+$.}
\end{figure}
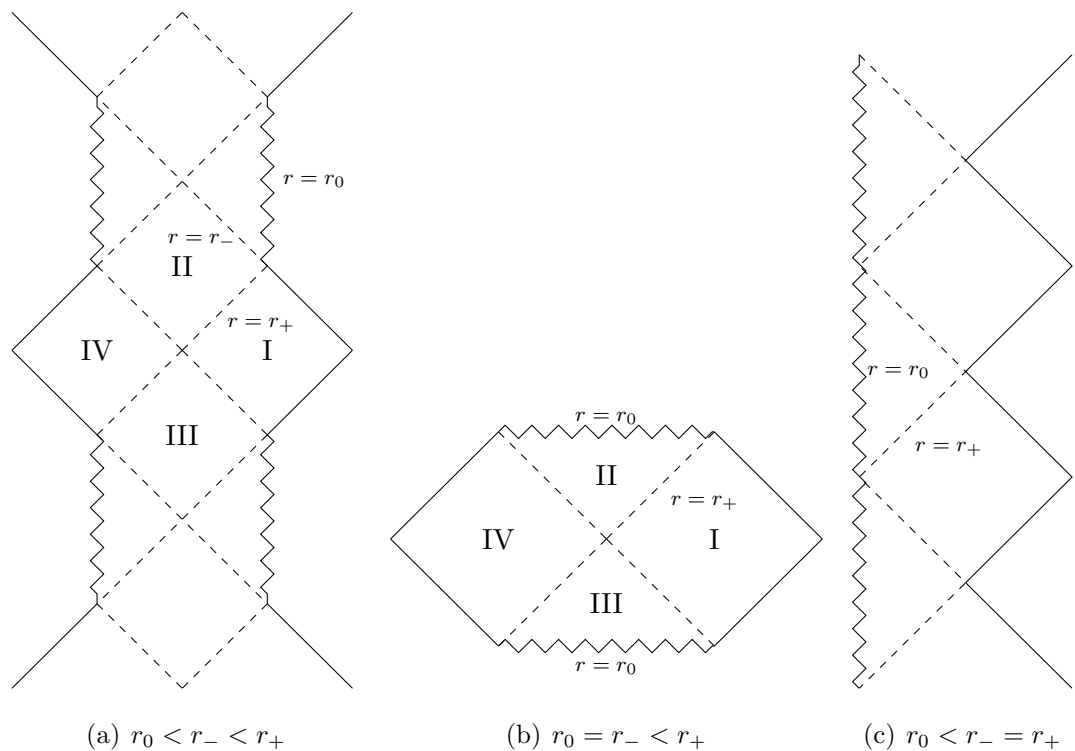


\chapter{Classical linear mode stability of the \texorpdfstring{WA\MakeLowercase{d}S$\mathbf{{}_3}$}{WAdS3} black holes}
\chaptermark{Classical linear mode stability}
\label{chap:classical-stability}

In this chapter, it is demonstrated that the warped AdS${}_3$ black hole solutions introduced in the previous chapter are classically stable against massive scalar field mode perturbations, even when the black hole is enclosed by a stationary timelike boundary with Dirichlet boundary conditions. Namely, it is shown that even though classical superradiance is present it does not give rise to superradiant instabilities. This is a surprising result given the similarity between the causal structure of the warped AdS${}_3$ black hole and the Kerr black hole in 3+1 dimensions. Having clarified the existence of the classical superradiance and the classical linear mode stability of the black hole, we then consider the quantised scalar field in the next chapter.


\section{Classical superradiance}
\label{sec:classical-superradiance}

In this section, we start by obtaining the solutions for the Klein-Gordon equation for a real massive scalar field on the spacelike stretched black hole, in both closed form and in the form of asymptotic approximations near the horizon and infinity. The latter allows us to construct bases for the space of solutions in the exterior region. Finally, we use these constructions to discuss the existence of classical superradiance.


\subsection{Klein-Gordon field equation}
\label{sec:fieldeq}

A real massive scalar field $\Phi$ on the background of a spacelike stretched black hole, whose metric is \eqref{eq:metricbh}, satisfies the Klein-Gordon equation \eqref{eq:KGequation},
\begin{equation}
\left(\nabla^2 - m_0^2 - \xi R \right) \Phi = 0 \, ,
\label{eq:fieldequation1}
\end{equation}
where $m_0$ is the mass of the field, $R$ is the Ricci scalar and $\xi$ is the curvature coupling parameter. For the spacelike stretched black hole, the Ricci scalar is a constant, $R =  - (\nu^2+3)+(\nu^2-3)/\ell^2$, so \eqref{eq:fieldequation1} can be rewritten as
\begin{equation}
\left(\nabla^2 - m^2 \right) \Phi = 0 \, ,
\label{eq:fieldequation2}
\end{equation}
where $m^2 := m_0^2 + \xi R$ is the ``effective squared mass'' of the scalar field.

Since $\partial_t$ and $\partial_{\theta}$ are Killing vector fields of the spacetime, one considers mode solutions of \eqref{eq:fieldequation2} of the form
\begin{equation}
\Phi_{\omega k}(t,r,\theta) = e^{-i\omega t + i k \theta} \, \phi_{\omega k}(r) \, ,
\label{eq:fieldansatz}
\end{equation}
where $\omega \in \mathbb{R}$ and $k \in \mathbb{Z}$.

\begin{remark}
Since the Killing vector field $\partial_t$ is not timelike anywhere in the exterior region of black hole, the parameter $\omega$ cannot be strictly regarded as a ``frequency'' in the usual sense. We will come back to this detail in section \ref{sec:existence-superradiance}, but for simplicity we sometimes refer to $\omega$ as the ``frequency''.
\end{remark}

Using \eqref{eq:fieldansatz} and \eqref{eq:metricbh}, the radial equation can be easily obtained,
\begin{equation}
\frac{4}{\ell^2} R^2 N^2 \frac{\dd}{\dd r} \left( R^2 N^2 \frac{\dd \phi_{\omega k}}{\dd r} \right) + \left[ R^2 (\omega + k N^{\theta})^2 - N^2 (k^2 + m^2 R^2) \right] \phi_{\omega k} = 0 \, .
\label{eq:radialeq}
\end{equation} 
By performing the rescalings $r \to r \ell$, $t \to t \ell$, $m \to m/\ell$, and $\omega \to \omega/\ell$, one can set $\ell=1$, as is assumed from now on.

In this (2+1)-dimensional setting, it is possible to write the general solution to the radial equation in closed form. Introducing a new radial coordinate
\begin{equation}
z = \frac{r-r_+}{r-r_-} \, 
\label{eq:definitionz1}
\end{equation}
the general real solution can be written as
\begin{equation}
\phi_{\omega k} (z) = A_{\omega k} \, z^{\alpha} (1-z)^{\beta} F(a,b;c;z) + B_{\omega k} \, \overline{z^{\alpha} (1-z)^{\beta} F(a,b;c;z)} \, ,
\label{eq:exactsolutionLorentzian}
\end{equation}
where $A_{\omega k}$ and $B_{\omega k}$ are constants, $F(a,b;c;z) := {}_2 F_1(a,b;c;z)$ is the Gaussian hypergeometric function (see Appendix~\ref{app:hypergeometric}) and its parameters are given by
\begin{align}
a = \alpha + \beta + \gamma \, , \qquad
b = \alpha + \beta - \gamma \, , \qquad
c = 2\alpha + 1 \, ,
\label{eq:defabcLorentzian}
\end{align}
where
\begin{subequations}
\begin{align}
\alpha &= -i \tilde{\omega}_+ \frac{2\nu r_+ - \sqrt{r_+ r_-(\nu^2+3)}}{(\nu^2+3)(r_+ - r_-)} \, , \\
\beta &= \frac{1}{2} - \hat{\varpi} \frac{\sqrt{3(\nu^2-1)}}{\nu^2+3} \, , \\ 
\gamma & = -i \tilde{\omega}_- \frac{2\nu r_- - \sqrt{r_+ r_-(\nu^2+3)}}{(\nu^2+3)(r_+ - r_-)} \, ,
\end{align}
\end{subequations}
and
\begin{align}
\tilde{\omega}_{\pm} := \omega + k N^{\theta}(r_{\pm}) \, , \qquad
\hat{\varpi} := \sqrt{\frac{(\nu^2+3)^2}{12(\nu^2-1)} \left( 1 + \frac{4m^2}{\nu^2+3} \right) - \omega^2} \, .
\end{align}

This exact solution will be useful for the stability analysis below. To discuss the existence of classical superradiance, it will be sufficient to consider asymptotic approximations near the horizon and infinity, as is done in the next subsection.


\subsection{Asymptotic mode solutions}
\label{eq:stability-asymptoticsolutions}

In order to construct convenient bases of mode solutions, the asymptotic approximations near the horizon and infinity are obtained by rewriting the radial field equation as a Schr\"{o}dinger-like equation. To do that, the first step is to derive the effective potential seen by the scalar field. Define the tortoise coordinate $r_*$ by
\begin{equation}
\frac{\dd r_*}{\dd r} = \frac{1}{2 R N^2} \, ,
\label{eq:tortoiseeq}
\end{equation}
which maps $(r_+, \infty)$ to $(-\infty,\infty)$ for $\nu>1$ and to $(-\infty,\hat{r}_*)$, where $\hat{r}_*$ is a finite value, for $\nu = 1$. Introduce the new radial function $\varphi_{\omega k}$,
\begin{equation}
\phi_{\omega k}(r) =: R(r)^{-1/2} \, \varphi_{\omega k} (r) \, .
\label{eq:varphi}
\end{equation}
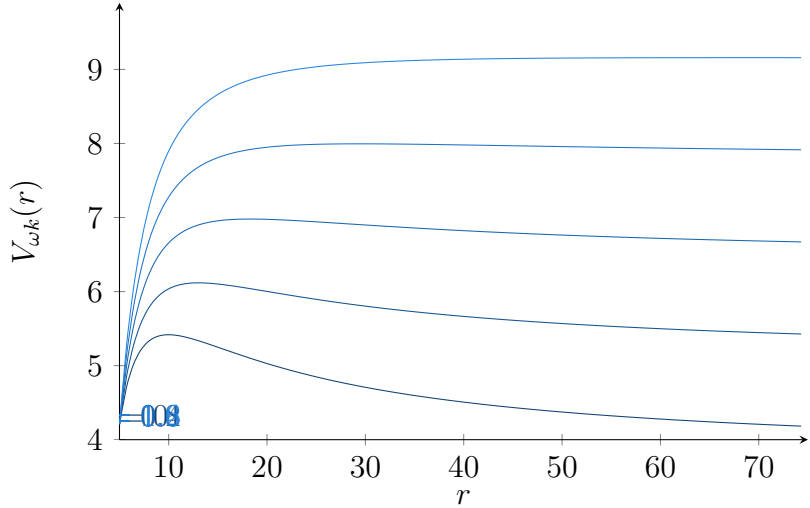
\begin{figure}[t!]
\centering
\begin{tikzpicture}
\begin{axis}[
	width=0.57\textwidth,
	x post scale=1.3,
	xlabel={$r$},
	ylabel={$V_{\omega k}(r)$},
	axis lines=left,
	xmin=5, xmax=75,
	ymin=4, ymax=9.9,
]
\addplot[smooth, no markers, blue1] 
	table {data/effectivepotential-1.txt}
	[yshift=10pt]
		node[pos=0.65]{$m^2=0$}
	;
\addplot[smooth, no markers, blue2] 
	table {data/effectivepotential-2.txt}
	[yshift=10pt]
		node[pos=0.65]{$m^2=0.4$}
	;
\addplot[smooth, no markers, blue3] 
	table {data/effectivepotential-3.txt}
	[yshift=10pt]
		node[pos=0.65]{$m^2=0.8$}
	;
\addplot[smooth, no markers, blue4] 
	table {data/effectivepotential-4.txt}
	[yshift=10pt]
		node[pos=0.65]{$m^2=1.2$}
	;
\addplot[smooth, no markers, blue5] 
	table {data/effectivepotential-5.txt}
	[yshift=10pt]
		node[pos=0.65]{$m^2=1.6$}
	;
\end{axis}
\end{tikzpicture}
\caption[Effective potential of the scalar field for selected values of the squared mass.]{Effective potential $V_{\omega k}(r)$ for selected values of $m^2$ with $r_+=5$, $r_-=2.5$, $\nu=1.2$, $\omega = 5$ and $k=-1$. For smaller values of $m^2$ (or larger values of $\omega$) the potential has a local maximum near the horizon, around which a potential barrier stands. As one considers fields with larger $m^2$ (or smaller $\omega$), the potential barrier eventually disappears. 
\vspace*{2ex}   
}
\label{fig:potentialmass}
\end{figure}
The radial field equation \eqref{eq:radialeq} can then be written in a Schr\"{o}dinger-like form
\begin{equation}
\left( \frac{\dd^2}{\dd r_*^2} + (\omega^2 - V_{\omega k}(r)) \right) \varphi_{\omega k}(r) = 0 \, ,
\label{eq:effectiveeq}
\end{equation}
with:
\begin{align}
V_{\omega k} &:= \omega^2 - (\omega + k N^{\theta})^2 + 2 N^3 \left( R N \frac{\dd^2 R}{\dd r^2} + \frac{1}{2} N \left(\frac{\dd R}{\dd r}\right)^2 + 2 R \frac{\dd R}{\dd r} \frac{\dd N}{\dd r} \right) \notag \\
&\quad\; + N^2 \left( m^2 + \frac{k^2}{R^2} \right) \, .
\label{eq:effectivepotential}
\end{align}
The function $V_{\omega k}(r)$ can hence be regarded as the \emph{effective potential} experienced by the scalar field of effective squared mass $m^2$, frequency $\omega$, and angular momentum number $k$. Figure~\ref{fig:potentialmass} shows the form of $V_{\omega k}(r)$ for selected values of $m^2$.

\begin{remark}
Similarly to what happens in the Kerr spacetime \cite{Simone:1991wn}, $V_{\omega k}(r)$ depends on the frequency $\omega$ of the scalar field (when $k \neq 0$). Also, $V_{\omega k}(r) \to +\infty$ as $r \to +\infty$ and $\nu \to 1$, as expected for the BTZ black hole.
\end{remark}

One can now find the asymptotic solutions of \eqref{eq:effectiveeq} near the horizon and near infinity by analysing the behaviour of the effective potential in those regions.
\begin{enumerate}
\item In the near-horizon limit,
\begin{equation}
V_{\omega k}(r) - \omega^2 \to - \tilde{\omega}^2 \, , \qquad r \to r_+ \, ,
\end{equation}
where
\begin{equation}
\tilde{\omega} := \omega + k N^{\theta}(r_+) = \omega - k \Omega_{\mathcal{H}} \, .
\label{eq:omegatilde}
\end{equation}
Thus, the solution near the horizon is of the form
\begin{equation}
\varphi_{\omega k}(r_*) = A_{\omega k} \, e^{i \tilde{\omega} r_*} + B_{\omega k} \, e^{- i \tilde{\omega} r_*} \, .
\label{eq:asymptoticmodesatthehorizon}
\end{equation}
Modes of the form $e^{i \tilde{\omega} r_*}$ are outgoing from the past event horizon, while modes of the form $e^{-i \tilde{\omega} r_*}$ are ingoing to the future event horizon.
\item At infinity,
\begin{align}
V_{\omega k}(r) \to \omega_m^2 \, , \qquad r \to \infty \, ,
\label{eq:effectivepotentialinfinity}
\end{align}
where
\begin{equation}
\omega_m := \frac{1}{2} \frac{\nu^2+3}{\sqrt{3(\nu^2-1)}} \sqrt{ 1 + \frac{4m^2}{\nu^2+3} } \, .
\label{eq:omegam}
\end{equation}

Two cases need now to be distinguished.
\begin{itemize}
\item[(a)] In the case $|\omega| > \omega_m$, the asymptotic solution is of the form
\begin{equation}
\varphi_{\omega k}(r_*) = C_{\omega k} \, e^{i \hat{\omega} r_*} + D_{\omega k} \, e^{- i \hat{\omega} r_*} \, ,
\label{eq:asymptoticmodesatinfinitycasea}
\end{equation}
where
\begin{equation}
\hat{\omega} := \begin{cases}
\sqrt{\omega^2-\omega_m^2} \, , & \omega > \omega_m \geq 0 \, , \\
-\sqrt{\omega^2-\omega_m^2} \, , & \omega < -\omega_m \leq 0 \, .
\end{cases}
\label{eq:hatomegadef}
\end{equation}
When $\hat{\omega} > 0$, modes of the form $e^{i \hat{\omega} r_*}$ correspond to outgoing flux at infinity, while the modes of the form $e^{-i \hat{\omega} r_*}$ correspond to incoming flux at infinity, and vice versa when $\hat{\omega} < 0$.
\item[(b)] In the case $|\omega| < \omega_m$, the asymptotic solutions are
\begin{equation}
\varphi_{\omega k}(r_*) = E_{\omega k} \, e^{\hat{\varpi} r_*} + F_{\omega k} \, e^{-\hat{\varpi} r_*} \, ,
\label{eq:asymptoticmodesatinfinitycaseb}
\end{equation}
where
\begin{equation}
\hat{\varpi} \equiv \begin{cases}
\sqrt{\omega_m^2 - \omega^2} \, , & 0 < \omega < \omega_m \, , \\
-\sqrt{\omega_m^2 - \omega^2} \, , & -\omega_m < \omega < 0 \, .
\end{cases}
\label{eq:hatvarpidef}
\end{equation}
To exclude the solution that diverges exponentially at infinity, impose that $E_{\omega k} = 0$ when $0 < \omega < \omega_m$ and  $F_{\omega k} = 0$ when $-\omega_m < \omega < 0$.
\end{itemize}
\end{enumerate}

\begin{remark} \label{rem:omegam}
The behaviour of the effective potential at infinity given by \eqref{eq:effectivepotentialinfinity} contrasts with that in asymptotically flat spacetimes such as Kerr, where the effective potential tends to $m^2$ at infinity \cite{Simone:1991wn}, and with asymptotically AdS spacetimes such as the BTZ or Kerr-AdS, where the effective potential grows without bound at infinity \cite{Winstanley:2001nx}. It is assumed that the asymptotic value of $V_{\omega k}$ at infinity, $\omega_m^2$, is non-negative and, by \eqref{eq:omegam}, this implies that $m^2$ may be negative provided it satisfies $m^2 \geq -\frac{\nu^2+3}{4}$. As a consistency check, in the BTZ limit $\nu \to 1$ this inequality reduces to the Breitenlohner-Freedman bound for AdS${}_3$ spacetimes $m^2 \geq -1$ \cite{Breitenlohner:1982bm}.
\end{remark}

\begin{remark}
The interpretation of the modes in \eqref{eq:asymptoticmodesatinfinitycasea} corresponding to incoming and outgoing flux at infinity can be explained by calculating the radial flux $j^r$ of the field mode $\phi_{\omega k}$ at infinity,
\begin{equation}
j^r = - i \, g^{rr} \left( \overline{\phi_{\omega k}} \, \frac{\dd \phi_{\omega k}}{\dd r} - \phi_{\omega k} \, \frac{\dd \overline{\phi_{\omega k}}}{\dd r} \right) \, .
\end{equation}
The result turns out to be
\begin{equation}
j^r = 4 \hat{\omega} \left( |C_{\omega k}|^2 - |D_{\omega k}|^2 \right) \, , \qquad r \to +\infty \, .
\end{equation}
Since a positive (negative) radial flux at infinity corresponds to outgoing (incoming) flux, the interpretation above follows.
\end{remark}

\begin{remark}
Note that so far no choice of ``positive frequency'' has been made, for instance, by taking $\omega > 0$. We will return to this point when discussing the existence of superradiance in Section~\ref{sec:existence-superradiance}.
\end{remark}


\subsection{Basis of mode solutions}
\label{sec:stability-basismodes}

Using the asymptotic mode solutions described above, one can construct a basis of mode solutions. Two particular basis modes will be of particular importance in the following, the ``in'' and ``up'' modes, which are specified by the boundary conditions they obey at the event horizon and at infinity. These modes are defined in analogy with the Kerr spacetime \cite{Ford:1975tp,Frolov:1989jh}. We will also define the so-called ``bound state'' modes for the case $|\omega| < \omega_m$.

\begin{definition}
For $|\omega| > \omega_m$, the \emph{in modes} are mode solutions of \eqref{eq:effectiveeq} of the form \eqref{eq:fieldansatz} which satisfy the boundary conditions
\begin{equation}
\varphi^{\rm in}_{\omega k}(r_*) = \begin{cases}
B^{\rm in}_{\omega k} \, e^{- i \tilde{\omega} r_*} \, , & r_* \to -\infty \, , \\
e^{- i \hat{\omega} r_*} + C^{\rm in}_{\omega k} \, e^{i \hat{\omega} r_*} \, , & r_* \to +\infty \, ,
\end{cases}
\label{eq:inmodes}
\end{equation}
whereas the \emph{up modes} are the ones which satisfy
\begin{equation}
\varphi^{\rm up}_{\omega k}(r_*) = \begin{cases}
e^{i \tilde{\omega} r_*} + B^{\rm up}_{\omega k} \, e^{- i \tilde{\omega} r_*} \, , & r_* \to -\infty \, , \\
C^{\rm up}_{\omega k} \, e^{i \hat{\omega} r_*} \, , & r_* \to +\infty \, .
\end{cases}
\label{eq:upmodes}
\end{equation}
In the case in which $|\omega| < \omega_m$, the \emph{bound state modes} are the mode solutions which satisfy the boundary conditions
\begin{equation}
\varphi^{\rm bs}_{\omega k}(r_*) = \begin{cases}
A^{\rm bs}_{\omega k} \, e^{i \tilde{\omega} r_*} + B^{\rm bs}_{\omega k} \, e^{- i \tilde{\omega} r_*} \, , & r_* \to -\infty \, , \\
e^{- \hat{\varpi} r_*} \, , & r_* \to +\infty \, .
\end{cases}
\label{eq:bsmodes}
\end{equation}
\end{definition}

\begin{remark}
The in modes correspond to flux coming from infinity which is partially reflected back to infinity and partially absorbed by the black hole. The up modes correspond to flux coming from the black hole which is partially reflected back to the black hole and partially sent to infinity. This is represented in Fig.~\ref{fig:inupmodes}. The bound state modes are localised near the event horizon and exponentially decay as infinity is approached.
\end{remark}

\begin{figure}[t!]
\begin{center}
\subfigure[\; `in' modes]{
\def\svgwidth{0.4\textwidth}
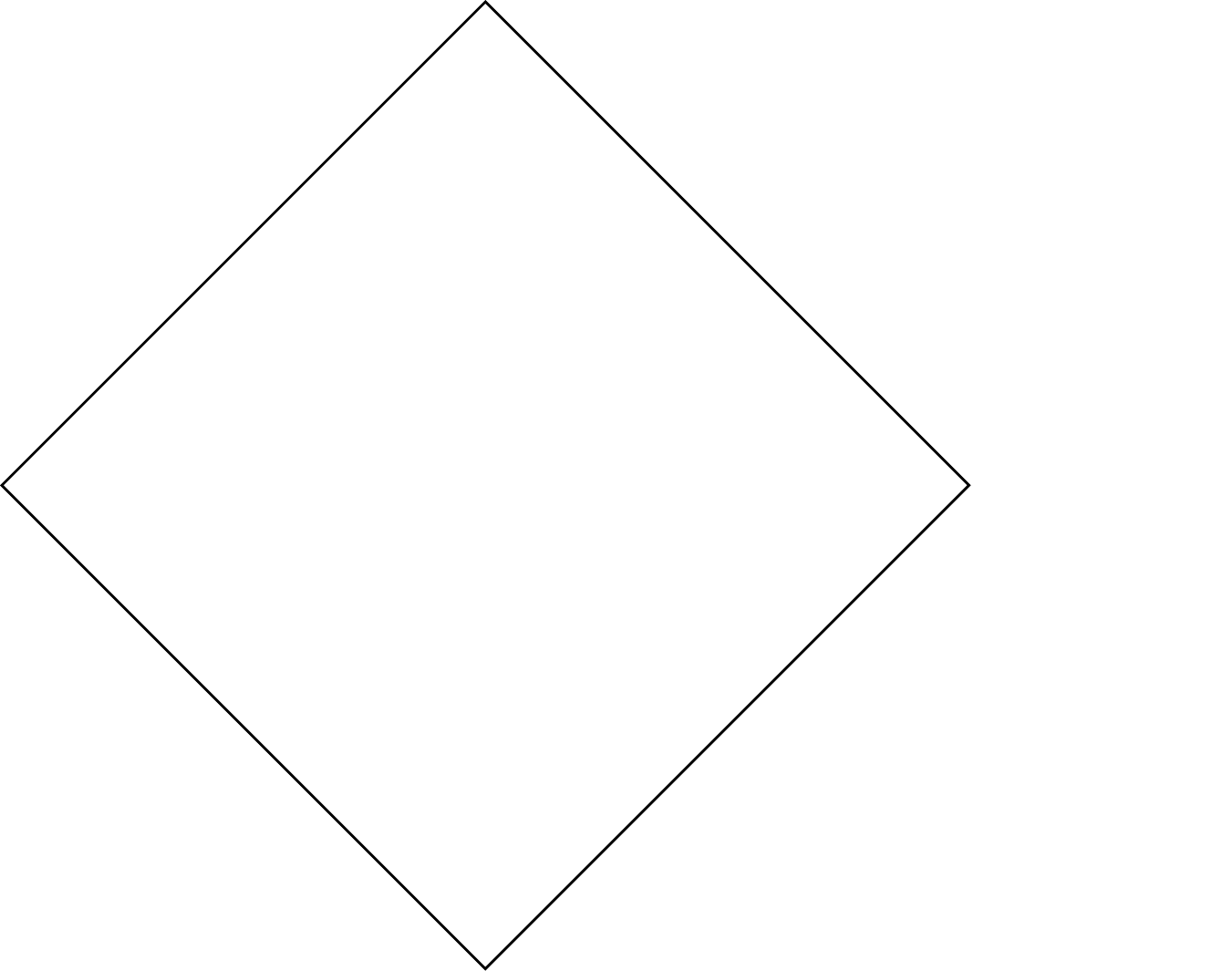}
\subfigure[\; `up' modes]{
\def\svgwidth{0.4\textwidth}
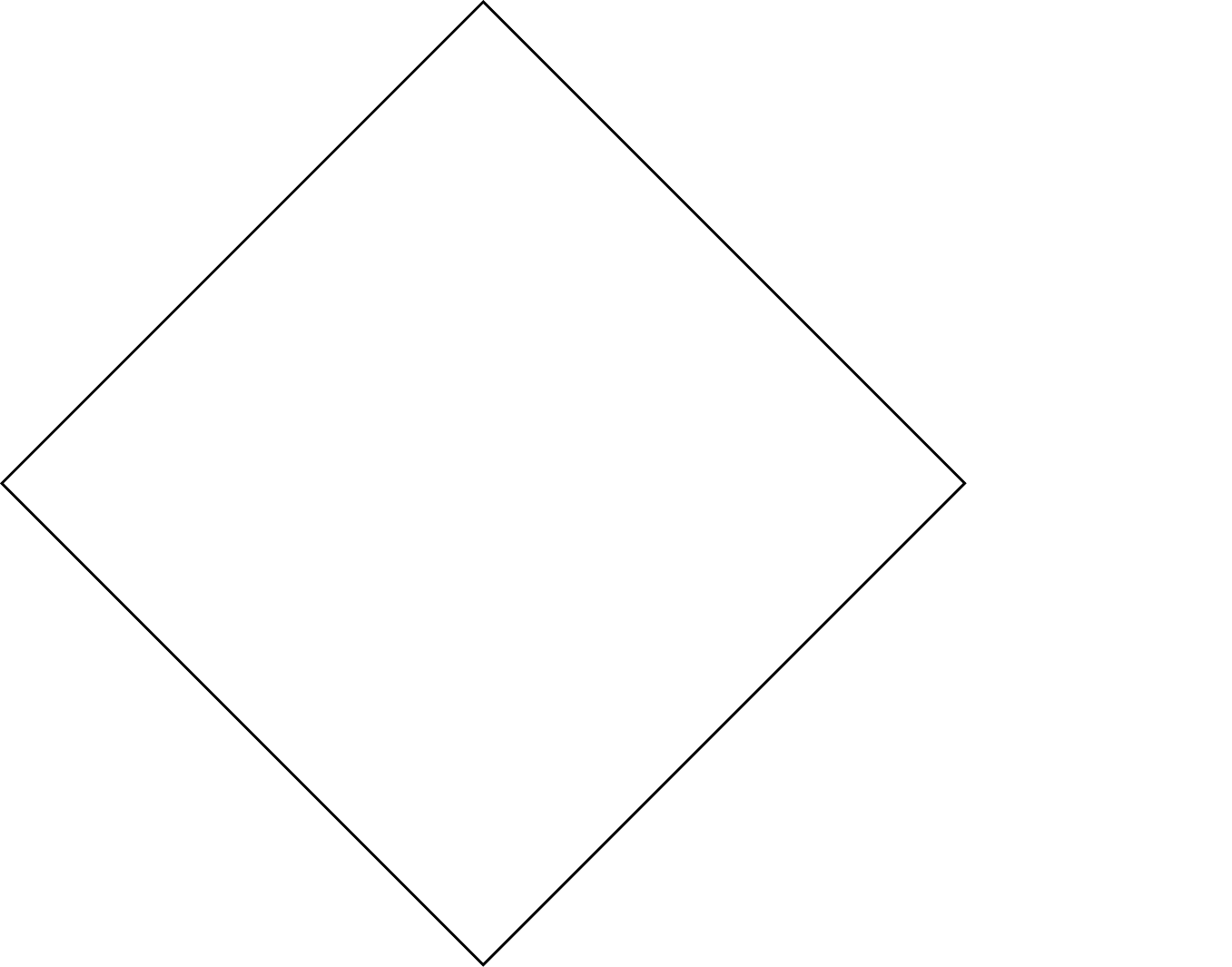}
\caption{\label{fig:inupmodes} In and up modes in the exterior region of the spacetime.}
\end{center}
\end{figure}

An immediate but important property relating the $A$, $B$ and $C$ coefficients in \eqref{eq:inmodes}, \eqref{eq:upmodes} and \eqref{eq:bsmodes} is given in the following lemma.

\begin{lemma} \label{eq:lemmaWronskiancoefficients}
The coefficients in \eqref{eq:inmodes}, \eqref{eq:upmodes} and \eqref{eq:bsmodes} satisfy
\begin{equation}
\tilde{\omega} \, |B^{\rm in}_{\omega k}|^2 = \hat{\omega} \left( 1 - |C^{\rm in}_{\omega k}|^2 \right) \, , \quad
\tilde{\omega} \left( 1 - |B^{\rm up}_{\omega k}|^2 \right) = \hat{\omega} \, |C^{\rm up}_{\omega k}|^2 \, , \quad 
|A^{\rm bs}_{\omega k}| = |B^{\rm bs}_{\omega k}| \, .
\label{eq:Wronskianrelations}
\end{equation}
\end{lemma}

\begin{proof}
The expressions relating the coefficients for each type of mode solution follow straightforwardly from the observation that, given any two linearly independent solutions $\varphi_1(r_*)$ and $\varphi_2(r_*)$ of \eqref{eq:effectiveeq}, their Wronskian is independent of $r_*$, i.e.,
\begin{equation}
W(\varphi_1, \varphi_2) := \varphi_1 \frac{\dd \overline{\varphi_2}}{\dd r_*} - \frac{\dd \varphi_1}{\dd r_*} \overline{\varphi_2} = \text{constant.}
\end{equation}
By comparing the Wronskians at the horizon and at infinity, the relations follow.
\end{proof}


\subsection{Existence of classical superradiance}
\label{sec:existence-superradiance}

In Appendix \ref{app:superradiance}, a brief overview of the classical superradiance phenomenon on black holes is given. In short, a given mode solution coming from either infinity or the horizon is called \emph{superradiant} if, after it gets reflected in the neighbourhood of the event horizon, its amplitude is increased. The superradiant nature of a given mode depends on its type and on its frequency, as formulated in the next proposition.

\begin{proposition} \label{prop:superradiantcondition}
In and up modes of a given frequency $\omega$ are superradiant if and only if $\tilde{\omega} \hat{\omega} < 0$. Bound state modes are never superradiant.
\end{proposition}

\begin{proof}
An in mode is superradiant if it is reflected back to infinity with a greater amplitude than the original one, i.e.~if $|C^{\text{in}}_{\omega k}|>1$. Using Lemma~\ref{eq:lemmaWronskiancoefficients}, this occurs when $\tilde{\omega} \hat{\omega} < 0$.

Similarly, an up mode is superradiant if it is reflected back to the horizon with a greater amplitude than the original one, i.e.~if $|B^{\text{up}}_{\omega k}|>1$. Using Lemma~\ref{eq:lemmaWronskiancoefficients}, this also occurs when $\tilde{\omega} \hat{\omega} < 0$.

Finally, for a bound state mode, $|A^{\text{bs}}_{\omega k}| = |B^{\text{bs}}_{\omega k}|$ implies that all flux coming from the horizon is reflected back. Consequently, the mode is not superradiant.
\end{proof}

Proposition~\ref{prop:superradiantcondition} gives the condition that a given mode solution of frequency $\omega$ needs to satisfy in order to be superradiant. It remains to verify if that condition can actually be fulfilled. 

At this point, one needs to discuss the notion of positive frequency, as described in detail in Section~\ref{sec:qftcst-positivefrequency}. The question of positive frequency becomes subtle for spacetimes which do not have a globally timelike Killing vector field, as we have seen, and in practice one needs to decide the location of a locally non-rotating observer with respect to whom only positive frequency modes are observed.

First, we define these locally non-rotating observers, who we have briefly mentioned in point 3 of Section~\ref{sec:wadsbh-wads3bhsolutions}.

\begin{definition} \label{def:LNRO}
An observer in the exterior region of the black hole is a \emph{locally non-rotating observer} (LNRO) if its radial coordinate $r$ is fixed and it has zero angular momentum, i.e.~$u_{\mu} (\partial_{\theta})^{\mu} = u_{\theta} = 0$, where $u^{a}$ is the future-directed unit vector tangent to the observer worldline.
\end{definition}

\begin{remark}
Locally non-rotating observers are also known as \emph{zero angular momentum observers} (ZAMO) in the literature.
\end{remark}

\begin{proposition}
A LNRO in the exterior region of the spacelike stretched black hole follows orbits of the vector field $\xi(r) := \partial_t - N^{\theta}(r) \, \partial_{\theta}$, which is timelike everywhere in the exterior region and is perpendicular to constant-$t$ surfaces.
\end{proposition}

\begin{proof}
Given that the coordinate $r$ is fixed, $u^r = 0$. The other components are
\begin{equation}
u^t = g^{tt} \, u_t \, , \qquad u^{\theta} = g^{\theta t} \, u_t \, .
\end{equation}
Hence, a vector field proportional to $u^a$ is
\begin{equation}
\partial_t + \frac{g^{\theta t}}{g^{tt}} \, \partial_{\theta} = \partial_t - N^{\theta} \, \partial_{\theta} =: \xi \, .
\end{equation}
Therefore, a LNRO follows orbits of $\xi$. This vector field is timelike as $\Omega(r) = - N^{\theta}(r)$ satisfies \eqref{eq:conditionalmoststationary} for all $r > r_+$. Furthermore, we have that $\xi_r = 0$ and
\begin{align}
\xi_t &= g_{t\mu} \xi^{\mu} = g_{tt} + g_{t\theta} (- N^{\theta}) = - N^2 \, , \\
\xi_{\theta} &= g_{\theta\mu} \xi^{\mu} = g_{\theta t} + g_{\theta\theta} (-N^{\theta}) = g_{\theta \theta} N^{\theta} + g_{\theta \theta} (- N^{\theta}) = 0 \, ,
\end{align}
thus, $\xi_a \propto (dt)_a$, which shows that $\xi^a$ is perpendicular to constant-$t$ surfaces.
\end{proof}

\begin{remark}
Note that $\Omega(r) = - N^{\theta}(r)$ is such that $\Omega(r_+) = \Omega_{\mathcal{H}}$ and $\Omega(r) \to 0$ as $r \to \infty$, as we have seen in point 3 of Section~\ref{sec:wadsbh-wads3bhsolutions}.
\end{remark}

Given these remarks, one now considers two new coordinate charts and convenient timelike Killing vector fields in each of them, near the event horizon and near spatial infinity, respectively.
\begin{enumerate}
\item Consider the open set $N_{\mathcal{H}} := \{ r_+ < r < r' \}$, for some $r' < r_{\mathcal{C}}$ (the location of the speed-of-light surface, cf.~point 5 of Section~\ref{sec:wadsbh-wads3bhsolutions}). The Killing vector field
\begin{equation}
\xi_{\mathcal{H}} := \chi = \partial_t + \Omega_{\mathcal{H}} \, \partial_{\theta}
\end{equation}
(the horizon generator) is clearly timelike in $N_{\mathcal{H}}$. In this set, consider the coordinate system $(\tilde{t}, r, \tilde{\theta})$ such that $\partial_{\tilde{t}} = \xi_{\mathcal{H}}$. It follows that $\tilde{t} = t$ and $\tilde{\theta} = \theta - \Omega_{\mathcal{H}} t$. Furthermore, for a mode solution $\Phi_{\omega k}$ of frequency $\omega$,
\begin{equation}
\frac{\partial}{\partial \tilde{t}} \Phi_{\omega k}(\tilde{t},r,\tilde{\theta}) = - i \tilde{\omega} \, \Phi_{\omega k}(\tilde{t},r,\tilde{\theta}) \, ,
\end{equation}
where $\tilde{\omega} = \omega - k \Omega_{\mathcal{H}}$ (cf.~\eqref{eq:omegatilde}). This is just the co-rotating coordinate system introduced in Section~\ref{sec:rotatingBHgeneral}.
\item Fix $r = r_*$, which can be taken to be very large, so that $r_* \gg r_+$. Define the Killing vector field
\begin{equation}
\xi_* := \partial_t + \Omega_* \, \partial_{\theta} \, ,
\end{equation}
with $\Omega_* := - N^{\theta}(r_*)$. There exists a small enough neighbourhood $N_*$ of $r=r_*$ such that $\xi_*$ is timelike in $N_*$, since the spacetime is locally stationary. In this neighbourhood, consider a new coordinate system $(t_*, r, \theta_*)$, such that $\partial_{t_*} = \xi_*$. It follows that $t_* = t$ and $\theta_* = \theta - \Omega_* t$. Furthermore, for a mode solution $\Phi_{\omega k}$ of frequency $\omega$,
\begin{equation}
\frac{\partial}{\partial t_*} \Phi_{\omega k}(t_*, r, \theta_*) = - i \omega_* \, \Phi_{\omega k}(t_*, r, \theta_*) \, ,
\end{equation}
where
\begin{equation}
\omega_* := \omega - k \Omega_* \, .
\end{equation}
\end{enumerate}

\begin{remark}
The idea behind the definition of the Killing vector field $\xi_*$ is the fact that one cannot use the everywhere spacelike Killing vector field $\partial_t$ to define positive frequency modes in a neighbourhood of spatial infinity. For the purpose of checking the existence of superradiant modes, it is enough to consider a neighbourhood of $r = r_*$, which can be taken to be as far from the black hole as desired.
\end{remark}

We now have all the necessary ingredients to pick appropriate notions of positive frequency near the event horizon and near spatial infinity. We adopt the terminology of Ref.~\cite{Frolov:1989jh} concerning ``near-horizon'' and ``distant'' observers.

\begin{enumerate}
\item For the up modes, one chooses to have positive frequency as measured by a LNRO close to the horizon (the `near-horizon observer' viewpoint), i.e.~positive frequency is defined with respect to $\tilde{t}$, which requires $\tilde{\omega} > 0$. 
\item For the in modes, one chooses to have positive frequency as measured by a LNRO near spatial infinity (the `distant observer' viewpoint), i.e.~positive frequency is defined with respect to $t_*$, which requires $\omega_* > 0$. If $\omega_m > 0$, one must additionally have $\omega > \omega_m$ for the in mode to exist, so that the positive frequency condition altogether is $\omega_* > \omega_m$. 
\end{enumerate}

The main conclusion of this section, expressed in the following theorem, is that classical superradiance is present in the spacelike stretched black hole.

\begin{theorem} \label{thm:existenceofsuperradiance}
Superradiant mode solutions exist for a massive scalar field on the background of a spacelike stretched black hole.
\end{theorem}

\begin{proof}
Proposition \ref{prop:superradiantcondition} states that up and in modes of frequency $\omega$ are superradiant if $\tilde{\omega} \hat{\omega} < 0$. We check that this condition is indeed possible for each mode.

\begin{enumerate}
\item An up mode has $\tilde{\omega} > 0$, as measured by a LNRO near the horizon. Therefore, the mode is superradiant if and only if $\hat{\omega} < 0$. This occurs when $\omega < - \omega_m$.
\item An in mode has $\omega_* > \omega_m$, as measured by a LNRO near spatial infinity, which is equivalent to $\omega > \omega_m + k \Omega_*$. This condition does not fix the sign of $\omega$, so there are two cases to consider.
\begin{itemize}
\item[(i)] In the case $\omega > \omega_m$, the condition $\tilde{\omega} \hat{\omega} < 0$ is equivalent to
\begin{equation}
\omega_m + k \Omega_* < \omega < k \Omega_{\mathcal{H}} \, ,
\label{eq:insuperradiantcondition1}
\end{equation}
with $k < 0$.
\item[(ii)] In the case $\omega < - \omega_m$, the condition $\tilde{\omega} \hat{\omega} < 0$ cannot be satisfied unless $k > 2 \omega_m / |\Omega_*|$, in which case it is equivalent to
\begin{equation}
k \Omega_{\mathcal{H}} < \omega < - \omega_m \, .
\label{eq:insuperradiantcondition2}
\end{equation}
\end{itemize}
When either \eqref{eq:insuperradiantcondition1} or \eqref{eq:insuperradiantcondition2} is satisfied, the in mode is superradiant.
\end{enumerate}

Hence, with the above choice of viewpoints, there can be superradiant mode solutions for a massive scalar field on a spacelike stretched black hole. (Note that with the above choice the modes have positive Klein-Gordon norm.) Since the in and up modes constitute a basis with which any solution of the scalar field equation can be expressed at any point in the exterior region of the spacetime, it can be concluded that classical superradiance is present in this spacetime.
\end{proof}

\begin{remark}
Note that the parameter $\Omega_*$ can be made arbitrarily small by fixing the location of the LNRO near spatial infinity to have arbitrarily large radial coordinate. If there was a Killing vector field which was timelike in a neighbourhood of spatial infinity (i.e.~for $r > r''$ for some $r'' > r_+$), then the limit $r_* \to \infty$ could be taken (sending the LNRO to infinity) and the familiar superradiance condition $\omega_m < \omega < k \Omega_{\mathcal{H}}$ would be recovered \cite{Brito:2015oca}.
\end{remark}

\begin{remark}
As stated in Proposition~\ref{prop:superradiantcondition}, the bound state modes cannot be superradiant. Therefore, in the frequency range $|\omega| < \omega_m$ there are no superradiant modes. This is similar to the situation with the BTZ black hole when reflective boundary conditions are imposed \cite{Ortiz:2011wd}.
\end{remark}

This result is in agreement with the expectation that the behaviour of the field modes should be similar to the Kerr spacetime case, given the similar causal structure and boundary conditions that we imposed. It is also interesting to note that the situation is significantly different for the Kerr-AdS spacetime, where classical superradiance is not inevitable \cite{Winstanley:2001nx}.


\sectionmark{Quasinormal and bound state modes}  
\section{Quasinormal and bound state modes and classical linear mode stability}
\sectionmark{Quasinormal and bound state modes}
\label{sec:quasinormal-modes}

In this section, we find the quasinormal and bound state scalar field modes and use the results to discuss the classical linear mode stability of the black hole solutions.

\subsection{Quasinormal and bound state modes}
\label{sec:stability-QNandBS}

Suppose that a spacelike stretched black hole is perturbed by a massive scalar field propagating in the spacetime. Once the black hole is perturbed it responds by releasing gravitational and scalar waves in the form of characteristic \emph{quasinormal modes} of discrete complex frequencies (for recent reviews on quasinormal modes see \cite{Konoplya:2011qq,Berti:2009kk}). For a stable black hole the quasinormal modes are exponentially decaying in time; conversely, if any of the modes are increasing in time, the black hole is unstable. Moreover, as seen in the previous section, there can be superradiant modes in this spacetime. If any of these superradiant modes are localised near the event horizon in the form of \emph{bound state modes} (possibly due to a potential well in the effective potential felt by the scalar field), the repeated amplitude increases due to reflections on the walls of the potential well lead to the so-called superradiant instabilities \cite{Press:1972zz,Cardoso:2004nk,Cardoso:2005vk,Dolan:2007mj,Pani:2012vp,Dolan:2012yt,Witek:2012tr}.

The quasinormal and bound state modes are defined by appropriate boundary conditions at the horizon and at infinity. Since the system under consideration is classical, there must be no flux from the horizon, and thus one imposes that only ingoing modes are present. Furthermore, no perturbations coming in from infinity should be allowed, and hence it is required that the quasinormal modes obey outgoing boundary conditions at infinity. As for the bound state modes, since they are localized in the vicinity of the black hole, one imposes that they decrease exponentially at infinity.

These ideas can be formalised in the following two definitions.

\begin{definition}
A \emph{quasinormal} mode is a mode solution of the field equation \eqref{eq:effectiveeq} of the form \eqref{eq:fieldansatz} with the following boundary conditions:
\begin{enumerate}
\item only ingoing modes at the horizon, cf.~\eqref{eq:asymptoticmodesatthehorizon},
\begin{equation}
\varphi_{\omega k}(r_*) \sim e^{-i\tilde{\omega}r_*} \, , \qquad r_* \to - \infty \, ;
\label{eq:QNingoingcondition}
\end{equation}
\item outgoing modes at spatial infinity, cf.~\eqref{eq:asymptoticmodesatinfinitycasea},
\begin{equation}
\varphi_{\omega k}(r_*) \sim e^{i\hat{\omega}r_*} \, , \qquad r_* \to \infty \, .
\label{eq:QNoutgoingcondition}
\end{equation}
\end{enumerate}
\end{definition}

\begin{definition}
A \emph{bound state} mode is a mode solution of the field equation \eqref{eq:effectiveeq} of the form \eqref{eq:fieldansatz} with the following boundary conditions:
\begin{enumerate}
\item only ingoing modes at the horizon, cf.~\eqref{eq:asymptoticmodesatthehorizon},
\begin{equation}
\varphi_{\omega k}(r_*) \sim e^{-i\tilde{\omega}r_*} \, , \qquad r_* \to - \infty \, ;
\label{eq:BSingoingcondition}
\end{equation}
\item exponentially decreasing modes at spatial infinity, cf.~\eqref{eq:asymptoticmodesatinfinitycaseb},
\begin{equation}
\varphi_{\omega k}(r_*) \sim e^{-\hat{\varpi}r_*} \, , \qquad r_* \to \infty \, .
\label{eq:BSexponentialcondition}
\end{equation}
\end{enumerate}
\end{definition}

These boundary conditions on field modes with an $e^{-i \omega t}$ time dependence restrict the allowed frequencies $\omega$ to a discrete set of complex values. The real part of $\omega$ represents the physical frequency of the oscillation, whereas the imaginary part gives the decay (or growth) in time of the mode. This occurs because the field can escape to the black hole or to infinity. One is then faced with an eigenvalue problem in which the quasinormal or the bound state modes are the eigenmodes. By obtaining the eigenfrequencies the stability of a given mode can be inferred by the sign of the imaginary part: if the imaginary part is negative then the mode decays in time and does not create an instability.

Contrary to higher dimensional black hole spacetimes, one does not have to resort to numerical methods, since an analytical expression for the field modes \eqref{eq:exactsolutionLorentzian} is available, to which the above boundary conditions can be applied.

\begin{proposition}
The eigenfrequencies for the quasinormal and bound state modes are given as follows. 
The ``right frequencies'' $(\omega_{\pm})_n^{({\rm R})}$ are given by
\begin{equation}
(\omega_{\pm})_n^{({\rm R})} = \frac{\nu^2+3}{d^2 \delta^2 - 3(\nu^2-1)} \left\{ -d\delta \left( \frac{4kd}{\nu^2+3} + i \left( n + \frac{1}{2} \right) \right) \pm i (e - i \, {\rm sgn}(k) f)  \right\} \, ,
\label{eq:rightfrequency}
\end{equation}
where
\begin{gather}
d = \frac{1}{r_+ - r_-} \, , \qquad
\delta = 2\nu(r_+ + r_-) - 2 \sqrt{(\nu^2+3) r_+ r_-} \, , \label{eq:defddelta} \displaybreak[0]  \\
e = \sqrt{\frac{\sqrt{E^2+F^2}+E}{2}} \, , \qquad
f 	= \sqrt{\frac{\sqrt{E^2+F^2}-E}{2}} \, , \displaybreak[0] \\
E = \frac{1}{4} \left( 1 + \frac{4m^2}{\nu^2+3} \right) d^2 \delta^2 - 3(\nu^2-1) \left[ \frac{1}{4} \left( 1 + \frac{4m^2}{\nu^2+3} \right) + \left( \frac{4kd}{\nu^2+3} \right)^2 - \left( n + \frac{1}{2} \right)^2 \right] \, , \displaybreak[0] \\
F = - 3(\nu^2-1) \left( n + \frac{1}{2} \right) \frac{8kd}{\nu^2+3} \, .
\end{gather}
The ``left frequencies'' $(\omega_{\pm})_n^{({\rm L})}$ are given by
\begin{equation}
(\omega_{\pm})_n^{({\rm L})} = - i \left[ (2n+1) \nu \mp \sqrt{ 3(\nu^2-1) \left( n + \frac{1}{2} \right)^2 + \frac{\nu^2+3}{4} \left( 1 + \frac{4m^2}{\nu^2+3} \right)} \right] \, .
\label{eq:leftfrequency}
\end{equation}
In the expressions above, the ``+ solutions'' correspond to the quasinormal eigenfrequencies, whereas the ``$-$ solutions'' correspond to the bound state eigenfrequencies.  Each of the modes has two types of eigenfrequencies, which are denoted by ``right'' and ``left'' frequencies, respectively.
\end{proposition}

\begin{remark}
This follows the AdS/CFT-inspired terminology \cite{Birmingham:2001hc} and the notation here follows closely the notation of Ref.~\cite{Chen:2009rf}.
\end{remark}

\begin{proof}
We start by imposing the boundary conditions which define each of these types of mode solutions. By imposing the ingoing boundary condition \eqref{eq:QNingoingcondition} and \eqref{eq:BSingoingcondition} at the horizon, one is left with
\begin{equation}
\phi_{\omega k}(z) = A_{\omega k} \, z^{\alpha} (1-z)^{\beta} F(a,b,c;z) \, .
\end{equation}
In order to impose the boundary condition at infinity, one uses the transformation formula \eqref{eq:transformationformula} of Appendix~\ref{app:hypergeometric}, resulting in
\begin{align}
\phi_{\omega k}(z) &= A_{\omega k} \, \Gamma(c) \, z^{\alpha} \Bigg[ (1-z)^{\beta} \frac{\Gamma(c-a-b)}{\Gamma(c-a) \Gamma(c-b)} F(a,b;a+b-c+1;1-z) \notag \\
&\qquad\qquad\qquad\; + (1-z)^{\overline{\beta}} \, \frac{\Gamma(a+b-c)}{\Gamma(a) \Gamma(b)} F(c-a,c-b;c-a-b+1;1-z) \Bigg] \, .
\label{eq:QNandBSmodes}
\end{align}
Note that at infinity,
\begin{equation}
(1-z)^{\beta} \sim r^{-1/2} e^{i \hat{\omega} r_*} = r^{-1/2} e^{\hat{\varpi} r_*} \, , \qquad
(1-z)^{\overline{\beta}} \sim r^{-1/2} e^{-i \hat{\omega} r_*} = r^{-1/2} e^{-\hat{\varpi} r_*} \, ,
\label{eq:behaviouratinfinity}
\end{equation}
cf.~\eqref{eq:hatomegadef} and \eqref{eq:hatvarpidef}. The frequency $\omega$ is complex and can be written as $\omega = \omega_{\text{R}} + i \omega_{\text{I}}$. One can use the $(t,\theta) \to (-t,-\theta)$ symmetry to only consider solutions with $\omega_{\text{R}} \geq 0$ and thus $\ReC [\hat\omega] \geq 0$, $\ReC [\hat\varpi] \geq 0$. 

For a quasinormal mode, condition \eqref{eq:QNoutgoingcondition} at infinity implies that the $(1-z)^{\overline{\beta}}$ term in \eqref{eq:QNandBSmodes} must vanish. This happens when
\begin{equation}
a = -n \quad \text{or} \quad b = - n \, ,
\label{eq:anb}
\end{equation}
where $n \in \mathbb{N}_0$ is the overtone number.

On the other hand, for a bound state mode, condition \eqref{eq:BSexponentialcondition} implies that the $(1-z)^{\beta}$ term in \eqref{eq:QNandBSmodes} must vanish. This happens when
\begin{equation}
c - a = -n \quad \text{or} \quad c - b = - n \, .
\label{eq:cacbn}
\end{equation}

Since $a$, $b$, and $c$ are functions of $\omega$ [from \eqref{eq:defabcLorentzian}], these relations imply that there is a discrete set of frequencies $\{\omega_n\}$ for which the boundary conditions are satisfied. These frequencies are given by \eqref{eq:rightfrequency} and \eqref{eq:leftfrequency}, each corresponding to the two possible relations in \eqref{eq:anb} and \eqref{eq:cacbn}.
\end{proof}

\begin{remark}
The ``left'' frequencies \eqref{eq:leftfrequency} have no real part, therefore $\hat{\varpi}$ in \eqref{eq:behaviouratinfinity} is real and there are no ``left'' quasinormal modes (in the sense that they do not have the expected outgoing wavelike behavior at infinity). Only the bound state solutions are relevant for the ``left'' modes.
\end{remark}

\begin{remark}
Note that, although there is a region of the parameter space for which the imaginary part of the ``right'' quasinormal frequencies \eqref{eq:rightfrequency} is positive, it is easy to check that in this case either the mode is not outgoing at infinity or it decreases exponentially at infinity and hence is not a quasinormal mode. Otherwise, both the quasinormal and bound state frequencies have a negative imaginary part.
\end{remark}

\begin{remark}
It should also be noted that the bound state modes presented here are called quasinormal modes in some of the literature \cite{Oh:2008tc,Chen:2009rf,Chen:2009hg}. This is due to the adoption of different boundary conditions at infinity, motivated by AdS/CFT purposes \cite{Konoplya:2011qq,Berti:2009kk}. In fact, in the BTZ limit $\nu \to 1$, the bound state frequencies reduce to the quasinormal frequencies of the BTZ black hole in a rotating frame \cite{Cardoso:2001hn,Birmingham:2001hc}. This is expected since the BTZ black hole quasinormal modes must vanish at infinity.
\end{remark}

\subsection{Classical linear mode stability}

Before concluding about the stability of the spacelike stretched black hole, we describe the usual notions of ``stability'' present in the literature.

\subsubsection{Linear mode, linear and non-linear stability}
\label{sec:typestability}

At this point, it is useful to clarify what is meant by the ``linear mode stability'' of a black hole. We emphasise that the notion of stability discussed in this section is classical in nature, i.e.~no quantum effects are taken into account. A good summary of these concepts can be found in \cite{Dafermos:2010hd}.

We distinguish between three notions of classical stability:
\begin{enumerate}
\item \emph{Linear mode stability}. In this case, we study individual mode solutions of a linear field equation (or the linearisation of a non-linear field equation) and not general solutions. For a scalar field on a spacelike stretched black hole, we consider mode solutions of the form of \eqref{eq:fieldansatz} and say that a mode with finite energy is unstable if $\ImC \omega > 0$.

This analysis is possible if the spacetimes involved have enough symmetries generated by Killing vectors. In the case of scalar field mode perturbations, the mode stability of Schwarszchild was shown in \cite{Regge:1957td} and of Kerr in \cite{Whiting:1988vc}.
\item \emph{Linear stability}. Stability of individual modes does not necessarily imply stability of the superposition of infinitely many modes. In particular, stability of individual modes is not incompatible with general linear perturbations with finite initial energy growing without bound in time. Therefore, to prove linear stability, we need to find bounds for quantities involving the fields (the so-called ``energy-type quantities'', such as the energy of the field).

For the case of a scalar field perturbation on Schwarszchild, it was shown in \cite{Wald1979JMP,Kay:1987ax} that such a bound can be found for general solutions. A similar result for Kerr was only obtained very recently \cite{Dafermos:2014cua}.
\item \emph{Non-linear stability}. In the context of black hole solutions to general relativity (or other gravity theory), non-linear stability usually refers to the dynamical stability of those solutions to the Cauchy problem associated with the Einstein's equations. The only currently existent proof is the non-linear stability of Minkowski spacetime \cite{christodoulou2014global}. In the non-linear case, finding bounds for energy-type quantities is not enough to prove stability and we also need to find decay bounds, showing that those quantities are bounded by a fixed decaying function. This is the only known mechanism for non-linear stability.
\end{enumerate}

In the following, we are only concerned with linear \emph{mode} stability of the spacelike stretched black hole to massive scalar field \emph{mode} perturbations and we do not attempt to prove its full linear stability.

\vspace{2ex}
\subsubsection{Mode stability of the spacelike stretched black hole}

Having clarified what is meant by ``linear mode stability'', we now verify if the spacelike stretched black hole is stable in this sense. By analysing \eqref{eq:rightfrequency} and \eqref{eq:leftfrequency} for the quasinormal and bound state frequencies and taking into account the remarks above, one sees that both frequencies have a negative imaginary real part, and therefore these modes are classically stable. We have then proved the following theorem:

\begin{theorem} \label{thm-stability}
The spacelike stretched black hole is classically stable to massive scalar field mode perturbations.
\end{theorem}

In particular, there are \emph{no} superradiant instabilities, even though superradiant modes can exist in this spacetime, as it was shown in Theorem~\ref{thm:existenceofsuperradiance}. This is related to the fact that the effective potential $V_{\omega k}(r)$ does not have a potential well where these superradiant modes could be localized, as illustrated in the plots of Fig.~\ref{fig:potentialmass}. The absence of superradiant instabilities is the main conceptual difference in classical scalar field theory between the spacelike stretched black hole and  Kerr \cite{Cardoso:2004nk,Cardoso:2005vk,Dolan:2007mj,Pani:2012vp,Dolan:2012yt,Witek:2012tr}.

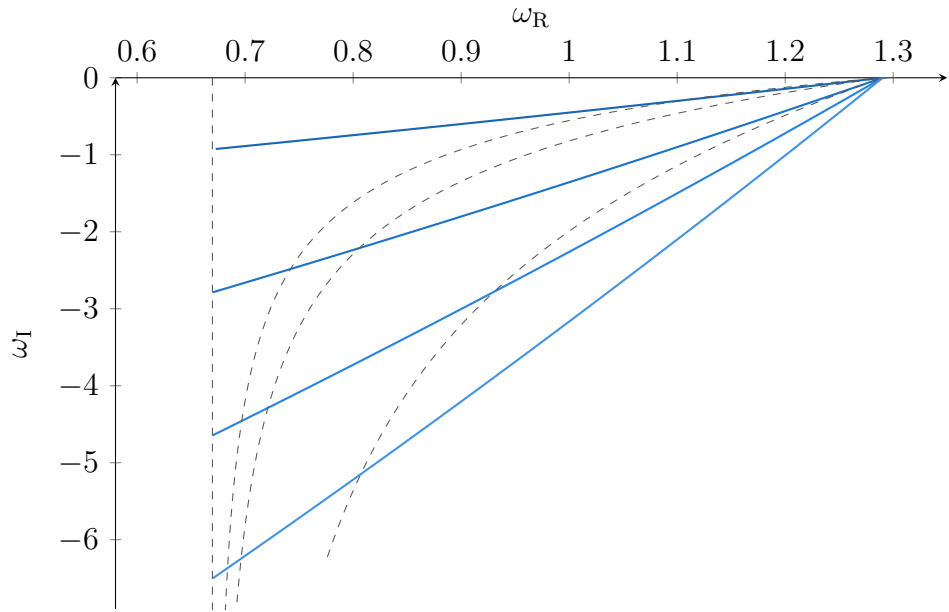
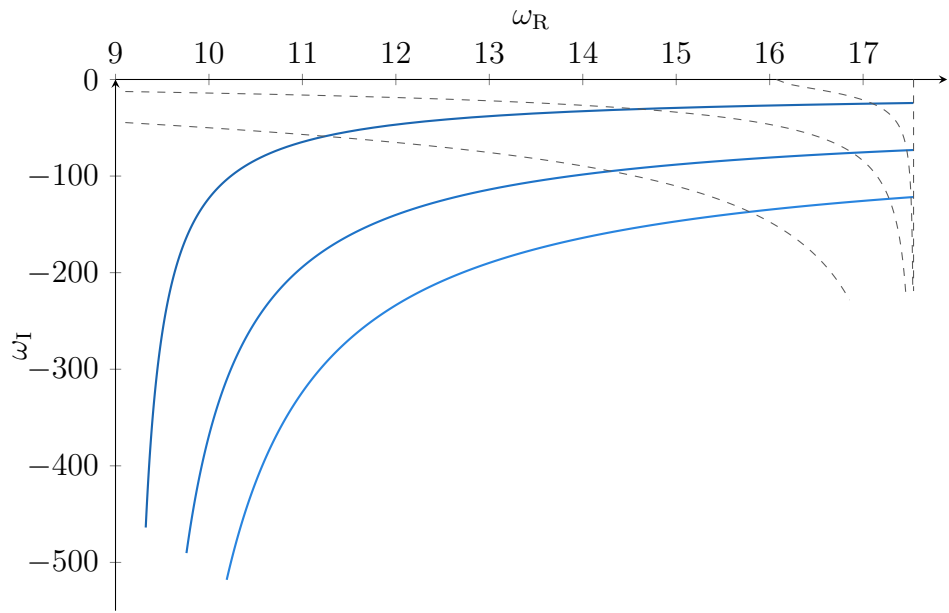
\begin{figure}[H]
\begin{center}
\subfigure[\; Quasinormal frequencies.]{
\begin{tikzpicture}
\begin{axis}[
	x post scale=1.3,
	xlabel={$\omega_{\rm R}$},
	ylabel={$\omega_{\rm I}$},
	axis x line=top,
	axis y line=left,
	xmin=0.58, xmax=1.35,
	ymin=-6.9, ymax=0,
]
%
%
\addplot[smooth, thick, no markers, blue3] 
	table {data/qnmodes-0.txt}
	[xshift=-20pt]
		node[pos=0]{{\small $n=0$}}
	;
\addplot[smooth, thick, no markers, blue4] 
	table {data/qnmodes-1.txt}
	[xshift=-20pt]
		node[pos=0]{{\small $n=1$}}
	;
\addplot[smooth, thick, no markers, blue5] 
	table {data/qnmodes-2.txt}
	[xshift=-20pt]
		node[pos=0]{{\small $n=2$}}
	;
\addplot[smooth, thick, no markers, blue6] 
	table {data/qnmodes-3.txt}
	[xshift=-20pt]
		node[pos=0]{{ \small $n=3$}}
	;
%
%
\addplot[smooth, no markers, gray!70!black, dashed] 
	table {data/qnmodes-mneg.txt}
	[xshift=30pt]
		node[pos=0.06]{{\footnotesize $m^2=-\frac{\nu^2+3}{4}$}}
	;
\addplot[smooth, no markers, gray!70!black, dashed] 
	table {data/qnmodes-m0.txt}
	[xshift=-25pt]
		node[pos=0.32]{{\footnotesize $m^2=0$}}
	;
\addplot[smooth, no markers, gray!70!black, dashed] 
	table {data/qnmodes-m1.txt}
	[xshift=23pt]
		node[pos=0.6]{{\footnotesize $m^2=1$}}
	;
\addplot[smooth, no markers, gray!70!black, dashed] 
	table {data/qnmodes-m10.txt}
	[xshift=23pt]
		node[pos=0.98]{{\footnotesize $m^2=10$}}
	;
\end{axis}
\end{tikzpicture}
}
\subfigure[\; Bound state frequencies.]{
\begin{tikzpicture}
\begin{axis}[
	x post scale=1.3,
	xlabel={$\omega_{\rm R}$},
	ylabel={$\omega_{\rm I}$},
	axis x line=top,
	axis y line=left,
	xmin=9, xmax=17.9,
	ymin=-550, ymax=0,
]
%
%
\addplot[smooth, thick, no markers, blue3] 
	table {data/bsmodes-0.txt}
	[xshift=-20pt]
		node[pos=0.18]{{\small $n=0$}}
	;
\addplot[smooth, thick, no markers, blue4] 
	table {data/bsmodes-1.txt}
	[xshift=-20pt]
		node[pos=0.27]{{\small $n=1$}}
	;
\addplot[smooth, thick, no markers, blue5] 
	table {data/bsmodes-2.txt}
	[xshift=-20pt]
		node[pos=0.3]{{\small $n=2$}}
	;
%
%
\addplot[smooth, no markers, gray!70!black, dashed] 
	table {data/bsmodes-mneg.txt}
	[xshift=0pt, yshift=1pt]
		node[pos=1,pin=below:{\footnotesize $m^2=-\frac{\nu^2+3}{4}$ \hspace*{25pt}}]{}
	;
\addplot[smooth, no markers, gray!70!black, dashed] 
	table {data/bsmodes-m0.txt}
	[xshift=-20pt]
		node[pos=0.2]{{\footnotesize $m^2=0$}}
	;
\addplot[smooth, no markers, gray!70!black, dashed] 
	table {data/bsmodes-m10.txt}
	[xshift=-50pt]
		node[pos=0.22]{{\footnotesize $m^2=10$}}
	;
\addplot[smooth, no markers, gray!70!black, dashed] 
	table {data/bsmodes-m100.txt}
	[xshift=-60pt]
		node[pos=0.25]{{\footnotesize $m^2=100$}}
	;
\end{axis}
\end{tikzpicture}
}
\caption[Eingenfrequencies in the complex plane for a spacelike stretched black hole.]{\label{fig:complexplane} Eigenfrequencies in the complex plane for a spacelike stretched black hole with $r_+ = 5$, $r_- = 0.5$, and $\nu=1.2$ and a scalar field with $k = -1$ and varying $m^2$. The different solid lines represent the eigenfrequencies for different overtone numbers $n$, and the dotted lines are lines of constant $m^2$.}
\end{center}
\end{figure}

In Fig.~\ref{fig:complexplane}, we plot the quasinormal and bound state frequencies in the complex plane for varying squared mass $m^2$. For a scalar field with a given $m^2$, there is a discrete set of complex eigenfrequencies for both the quasinormal and bound state modes at the intersection of the dotted and solid curves. It is clear the discrete nature of the allowed frequencies and the fact that their imaginary part is always negative. For the quasinormal modes the real part of the frequency increases as we consider scalar fields of larger $m^2$, while the imaginary part decreases. The bound $m^2 \geq -\frac{\nu^2+3}{4}$ is a consequence of the constraint on the effective potential at infinity, as noted in Remark~\ref{rem:omegam}. For each overtone number $n$ there is a maximum value of $m^2$ beyond which the corresponding quasinormal mode ceases to exist. This behaviour is similar to that of a massive scalar field in Schwarzschild and Kerr spacetimes \cite{Simone:1991wn}, as expected. Finally, the real and imaginary parts of the bound state mode frequencies are generally larger in absolute value than those of the quasinormal modes, but no growing modes are present.


\section{Case with a mirror-like boundary}
\label{sec:stability-mirror}

As seen above, massive scalar fields propagating in spacelike stretched black holes do not give rise to classical instabilities. In particular, there are no superradiant bound state modes, as the effective potential never develops a potential well. We now investigate whether these properties persist when a mirror-like, timelike boundary is introduced outside the event horizon, as discussed in Section~\ref{sec:qftcst-stwithboundaries}. One reason to consider this situation is that a ``mirror wall'' in the effective potential might give rise to a superradiant ``black hole bomb'' instability \cite{Press:1972zz}, as is shown to happen for a massless scalar field in Kerr \cite{Cardoso:2004nk}, even though there are no superradiant instabilities when no mirror is present. Another reason is that to treat the quantised scalar field, the existence of a speed of light surface implies that there is no well defined Hartle-Hawking state and one way to solve this problem is precisely to add a mirror between the horizon and the speed of light surface, as discussed in Section~\ref{sec:qftcst-stwithboundaries}.

Suppose then that a timelike boundary $\mathcal{M}$ is introduced at the radius $r=r_{\mathcal{M}}$, such that $r_+ < r_{\mathcal{M}} < \infty$. We impose Dirichlet boundary conditions at the boundary:
\begin{equation}
\Phi(t,z_{\mathcal{M}},\theta) = A_{\omega k} \, e^{-i\omega t + i k \theta} \, z^{\alpha}_{\mathcal{M}} (1-z_{\mathcal{M}})^{\beta} F(a,b,c;z_{\mathcal{M}}) = 0 \, ,
\end{equation}
where $z_{\mathcal{M}} = (r_{\mathcal{M}} - r_+)/(r_{\mathcal{M}} - r_-)$, cf.~\eqref{eq:definitionz1}. The bound state modes now have eigenfrequencies $\omega_{\mathcal{M}}$ determined by
\begin{equation}
z^{\alpha(\omega_{\mathcal{M}})}_{\mathcal{M}} (1-z_{\mathcal{M}})^{\beta(\omega_{\mathcal{M}})} F(a(\omega_{\mathcal{M}}),b(\omega_{\mathcal{M}}),c(\omega_{\mathcal{M}});z_{\mathcal{M}}) = 0 \, .
\end{equation}
Unfortunately, this equation cannot be analytically solved for $\omega_{\mathcal{M}}$, so these eigenfrequencies are found numerically, by truncating the hypergeometric series (see Appendix~\ref{app:hypergeometric}) to the desired accuracy and using {\sc Mathematica}'s root finding algorithm. A check on the numerics is done by considering the limit $z_{\mathcal{M}} \to 1$ ($r_{\mathcal{M}} \to + \infty$), in which $\omega_{\mathcal{M}}$ approaches the previously derived bound state frequency $\omega_-$ without the boundary. Then continuity can be used to obtain the eigenfrequencies for any value of $z_{\mathcal{M}} \in (0,1)$.

Even though no explicit expression for the frequencies is available, it is possible to obtain a useful piece of information by using the following heuristic argument. On the one hand, one can only expect superradiant instabilities if the frequencies of the bound state modes are such that $\omega_{\rm R} \lesssim \Omega_{\mathcal{H}}$, cf.~Eq.~\eqref{eq:insuperradiantcondition1}, or, in other words, if their wavelengths are $\lambda \gtrsim \Omega_{\mathcal{H}}^{-1}$. On the other hand, a mirror at $r = r_{\mathcal{M}}$ can only ``see'' these modes if $r_{\mathcal{M}} \gtrsim \lambda \gtrsim \Omega_{\mathcal{H}}^{-1}$. Therefore, superradiant instabilities, if they exist, can only occur if the mirror is placed beyond a critical radius which depends on the parameters of the spacetime. If one is not able to find any instabilities beyond this critical radius then one can assert with confidence that there are no superradiant instabilities wherever the timelike boundary is placed.

In Figs.~\ref{fig:mirrorsize}--\ref{fig:mirrorn} we present the results for the real and imaginary parts of the ``right'' eigenfrequencies $\omega_{\mathcal{M}}$ as functions of the mirror's position for selected values of the parameters. Note that only negative values of $k$ are considered in these examples, since, by \eqref{eq:rightfrequency}, $(\omega_-)_n^{({\rm R})}(k) = \overline{(\omega_-)_n^{({\rm R})}(-k)}$, and thus it suffices to consider modes with $\omega_{\rm R} \geq 0$.
\begin{figure}[H]
\begin{center}
\begin{tikzpicture}
\begin{axis}[
	width=8.5cm,
	x post scale=1.2,
	xlabel={$z_{\mathcal{M}}$},
	ylabel={$\ReC \, [\omega_{\mathcal{M}}]$},
	axis lines=left,
	xmin=0.6, xmax=1,
	ymin=0, ymax=2.9,
	max space between ticks=50pt,
]
\addplot[smooth, no markers, blue5] 
	table {data/bsmodesmirror-1-1r.txt}
	[xshift=5pt]
		node[pos=0.3,pin=100:{{\small $r_+=5, \, r_-=0.5$}}]{}
	;
\addplot[smooth, no markers, blue3] 
	table {data/bsmodesmirror-1-2r.txt}
	[xshift=5pt]
		node[pos=0.45,pin={[pin distance=20pt]120:{{\small $r_+=7, \, r_-=0.7$}}}]{}
	;
\addplot[smooth, no markers, blue1] 
	table {data/bsmodesmirror-1-3r.txt}
	[xshift=7pt]
		node[pos=0.77,pin={[pin distance=25pt]130:{{\small $r_+=10, \, r_-=1$}}}]{}
	;
\end{axis}
\end{tikzpicture}
\begin{tikzpicture}
\begin{axis}[
	width=8.5cm,
	x post scale=1.2,
	xlabel={$z_{\mathcal{M}}$},
	ylabel={$\ImC \, [\omega_{\mathcal{M}}]$},
	axis lines=left,
	xmin=0.6, xmax=1,
	ymin=-8.15, ymax=-7.8,
	max space between ticks=50pt,
]
\addplot[smooth, no markers, blue5] 
	table {data/bsmodesmirror-1-1i.txt}
	[xshift=5pt]
		node[pos=0.48,pin=100:{{\small $r_+=5, \, r_-=0.5$}}]{}
	;
\addplot[smooth, no markers, blue3] 
	table {data/bsmodesmirror-1-2i.txt}
	[xshift=6pt]
		node[pos=0.3,pin={[pin distance=35pt]95:{{\small $r_+=7, \, r_-=0.7$ }\hspace*{15pt}}}]{}
	;
\addplot[smooth, no markers, blue1] 
	table {data/bsmodesmirror-1-3i.txt}
	[xshift=10pt]
		node[pos=0.29,pin={[pin distance=45pt]95:{{\small $r_+=10, \, r_-=1$ \hspace*{20pt}}}}]{}
	;
\end{axis}
\end{tikzpicture}
\caption[``Right'' frequencies as functions of the mirror's location for selected values of $r_+$ and $r_-$.]{\label{fig:mirrorsize} ``Right'' frequencies as functions of the mirror's location for selected values of $r_+$ and $r_-$, with fixed $\nu=1.2$, $n=0$, $k=-1$, and $m=0$. The real part of the frequency is such that $\omega_{\text{R}} \to \ReC  \big[(\omega_-)_0^{({\rm R})} \big]$ in \eqref{eq:rightfrequency} as $z_{\mathcal{M}} \to 1$ and $\omega_{\text{R}} \to k \Omega_{\mathcal{H}}$ as $z_{\mathcal{M}} \to 0$. $k \Omega_{\mathcal{H}}$ equals 0.2307, 0.1648, and 0.1154 in the cases $r_+ = 5$, $r_+=7$, and $r_+=10$, respectively. The imaginary part of the frequency is such that $\omega_{\text{I}} \to \ImC  \big[(\omega_-)_0^{({\rm R})} \big]$ in \eqref{eq:rightfrequency} as $z_{\mathcal{M}} \to 1$.}
\end{center}
\end{figure}
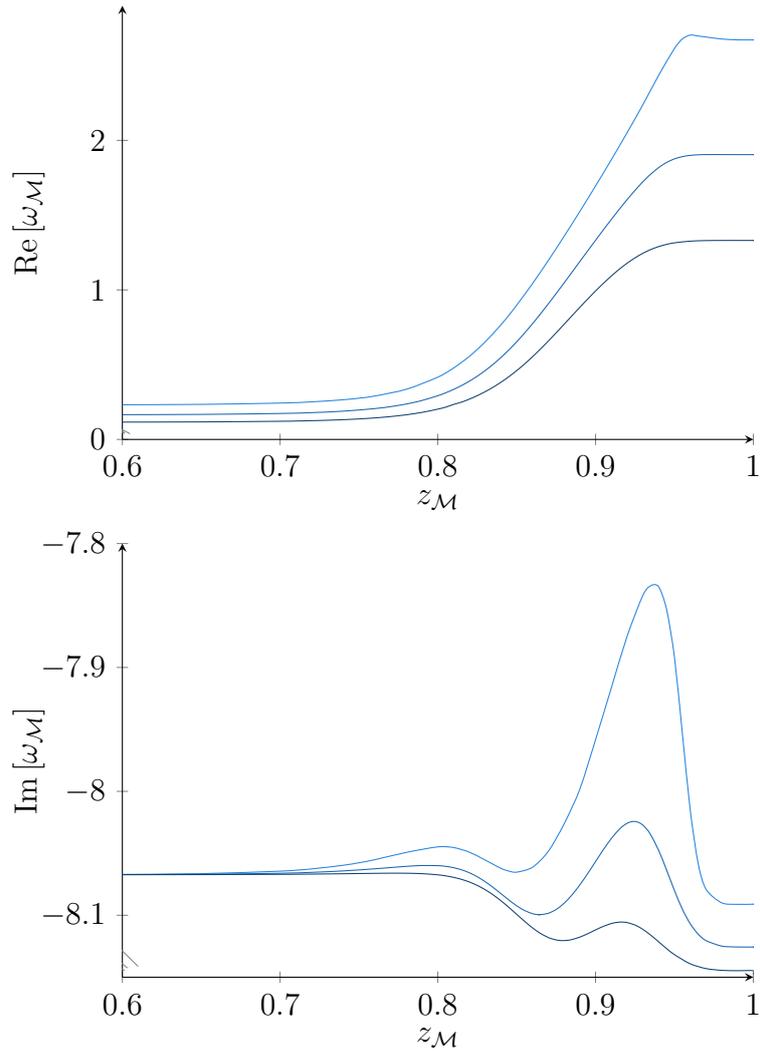
\begin{figure}[H]
\begin{center}
\begin{tikzpicture}
\begin{axis}[
	width=8.5cm,
	x post scale=1.2,
	xlabel={$z_{\mathcal{M}}$},
	ylabel={$\ReC \, [\omega_{\mathcal{M}}]$},
	axis lines=left,
	xmin=0.7, xmax=1,
	ymin=0, ymax=3.4,
	max space between ticks=50pt,
]
\addplot[smooth, no markers, blue1] 
	table {data/bsmodesmirror-1-1r.txt}
	[xshift=5pt, yshift=-2pt]
		node[pos=0.035,pin={[pin distance=10pt]90:{{\small $m=0$}}}]{}
	;
\addplot[smooth, no markers, blue3] 
	table {data/bsmodesmirror-2-1r.txt}
	[xshift=5pt]
		node[pos=0.09,pin=150:{{\small $m=1$}}]{}
	;
\addplot[smooth, no markers, blue4] 
	table {data/bsmodesmirror-2-2r.txt}
	[xshift=5pt]
		node[pos=0.28,pin=160:{{\small $m=1.5$}}]{}
	;
\addplot[smooth, no markers, blue6] 
	table {data/bsmodesmirror-2-3r.txt}
	[xshift=6pt]
		node[pos=0.77,pin={[pin distance=20pt]130:{{\small $m=2$}}}]{}
	;
\end{axis}
\end{tikzpicture}
\begin{tikzpicture}
\begin{axis}[
	width=7.1cm,
	x post scale=1.5,
	xlabel={$z_{\mathcal{M}}$},
	ylabel={$\ImC \, [\omega_{\mathcal{M}}]$},
	axis lines=left,
	xmin=0.7, xmax=1,
	ymin=-8.76, ymax=-7.81,
	max space between ticks=50pt,
]
\addplot[smooth, no markers, blue1] 
	table {data/bsmodesmirror-1-1i.txt}
	[xshift=5pt]
		node[pos=0.6,pin=170:{{\small $m=0$}}]{}
	;
\addplot[smooth, no markers, blue3] 
	table {data/bsmodesmirror-2-1i.txt}
	[xshift=3pt]
		node[pos=0.3,pin=200:{{\small $m=1$}}]{}
	;
\end{axis}
\end{tikzpicture}
\begin{tikzpicture}
\begin{axis}[
	width=7.1cm,
	x post scale=1.5,
	xlabel={$z_{\mathcal{M}}$},
	ylabel={$\ImC \, [\omega_{\mathcal{M}}]$},
	axis lines=left,
	xmin=0.7, xmax=1,
	ymin=-11.6, ymax=-9.3,
	max space between ticks=50pt,
]
\addplot[smooth, no markers, blue4] 
	table {data/bsmodesmirror-2-2i.txt}
	[xshift=5pt]
		node[pos=0.35,pin=150:{{\small $m=1.5$}}]{}
	;
\addplot[smooth, no markers, blue6] 
	table {data/bsmodesmirror-2-3i.txt}
	[xshift=0pt, yshift=5pt]
		node[pos=0.08,pin=280:{{\small $m=2$}}]{}
	;
\end{axis}
\end{tikzpicture}
\caption[``Right'' frequencies as functions of the mirror's location for selected values of $m$.]{\label{fig:mirrormass} ``Right'' frequencies as functions of the mirror's location for selected values of $m$, with fixed $r_+ = 5$, $r_- = 0.5$, $\nu=1.2$, $n=0$, and $k = -1$. The real part of the frequency is such that $\omega_{\text{R}} \to \ReC  \big[(\omega_-)_0^{({\rm R})} \big]$ in \eqref{eq:rightfrequency} as $z_{\mathcal{M}} \to 1$ and $\omega_{\text{R}} \to k \Omega_{\mathcal{H}} = 0.2307$ as $z_{\mathcal{M}} \to 0$. The imaginary part of the frequency is such that $\omega_{\text{I}} \to \ImC  \big[(\omega_-)_0^{({\rm R})} \big]$ in \eqref{eq:rightfrequency} as $z_{\mathcal{M}} \to 1$.}
\end{center}
\end{figure}
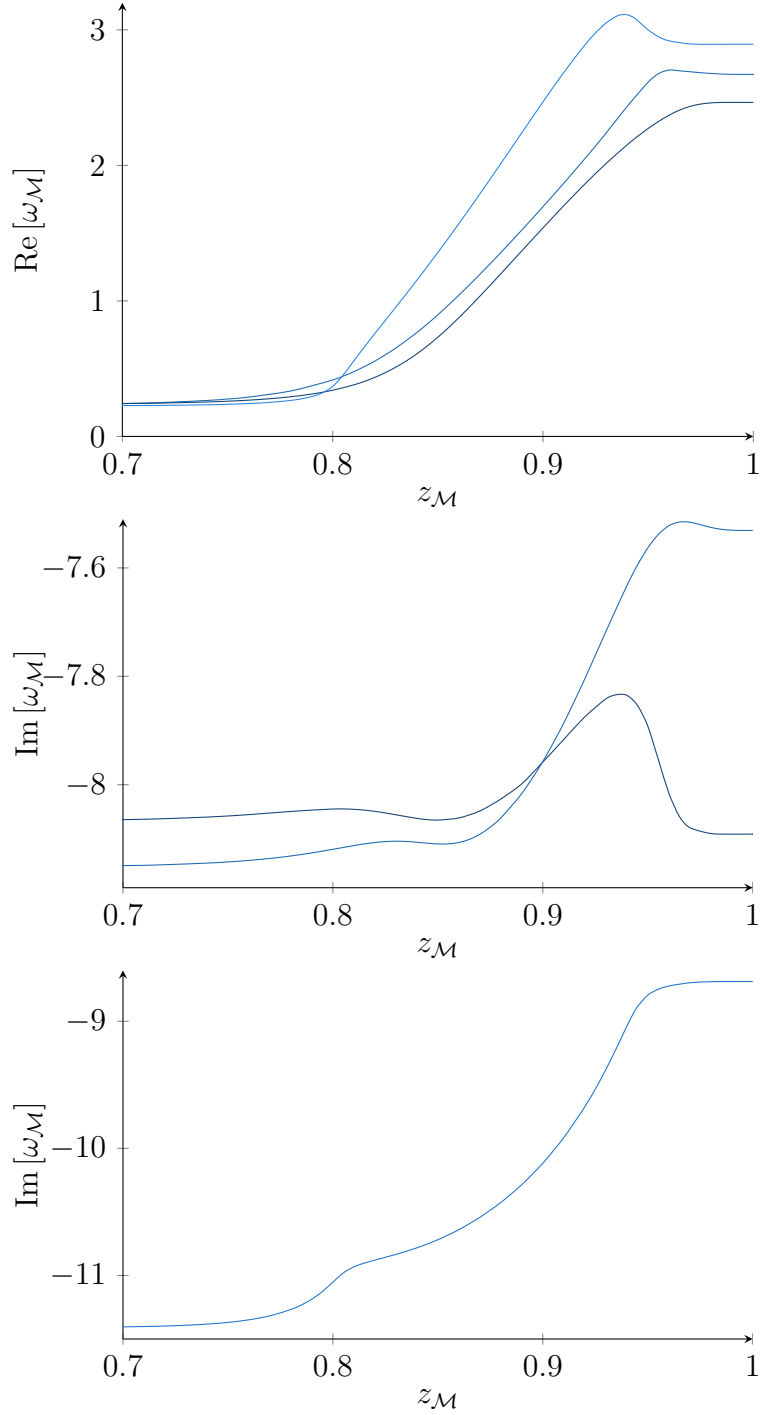
\begin{figure}[H]
\begin{center}
\begin{tikzpicture}
\begin{axis}[
	width=8.5cm,
	x post scale=1.2,
	xlabel={$z_{\mathcal{M}}$},
	ylabel={$\ReC \, [\omega_{\mathcal{M}}]$},
	axis lines=left,
	xmin=0.7, xmax=1,
	ymin=0, ymax=3.2,
	max space between ticks=50pt,
]
\addplot[smooth, no markers, blue3] 
	table {data/bsmodesmirror-1-1r.txt}
	[xshift=5pt, yshift=-2pt]
		node[pos=0.4,pin={[pin distance=35pt]150:{{\small $\nu=1.2$}}}]{}
	;
\addplot[smooth, no markers, blue1] 
	table {data/bsmodesmirror-3-1r.txt}
	[xshift=-5pt]
		node[pos=0.55,pin=300:{{\small $\nu=1.18$}}]{}
	;
\addplot[smooth, no markers, blue5] 
	table {data/bsmodesmirror-3-2r.txt}
	[xshift=5pt]
		node[pos=0.22,pin=160:{{\small $\nu=1.22$}}]{}
	;
\end{axis}
\end{tikzpicture}
\begin{tikzpicture}
\begin{axis}[
	width=7.5cm,
	x post scale=1.4,
	xlabel={$z_{\mathcal{M}}$},
	ylabel={$\ImC \, [\omega_{\mathcal{M}}]$},
	axis lines=left,
	xmin=0.7, xmax=1,
	ymin=-8.19, ymax=-7.51,
	max space between ticks=50pt,
]
\addplot[smooth, no markers, blue1] 
	table {data/bsmodesmirror-1-1i.txt}
	[yshift=-4pt]
		node[pos=0.7,pin=100:{{\small $\nu=1.2$}}]{}
	;
\addplot[smooth, no markers, blue3] 
	table {data/bsmodesmirror-3-1i.txt}
	[xshift=5pt]
		node[pos=0.3,pin=160:{{\small $\nu=1.18$}}]{}
	;
\end{axis}
\end{tikzpicture}
\begin{tikzpicture}
\begin{axis}[
	width=7.5cm,
	x post scale=1.4,
	xlabel={$z_{\mathcal{M}}$},
	ylabel={$\ImC \, [\omega_{\mathcal{M}}]$},
	axis lines=left,
	xmin=0.7, xmax=1,
	ymin=-11.5, ymax=-8.6,
	max space between ticks=50pt,
]
\addplot[smooth, no markers, blue4] 
	table {data/bsmodesmirror-3-2i.txt}
	[xshift=5pt]
		node[pos=0.45,pin=150:{{\small $\nu=1.22$}}]{}
	;
\end{axis}
\end{tikzpicture}
\caption[``Right'' frequencies as functions of the mirror's location for selected values of $\nu$.]{\label{fig:mirrornu} ``Right'' frequencies as functions of the mirror's location for selected values of $\nu$, with fixed $r_+ = 5$, $r_- = 0.5$, $n=0$, $k = -1$, and $m=0$. The real part of the frequency is such that $\omega_{\text{R}} \to \ReC  \big[(\omega_-)_0^{({\rm R})} \big]$ in \eqref{eq:rightfrequency} as $z_{\mathcal{M}} \to 1$ and $\omega_{\text{R}} \to k \Omega_{\mathcal{H}}$ as $z_{\mathcal{M}} \to 0$. $k \Omega_{\mathcal{H}}$ equals 0.2357, 0.2307, and 0.2260 in the cases $\nu=1.18$, $\nu=1.2$, and $\nu=1.22$, respectively. The imaginary part of the frequency is such that $\omega_{\text{I}} \to \ImC  \big[(\omega_-)_0^{({\rm R})} \big]$ in \eqref{eq:rightfrequency} as $z_{\mathcal{M}} \to 1$.}
\end{center}
\end{figure}
\begin{figure}[H]
\begin{center}
\begin{tikzpicture}
\begin{axis}[
	width=8.5cm,
	x post scale=1.2,
	xlabel={$z_{\mathcal{M}}$},
	ylabel={$\ReC \, [\omega_{\mathcal{M}}]$},
	axis lines=left,
	xmin=0.6, xmax=1,
	ymin=0, ymax=5.4,
	max space between ticks=50pt,
]
\addplot[smooth, no markers, blue5] 
	table {data/bsmodesmirror-1-1r.txt}
	[xshift=-7pt]
		node[pos=0.4,pin=300:{{\small $k=-1$}}]{}
	;
\addplot[smooth, no markers, blue3] 
	table {data/bsmodesmirror-4-1r.txt}
	[xshift=6pt]
		node[pos=0.45,pin=120:{{\small $k=-2$}}]{}
	;
\end{axis}
\end{tikzpicture}
\begin{tikzpicture}
\begin{axis}[
	width=8.5cm,
	x post scale=1.2,
	xlabel={$z_{\mathcal{M}}$},
	ylabel={$\ImC \, [\omega_{\mathcal{M}}]$},
	axis lines=left,
	xmin=0.6, xmax=1,
	ymin=-11.2, ymax=-7.6,
	max space between ticks=50pt,
]
\addplot[smooth, no markers, blue5] 
	table {data/bsmodesmirror-1-1i.txt}
	[yshift=3pt]
		node[pos=0.8,pin=300:{{\small $k=-1$}}]{}
	;
\addplot[smooth, no markers, blue3] 
	table {data/bsmodesmirror-4-1i.txt}
	[xshift=7pt]
		node[pos=0.6,pin=95:{{\small $k=-2$ }\hspace*{15pt}}]{}
	;
\end{axis}
\end{tikzpicture}
\caption[``Right'' frequencies as functions of the mirror's location for selected values of $k$.]{\label{fig:mirrork} ``Right'' frequencies as functions of the mirror's location for selected values of $k$, with fixed $r_+ = 5$, $r_- = 0.5$, $\nu=1.2$, $n=0$, and $m=0$. The real part of the frequency is such that $\omega_{\text{R}} \to \ReC  \big[(\omega_-)_0^{({\rm R})} \big]$ in \eqref{eq:rightfrequency} as $z_{\mathcal{M}} \to 1$ and $\omega_{\text{R}} \to k \Omega_{\mathcal{H}}$ as $z_{\mathcal{M}} \to 0$. $k \Omega_{\mathcal{H}}$ equals 0.2307 and 0.4615 in the cases $k=-1$ and $k=-2$, respectively. The imaginary part of the frequency is such that $\omega_{\text{I}} \to \ImC  \big[(\omega_-)_0^{({\rm R})} \big]$ in \eqref{eq:rightfrequency} as $z_{\mathcal{M}} \to 1$.}
\end{center}
\end{figure}
\begin{figure}[H]
\begin{center}
\begin{tikzpicture}
\begin{axis}[
	width=8.5cm,
	x post scale=1.2,
	xlabel={$z_{\mathcal{M}}$},
	ylabel={$\ReC \, [\omega_{\mathcal{M}}]$},
	axis lines=left,
	xmin=0.6, xmax=1,
	ymin=0, ymax=3.2,
	max space between ticks=50pt,
]
\addplot[smooth, no markers, blue5] 
	table {data/bsmodesmirror-1-1r.txt}
	[xshift=-6pt]
		node[pos=0.35,pin=320:{{\small $n=0$}}]{}
	;
\addplot[smooth, no markers, blue3] 
	table {data/bsmodesmirror-5-1r.txt}
	[xshift=5pt]
		node[pos=0.45,pin=160:{{\small $n=1$}}]{}
	;
\end{axis}
\end{tikzpicture}
\begin{tikzpicture}
\begin{axis}[
	width=8.5cm,
	x post scale=1.2,
	xlabel={$z_{\mathcal{M}}$},
	ylabel={$\ImC \, [\omega_{\mathcal{M}}]$},
	axis lines=left,
	xmin=0.6, xmax=1,
	ymin=-23.39, ymax=-22.21,
	max space between ticks=50pt,
]
\addplot[smooth, no markers, blue3] 
	table {data/bsmodesmirror-5-1i.txt}
	[xshift=5pt]
		node[pos=0.3,pin=120:{{\small $n=1$}}]{}
	;
\end{axis}
\end{tikzpicture}
\caption[``Right'' frequencies as functions of the mirror's location for selected values of $n$.]{\label{fig:mirrorn} ``Right'' frequencies as functions of the mirror's location for selected values of $n$, with fixed $r_+ = 5$, $r_- = 0.5$, $\nu=1.2$, $k=-1$, and $m=0$. The real part of the frequency is such that $\omega_{\text{R}} \to \ReC  \big[(\omega_-)_n^{({\rm R})} \big]$ in \eqref{eq:rightfrequency} as $z_{\mathcal{M}} \to 1$ and $\omega_{\text{R}} \to k \Omega_{\mathcal{H}} = 0.2307$ as $z_{\mathcal{M}} \to 0$. The imaginary part of the frequency is such that $\omega_{\text{I}} \to \ImC  \big[(\omega_-)_n^{({\rm R})} \big]$ in \eqref{eq:rightfrequency} as $z_{\mathcal{M}} \to 1$.}
\end{center}
\end{figure}
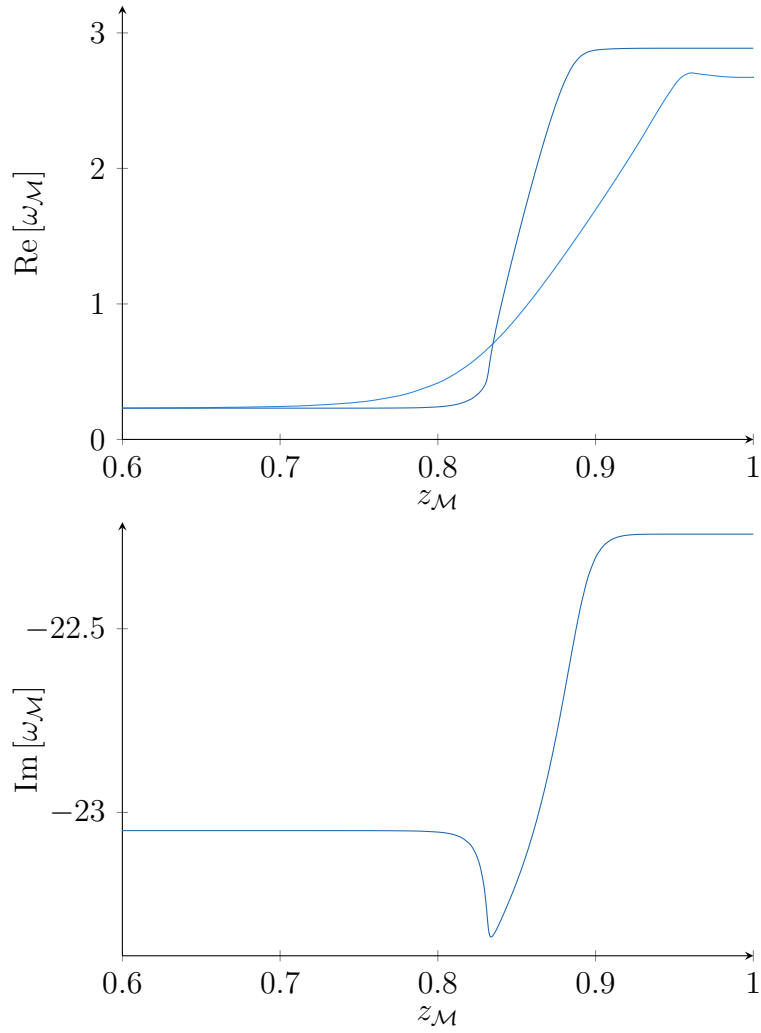
%

\newpage

First, note that the imaginary part of the eigenfrequencies is negative in all the presented cases, and therefore no superradiant instabilities are present. This again contrasts with the Kerr spacetime surrounded by a mirror, where even a massless scalar field has superradiant instabilities \cite{Cardoso:2004nk}. 

In regard to the size of the black hole, from Fig.~\ref{fig:mirrorsize} one observes that the real part of the frequency generally decreases as the horizon grows, while the imaginary part of the frequency increases in absolute value. The dependence on the scalar field mass as shown in Fig.~\ref{fig:mirrormass} is more complicated, but it is clear that the imaginary part of the frequency also increases in absolute value as the field mass increases. A similar conclusion can be drawn from Figs.~\ref{fig:mirrornu}--\ref{fig:mirrorn} concerning the warp factor $\nu$, the angular momentum number $k$ (in absolute value), and the overtone number $n$.

In order to understand these results, it is useful to keep in mind the effective potential picture described in section~\ref{eq:stability-asymptoticsolutions}. Note that in the current situation the frequencies take imaginary values, and hence this picture is not entirely accurate. Recall that the ``right'' bound state frequencies \eqref{eq:rightfrequency} have a real part that always exceeds $k \Omega_{\mathcal{H}}$ and, therefore, there are no superradiant bound state modes when the mirror is placed far from the event horizon. As seen previously, this can be explained by the fact that the effective potential does not develop a potential well near the horizon where the field mode could be trapped. However, as we move the mirror closer to the horizon, it is possible that a potential well can be artificially created, since the mirror works as an infinite potential wall. If we place the mirror close to the horizon, the real part of the frequency is approximately $k \Omega_{\mathcal{H}}$, due to the dragging of the inertial frames. In the general case in which the mirror is somewhere in between the horizon and infinity, we expect the real part of the frequency to be greater than $k \Omega_{\mathcal{H}}$ but smaller than the asymptotic value, with possibly an increasing profile as the mirror is moved towards infinity. This expectation is in good agreement with the numerical results. We thus conclude that the real part of the ``right'' frequency does not satisfy the superradiant condition irrespective of the mirror's position in the exterior region. 

One may ask why the mirror does not create an artificial potential well. The well might have been expected to arise in cases where the effective potential has a local maximum near the horizon and the mirror is placed close to the horizon. We find, however, that when the mirror approaches the maximum of the effective potential from the right, the real part of the frequency does not decrease quickly enough to create superradiant bound state modes. When the mirror is moved even closer to the horizon, the real part of the frequency has no other choice but to approach $k \Omega_{\mathcal{H}}$.

The dependence of the imaginary part of the frequency (and consequently the decay rate) on the several parameters of the system can be interpreted in the same way. If one increases the absolute values of $m^2$, $\nu$, $k$, and $n$, the effective potential is changed in such a way that the local maximum tends to disappear (as in Fig.~\ref{fig:potentialmass}) and, as a result, the field is more stable. On the other hand, the effective potential itself depends on the frequency of the field, and in this case the previous behaviour roughly occurs if we decrease the real part of the frequency. 

\begin{figure}[ht!]
\begin{center}
\begin{tikzpicture}
\begin{axis}[
	width=8.5cm,
	x post scale=1.2,
	xlabel={$z_{\mathcal{M}}$},
	ylabel={$\ReC \, [\omega_{\mathcal{M}}]$},
	axis lines=left,
	xmin=0, xmax=1,
	ymin=0, ymax=0.25,
	max space between ticks=50pt,
]
\addplot[smooth, no markers, blue3] 
	table {data/bsmodesmirror-6-1r.txt}
	;
\end{axis}
\end{tikzpicture}
\begin{tikzpicture}
\begin{axis}[
	width=8.5cm,
	x post scale=1.2,
	xlabel={$z_{\mathcal{M}}$},
	ylabel={$\ImC \, [\omega_{\mathcal{M}}]$},
	axis lines=left,
	xmin=0, xmax=1,
	ymin=-2.45, ymax=-1.05,
	max space between ticks=50pt,
]
\addplot[smooth, no markers, blue3] 
	table {data/bsmodesmirror-6-1i.txt}
	;
\end{axis}
\end{tikzpicture}
\caption[``Left'' frequencies as functions of the mirror's location.]{\label{fig:mirrorleft} ``Left'' frequencies as functions of the mirror's location for $r_+ = 5$, $r_- = 0.5$, $\nu=1.2$, $n=0$, $k=-1$ and $m=0$. The real part of the frequency is such that $\omega_{\text{R}} \to \ReC  \big[(\omega_-)_0^{({\rm L})} \big] = 0$ as $z_{\mathcal{M}} \to 1$ and $\omega_{\text{R}} \to k \Omega_{\mathcal{H}} = 0.2307$ as $z_{\mathcal{M}} \to 0$. The imaginary part of the frequency is such that $\omega_{\text{I}} \to \ImC  \big[(\omega_-)_0^{({\rm L})} \big]$ in \eqref{eq:rightfrequency} as $z_{\mathcal{M}} \to 1$.}
\end{center}
\end{figure}
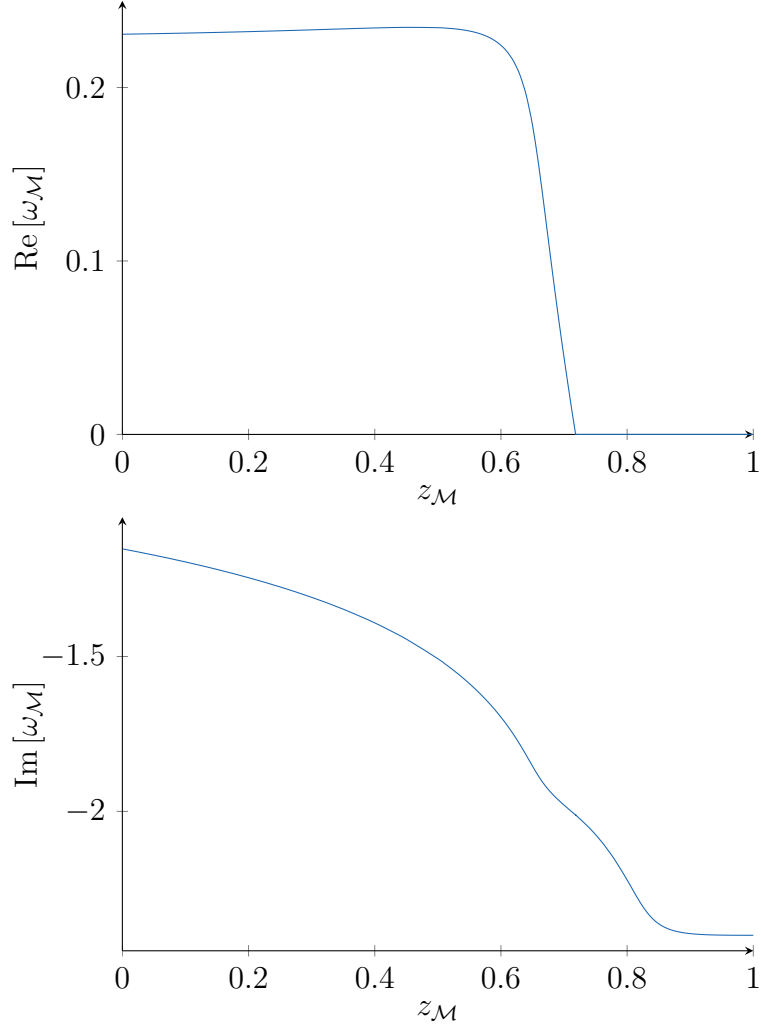

The analysis for the ``left'' frequencies \eqref{eq:leftfrequency} has some similarities but also some significant differences. The dependence of the frequencies on the parameters of the system is largely identical, but not on the mirror's position. As seen in section~\ref{sec:stability-QNandBS}, without the mirror the real part of the frequency is zero for the ``left'' bound state frequencies, and so there is no superradiance. As seen in Fig.~\ref{fig:mirrorleft}, when the mirror is brought in from infinity, this situation persists until a critical radius (call it $r_1$) beyond which the real part of the frequency sharply increases up to a value which is slightly greater than $k \Omega_{\mathcal{H}}$ (denote by $r_2$ the radius at which $\omega_{\text{R}} = k \Omega_{\mathcal{H}}$, such that $r_+ < r_2 < r_1$). When the mirror is placed at $r_{\mathcal{M}} \in (r_2,r_1)$ the bound state mode is indeed superradiant, but the imaginary part of the eigenfrequency is still negative. Again, this can be understood by analyzing the effective potential, which we recall depends on the frequency. The somewhat narrow interval of the mirror's position in which the real part of the frequency satisfies the superradiant condition is already past the local maximum of the effective potential, when it exists. Therefore, no potential well is created and thus no instabilities are present.


\section{Conclusions on the classical stability}
\label{sec:stability-conclusions}

In this chapter, we have investigated the classical linear mode stability of a massive scalar field on the background of a warped AdS${}_3$ black hole. The first main result, Theorem~\ref{thm:existenceofsuperradiance}, is that classical superradiance is present when physically motivated boundary conditions are imposed at infinity; the second main result, cf.~Theorem~\ref{thm-stability} and numerical results of section~\ref{sec:stability-mirror}, is that the black hole is nevertheless classically stable against the scalar mode perturbations, both with and without a stationary mirror in the exterior region. 

Taken together, these results are surprising at a first glance, as one might have expected the superradiant modes to create instabilities as in the (3+1)-dimensional Kerr spacetime. It was shown, however, that instabilities are not present, because the effective potential never develops a potential well near the horizon where the superradiant modes could be trapped. This stability might be a general characteristic of (2+1)-dimensional spacetimes, and it is a particularly interesting result as almost all of the research to date on the classical stability of black holes has addressed spacetimes in four or more dimensions. Compare, for instance, with the results from Ref.~\cite{Cardoso:2005vk}, in which it was shown that black branes of the type Kerr${}_d \times \mathbb{R}^p$ (where $p \in \mathbb{N}$ and Kerr${}_d$ is the Kerr black hole if $d=4$ or the Myers-Perry black hole \cite{Myers:1986un} if $d>4$) have superradiant instabilities if $d=4$ but not if $d > 4$.

Additionally, the analysis in section~\ref{sec:classical-superradiance} helps to clarify the role of boundary conditions in classical superradiance. The in and up modes described in section~\ref{sec:stability-basismodes} can be superradiant, whatever the choice of positive frequency, similarly to what happens in the Kerr spacetime. This differs from the situation in the BTZ and Kerr-AdS spacetimes, where superradiance is not present if reflective boundary conditions are chosen at infinity, which are motivated by their asymptotic structure. 

Having addressed the mode stability and the existence of classical superradiance for the spacelike stretched black hole, in the following chapter, we return to the quantum theory and address the computation of the renormalised vacuum polarisation, as explained in general in Chapter~\ref{chap:localobservables}.


\chapter{Computation of \texorpdfstring{$\mathbf{\langle \Phi^2(\boldsymbol{x}) \rangle}$}{<Phi2>} on a \texorpdfstring{WA\MakeLowercase{d}S$\mathbf{{}_3}$}{WAdS3} black hole}
\chaptermark{Computation of $\langle \Phi^2(\MakeLowercase{x}) \rangle$ on a WAdS${}_3$ black hole}
\label{chap:computation}

In this chapter, we apply the method described in Chapter~\ref{chap:localobservables} to compute the renormalised vacuum polarisation of a massive scalar field in the Hartle-Hawking state on a spacelike stretched black hole surrounded by a Dirichlet mirror. We present numerical results which demonstrate the numerical efficacy of the method.

This chapter is mostly based on \cite{Ferreira:2014ina}.

\section{Vacuum polarisation on a WAdS$\mathbf{{}_3}$ black hole}
\label{sec:computation-vacuumpolarisation}

\begin{figure}[t!]
\begin{center}
{\small
\def\svgwidth{0.4\textwidth}
\input{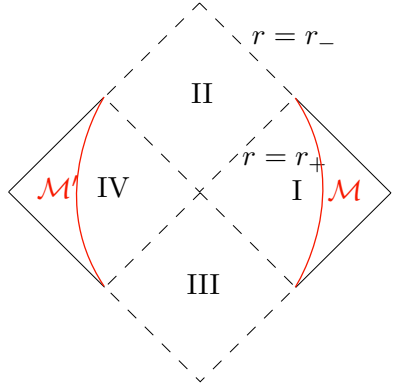}
}
\caption{\label{fig:CPdiagramWAdSBH} Carter-Penrose diagram of a non-extremal spacelike stretched black hole surround by mirrors.}
\end{center}
\end{figure}

As described in Section~\ref{section:qftcst-rotatingbhs}, in order to have a well defined, regular, isometry-invariant vacuum state (the Hartle-Hawking state) on a spacelike stretched black hole, we introduce a boundary $\mathcal{M}$ at a fixed radial coordinate $r = r_{\mathcal{M}}$ in region I and a similar boundary $\mathcal{M}'$ in region IV on which we impose Dirichlet boundary conditions (see Fig.~\ref{fig:CPdiagramWAdSBH}). We require that $r_{\mathcal{M}} \in (r_+, r_{\mathcal{C}})$, where $r=r_{\mathcal{C}}$ is the radial location of the speed of light surface, given by \eqref{eq:solsurface}. As before, we denote by $\widetilde{\text{I}}$ the portion of region I from the horizon up to the boundary and by the region $\widetilde{\text{IV}}$ the portion of region IV from the horizon up to the boundary.

The Hartle-Hawking state is defined in the union of regions 
$\widetilde{\text{I}}$, II, III and $\widetilde{\text{IV}}$, which we take to be the manifold $M$ of interest. Nevertheless, we are interested in computing the renormalised vacuum polarisation in region $\widetilde{\text{I}}$, as described in Chapter~\ref{chap:localobservables}. In this region, there is a timelike Killing vector field  and we introduce the co-rotating coordinate system 
$(\tilde{t}=t, \, r, \,\tilde{\theta} = \theta - \Omega_{\mathcal{H}} t)$, with $\Omega_{\mathcal{H}}$ given by \eqref{eq:omegaH}, such that the timelike Killing vector field is $\chi = \partial_{\tilde{t}}$ and the metric is given by
\begin{equation}
\dd s^2 = - N(r)^2 \, \dd\tilde{t}^2 + \frac{\dd r^2}{4 R(r)^2 N(r)^2} + R(r)^2 \left( \dd\tilde{\theta} + \big( N^{\theta}(r) + \Omega_{\mathcal{H}} \big) \dd\tilde{t} \right)^2 \, .
\end{equation}

In the following, we go through the steps of the method described in Chapter~\ref{chap:localobservables} to renormalise the vacuum polarisation for a massive scalar field.

\subsection{Scalar field and the Hartle-Hawking state}

We consider again a real massive scalar field $\Phi$ on a spacelike stretched black hole which satisfies the Klein-Gordon equation \eqref{eq:fieldequation2}
\begin{equation}
\left(\nabla^2 - m^2 \right) \Phi = 0 \, ,
\end{equation}
and mode solutions of the form
\begin{equation}
\Phi_{\tilde{\omega} k} (\tilde{t}, r, \tilde{\theta}) = e^{-i \tilde{\omega} \tilde{t} + i k \tilde{\theta}} \phi_{\tilde{\omega} k} (r) \, ,
\end{equation}
where $\tilde{\omega} \in \mathbb{R}$ and $k \in \mathbb{Z}$, cf.~Eq~\eqref{eq:modesolutions}.

We now repeat the construction of the L and R modes as in Section~\ref{sec:HHstateconstruction}, which are modes defined everywhere in $M$ and, from Proposition~\ref{prop:LRmodespositivefreq}, of positive frequency with respect to the affine parameters of the horizons. We take the one-particle Hilbert space $\mathscr{H}$ to consist of the L and R mode solutions and define the Hartle-Hawking state $|H \rangle$ as the vacuum state of $\mathscr{F}_{\rm s}(\mathscr{H})$, the Fock space associated with $\mathscr{H}$. The Feynman propagator $G^{\rm F}$ evaluated for the Hartle-Hawking state is then defined as in \eqref{eq:Feynmanpropdef}.

\subsection{Complex Riemannian section}

At this stage, as in Chapter~\ref{chap:localobservables}, we are faced with the challenge of explicitly computing the Feynman propagator as a sum over mode solutions of \eqref{eq:GFdiffeq}. For that, we consider the complex Riemannian section of region $\widetilde{\text{I}}$ of the black hole.

The complex Riemannian section of a stationary spacetime was defined in Definition~\ref{eq:complexRiemanniansection} and its metric for a (2+1)-dimensional rotating black hole was given by \eqref{eq:metriccomplexRiemanniansectiongen}. For the spacelike stretched black hole, the metric is
\begin{equation}
\dd s^2 = N(r)^2 \dd\tau^2 + \frac{\dd r^2}{4 R(r)^2 N(r)^2} + R(r)^2 \left( \dd\tilde{\theta} - i \, \big( N^{\theta}(r) + \Omega_{\mathcal{H}} \big) \dd\tau \right)^2 \, .
\label{eq:metriccomplexRiemanniansection}
\end{equation}
where we performed  a Wick rotation $\tilde{t} = -i \tau$, with $\tau \in \mathbb{R}$. We denote the complex Riemannian section of region I of the black hole by $I^{\mathbb{C}}$. This metric is regular at the horizon if $\tau$ is periodic with period $2\pi/\kappa_+$, where $\kappa_+$ is the surface gravity,
\begin{equation}
\kappa_+ = \frac{(\nu^2+3)(r_+-r_-)}{2\left(2\nu r_+ - \sqrt{(\nu^2+3)r_+ r_-}\right)}
= \frac{|\Omega_{\mathcal{H}}|}{4} (\nu^2+3)(r_+-r_-) \, .
\end{equation}

\subsection{Green's distribution}

Next, we find the Green's distribution $G$ associated with the Klein-Gordon equation in the complex Riemannian section, which satisfies the distributional equation \eqref{eq:GCfunctioneq},
\begin{equation}
\left( \nabla^2 - m^2 \right) G(x,x') = - \frac{\delta^3(x,x')}{\sqrt{g(x)}} = - 2 \delta(\tau-\tau') \delta(r-r') \delta(\tilde{\theta}-\tilde{\theta}') \, .
\end{equation}
As noted before, in the complex Riemannian section there is a unique solution to this equation which is regular at the horizon and which satisfies the Dirichlet boundary conditions at the boundary. 

Given the periodicity conditions of $\tau$ and $\tilde{\theta}$, one has
\begin{align}
\delta (\tau - \tau') &= \frac{\kappa_+}{2\pi} \sum_{n = -\infty}^{\infty} e^{i \kappa_+ n (\tau - \tau')} \, ,
\label{eq:deltatau} \\
\delta (\tilde{\theta} - \tilde{\theta}') &= \frac{1}{2\pi} \sum_{k = -\infty}^{\infty} e^{i k(\tilde{\theta} - \tilde{\theta}')} \, ,
\label{eq:deltatheta}
\end{align}
and, if we expand $G(x,x')$ as
\begin{equation}
G(x,x') = \frac{\kappa_+}{4\pi^2} \sum_{n = -\infty}^{\infty} e^{i \kappa_+ n (\tau - \tau')} \sum_{k = -\infty}^{\infty} e^{i k(\tilde{\theta} - \tilde{\theta}')} \, G_{nk}(r,r')
\label{eq:Greenfunction0}
\end{equation}
and use \eqref{eq:deltatau} and \eqref{eq:deltatheta}, one obtains a differential equation for $G_{nk}$,
\begin{multline}
\frac{\dd}{\dd r} \left( 4 R(r)^2 N(r)^2 \frac{\dd G_{nk}(r)}{\dd r} \right) - \frac{1}{R(r)^2 N(r)^2} \bigg[ R(r)^2 \left( \kappa_+ n + i k (N^{\theta}(r) + \Omega_{\mathcal{H}}) \right)^2 \\ + N(r)^2 \left(k^2 + m^2 R(r)^2 \right) \bigg] G_{nk}(r) = - 2 \delta(r-r') \, ,
\label{eq:KGBH3}
\end{multline}
cf.~\eqref{eq:radialGreenseq}. The solutions of this equation may be given as a product of solutions of the corresponding homogeneous equation. Using Appendix~\ref{app:hypergeometric}, a pair of independent solutions of the homogeneous equation is
\begin{align}
\phi^1_{n k} (z) &= z^{\alpha} (1-z)^{\beta} F(a,b;c;z) \, , \\
\phi^2_{n k} (z) &= z^{\alpha} (1-z)^{\beta} F(a,b;a+b+1-c;1-z) \, ,
\label{eq:exactsolution}
\end{align}
where we introduce the radial coordinate
\begin{equation}
z = \frac{r-r_+}{r-r_-} \, ,
\label{eq:definitionz}
\end{equation}
and where the parameters of the hypergeometric functions are given by
\begin{align}
a = \alpha + \beta + \gamma \, , \qquad
b = \alpha + \beta - \gamma \, , \qquad
c = 2\alpha + 1 \, ,
\label{eq:defabc}
\end{align}
with
\begin{subequations}
\begin{align}
\alpha &= \frac{|n|}{2} \, , \\
\beta &= \frac{1}{2} + \frac{\sqrt{3(\nu^2-1)}}{\nu^2+3} \sqrt{\frac{(\nu^2+3)^2}{12(\nu^2-1)} \left( 1 + \frac{4m^2}{\nu^2+3} \right) + (\kappa_+ n + i k \Omega_{\mathcal{H}})^2} \, , \\ 
\gamma & = \frac{2\nu r_- - \sqrt{r_+ r_-(\nu^2+3)}}{(\nu^2+3)(r_+ - r_-)} \sqrt{\left[\kappa_+ n + i k \left(N^{\theta}(r_-) + \Omega_{\mathcal{H}} \right) \right]^2} \, .
\end{align}
\label{eq:alphabetagamma}
\end{subequations}
Our convention for the the branch of the square roots  in \eqref{eq:alphabetagamma} is the one with non-negative real part.

Taking into account the boundary conditions, the regular solution at event horizon, $z=0$, is
\begin{equation}
p_{nk}(z) = \phi^1_{nk}(z) \, ,
\label{eq:pnk}
\end{equation}
whereas the Dirichlet solution near the mirror, $z=z_M$, is given by
\begin{equation}
q_{nk}(z) = \phi^2_{nk}(z) - \frac{\phi^2_{nk}(z_M)}{\phi^1_{nk}(z_M)} \phi^1_{nk}(z) \, .
\label{eq:qnk}
\end{equation}
The radial part of the Green's function, as in \eqref{eq:Gnkgeneral}, is then
\begin{equation}
G_{nk}(z,z') = C_{nk} \, p_{nk}(z_<) \, q_{nk}(z_>) \, ,
\label{eq:Gnk}
\end{equation}
where $z_< := \min \{ z, z' \}$, $z_> := \max \{ z, z' \}$ and $C_{nk}$ is the normalization constant determined by \eqref{eq:Cnkdefinition}. For convenience, we rewrite \eqref{eq:Greenfunction0} as
\begin{equation}
G(x,x') = \frac{|\Omega_{\mathcal{H}}|}{8\pi^2} \sum_{n = -\infty}^{\infty} e^{i \kappa_+ n (\tau - \tau')} \sum_{k = -\infty}^{\infty} e^{i k(\tilde{\theta} - \tilde{\theta}')} \, G_{nk}(r,r') \, .
\label{eq:Greenfunction1}
\end{equation}
Then, $C_{nk}$ is given by
\begin{equation}
C_{nk} = \frac{\Gamma(a) \Gamma(b)}{|n|! \, \Gamma(a+b-|n|)} \, . 
\label{eq:Cnk}
\end{equation}

\subsection{Hadamard renormalisation}

Having computed the Green's distribution $G$, it now remains to follow the Hadamard renormalisation procedure described in Section~\ref{sec:renormalisation-procedure} to subtract its short-distance divergences and obtain the renormalised vacuum polarisation,
\begin{equation} \label{eq:phi2wads}
\langle \Phi^2(x) \rangle = \lim_{x' \to x} \left[ G(x,x') - G_{\text{Had}}(x,x') \right] \, .
\end{equation}
where $G_{\text{Had}}$ is the Hadamard singular part,
\begin{equation}
G_{\text{Had}}(x,x') = \frac{1}{4\sqrt{2}\pi} \frac{1}{\sqrt{\sigma(x,x')}} + \mathcal{O}(\sigma^{1/2}) \, .
\end{equation}

To do that, we rewrite $G_{\text{Had}}(x,x')$ as a sum over mode solutions for the complex Riemannian section of the Minkowski spacetime, plus a term which is finite when the coincidence limit is taken, as given by \eqref{eq:GnkMink} and explained in detail in the text after Eq.~\eqref{eq:GnkMink} and in Appendix~\ref{app:Minkowski}.

For concreteness, assume that the points $x$ and $x'$ are angularly separated, i.e.~assume that the black hole metric is given in coordinates $(\tau,r,\tilde{\theta})$, whereas the Minkowski metric is given in coordinates $(\tau,\rho,\tilde{\theta})$, and let $x = (\tau, r, 0)$ and $x' = (\tau, r, \tilde{\theta})$, with $\tilde{\theta} > 0$, for the black hole case, and similarly for the Minkowski case.

The Synge's world function for the black hole is
\begin{equation}
\sigma (x,x') = \frac{1}{2} R(r)^2 \tilde{\theta}^2  + \mathcal{O}(\tilde{\theta}^3) \, .
\end{equation}
and, hence, the Hadamard singular part of the Green's distribution is
\begin{equation} \label{eq:HadamardsingBH}
G_{\text{Had}}(x,x') = \frac{1}{4\pi} \frac{1}{R(r) \tilde{\theta}} + \mathcal{O}(\tilde{\theta}) \, .
\end{equation}

For the Minkowski case, the Hadamard singular part is given by
\begin{equation} \label{eq:HadamardsingM}
G_{\text{Had}}^{\mathbb{M}}(x,x') = \frac{1}{4\pi} \frac{1}{\rho \tilde{\theta}} + \mathcal{O}(\tilde{\theta}) \, ,
\end{equation}
which, using the notation and definitions of Appendix~\ref{app:Minkowski}, can also be expressed as the mode sum,
\begin{equation}
G_{\text{Had}}^{\mathbb{M}}(x,x') = \frac{T_{\mathbb{M}}}{2\pi} \sum_{k=-\infty}^{\infty} \left( e^{i k \tilde{\theta}} \sum_{n=-\infty}^{\infty} \, \hat{G}_{n k}^{\mathbb{M}} (\rho,\rho') \right) - \hat{G}_{\text{reg}}^{\mathbb{M}}(x,x') + \mathcal{O}(\tilde{\theta}) \, ,
\end{equation}

At this stage, we set the leading terms of Hadamard singular parts \eqref{eq:HadamardsingBH} and \eqref{eq:HadamardsingM} to be equal up to a function $\gamma(r)>0$ by identifying the two radial coordinates as in \eqref{eq:gthetathetaidentification}, i.e.
\begin{equation} \label{eq:matchrhor}
\rho(r) = \gamma(r)^{-1} \, R(r) \, .
\end{equation}
Given this identification, we can now write
\begin{align} \label{eq:Gren}
G(x,x') - G_{\text{Had}}(x,x') 
&= \sum_{k=-\infty}^{\infty} e^{i k \tilde{\theta}} \sum_{n=-\infty}^{\infty} \, \left[ \frac{|\Omega_{\mathcal{H}}|}{8\pi^2} \, G_{nk}(r,r) - \frac{T_{\mathbb{M}}}{2\pi \gamma(r)} \, \hat{G}^{\mathbb{M}}_{nk}(\rho(r),\rho(r)) \right] \notag \\ 
&\quad + \gamma(r)^{-1} \, \hat{G}_{\text{reg}}^{\mathbb{M}}(x,x') + \mathcal{O}(\tilde{\theta}) \, .
\end{align}

We can now use Theorem~\ref{thm:matchingpolarisation} to guarantee that the double sum in the RHS of \eqref{eq:Gren} is convergent in the coincidence limit $\tilde{\theta} \to 0$ if the parameters of the Minkowski's Green's distribution are chosen to be
\begin{equation}
\gamma(z) = N(z) \, , \qquad T_{\mathbb{M}} = \frac{\kappa_+}{2\pi} \, , \qquad \Omega_{\mathbb{M}} = N^{\theta}(z) + \Omega_{\mathcal{H}} \, .
\label{eq:matchingWADS}
\end{equation}

We have all the necessary ingredients to compute the renormalised vacuum polarisation \eqref{eq:phi2wads}. The modes sums need to be computed numerically, as done in the next section.

\section{Numerical results}
\label{sec:computation-results}

%
%
\begin{figure}[t!]
\begin{center}
\begin{tikzpicture}
\begin{axis}[
	xlabel={$z/z_{\mathcal{M}}$},
	ylabel={$\langle \Phi^2(x) \rangle$},
	xmin=0, xmax=1,
	ymin=-0.85, ymax=0,
]
\addplot+[smooth, no markers, blue3] table {data/data2.txt};
\end{axis}
\end{tikzpicture}
\end{center}
\caption[Vacuum polarisation for the scalar field in the exterior region.]{\label{fig:numerics1} Vacuum polarization for the scalar field as a function of $z/z_{\mathcal{M}}$ for $\nu = 1.2$, $r_+ = 15$, $r_- = 1$, $r_{\mathcal{M}} = 62$ and $m = 1$.}
\end{figure}
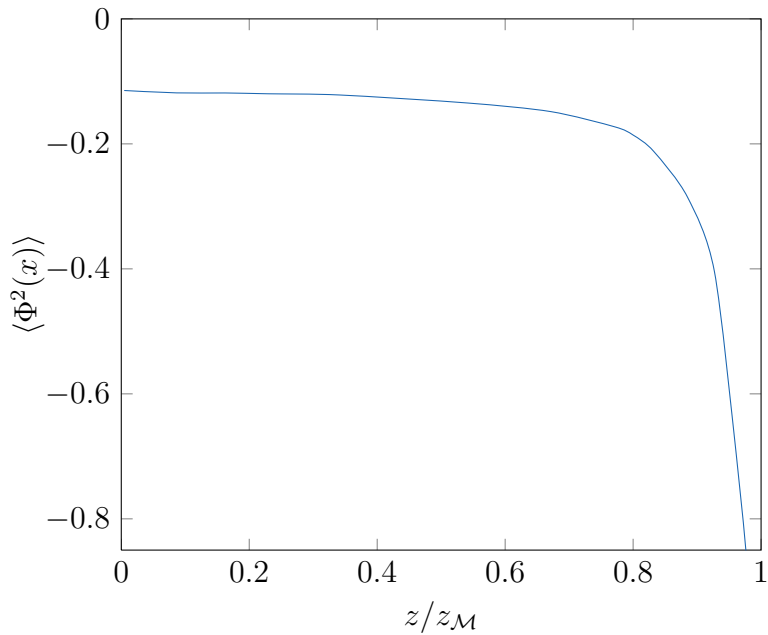

In this section, we present the numerical results for the computation of the renormalised vacuum polarisation of the scalar field in the Hartle-Hawking state in region $\widetilde{\text{I}}$ of the spacelike stretched black hole.

The numerical computation uses expressions \eqref{eq:Gren} and \eqref{eq:phi2wads} with the Minkowski parameters chosen as in \eqref{eq:matchingWADS}:
\begin{align}
\langle \Phi^2 (x) \rangle &= \sum_{k=-\infty}^{\infty} \sum_{n=-\infty}^{\infty} \, \left[ \frac{|\Omega_{\mathcal{H}}|}{8\pi^2} \, G_{nk}(z,z) - \frac{\kappa_+}{4\pi^2 N(z)} \, \hat{G}^{\mathbb{M}}_{nk}\left(\tfrac{R(z)}{N(z)},\tfrac{R(z)}{N(z)}\right)\Big|_{\Omega_{\mathbb{M}} = N^{\theta}(z)+\Omega_{\mathcal{H}}} \right] \notag \\ 
&\quad + \frac{1}{4\pi N(z)} \left[ {-m_{\mathbb{M}}} + \sum_{N \neq 0} \frac{e^{-m_{\mathbb{M}} \sqrt{\left(\frac{N}{T_{\mathbb{M}}}\right)^2 - 4 \tfrac{R^2(z)}{N^2(z)} \sinh^2 \left( \frac{\Omega_{\mathbb{M}} N}{2T_{\mathbb{M}}} \right) + i \epsilon \, \text{sgn}(\Omega_{\mathbb{M}} N)}}}{\sqrt{\left(\frac{N}{T_{\mathbb{M}}}\right)^2 - 4 \tfrac{R^2(z)}{N^2(z)} \sinh^2 \left( \frac{\Omega_{\mathbb{M}} N}{2T_{\mathbb{M}}} \right) + i \epsilon \, \text{sgn}(\Omega_{\mathbb{M}} N)}} \right] \, ,
\label{eq:vacuumpolcalculation}
\end{align}
with $\epsilon \to 0+$ indicating the choice of branch of the square root (see details in Appendix~\ref{app:Minkowski}).

As described previously, the sums in \eqref{eq:vacuumpolcalculation} are convergent. For the numerical evaluation of the sums, cutoffs are imposed appropriately. Note that the parameter $m_{\mathbb{M}}^2$ is not fixed and it is chosen in such a way to improve the numerical convergence of the double sum over $k$ and $n$.

The numerical results for selected values of the parameters are presented in Fig.~\ref{fig:numerics1}. In the plot, $\langle \Phi^2(x) \rangle$ is shown as a function of the normalized radial coordinate $z/z_{\mathcal{M}}$, where $z = (r-r_+)/(r-r_-)$. The plot is very similar to the one obtained in Ref.~\cite{Duffy:2002ss} for a scalar field in the (3+1)-dimensional Minkowski spacetime surrounded by a mirror with Dirichlet boundary conditions (note that ``rotating Minkowski spacetime'' is related to ``static Minkowski spacetime'' by a coordinate transformation, hence the results for $\langle \Phi^2(x) \rangle$ are the same for both cases).

Furthermore, note that $\langle \Phi^2(x) \rangle$ gets arbitrarily large and negative as the mirror is approached. This is to be expected, as we imposed that the Green's function $G(x,x')$ must vanish when $x$ is at the boundary, even when $x' \to x$, whereas the subtraction term still diverges when $x' \to x$ (see Section~4.3 of \cite{birrell1984quantum} for more details). 

We reemphasize that the result shown in Fig.~\ref{fig:numerics1} is the full renormalized vacuum polarization in the Hartle-Hawking state. To find the renormalized vacuum polarization in other Hadamard states of interest, such as the Boulware vacuum state, it would suffice to use the Hartle-Hawking state as a reference and just to calculate the difference, which is finite without further renormalization. For comparison, we note that in Kerr with a mirror the difference of the vacuum polarization in the Boulware and Hartle-Hawking states was found in \cite{Duffy:2005mz}, while the renormalized vacuum polarization in the individual states is still unknown.

\pagestyle{myheadings}

\newpage
\phantomsection
\addcontentsline{toc}{chapter}{Conclusions}
\addtocontents{toc}{\protect\vspace{-9pt}}

\chapter*{Conclusions}
\label{chap:conclusions}
\markboth{CONCLUSIONS}{CONCLUSIONS}

In this thesis we have developed a method to compute a class of renormalised local observables which includes the vacuum polarisation for a quantised matter field, in a given quantum state, on a rotating black hole spacetime. The rotating black hole is surrounded by a Dirichlet mirror, if necessary, such that the resulting exterior region possesses a timelike Killing vector field and on which a regular, isometry-invariant state for the matter field can be defined as a result. For simplicity, we have focused on the case of a massive scalar field on a (2+1)-dimensional rotating black hole, but the method can be straightforwardly extended to other types of fields and higher-dimensional rotating black holes.

The main results of this thesis were presented in Chapter~\ref{chap:localobservables}. Here, we have described the steps involved to explicitly renormalise and compute a local observable which is non-linear in the field operators, but which does not involve covariant derivatives of the field operators. We implement the renormalisation at the level of the Feynman propagator evaluated for the regular, isometry-invariant state, from which we subtract the singular, purely geometric part and after which the coincidence limit can be taken. For instance, for a scalar field $\Phi$, we have seen how the renormalised vacuum polarisation can be obtained by a careful use of the formula
\begin{equation*}
\langle \Phi^2(x) \rangle = -i \, \lim_{x' \to x} \left[ G^{\rm F}(x,x') - G_{\rm Had}(x,x') \right] \, ,
\end{equation*}
where $G^{\rm F}$ is the Feynman propagator evaluated for that quantum state and $G_{\rm Had}$ is its Hadamard singular part. This formula summarises the three main steps necessary to perform the computation:
\begin{enumerate}[label={(\roman*)}]
\item Expressing $G^{\rm F}(x,x')$ as a sum over mode solutions of the differential equation satisfied by $G^{\rm F}$. It is advantageous to consider the complex Riemannian section of the exterior region of the rotating black hole, on which the Green's distribution associated with the field equation is unique and can be obtained using standard techniques of the theory of Green's functions. We then analytically continue the result back to the original spacetime, where $G^{\rm F}(x,x')$ is written as a discrete sum over mode solutions. This step was described in Section~\ref{sec:quasi-euclidean-method}.
\item Expressing $G_{\rm Had}(x,x')$, which is known in closed form for any spacetime dimension, as a sum over mode solutions, so that the short-distance divergences of $G^{\rm F}(x,x')$ may be subtracted term by term. As we have seen in Section~\ref{sec:renormalisation-procedure}, we have done this by writing $G_{\rm Had}(x,x')$ as sum over mode solutions on the complex Riemannian section of Minkowski, for which the Green's distribution is known both in closed form and as a mode sum. We have then succeeded in expressing $[ G^{\rm F}(x,x') - G_{\rm Had}(x,x') ]$ as a mode sum (plus a regular term), which was made convergent in the coincidence limit by a natural choice of the parameters of the Minkowski Green's distribution.
\item Having guaranteed the convergence of the mode sum in $[ G^{\rm F}(x,x') - G_{\rm Had}(x,x') ]$ when $x' \to x$, we can safely take the coincidence limit and obtain the renormalised vacuum polarisation.
\end{enumerate}

A few remarks are in order. First, as we have emphasised in this thesis, in step (ii) above, the procedure involved in guaranteeing the convergence of the mode sum in $[ G^{\rm F}(x,x') - G_{\rm Had}(x,x') ]$ in the coincidence limit, where $G_{\rm Had}(x,x')$ is expressed as a sum over mode solutions on the complex Riemannian section of Minkowski, does \emph{not} require the knowledge of the mode solutions of the field equation in closed form, but only the first terms of the asymptotic expansion for large values of the quantum numbers. These were obtained in Section~\ref{sec:renormalisation-procedure} and Appendix~\ref{app:WKBexpansions}. This allows the method to be extended to the Kerr black hole and other higher-dimensional rotating black holes for which the mode solutions have to be constructed fully numerically. Hence, the implementation of our method for Kerr would seem feasible in principle, and it should prove interesting to attempt the implementation in practice.

Second, as seen in Section~\ref{sec:stress-energy-tensor}, this method is not suitable to renormalise local observables which involve covariant derivatives of the field operators, such as the expectation value of the stress-energy tensor, $\langle T_{ab}(x) \rangle$. For observables of this type, we were not able to subtract the short-distance divergences by expressing the singular terms as sums over mode solutions, or derivatives of mode solutions, on the complex Riemannian section of Minkowski. This is due to the fact that the shift function of the metric of Minkowski written in some rotating coordinate system is a constant in spacetime, whereas the shift function of the metric of a rotating black hole is a function of the radial coordinate in some coordinate system.

For the specific case of the (2+1)-dimensional warped AdS${}_3$ black hole considered in Part II of the thesis, one possibility is to consider the rotating BTZ black hole as a reference background, instead of Minkowski, since it is possible to explicitly compute the renormalised expectation value of the stress-energy tensor by using the fact that the rotating BTZ corresponds to AdS${}_3$ with discrete identifications. The calculation of the renormalised expectation value of the stress-energy tensor for the rotating BTZ was done in \cite{Steif:1993zv}. We hope to look into this case in the future.

For other rotating black holes, a more general method is required. Our method requires the knowledge of the Feynman propagator in both closed form and as a mode sum for a reference spacetime and the only available examples are Minkowski, AdS and dS. These are sufficient for the renormalisation of local observables on static black hole spacetimes, for which there are coordinate systems such that the shift function vanishes, but not for stationary, but non-static, black hole spacetimes. This remains as an open question, one among several important open questions concerning classical and quantum aspects of rotating black holes, especially Kerr, such as its non-linear stability as a solution of the Einstein equations.

\part{Appendices}
\setcounter{secnumdepth}{1}

\pagestyle{headings}

\appendix

\chapter{Complex Riemannian section of the Minkowski spacetime}
\chaptermark{Complex Riemannian section of Minkowski}
\label{app:Minkowski}

Consider (2+1)-dimensional rotating Minkowski spacetime. Choosing rotating, spherical coordinates $(\tilde{t},\rho,\tilde{\theta})$, its metric is
\begin{equation}
\dd s^2 = -\dd \tilde{t}^2 + \dd \rho^2 + \rho^2 \big( \dd \tilde{\theta} + \Omega_{\mathbb{M}} \, \dd \tilde{t} \big)^2 \, ,
\end{equation}
with $(\tilde{t},\rho,\tilde{\theta}) \sim (\tilde{t},\rho,\tilde{\theta}+2\pi)$ and $\Omega_{\mathbb{M}} \in \mathbb{R}$. In the complex Riemannian section, the metric is given by
\begin{equation} \label{eq:Minkmetric}
\dd s^2 = \dd \tau^2 + \dd \rho^2 + \rho^2 \big(\dd \tilde{\theta} - i \Omega_{\mathbb{M}} \, \dd \tau \big)^2 \, ,
\end{equation}
with $t = - i \tau$.

Note that in the real Lorentzian section, for $\Omega_{\mathbb{M}} \neq 0$, the Killing vector field $\chi = \partial_{\tilde{t}}$ becomes spacelike when $\rho > |\Omega_{\mathbb{M}}|^{-1}$. We restrict our attention to the part of the spacetime where $\rho < \rho_{\mathcal{M}}$, such that at $\rho = \rho_{\mathcal{M}} < |\Omega_{\mathbb{M}}|^{-1}$ there is a timelike boundary at which Dirichlet boundary conditions are imposed.

Moreover, we will require that
\begin{equation}
(\tau, \rho, \tilde{\theta}) \sim (\tau + T_{\mathbb{M}}^{-1}, \rho, \tilde{\theta}) \, ,
\end{equation}
where $T_{\mathbb{M}} > 0$ is to be interpreted as the temperature.

Consider the Klein-Gordon equation for a real scalar field of mass $m_{\mathbb{M}}$,
\begin{equation}
\left( \nabla^2 - m_{\mathbb{M}}^2 \right) \Phi(\tau,\rho,\tilde{\theta}) = 0 \, ,
\label{eq:KGeq}
\end{equation}
which in this coordinate system is given by
\begin{equation}
\left[ \frac{\partial^2}{\partial\tau^2} + \frac{1}{\rho} \frac{\partial}{\partial \rho} \left( \rho \frac{\partial}{\partial \rho} \right) + \frac{1-\rho^2 \Omega_{\mathbb{M}}^2}{\rho^2} \frac{\partial^2}{\partial\tilde{\theta}^2} + 2i \Omega_{\mathbb{M}} \partial_{\tau} \partial_{\tilde{\theta}} - m_{\mathbb{M}}^2 \right] \Phi (\tau,\rho,\tilde{\theta}) = 0 \, .
\end{equation}
Using the ansatz $\Phi(\tau,\rho,\theta) = e^{i \tilde{\omega} \tau + i k \tilde{\theta}} \phi(\rho)$ one gets
\begin{equation}
\frac{\dd^2}{\dd \rho^2} \phi(\rho) + \frac{1}{\rho} \frac{\dd}{\dd \rho} \phi(\rho) - \left( (\tilde{\omega}+ik\Omega_{\mathbb{M}})^2 + m_{\mathbb{M}}^2 + \frac{k^2}{\rho^2} \right) \phi(\rho) = 0 \, .
\label{eq:fieldeq}
\end{equation}
Two independent solutions are
\begin{equation}
\phi^1_{\tilde{\omega} k}(\rho) = I_k \! \left(\!  \sqrt{(\tilde{\omega} \! +ik\Omega_{\mathbb{M}})^2+m_{\mathbb{M}}^2 \! } \, \rho \! \right) \! , \;
\phi^2_{\tilde{\omega} k}(\rho) = K_k \! \left(\! \sqrt{(\tilde{\omega} \! +ik\Omega_{\mathbb{M}} \! )^2+m_{\mathbb{M}}^2} \, \rho \! \right) \! .
\end{equation}
where $I_k$ and $K_k$ are the modified Bessel functions and the principal branch of the square root is understood.

The Green's distribution $G^{\mathbb{M}}(x,x')$ associated with \eqref{eq:KGeq} satisfies the equation
\begin{equation}
\left( \nabla^2 - m^2 \right) G^{\mathbb{M}}(x,x') = - \frac{\delta^3(x,x')}{\sqrt{g(x)}} = - \frac{1}{\rho} \delta(\tau-\tau') \delta(\rho-\rho') \delta(\tilde{\theta}-\tilde{\theta}') \, .
\end{equation}

Given the periodicities of $\tau$ and $\tilde{\theta}$, $\tilde{\omega} = 2\pi T_{\mathbb{M}} n$, with $n \in \mathbb{Z}$, and $k \in \mathbb{Z}$. Thus,
\begin{align}
\delta (\tilde{\theta} - \tilde{\theta}') &= \frac{1}{2\pi} \sum_{k=-\infty}^{\infty} e^{i k (\tilde{\theta}-\tilde{\theta}')} \, , \\
\delta(\tau-\tau') &= T_{\mathbb{M}} \sum_{n=-\infty}^{\infty} e^{i n 2\pi T_{\mathbb{M}} (\tau - \tau')} \, .
\end{align}

If one now expands the Green's distribution $G^{\mathbb{M}}(x,x')$ as
\begin{equation}
G^{\mathbb{M}}(x,x') = \frac{T_{\mathbb{M}}}{2\pi} \sum_{n=-\infty}^{\infty} e^{i 2\pi T_{\mathbb{M}} n (\tau - \tau')} \sum_{k=-\infty}^{\infty} e^{i k (\tilde{\theta}-\tilde{\theta}')} \, G^{\mathbb{M}}_{n k} (\rho,\rho') \, ,
\label{eq:GMinkunren}
\end{equation}
then $G^{\mathbb{M}}_{n k} (\rho,\rho')$ satisfies
\begin{multline}
\frac{\dd^2}{\dd \rho^2} G^{\mathbb{M}}_{n k} (\rho,\rho') + \frac{1}{\rho} \frac{\dd}{\dd \rho} G^{\mathbb{M}}_{n k} (\rho,\rho') - \left( (2\pi T_{\mathbb{M}} n+ik\Omega_{\mathbb{M}})^2 + m_{\mathbb{M}}^2 + \frac{k^2}{\rho^2} \right) G^{\mathbb{M}}_{n k} (\rho,\rho') \\ 
= - \frac{\delta(\rho-\rho')}{\rho} \, .
\label{eq:Greenfunctioneq}
\end{multline}

Consider the homogeneous equation associated with \eqref{eq:Greenfunctioneq} and let $p^{\mathbb{M}}_{n k}(\rho)$ be the regular solution near $\rho=0$ and $q^{\mathbb{M}}_{n k}(\rho)$ be the Dirichlet solution near $\rho=\rho_{\mathcal{M}}$. Then, the unique solution to the inhomogeneous equation is
\begin{equation}
G^{\mathbb{M}}_{n k} (\rho,\rho') = C^{\mathbb{M}}_{n k} \, p^{\mathbb{M}}_{n k}(\rho_<) \, q^{\mathbb{M}}_{n k}(\rho_>) \, ,
\end{equation}
where $\rho_< := \min \{ \rho, \rho' \}$, $\rho_> := \max \{ \rho , \rho' \}$ and $C^{\mathbb{M}}_{n k}$ is a normalization constant which is determined from the Wronskian relation
\begin{equation}
C^{\mathbb{M}}_{n k} \left( p^{\mathbb{M}}_{n k} \frac{\dd q^{\mathbb{M}}_{n k}}{\dd \rho} - q^{\mathbb{M}}_{n k} \frac{\dd p^{\mathbb{M}}_{n k}}{\dd \rho} \right) = - \frac{1}{\rho} \, .
\label{eq:Comegak}
\end{equation}

Comparing \eqref{eq:fieldeq} and \eqref{eq:Greenfunctioneq} one concludes that the solutions to the homogeneous equation corresponding to \eqref{eq:Greenfunctioneq} are
\begin{equation}
p^{\mathbb{M}}_{n k}(\rho) = \phi^1_{n k}(\rho) \, , \qquad 
q^{\mathbb{M}}_{n k}(\rho) = \phi^2_{n k}(\rho) - \frac{\phi^2_{n k}(\rho_{\mathcal{M}})}{\phi^1_{n k}(\rho_{\mathcal{M}})} \phi^1_{n k}(\rho) \, ,
\end{equation}
where $\phi^i_{nk}(\rho) := \phi^i_{\tilde{\omega}k}(\rho)|_{\tilde{\omega} = 2 \pi T_{\mathbb{M}} n}$.
Moreover, Eq.~\eqref{eq:Comegak} leads to $C^{\mathbb{M}}_{n k} = 1$, thus,
\begin{equation}
G^{\mathbb{M}}_{n k} (\rho,\rho') = \phi^1_{n k}(\rho_<)  \left[ \phi^2_{n k}(\rho_>) - \frac{\phi^2_{n k}(\rho_{\mathcal{M}})}{\phi^1_{n k}(\rho_{\mathcal{M}})} \phi^1_{n k}(\rho_>) \right] \, .
\end{equation}

The Hadamard singular part $G^{\mathbb{M}}_{\text{Had}}$ of the Green's distribution is given in closed form by \eqref{eq:Hadamardsingpart}. We also want to express the Hadamard singular part of this Green's distribution as a mode sum.

We can write the Green's distribution $G^{\mathbb{M}}(x,x')$ \eqref{eq:GMinkunren} as
\begin{equation} \label{eq:GHadMinkappendix}
G^{\mathbb{M}}(x,x') = G^{\mathbb{M}}_{\text{Had}}(x,x') + G^{\mathbb{M}}_{\text{reg}}(x,x') \, ,
\end{equation}
where $G^{\mathbb{M}}_{\text{reg}}(x,x')$ is finite when $x' \to x$. As $G^{\mathbb{M}}_{\text{Had}}$ has no mirror dependence, it is convenient to express it as
\begin{equation}
G^{\mathbb{M}}_{\text{Had}}(x,x') = \frac{T_{\mathbb{M}}}{2\pi} \left( \sum_{k=-\infty}^{\infty} e^{i k (\tilde{\theta} - \tilde{\theta}')} \sum_{n=-\infty}^{\infty} e^{i n (\tau - \tau')} \, \hat{G}_{n k} (\rho,\rho') \right) -  \hat{G}^{\mathbb{M}}_{\text{reg}}(x,x') \, ,
\label{eq:GHadMink}
\end{equation}
with
\begin{equation}
\hat{G}^{\mathbb{M}}_{n k} (\rho,\rho') := \phi^1_{n k}(\rho) \, \phi^2_{n k}(\rho') \, ,
\end{equation}
and $\hat{G}^{\mathbb{M}}_{\text{reg}}(x,x')$ finite when $x' \to x$. In this form, neither of the terms on the RHS of \eqref{eq:GHadMink} has any mirror dependence. We have written $G^{\mathbb{M}}_{\text{Had}}$ as a mode sum (plus a regular term), which can be used to subtract the divergences in the black hole Green's distribution, as detailed in Sec.~\ref{sec:renormalisation-procedure}. It remains to compute $\hat{G}^{\mathbb{M}}_{\text{reg}}(x,x')$. Since this term is finite in the coincidence limit, we only need to determine the limit of this term when $x' \to x$.

First, it will be useful to determine $G^{\mathbb{M}}_{\text{Had}}$ in closed form. Suppose that $x$ and $x'$ are \emph{angularly} separated, i.e. $\tau = \tau'$ and $\rho = \rho'$. Then, the complex Synge's world function is given by
\begin{equation}
\sigma (x,x') = \frac{1}{2} \, \rho^2(\tilde{\theta}' - \tilde{\theta})^2 + \mathcal{O}(\tilde{\theta}' - \tilde{\theta})^3 \, .
\end{equation}
The Hadamard singular part of the Green's distribution is then
\begin{equation}
G^{\mathbb{M}}_{\text{Had}}(x,x') = \frac{1}{4\pi} \frac{1}{\rho|\tilde{\theta}' - \tilde{\theta}|} + \mathcal{O}(\tilde{\theta}' - \tilde{\theta}) \, .
\end{equation}

Without loss of generality, let $x = (\tau, \rho, 0)$ and $x' = (\tau, \rho, \tilde{\theta})$, with $\tilde{\theta} > 0$, such that
\begin{equation}
G^{\mathbb{M}}_{\text{Had}}(x,x') = \frac{1}{4\pi} \frac{1}{\rho \tilde{\theta}} + \mathcal{O}(\tilde{\theta}) \, .
\label{eq:HadamardsingMink}
\end{equation}

Note that we can relate the thermal Green's distribution $G^{\mathbb{M}}(x,x')$ at temperature $T_{\mathbb{M}}$ to the Green's distribution $G^{\mathbb{M}}_0(x,x')$ of a scalar field at zero temperature using the image sum formula \eqref{eq:imagemsumG},
\begin{equation}
G^{\mathbb{M}}(\tau, \rho, \tilde{\theta}; \, \tau', \rho', \tilde{\theta}') = \sum_{N=-\infty}^{\infty} G^{\mathbb{M}}_0(\tau + \tfrac{N}{T_{\mathbb{M}}}, \rho, \tilde{\theta}; \, \tau', \rho', \tilde{\theta}') \, .
\end{equation}
The zero-temperature Green's distribution can be written as
\begin{equation}
G^{\mathbb{M}}_0(x,x') = \frac{1}{4\pi} \frac{e^{-m_{\mathbb{M}} \Delta s}}{\Delta s} + G_0^{\mathcal{M}}(x,x') \, ,
\end{equation}
where $G_0^{\mathcal{M}}(x,x')$ is the contribution which contains the mirror dependence and is finite when $x' \to x$. For Minkowski spacetime in the complex Riemannian section, $\Delta s$ is given by
\begin{equation}
\Delta s^2 = (\tau'-\tau)^2 + (\rho - \rho')^2 + 4 \rho \rho' \sin^2  \left[ \frac{1}{2} \left(\tilde{\theta}' - \tilde{\theta} - i\Omega_{\mathbb{M}} (\tau' - \tau) \right) \right] \, .
\end{equation}

In the case of angular separation, the Green's function becomes
\begin{align}
{}& G^{\mathbb{M}}(\tau, \rho, 0; \, \tau, \rho, \tilde{\theta}) \notag \\
&= \frac{1}{4\pi} \sum_{N=-\infty}^{\infty} \left[\rule{0ex}{6ex}\right. \frac{e^{-m_{\mathbb{M}} \sqrt{ \left(\frac{N}{T_{\mathbb{M}}}\right)^2 + 4 \rho^2 \sin^2 \left( \frac{\tilde{\theta}}{2} + i \frac{\Omega_{\mathbb{M}} N}{2T_{\mathbb{M}}} \right)}}}{\sqrt{ \left(\frac{N}{T_{\mathbb{M}}}\right)^2 + 4 \rho^2 \sin^2 \left( \frac{\tilde{\theta}}{2} + i \frac{\Omega_{\mathbb{M}} N}{2T_{\mathbb{M}}} \right)}} 
+ G_0^{\mathcal{M}}(\tau + \tfrac{N}{T_{\mathbb{M}}}, \rho, 0; \, \tau, \rho, \tilde{\theta}) \left.\rule{0ex}{6ex}\right] \notag \\
&= \frac{1}{4\pi} \left\{\rule{0ex}{6ex}\right. \sum_{N \neq 0} \left[\rule{0ex}{6ex}\right. \frac{e^{-m_{\mathbb{M}} \sqrt{ \left(\frac{N}{T_{\mathbb{M}}}\right)^2 + 4 \rho^2 \sin^2 \left( \frac{\tilde{\theta}}{2} + i \frac{\Omega_{\mathbb{M}} N}{2T_{\mathbb{M}}} \right)}}}{\sqrt{ \left(\frac{N}{T_{\mathbb{M}}}\right)^2 + 4 \rho^2 \sin^2 \left( \frac{\tilde{\theta}}{2} + i \frac{\Omega_{\mathbb{M}} N}{2T_{\mathbb{M}}} \right)}} 
+ G_0^{\mathcal{M}}(\tau + \tfrac{N}{T_{\mathbb{M}}}, \rho, 0; \, \tau, \rho, \tilde{\theta}) \left.\rule{0ex}{6ex}\right] \notag \\
&\qquad\qquad + \frac{e^{-2 m_{\mathbb{M}} \rho \sin (\tilde{\theta}/2)}}{2 \rho \sin (\tilde{\theta}/2)} + G_0^{\mathcal{M}}(\tau, \rho, 0; \, \tau, \rho, \tilde{\theta})   \left.\rule{0ex}{6ex}\right\}  \notag \\
&= G^{\mathbb{M}}_{\text{Had}}(x,x') + \hat{G}^{\mathbb{M}}_{\text{reg}}(x,x') + G^{\mathcal{M}}(x,x')  \, ,
\end{align}
with
\begin{align}
\hat{G}^{\mathbb{M}}_{\text{reg}}(x,x') &:= \frac{1}{4\pi} \left[ \frac{e^{-2 m_{\mathbb{M}} \rho \sin (\tilde{\theta}/2)}}{2 \rho \sin (\tilde{\theta}/2)} - \frac{1}{\rho \tilde{\theta}} + \sum_{N \neq 0} \frac{e^{-m_{\mathbb{M}} \sqrt{\left(\frac{N}{T_{\mathbb{M}}}\right)^2 + 4 \rho^2 \sin^2 \left( \frac{\tilde{\theta}}{2} + i \frac{\Omega_{\mathbb{M}} N}{2T_{\mathbb{M}}} \right)}}}{\sqrt{ \left(\frac{N}{T_{\mathbb{M}}}\right)^2 + 4 \rho^2 \sin^2 \left( \frac{\tilde{\theta}}{2} + i \frac{\Omega_{\mathbb{M}} N}{2T_{\mathbb{M}}} \right)}} \right]  \, , \label{eq:Greg} \\
G^{\mathcal{M}}(x,x') &:= \frac{1}{4\pi} \sum_{N=-\infty}^{\infty} G_0^{\mathcal{M}}(\tau + \tfrac{N}{T_{\mathbb{M}}}, \rho, 0; \, \tau, \rho, \tilde{\theta}) \, .
\end{align}

$\hat{G}^{\mathbb{M}}_{\text{reg}}(x,x')$ has a finite limit when $\tilde{\theta} \to 0$, except for isolated values of the parameters at which the the square root in \eqref{eq:Greg} vanishes. To see this, consider the expansion of the argument of the square root for small positive values of $\tilde{\theta}$:
\begin{multline}
\left(\frac{N}{T_{\mathbb{M}}}\right)^2 + 4 \rho^2 \sin^2 \left( \frac{\tilde{\theta}}{2} + i \frac{\Omega_{\mathbb{M}} N}{2T_{\mathbb{M}}} \right) \\
= \left(\frac{N}{T_{\mathbb{M}}}\right)^2 - 4 \rho^2 \left[ \sinh^2 \left( \frac{\Omega_{\mathbb{M}} N}{2T_{\mathbb{M}}} \right) - i \tilde{\theta} \sinh \left( \frac{\Omega_{\mathbb{M}} N}{2T_{\mathbb{M}}} \right) \cosh \left( \frac{\Omega_{\mathbb{M}} N}{2T_{\mathbb{M}}} \right) \right] + \mathcal{O}(\tilde{\theta})^2 \, .
\end{multline}
When $\left(\frac{N}{T_{\mathbb{M}}}\right)^2 - 4 \rho^2 \sinh^2 \left( \frac{\Omega_{\mathbb{M}} N}{2T_{\mathbb{M}}} \right) > 0$, the positive branch of the square root is to be used when $\tilde{\theta} \to 0$. Otherwise, when $\left(\frac{N}{T_{\mathbb{M}}}\right)^2 - 4 \rho^2 \sinh^2 \left( \frac{\Omega_{\mathbb{M}} N}{2T_{\mathbb{M}}} \right) < 0$, the square root when $\tilde{\theta} \to 0$ is given by
\begin{multline}
i \, \text{sgn}(\Omega_{\mathbb{M}} N) \sqrt{4 \rho^2 \sinh^2 \left( \frac{\Omega_{\mathbb{M}} N}{2T_{\mathbb{M}}} \right) - \left(\frac{N}{T_{\mathbb{M}}}\right)^2} \\
= i \, \text{sgn}(\Omega_{\mathbb{M}}) \frac{N}{T_{\mathbb{M}}} \sqrt{\frac{4 \rho^2 T_{\mathbb{M}}^2}{N^2} \sinh^2 \left( \frac{\Omega_{\mathbb{M}} N}{2T_{\mathbb{M}}} \right) - 1}  \, .
\end{multline}
Hence, one can take the limit $\tilde{\theta} \to 0$ in $\hat{G}^{\mathbb{M}}_{\text{reg}}(x,x')$ to obtain
\begin{align}
\lim_{x' \to x} \hat{G}^{\mathbb{M}}_{\text{reg}}(x,x') = \frac{1}{4\pi} \left[ {-m_{\mathbb{M}}} + \sum_{N \neq 0} \frac{e^{-m_{\mathbb{M}} \sqrt{\left(\frac{N}{T_{\mathbb{M}}}\right)^2 - 4 \rho^2 \sinh^2 \left( \frac{\Omega N}{2T_{\mathbb{M}}} \right) + i \epsilon \, \text{sgn}(\Omega_{\mathbb{M}} N)}}}{\sqrt{\left(\frac{N}{T_{\mathbb{M}}}\right)^2 - 4 \rho^2 \sinh^2 \left( \frac{\Omega_{\mathbb{M}} N}{2T_{\mathbb{M}}} \right) + i \epsilon \, \text{sgn}(\Omega_{\mathbb{M}} N)}} \right] \, ,
\end{align}
with $\epsilon \to 0+$, if $\left(\frac{N}{T_{\mathbb{M}}}\right)^2 - 4 \rho^2 \sinh^2 \left( \frac{\Omega_{\mathbb{M}} N}{2T_{\mathbb{M}}} \right) \neq 0$.


\chapter{WKB expansions}
\label{app:WKBexpansions}

In this appendix we describe the WKB method used to obtain asymptotic expansions for solutions of differential equations which can be written in a Schr\"{o}dinger-like form. A standard reference is \cite{bender1999advanced}.

Let $\phi_{n k}^1$ and $\phi_{n k}^2$ be two independent solutions of the radial equation of a field equation for which there is a radial coordinate $\xi$ such that the equation can be written in a Schr\"{o}dinger-like form
\begin{equation} \label{eq:fieldeqxi}
\frac{\dd^2 \phi_{n k}(\xi)}{\dd \xi^2} - Q_{nk}(\xi) \, \phi_{n k}(\xi) = 0 \, ,
\end{equation}
and the Wronskian relation is given by
\begin{equation} \label{eq:wronskianxi}
\phi_{nk}^1(\xi) \frac{\dd \phi_{nk}^2(\xi)}{\dd \xi} - \phi_{nk}^2(\xi) \frac{\dd \phi_{nk}^1(\xi)}{\dd \xi} = \frac{1}{C_{n k}} \, ,
\end{equation}
where $C_{n k}$ is a constant, $Q_{nk}(\xi) := \chi_{nk}^2 (\xi) + \eta^2 (\xi)$ and $\chi_{nk}^2 (\xi)$ contains all the $n$ and $k$ dependence and is large whenever $\lambda^2 := n^2+k^2$ is large.

We assume that $f(\xi; \lambda) := - Q_{nk}(\xi)$ has an asymptotic expansion of the form
\begin{equation}
f(\xi; \lambda) \sim \lambda^2 \sum_{j=0}^{\infty} f_j(\xi) a_j(\lambda) \, , \qquad \lambda \to +\infty \, ,
\end{equation}
where $\{ a_j (\lambda) \}_{j=0}^{\infty}$ is an asymptotic sequence such that $a_0(\lambda) = 1$. In this case, standard WKB theory guarantees that there is an asymptotic expansion for the solutions $\phi_{nk}^i(\xi)$, $i=1,2$, when $\lambda \to +\infty$, given by the so-called WKB method.

\begin{lemma} \label{lemma:WKBexpansions}
Rewrite the differential equation \eqref{eq:fieldeqxi} as
\begin{equation}
\epsilon^2 \frac{\dd^2 \phi_{n k}(\xi)}{\dd \xi^2} - Q_{nk}(\xi) \, \phi_{n k}(\xi) = 0 \, ,
\end{equation}
where $\epsilon > 0$ is an expansion parameter (which may be set to 1 at the end). The WKB expansions of $\phi_{nk}^1(\xi)$ and $\phi_{nk}^2(\xi)$ are the asymptotic expansions in $\epsilon$,
\begin{align}
\phi_{nk}^1(\xi) &= \frac{1}{\sqrt{2 C_{nk} Q_{nk}^{1/2}}} \,
\exp \Bigg\{ \frac{1}{\epsilon} \int^{\xi} dt \left[ Q_{nk}^{1/2} + \epsilon^2 \left( \frac{Q_{nk}''}{8 Q_{nk}^{3/2}} - \frac{5(Q_{nk}')^2}{32 Q_{nk}^{5/2}} \right) \right] \notag \\
&\quad + \epsilon^2 \left( - \frac{Q_{nk}''}{16 Q_{nk}^2} + \frac{5(Q_{nk}')^2}{64 Q_{nk}^3} \right) + \mathcal{O}(\epsilon^3) \Bigg\} \, , \label{eq:phi1WKB} \\
\phi_{nk}^2(\xi) &= \frac{1}{\sqrt{2 C_{nk} Q_{nk}^{1/2}}} \,
\exp \Bigg\{ {- \frac{1}{\epsilon}} \int^{\xi} dt \left[ Q_{nk}^{1/2} + \epsilon^2 \left( \frac{Q_{nk}''}{8 Q_{nk}^{3/2}} - \frac{5(Q_{nk}')^2}{32 Q_{nk}^{5/2}} \right) \right] \notag \\
&\quad + \epsilon^2 \left( - \frac{Q_{nk}''}{16 Q_{nk}^2} + \frac{5(Q_{nk}')^2}{64 Q_{nk}^3} \right) + \mathcal{O}(\epsilon^3) \Bigg\} \, .
\end{align}
\end{lemma}

\begin{proof}
See e.g.~Chapter 10 of \cite{bender1999advanced}.
\end{proof}

The WKB expansions give us the asymptotic behaviour of the solutions $\phi_{nk}^1(\xi)$ and $\phi_{nk}^2(\xi)$ for large values of $Q_{nk}(\xi)$.

We are interested in obtaining the large $\chi_{n k}$ expansion of
\begin{equation}
\mathcal{G}_{n k} (\xi) := C_{n k} \, \phi_{n k}^1(\xi) \, \phi_{n k}^2(\xi) \, .
\end{equation}

\begin{proposition} \label{prop:Gxiexpansion}
The asymptotic expansion of $\mathcal{G}_{n k}(\xi)$ for large values of $\chi_{n k}$ is
\begin{equation} \label{eq:Gxiexpansion}
\mathcal{G}_{n k} (\xi) = \frac{1}{2 \chi_{n k}} - \frac{\eta^2}{4 \chi_{n k}^3} - \frac{(\chi_{n k}^2)''}{16 \chi_{n k}^5} + \frac{5 [(\chi_{n k}^2)']^2}{64 \chi_{n k}^7} + \mathcal{O}(\chi_{n k}^{-5}) \, .
\end{equation}
\end{proposition}

\begin{remark}
Note that all of the second, third and fourth terms on the RHS of \eqref{eq:Gxiexpansion} are of order $\chi_{n k}^{-3}$.
\end{remark}

\begin{proof}
Lemma~\ref{lemma:WKBexpansions} allows us to write
\begin{align}
\mathcal{G}_{n k} (\xi) 
&= \frac{1}{2 Q_{nk}^{1/2}} \, \exp \left[ 2 \epsilon^2 \left( - \frac{Q_{nk}''}{16 Q_{nk}^2} + \frac{5(Q_{nk}')^2}{64 Q_{nk}^3} \right) + \mathcal{O}(\epsilon^3) \right] \notag \\
&= \frac{1}{2 Q_{nk}^{1/2}} + \epsilon^2 \left( - \frac{Q_{nk}''}{16 Q_{nk}^{5/2}} + \frac{5(Q_{nk}')^2}{64 Q_{nk}^{7/2}} \right) + \mathcal{O}(\epsilon^3) \, .
\label{eq:GxiexpansionQ}
\end{align}
Expanding $Q_{nk}$ for large values of $\chi_{nk}$,
\begin{equation}
\frac{1}{2 Q_{nk}^{1/2}} = \frac{1}{\left( \chi_{nk}^2 + \eta^2\right)^{1/2}}
= \frac{1}{2 \chi_{nk}} \left( 1 - \frac{\eta^2}{2 \chi_{nk}^2} + \mathcal{O}(\chi_{nk}^{-4}) \right) \, ,
\end{equation}
and setting the expansion parameter $\epsilon = 1$ gives the result.
\end{proof}

We will also be interested in the large $\chi_{n k}$ expansion of
\begin{equation}
\mathcal{G}_{n k}' (\xi) := C_{n k} \, \frac{\dd \phi_{n k}^1(\xi)}{\dd \xi} \, \phi_{n k}^2(\xi) \, .
\end{equation}

\begin{proposition} \label{prop:WKBexpansionGprime}
The asymptotic expansion of $\mathcal{G}_{n k}'(\xi)$ for large values of $\chi_{n k}$ is
\begin{equation} \label{eq:Gxiprimeexpansion}
\mathcal{G}_{n k}' (\xi) = \frac{1}{2} - \frac{(\chi_{n k}^2)'}{8 \chi_{n k}^3} + \mathcal{O}(\chi_{n k}^{-3}) \, .
\end{equation}
\end{proposition}

\begin{proof}
From \eqref{eq:phi1WKB},
\begin{align}
\frac{\dd \phi^1(\xi)}{\dd \xi} = \left[ Q_{nk}^{1/2} - \frac{Q_{nk}'}{4 Q_{nk}} + \frac{Q_{nk}''}{8 Q_{nk}^{3/2}} - \frac{5(Q_{nk}')^2}{32 Q_{nk}^{5/2}} + \mathcal{O}(Q_{nk}^{-1}) \right] \phi^1 \, .
\end{align}
Hence, using \eqref{eq:GxiexpansionQ},
\begin{align}
\mathcal{G}_{n k}' (\xi) &= C \phi^1 \phi^2 \left[ Q_{nk}^{1/2} - \frac{Q_{nk}'}{4 Q_{nk}} + \frac{Q_{nk}''}{8 Q_{nk}^{3/2}} - \frac{5(Q_{nk}')^2}{32 Q_{nk}^{5/2}} + \mathcal{O}(Q_{nk}^{-1}) \right] \notag \\
&= \left[ \frac{1}{2 Q_{nk}^{1/2}} + \left( - \frac{Q_{nk}''}{16 Q_{nk}^{5/2}} + \frac{5(Q_{nk}')^2}{64 Q_{nk}^{7/2}} \right) + \mathcal{O}(Q_{nk}^{-3/2}) \right] \notag \\
&\quad \times \left[ Q_{nk}^{1/2} - \frac{Q_{nk}'}{4 Q_{nk}} + \frac{Q_{nk}''}{8 Q_{nk}^{3/2}} - \frac{5(Q_{nk}')^2}{32 Q_{nk}^{5/2}} + \mathcal{O}(Q_{nk}^{-1}) \right] \notag \\
&= \frac{1}{2} - \frac{Q_{nk}'}{8Q_{nk}^{3/2}} + \mathcal{O}(Q_{nk}^{-3/2}) \, .
\end{align}
Expanding for large $\chi_{n k}$ gives the result.
\end{proof}


\chapter{Hypergeometric functions}
\label{app:hypergeometric}

In this appendix, we give a very brief overview of the hypergeometric differential equation, the hypergeometric function and a few of its properties. For a more complete overview, see e.g.~\cite{erdelyi1953,olver2010nist,jeffrey2007table}.

Consider the 2nd-order ordinary differential equation,
\begin{equation} \label{eq:2ndorderODE}
\frac{\dd^2 u}{\dd z^2} + P(z) \frac{\dd u}{\dd z} + Q(z) u = 0 \, .
\end{equation}
Recall that
\begin{enumerate}[label={(\roman*)}]
\item if $P(z)$ and $Q(z)$ remain finite at $z = z_0$, then $z_0$ is called an \emph{ordinary point};
\item if either $P(z)$ or $Q(z)$ diverges as $z \to z_0$, then $z_0$ is called a \emph{singular point};
\item if either $P(z)$ or $Q(z)$ diverges as $z \to z_0$, but $(z-z_0) P(z)$ and $(z-z_0)^2 Q(z)$ remain finite at $z = z_0$, then $z_0$ is called a \emph{regular singular point}.
\end{enumerate}

If \eqref{eq:2ndorderODE} has at most three singular points we may assume that these are $0, \, 1, \, \infty$. If these singular points are also regular, then \eqref{eq:2ndorderODE} can be reduced to the form
\begin{equation} \label{eq:hypergeomdiffeq}
z(1-z) \frac{\dd^2 u}{\dd z^2} + \left[ c-(a+b+1)z \right] \frac{\dd u}{\dd z} -ab u = 0 \, ,
\end{equation}
where $a, \, b, \, c \in \mathbb{C}$ are independent of $z$. This is the \emph{hypergeometric differential equation}.

If $c \not\in \mathbb{Z}^-_0$, then one solution which is regular at $z=0$ is the hypergeometric function.

\begin{definition}
The \emph{hypergeometric function} $F(a,b;c;z) := {}_2 F_1(a,b;c;z)$ is given by the series
\begin{equation}
F(a,b;c;z) := \frac{\Gamma(c)}{\Gamma(a)\Gamma(b)} \sum_{k=0}^{\infty} \frac{\Gamma(a+k)\Gamma(b+k)}{\Gamma(c+k)} \frac{z^k}{k!} \, ,
\end{equation}
when $|z| < 1$ and elsewhere by analytical continuation.
\end{definition}

The hypergeometric function $F(a,b;c;z)$ is not defined for $c \in \mathbb{Z}^-_0$ and the principal branch is the branch $|{\arg (1-z)}| \leq \pi$.

We now list the linearly independent solutions of the hypergeometric differential equation \eqref{eq:hypergeomdiffeq} when none of the numbers $a, \, b, \, c-a, \, c-b$ is an integer (for other cases, see the references listed above).
\begin{enumerate}[label={(\roman*)}]
\item If $c \not\in \mathbb{Z}$, then two independent solutions are
\begin{subequations}
\begin{align}
u_1(z) &= F(a,b;c;z) \, , \\
u_2(z) &= z^{1-c} \, F(a+1-c,b+1-c;2-c;z) \, .
\end{align}
\end{subequations}
\item If $c \in \mathbb{Z}$, then two independent solutions are
\begin{subequations}
\begin{align}
u_1(z) &= \begin{cases} 
F(a,b;c;z) \, , & c > 0 \, , \\
z^{1-c} \, F(a+1-c,b+1-c;2-c;z) \, , & c \leq 0 \, ,
\end{cases} \\
u_2(z) &= \begin{cases} 
F(a,b;a+b+1-c;1-z) \, , & a+b+1-c \not\in \mathbb{Z}_0^- \, , \\
(1-z)^{c-a-b} \, F(c-a,c-b;1+c-a-b;1-z) \, , & a+b+1-c \in \mathbb{Z}_0^- \, .
\end{cases}
\end{align}
\end{subequations}
\end{enumerate}

A few important properties of the hypergeometric function $F(a,b;c;z)$ which are necessary in the text are the following.
\begin{enumerate}
\item At $z=1$,
\begin{equation}
F(a,b;c;1) = \frac{\Gamma(c) \Gamma(c-a-b)}{\Gamma(c-a) \Gamma(c-b)} \, ,  \qquad \ReC [c-a-b] > 0 \, .
\end{equation}
\item One of the transformation formulas is
\begin{align} \label{eq:transformationformula}
F(a,b;c;z) &= \frac{\Gamma(c) \Gamma(c-a-b)}{\Gamma(c-a) \Gamma(c-b)} \, F(a,b;a+b+1-c;1-z) + (1-z)^{c-a-b} \notag \\
&\quad \times \frac{\Gamma(c) \Gamma(a+b-c)}{\Gamma(a) \Gamma(b)} \, F(c-a,c-b;1+c-a-b;1-z) \, .
\end{align}
\end{enumerate}
%


\chapter{Classical black hole superradiance}
\label{app:superradiance}

In this appendix, we give a very brief overview of the classical superradiance phenomenon on stationary black hole spacetimes. For more details, see e.g.~Ref.~\cite{Brito:2015oca}. 

We assume that the background black hole spacetime is asymptotically flat and, moreover, stationary and axisymmetric. Consider a classical matter field perturbation which may be expressed in terms of a single master variable $\Psi$ which obeys a Schr\"{o}dinger-type equation of the form
\begin{equation} \label{eq:mastereq} 
  \frac{\dd^2 \Psi}{\dd r_*^2}+V \Psi=0\, ,
\end{equation}
where $V$ is the effective potential and the tortoise coordinate $r_*$ maps a radial coordinate $r \in (r_+,\infty)$ to $(-\infty,\infty)$, where $r_+$ is the horizon radius. Given the symmetries of the spacetime, we consider a mode solution with frequency $\omega$ and angular momentum number $k$ of the form $e^{-i\omega t+ik\theta}$. If we assume that the effective potential $V$ is constant at the horizon and at infinity, then the mode solution has the following asymptotic behaviour 
\begin{equation}
 \Psi \sim \begin{cases}
A \, e^{i \omega_{\mathcal{H}} r_*} + B \, e^{-i \omega_{\mathcal{H}} r_*} \, , & r \to r_+ \, , \\
C \, e^{i \omega_{\infty} r_*}+ D \, e^{-i \omega_{\infty} r_*} \, , & r \to \infty \, .
\end{cases}
\label{eq:masterasymptotic}
\end{equation}
where $\omega_{\mathcal{H}}^2 := V(r_+)$ and $\omega_{\infty}^2 := \lim_{r\to \infty} V(r)$.

These boundary conditions correspond to an incident wave of amplitude $D$ from infinity, a reflected wave of amplitude $C$, a transmitted wave of amplitude $B$ at the horizon and an outgoing wave of amplitude $A$ from the horizon. Even though we do not expect outgoing flux from the horizon at a classical level, a term of this form can be useful to define bases of mode solutions, as in Section~\ref{sec:stability-basismodes}.

If we now assume that the effective potential $V$ is real-valued, then \eqref{eq:mastereq} is invariant under the transformations $t\to -t$ and $\omega\to-\omega$ and, hence, $\overline{\Psi}$ is also a solution of \eqref{eq:mastereq} and is linearly independent of $\Psi$. Therefore, the Wronskian
\begin{equation}
W(\Psi, \Psi) := \Psi \frac{\dd \overline{\Psi}}{\dd r_*} - \overline{\Psi} \frac{\dd \Psi}{\dd r_*}
\end{equation}
is independent of $r_*$. It thus follows that the Wronskian evaluated at the horizon, $W(\Psi, \Psi) = 2i \omega_{\mathcal{H}} \left(|A|^2-|B|^2\right)$, must equal the one evaluated at infinity, $W(\Psi, \Psi) = 2i \omega_{\infty} (|C|^2-|D|^2)$, so that 
\begin{equation} \label{eq:reflectivity}
 |C|^2 - |D|^2 = \frac{\omega_{\mathcal{H}}}{\omega_{\infty}}\left(|A|^2-|B|^2\right) \, .
\end{equation}

If there is no flux coming from the horizon, $A=0$, then $|C|^2<|D|^2$ when $\omega_{\mathcal{H}}/\omega_{\infty}>0$, i.e.~the amplitude of the reflected wave is smaller than of the incident wave. However, for 
$\omega_{\mathcal{H}}/\omega_{\infty}<0$, the wave is amplified, $|C|^2>|D|^2$. This is the phenomenon of \emph{superradiance}. If $A \neq 0$ but $B=0$, then superradiance occurs if $|D|^2 > |C|^2$, i.e.~for $\omega_{\mathcal{H}}/\omega_{\infty}<0$.

\addcontentsline{toc}{chapter}{Bibliography}
\addtocontents{toc}{\protect\vspace{-9pt}}
\bibliography{bibliography}

\end{document}